\documentclass[11pt]{article}

\usepackage{fullpage}
\usepackage{float}

\usepackage{mathpazo}

\usepackage{parskip}
\usepackage[normalem]{ulem} 

\def\focsmode{0}

\def\stocmode{0}

\def\arxivmode{1}
\def\fastmode{0}

\def\showauthornotes{2}

\def\showkeys{0}
\def\showdraftbox{0}
\def\showcolorlinks{1}
\def\usemicrotype{1}
\def\showfixme{1}

\def\blindmode{0}

\ifnum\focsmode=1
\def\arxivmode{0}
\def\showdraftbox{0}
\def\showauthornotes{0}
\def\usemicrotype{1}
\def\showfixme{0}
\def\blindmode{1}
\fi

\ifnum\arxivmode=1
\def\arxivmode{0}
\def\showdraftbox{0}
\def\showauthornotes{0}
\def\usemicrotype{1}
\def\showfixme{0}
\def\blindmode{0}
\fi

\ifnum\arxivmode=0  \else  \fi

\ifnum\fastmode=0
\usepackage{etex}
\fi

\usepackage{xspace,enumerate}

\usepackage[dvipsnames]{xcolor}

\usepackage{mathtools, bbm}

\usepackage{apxproof}
\usepackage{ifxetex}
\usepackage{amsmath,amsthm,amssymb}
\usepackage{algorithm}
\usepackage{algpseudocode}

\newtheorem{theorem}{Theorem}[section]
\newtheorem*{theorem*}{Theorem}

\newtheorem{proposition}[theorem]{Proposition}
\newtheorem*{proposition*}{Proposition}
\newtheorem{lemma}[theorem]{Lemma}
\newtheorem*{lemma*}{Lemma}
\newtheorem{corollary}[theorem]{Corollary}
\newtheorem*{conjecture*}{Conjecture}
\newtheorem{fact}[theorem]{Fact}
\newtheorem*{fact*}{Fact}

\newtheorem*{exercise*}{Exercise}

\newtheorem*{hypothesis*}{Hypothesis}
\newtheorem{conjecture}[theorem]{Conjecture}

\newtheorem{definition}[theorem]{Definition}

\newtheorem{example}[theorem]{Example}

\newtheorem{openquestion}[theorem]{Open Question}

\newtheorem{assumption}[theorem]{Assumption}
\newtheorem{exercise-easy}[theorem]{Exercise}
\newtheorem{exercise-med}[theorem]{Exercise}
\newtheorem{exercise-hard}[theorem]{Exercise$^\star$}

\newtheorem*{claim*}{Claim}

\newtheorem*{remark*}{Remark}

\newtheorem*{observation*}{Observation}

\usepackage{CJKutf8}

\ifnum\arxivmode=0

\usepackage{dsfont}
\usepackage{kpfonts}
\usepackage[T1]{fontenc}
\usepackage[full]{textcomp}

\DeclareFontFamily{OT1}{rsfs}{}
\DeclareFontShape{OT1}{rsfs}{n}{it}{<-> rsfs10}{}
\DeclareMathAlphabet{\mathscr}{OT1}{rsfs}{n}{it}

\let\mathbb\varmathbb

\usepackage{bm}
\usepackage{csquotes}
\fi

\ifnum\arxivmode=1

\usepackage{dsfont}
\usepackage{kpfonts}
\usepackage[T1]{fontenc}
\usepackage[full]{textcomp}

\DeclareFontFamily{OT1}{rsfs}{}
\DeclareFontShape{OT1}{rsfs}{n}{it}{<-> rsfs10}{}
\DeclareMathAlphabet{\mathscr}{OT1}{rsfs}{n}{it}

\usepackage{bm}
\usepackage{csquotes}

\let\mathbb\varmathbb
\fi

\ifnum\fastmode=0
\usepackage[T1]{fontenc}
\usepackage[full]{textcomp}
\fi

\ifnum\showkeys=1
\usepackage[color]{showkeys}
\fi

\makeatletter
\@ifpackageloaded{hyperref}
{}
{\usepackage[pagebackref, linktocpage]{hyperref}}

\definecolor{bleudefrance}{rgb}{0.01, 0.1, 1.0}
\definecolor{azure}{rgb}{0.0, 0.5, 1.0}

\ifnum\showcolorlinks=1
\hypersetup{
colorlinks=true,
urlcolor=blue,
linkcolor=blue,
citecolor=OliveGreen}
\fi

\ifnum\showcolorlinks=0
\hypersetup{
pagebackref=true,
colorlinks=false,
pdfborder={0 0 0}}
\fi

\usepackage{prettyref}
\usepackage[nameinlink]{cleveref}

\newcommand{\savehyperref}[2]{\texorpdfstring{\hyperref[#1]{#2}}{#2}}

\newrefformat{eq}{\savehyperref{#1}{\textup{(\ref*{#1})}}}
\newrefformat{lem}{\savehyperref{#1}{Lemma~\ref*{#1}}}
\newrefformat{def}{\savehyperref{#1}{Definition~\ref*{#1}}}
\newrefformat{obs}{\savehyperref{#1}{Observation~\ref*{#1}}}
\newrefformat{ass}{\savehyperref{#1}{Assumption~\ref*{#1}}}
\newrefformat{thm}{\savehyperref{#1}{Theorem~\ref*{#1}}}
\newrefformat{cor}{\savehyperref{#1}{Corollary~\ref*{#1}}}
\newrefformat{cha}{\savehyperref{#1}{Chapter~\ref*{#1}}}
\newrefformat{sec}{\savehyperref{#1}{Section~\ref*{#1}}}
\newrefformat{app}{\savehyperref{#1}{Appendix~\ref*{#1}}}
\newrefformat{tab}{\savehyperref{#1}{Table~\ref*{#1}}}
\newrefformat{fig}{\savehyperref{#1}{Figure~\ref*{#1}}}
\newrefformat{hyp}{\savehyperref{#1}{Hypothesis~\ref*{#1}}}
\newrefformat{alg}{\savehyperref{#1}{Algorithm~\ref*{#1}}}
\newrefformat{rem}{\savehyperref{#1}{Remark~\ref*{#1}}}
\newrefformat{item}{\savehyperref{#1}{Item~\ref*{#1}}}
\newrefformat{step}{\savehyperref{#1}{step~\ref*{#1}}}
\newrefformat{conj}{\savehyperref{#1}{Conjecture~\ref*{#1}}}
\newrefformat{fact}{\savehyperref{#1}{Fact~\ref*{#1}}}
\newrefformat{prop}{\savehyperref{#1}{Proposition~\ref*{#1}}}
\newrefformat{prob}{\savehyperref{#1}{Problem~\ref*{#1}}}
\newrefformat{claim}{\savehyperref{#1}{Claim~\ref*{#1}}}
\newrefformat{relax}{\savehyperref{#1}{Relaxation~\ref*{#1}}}
\newrefformat{red}{\savehyperref{#1}{Reduction~\ref*{#1}}}
\newrefformat{part}{\savehyperref{#1}{Part~\ref*{#1}}}
\newrefformat{ex}{\savehyperref{#1}{Exercise~\ref*{#1}}}
\newrefformat{property}{\savehyperref{#1}{Property~\ref*{#1}}}
\newrefformat{eqs}{\savehyperref{#1}{Equations~\ref*{#1}}}
\newrefformat{cond}{\savehyperref{#1}{Condition~\ref*{#1}}}
\newrefformat{Example}{\savehyperref{#1}{Example~\ref*{#1}}}

\newcommand{\Sref}[1]{\hyperref[#1]{\S\ref*{#1}}}

\usepackage{nicefrac}

\ifnum\fastmode=0
\ifnum\usemicrotype=1
\usepackage{microtype}
\fi
\fi

\ifnum\showauthornotes=2

\newcommand{\JSS}[1]{{\color{Red}[JSS: #1]}}
\newcommand{\JS}[1]{{\color{Emerald}[JS: #1]}}
\newcommand{\mynotes}[1]{{\sffamily\small\color{teal}{#1}}\medskip}
\newcommand{\Authornote}[2]{{\sffamily\small\color{Maroon}{[#1: #2]}}\medskip}
\newcommand{\Authornotecolored}[3]{{\sffamily\small\color{#1}{[#2: #3]}}}
\newcommand{\Authorcomment}[2]{{\sffamily\small\color{gray}{[#1: #2]}}}
\newcommand{\Authorstartcomment}[1]{\sffamily\small\color{gray}[#1: }

\newcommand{\Authorfnote}[2]{\footnote{\color{red}{#1: #2}}}
\newcommand{\Authorfixme}[1]{\Authornote{#1}{\textbf{??}}}
\newcommand{\Authormarginmark}[1]{\marginpar{\textcolor{red}{\fbox{\Large #1:!}}}}
\newcommand{\myexplain}[1]{{\sffamily\small\color{red}{\noindent [Explanation:\medskip\newline \begin{quote}#1\hfill]\end{quote}}}\medskip}
\newcommand{\explain}[1]{{\sffamily\small\color{red}{#1}}\medskip}

\else

\newcommand{\JSS}[1]{\PackageError{asldfjkas}{safiawe}{asfdsf}}
\newcommand{\JS}[1]{\PackageError{asldfjkas}{safiawe}{asfdsf}}
\newcommand{\mynotes}[1]{}
\newcommand{\Authornote}[2]{}
\newcommand{\Authornotecolored}[3]{}
\newcommand{\Authorcomment}[2]{}
\newcommand{\Authorstartcomment}[1]{}

\newcommand{\Authorfnote}[2]{}
\newcommand{\Authorfixme}[1]{}
\newcommand{\Authormarginmark}[1]{}
\newcommand{\myexplain}[1]{}
\newcommand{\explain}[1]{}

\ifnum\showauthornotes=1
\renewcommand{\myexplain}[1]{{\sffamily\small\color{red}{\noindent \begin{quote}{\bf Explanation:} \medskip\newline #1\end{quote}}}\medskip}
\fi

\fi

\ifnum\showfixme=0

\fi

\usepackage{boxedminipage}

\newcommand{\card}[1]{\lvert#1\rvert}

\newcommand{\norm}[1]{\left\lVert#1\right\rVert}

\newcommand{\iprod}[1]{\left\langle#1\right\rangle}

\newcommand{\Esymb}{\mathbb{E}}
\newcommand{\Psymb}{\mathbb{P}}

\DeclareMathOperator*{\E}{\Esymb}

\DeclareMathOperator*{\ProbOp}{\Psymb}

\renewcommand{\Pr}{\ProbOp}

\newcommand{\textparen}[1]{\text{(#1)}}

\ifx\because\undefined
\newcommand{\because}[1]{\textparen{because #1}}
\else
\renewcommand{\because}[1]{\textparen{because #1}}
\fi

\newcommand{\mper}{\,.}
\newcommand{\mcom}{\,,}

\newcommand\bdot\bullet

\ifx\mathds\undefined %

\else

\fi

\DeclareMathOperator{\Tr}{Tr}

\DeclareMathOperator{\poly}{poly}

\DeclareMathOperator{\supp}{supp}

\newcommand{\Z}{\mathbb Z}
\newcommand{\N}{\mathbb N}
\newcommand{\R}{\mathbb R}

\newcommand{\cA}{\mathcal A}
\newcommand{\cB}{\mathcal B}
\newcommand{\cC}{\mathcal C}
\newcommand{\cD}{\mathcal D}
\newcommand{\cE}{\mathcal E}

\newcommand{\cG}{\mathcal G}

\newcommand{\cL}{\mathcal L}
\newcommand{\cM}{\mathcal M}
\newcommand{\cN}{\mathcal N}

\newcommand{\cP}{\mathcal P}

\newcommand{\cS}{\mathcal S}
\newcommand{\cT}{\mathcal T}

\newcommand{\cV}{\mathcal V}

\newcommand{\cX}{\mathcal X}

\renewcommand{\leq}{\leqslant}
\renewcommand{\le}{\leqslant}
\renewcommand{\geq}{\geqslant}
\renewcommand{\ge}{\geqslant}

\ifnum\showdraftbox=1
\newcommand{\draftbox}{\begin{center}
  \fbox{%
    \begin{minipage}{2in}%
      \begin{center}%
          \Large\textsc{Working Draft}\\%
        Please do not distribute%
      \end{center}%
    \end{minipage}%
  }%
\end{center}
\vspace{0.2cm}}
\else
\newcommand{\draftbox}{}
\fi

\let\epsilon=\varepsilon

\numberwithin{equation}{section}

\newcommand\MYcurrentlabel{xxx}

\newcommand{\MYstore}[2]{%
  \global\expandafter \def \csname MYMEMORY #1 \endcsname{#2}%
}

\newcommand{\MYload}[1]{%
  \csname MYMEMORY #1 \endcsname%
}

\newcommand{\MYnewlabel}[1]{%
  \renewcommand\MYcurrentlabel{#1}%
  \MYoldlabel{#1}%
}

\newcommand{\MYdummylabel}[1]{}

\newcommand{\torestate}[1]{%
  \let\MYoldlabel\label%
  \let\label\MYnewlabel%
  #1%
  \MYstore{\MYcurrentlabel}{#1}%
  \let\label\MYoldlabel%
}

\newcommand{\restatetheorem}[1]{%
  \let\MYoldlabel\label
  \let\label\MYdummylabel
  \begin{theorem*}[Restatement of \prettyref{#1}]
    \MYload{#1}
  \end{theorem*}
  \let\label\MYoldlabel
}

\newcommand{\restatelemma}[1]{%
  \let\MYoldlabel\label
  \let\label\MYdummylabel
  \begin{lemma*}[Restatement of \prettyref{#1}]
    \MYload{#1}
  \end{lemma*}
  \let\label\MYoldlabel
}

\newcommand{\restateprop}[1]{%
  \let\MYoldlabel\label
  \let\label\MYdummylabel
  \begin{proposition*}[Restatement of \prettyref{#1}]
    \MYload{#1}
  \end{proposition*}
  \let\label\MYoldlabel
}

\newcommand{\restatefact}[1]{%
  \let\MYoldlabel\label
  \let\label\MYdummylabel
  \begin{fact*}[Restatement of \prettyref{#1}]
    \MYload{#1}
  \end{fact*}
  \let\label\MYoldlabel
}

\newcommand{\restate}[1]{%
  \let\MYoldlabel\label
  \let\label\MYdummylabel
  \MYload{#1}
  \let\label\MYoldlabel
}

\newcommand{\addreferencesection}{
  \phantomsection
\ifnum\stocmode=0
  \addcontentsline{toc}{section}{References}
\else
  \addcontentsline{toc}{section}{References \hspace*{1in} --------- End of extended abstract ---------}
\fi

}

\newcommand{\eps}{\epsilon}

\renewcommand{\paragraph}[1]{\medskip\noindent{\bf #1.}}

\allowdisplaybreaks

\usepackage{paralist}

\usepackage{comment}

\DeclareMathOperator{\tr}{tr}

\newcommand{\vertiii}[1]{{\left\vert\kern-0.25ex\left\vert\kern-0.25ex\left\vert #1
          \right\vert\kern-0.25ex\right\vert\kern-0.25ex\right\vert}}

\newcommand*{\ot}{\otimes}

\DeclareMathOperator{\Id}{\mathrm{Id}}

\DeclareMathOperator*{\pE}{\widetilde{\mathbb{E}}}

\DeclareMathOperator{\fE}{\mathbb{E}}

\newcommand{\iso}{\mathrm{iso}}

\usepackage{accents}

\usepackage{stmaryrd}

\usepackage{enumitem}

\renewcommand{\supp}{\mathrm{supp}}
\DeclareMathOperator{\mat}{Mat}

\Crefname{Problem}{Problem}{Problems}
\Crefname{Lemma}{Lemma}{Lemmas}
\Crefname{Proposition}{Proposition}{Propositions}
\Crefname{Open Question}{Open Question}{Open Questions}
\Crefname{Definition}{Definition}{Definitions}
\Crefname{Corollary}{Corollary}{Corollaries}
\Crefname{Fact}{Fact}{Facts}
\Crefname{Condition}{Condition}{Conditions}
\Crefname{Assumption}{Assumption}{Assumptions}
\Crefname{Conjecture}{Conjecture}{Conjectures}
\Crefname{Remark}{Remark}{Remarks}

\newcommand{\sop}{_\mathsf{op}}
\providecommand{\opnorm}[1]{\ensuremath{\left\lVert #1 \right\rVert_{\sop}}}

\renewcommand{\Z}{\mathbb{Z}}

\renewcommand{\R}{\mathbb{R}}
\renewcommand{\N}{\mathbb{N}}

\renewcommand{\bar}{\overline}
\renewcommand{\epsilon}{\varepsilon}
\renewcommand{\eps}{\varepsilon}

\renewcommand{\ot}{\otimes}

\newcommand{\sT}{\mathsf{T}}

\newcommand{\sym}{\mathrm{sym}}

\newcommand{\partialcum}[2]{{\kappa\kern-0.55em\raise0.42ex\hbox{$\scriptsize\smallsetminus$}_{#1}\left[{#2}\right]}}
\DeclareRobustCommand\proves{\mathrel{|}\joinrel\mkern-.5mu\mathrel{-}}
\newcommand{\He}{\mathsf{He}}

\newcommand{\pcum}[2]{\widetilde{\kappa}_{{#1}}\left[{#2}\right]}\newcommand{\cum}[2]{\kappa_{{#1}}\left[{#2}\right]}

\newcommand{\mon}[1]{\mathtt{mon}[{#1}]}

\usepackage[hang,flushmargin]{footmisc}

\usepackage{titling}
\usepackage{mleftright}
\mleftright

\author{
\ifnum\blindmode=1
Author names withheld
\fi
\ifnum\blindmode=0
    Juspreet Singh Sandhu\thanks{School of Engineering \& Applied Sciences, Harvard University, Cambridge, Massachusetts, USA. \newline Supported by a Simons Investigator Fellowship, NSF grant DMS-2134157, NSF STAQ award PHY-1818914, DARPA grant W911NF2010021, DARPA ONISQ program award HR001120C0068 and DOE grant DE-SC0022199. \newline Email: \textcolor{red}{jus065@g.harvard.edu}.}
    \and
    Jonathan Shi\thanks{
    Halıcıoğlu Data Science Institute, University of California San Diego, USA.
    \newline
    Supported by NSF TRIPODS. \newline Email: \textcolor{red}{joshi@ucsd.edu}.} \\\vspace{0.5em}
\fi
}

\title{Sum-of-Squares \& Gaussian Processes I: Certification
\ifnum\blindmode=0
\thanks{This project has received funding from the European Research Council (ERC) award No. 834861 (``SO-ReCoDi'').\vspace{-0.0em}}
\fi
\\\vspace{0.5em}
}

\date{\today}

\begin{document}
\maketitle
\draftbox

\thispagestyle{empty}

\begin{abstract}
\vspace*{-0.0em}
\setlength{\parindent}{0em}
\setlength{\parskip}{0.5em}

We introduce a class of distributions which may be considered as a smoothed probabilistic version of the ultrametric property that famously characterizes the Gibbs distributions of various spin glass models \cite{panchenko2013parisi}.
This class of \emph{high-entropy step} (HES) distributions is expressive enough to capture a distribution achieving near-optimal average energy on spin glass models in the so-called full Replica-Symmetry Breaking (fRSB) regime.

Simultaneously, with high probability, there are polynomial-size certificates on the average energy achievable by \emph{any} HES distribution which are tight within a constant factor. These certificates can be found in polynomial time by a semidefinite program corresponding to a sum-of-squares (SoS) hierarchy we introduce, termed the HES SoS hierarchy.
This improves over classical sum-of-squares certificates which are loose by a factor of $n^{\lfloor p/2 - 1\rfloor/2}$ \cite{bhattiprolu2017sum, hopkins2017power}.

\thispagestyle{empty}
\end{abstract}
\newpage

\addtocontents{toc}{\protect\thispagestyle{empty}}
\pagenumbering{gobble}
\setcounter{tocdepth}{2}
{
    \hypersetup{linkcolor=blue}
    \thispagestyle{empty}
    \tableofcontents
    \thispagestyle{empty}
}

\newpage
\clearpage
\pagenumbering{arabic}

\newpage

\section{Introduction}

In the Sum-of-Squares (SoS) meta-algorithm, the existence of a certain type of proof (known as a low-degree SoS proof) characterizing the solutions of a problem instance automatically gives rise to a corresponding polynomial-time algorithm for that problem~\cite{parrilo2004sum, barak2014sum}. In various average-case settings, however, no low-degree proof exists due to known lower bounds against the standard SoS hierarchy~\cite{bhattiprolu2017sum, hopkins2017power, jones2022sum}. Therefore, the central question that motivates this paper is:
\begin{center}
    \emph{What can SoS-type proof systems certify for large ensembles of average-case optimization problems?}
\end{center}

To understand this question, we propose a new SoS hierarchy to certify the value of a polynomial over points on the sphere that can be output by a parameterized family of distributions (\Cref{thm:main}). 
These distributions are termed high-entropy (HES) step processes (\Cref{def:hes}). 
Thus, the HES SoS hierarchy certifies the maximum value that can be achieved by \emph{any} algorithm whose output distribution on a fixed instance can be characterized by some HES process. 
To demonstrate the power of this hierarchy, we prove that for the spherical spin glass problem~(\Cref{def:spin-glass}) the HES SoS hierarchy provides polynomial-sized certificates with value that is only a constant factor off the true optimum (\Cref{thm:final-result-informal}). 
This improves over prior SoS certificates, which are off of the true optimum by superlinear factors at any constant level of the hierarchy \cite{bhattiprolu2017sum, hopkins2017power}. 
A celebrated result of Subag~\cite{subag2021following} provides a Hessian ascent algorithm for the spherical spin glass problem with behavior that is conjecturally optimal among all polynomial time algorithms~\cite{huang2021tight}. Indeed, we demonstrate that a randomized version of Subag's algorithm (\Cref{alg:subag-2}) is a HES process (\Cref{lem:feasibility}) and then certify that the value achieved by \emph{any} HES process is the same up to constant factors (\Cref{thm:final-result-informal}). 
Continuing further, we show that the HES SoS hierarchy diverges from standard Hessian ascent on modified ensembles of the spherical spin glass problem~(\Cref{sec:robustness}) and performs better than standard Hessian ascent, though it remains open to prove its performance is optimal. 
While a result of Huang and Sellke~\cite{huang2023algorithmic, huang2024optimization} provides a conjecturally-optimal approximate-message passing (AMP) algorithm for the models from which these modified ensembles are chosen, the striking fact is that the HES SoS hierarchy works \emph{without} being reprogrammed or changed. 
This is in stark constrast to the AMP algorithm~\cite{huang2023algorithmic, huang2024optimization} which is not in any obvious way relatable to the Hessian ascent algorithm~\cite{subag2021following}.
This leads to a conjecture that the HES SoS hierarchy is optimal among \emph{all} polynomial-time algorithms for optimizing a large family of Gaussian processes~(\Cref{conj:distributional-robustness}); these processes trivially subsume the spherical spin glass problem.

The HES SoS hierarchy is based on three key ideas.
\begin{itemize}
    \item The first key insight is to permit \emph{distributions over} solutions to be objects of optimization in their own right~\cite{ge2018simulated,moitra2020fast,doi:10.1137/S0036141096303359}.
    The upshot of this is that SoS proofs may refer to the moments of a hypothetical distribution over solutions, and prove polynomial relationships between these moments, as well as constraints on those moments.
    \item The second key insight is giving the hierarchy access to the \emph{interior} of the solution domain to parameterize such distributions via their low-degree moments. 
    This idea is inspired by a line of work initiated by Subag~\cite{subag2018free, chen2018generalized} that considers an extension of the original polynomial into the interior of the solution domain, and tracks how a smoothened version of the maximum changes as you move from the center outwards to the domain.
    \item The third key insight is to impose constraints on the moments of the distributions that maximize some notion of entropy over the near-optimizers. 
    This is inspired by the well-known fact that the Gibbs distribution is the maximum entropy distribution with a given energy.
\end{itemize}

In keeping with the insights mentioned above, the HES SoS hierarchy solves the following algorithmic task:
\begin{quote}
    \emph{Certify the value of a problem instance with low-degree SoS proofs over distributions of points that are generated by a ``high-entropy'' process.}
\end{quote}

Critically, note that this task is \emph{not} a relaxation of the original problem but a different problem that asks for the value of the instance achieved by some family of distributions having large entropy. 
Informally, HES processes are defined as a sum of sequential vector-valued steps, with each step conditionally distributed as some Gaussian with large effective dimension whose covariance varies smoothly under the values of the previous steps.
This family of distributions is sufficiently well-behaved to enable powerful SoS proofs on its properties, and yet it is also expressive enough to capture near-extrema of (for instance) spherial spin glasses on the sphere~(\Cref{thm:final-result-informal}).

In order to get a better understanding of the capabilities of this hierarchy, we briefly contrast it with three (broad) algorithmic approaches that have been used previously in the literature to study average-case optimization problems.
\begin{description}
    \item[HES SoS hierarchy vs. Hessian ascent.]
    The Hessian ascent algorithm proposed by Subag~\cite{subag2021following} follows the largest eigenvector of the Hessian of the instance from the origin outwards, eventually reaching a point on the sphere. 
    On spherical spin glasses, the HES SoS hierarchy \emph{automatically} finds a near-optimal HES distribution which is also output by a randomized variant of the prior Hessian ascent algorithm~\cite{subag2018free, subag2021following}. In doing so, it captures the Hessian ascent algorithm \emph{without simulating} it. 
    This allows the HES SoS hierarchy to perform better on modified input ensembles that are anisotropic~(see \Cref{sec:robustness} and~\Cref{conj:distributional-robustness}). 
    \item[HES SoS hierarchy vs. Local-Statistics hierarchy.] The HES SoS hierarchy is stronger than the previously proposed local-statistics hierarchy (LStH)~\cite{banks2021local, ivkov2023semidefinite}  in the sense that the structure of the LStH precludes any scope of meaningful certification on an \emph{instance} from the relevant average-case ensemble. 
    In contrast, the SoS HES hierarchy is a certification algorithm that \emph{directly} searches and optimizes over certain distributions over solutions.
    
    \item[HES SoS hierarchy vs. Low-degree polynomial method.] Previous algorithms that have successfully been applied to spin-glasses and random Max-CSPs are either some form of local algorithms~\cite{alaoui2023local, chen2023local} or low-degree polynomial (LDP) algorithms~\cite{gamarnik2020low, huang2021tight}.
    For certification problems, there are natural ways to set up the standard SoS hierarchy~\cite{parrilo2004sum}. However, as is evidenced by the certification lower bounds against the standard SoS hierarchy~\cite{bhattiprolu2017sum, hopkins2017power, jones2022sum} the relevant question becomes what can be achieved when certification is harder than optimization~\cite{subag2021following, montanari2021optimization} for such problems. In such a situation, the HES distributions provide a natural way to set up a SoS hierarchy that may be universal.  
    
    We conjecture (\Cref{conj:distributional-robustness}) that the SoS HES hierarchy subsumes these algorithms on a wide class of average-case optimization problems. The HES distributions considered by the hierarchy are not necessarily LDPs of the instance (although some parameters of the distributions are constrained to be near-LDPs of \emph{each other}): thus, the LDP method is not known to be equivalent to the SoS HES hierarchy.
\end{description}

In this paper, for technical simplicity, we study optimization problems over the sphere. 
Provided a similar HES SoS hierarchy can be instantiated on the cube, it would very likely have similar findings for random optimization problems over the hypercube~(see~\Cref{conj:hes-cube}).

A natural question, then, is the following: 
\begin{center}
    \emph{What is the advantage of the HES SoS hierarchy over prior approaches?}
\end{center}   
The answer to this question is threefold, as detailed below.
\begin{description}
    \item[Lower bounds computable in $\mathsf{poly}(n)$-time.] The HES distributions correspond to entire families of stochastic processes and, naively, it is not even possible to enumerate them in finite time. Even then, \Cref{thm:main} and \Cref{thm:final-result-informal} demonstrate that the HES SoS hierarchy can find efficient certificates in $\mathsf{poly}(n)$ time that give limitations on \emph{any} algorithm whose output distribution can be parameterized by some HES process. This is a non-standard and somewhat surprising use of semidefinite programs, and can be thought of as an ``algorithmic lower bound'' against HES processes on certain families of inputs.   

    \item[Computing Parisi-like constants in low-degree SoS.] Prior certification lower bounds for the standard SoS hierarchy~\cite{bhattiprolu2017sum, hopkins2017power, jones2022sum} demonstrate it cannot even certify the correct \emph{scaling} of the maxima of random low-degree polynomials until the underlying proofs have near-linear degree. This reflects the reality that the existing proofs of the Parisi-like formula~\cite{crisanti1992sphericalp} that captures the maximum value are extremely technical and deep, making heavy use of concentration of measure and the analytic toolkit used to analyze Gaussian processes~\cite{talagrand2006free, subag2018free}. However, we demonstrate that there exist low-degree SoS proofs over HES distributions that obtain the correct \emph{scaling} with a value that approximates the maximum up to an explicit constant factor~(\Cref{thm:final-result-informal}).

    \item[Ensemble agnosticism.] We demonstrate that the HES hierarchy handles some anisotropic variants of the spherical spin glass ensemble~(\Cref{sec:robustness}) where direct Hessian ascent takes suboptimal steps.
    While it is an open problem to provide low-degree HES SoS certificates that match the conjectural optimal value~\cite{huang2023algorithmic, huang2024optimization} on such anisotropic ensembles (\Cref{conj:distributional-robustness}), the striking feature is that the HES hierarchy does not require any such tweaking and works as-is without being told which ensemble the instance was drawn from. Au contraire, the standard Hessian ascent algorithm becomes provably suboptimal, as is shown in \Cref{sec:robustness}, and a completely new AMP algorithm needs to be designed for this different anisotropic ensemble~\cite{huang2023algorithmic, huang2024optimization}.
\end{description}

\Cref{conj:distributional-robustness} and \Cref{conj:hes-cube} also assert that there exist rounding schemes for the HES SoS hierarchy which achieves optimal value among all polynomial-time algorithms on these ensembles of problems. The rounding aspect of the conjecture is discussed further in the companion paper \cite{ss24rounding}, which introduces the rounding scheme for the spherical spin glass problem.

To complete the picture, it would be interesting to know what natural families of algorithms are captured by HES processes: \Cref{lem:feasibility} shows that this includes standard Hessian ascent, and a recent result~\cite{jekel2024potential} shows that a Hessian ascent algorithm modified with a potential function achieves equivalent performance with AMP~\cite{montanari2021optimization}, suggesting that the performance of AMP may also be generically captured by HES processes in certain settings. It may be possible to demonstrate a connection between sufficiently ``smooth'' algorithms---such as overlap-concentrated algorithms~\cite{huang2021tight}---and HES processes on various families of problems.   

This work makes progress towards a program proposed by Barak~\cite[Chapter 9, p.~105]{barak2014op} to understand the optimality of SoS on average-case optimization problems.

The remaining parts of the paper are structured as follows:
\begin{enumerate}
    \item In~\Cref{sec:contributions} we state our technical contributions, main results, the ``ensemble-agnosticism'' conjecture and an immediate open question, followed by an interpretation of the implications of these for SoS relaxations in~\Cref{sec:sos-implications} and a summary of the related literature in~\Cref{sec:related-work}.
    \item In~\Cref{sec:the-problem} we introduce the spherical spin-glass hamiltonian. In~\Cref{sec:parameterized} and~\Cref{sec:hes-sos-program} we motivate, state and explain the HES SoS hierarchy and its proof system. 
    We conclude in~\Cref{sec:proof-overview} by giving an overview of the proof of~\Cref{thm:final-result-informal} and the construction of the modified ensembles in~\Cref{sec:robustness} where the HES SoS hierarchy does better than the direct Hessian ascent algorithm.
    \item In~\Cref{sec:prelims} we introduce various technical preliminaries. 
    These critically include a description of the Hessian ascent algorithm~(\Cref{sec:subag}), the algorithmic threshold for spherical spin glasses~(\Cref{sec:alg-threshold}), graph matrices~(\Cref{sec:prelims-graph-matrices}), some important facts from Hermite analysis~(\Cref{sec:harmonic-analysis}), and customized matrix representations of polynomials~(\Cref{sec:matrix-representations}) that heavily feature in the analysis.
    \item In~\Cref{sec:sdp-proof-system} we formally introduce the grammar that generates the conditional moments of a filtrated space as the atoms of the SDP corresponding to the HES SoS hierarchy~(\Cref{sec:grammar}). 
    We explicitly explain how to encode conditional constraints into atoms of the SDP~(\Cref{sec:sos-poly-basis}), briefly mentioning the connection of the grammar with replicas~(\Cref{sec:replica-equivalance}) and how its encoding includes the low-degree Wiener-Hermite decomposition for functions at a particular level of the filtration~(\Cref{sec:encoding-hermite-decomps}).  
    \item In~\Cref{sec:hods} we use harmonic analysis over multi-variate Hermite polynomials to develop low-degree SoS proofs to control the nuclear norms of the joint moments of HES distributions~(\Cref{thm:high-entropy-nuclear-norm-bound},~\Cref{thm:high-entropy-nuclear-norm-bound-graph-matrices}).
    \item In~\Cref{sec:hessian-taylor} we demonstrate, using low-degree HES SoS proofs, that the spectral moments of the Hessian of the spherical spin glass ensemble converge to those of a scaling of the $\mathsf{GOE}(n)$ ensemble by a function that depends purely on the distance from the origin~(\Cref{lem:hessian-schatten-norm}). 
    We also demonstrate that the energy of the spin-glass hamiltonian is certifiably upper bounded by the value that comes directly from the contribution of the Hessian~(\Cref{lem:taylor}).
    \item In~\Cref{sec:feasibility-main} we show that the HES SoS hierarchy is feasible for the spherical spin-glass model by proving that its constraints are satisfied by a \emph{randomized} version of the Hessian ascent algorithm that has (provably) equivalent performance~(\Cref{lem:feasibility}).
    \item In~\Cref{sec:robustness} we show that the HES SoS hierarchy does better than the Hessian ascent algorithm on certain ensembles that come from anisotropic modifications of spherical spin glasses.
    \item In~\Cref{sec:discussion} we conclude with discussions about the possibility of obtaining meaningful lower-bounds against larger families of SDPs~(\Cref{sec:discussion-hardness}), the possibility (or lack thereof) of certifying relaxed versions of the Ghirlanda-Guerra identities~(\Cref{sec:sos-cert-gg-identities}), the ability to create a SoS HES hierarchy on the cube~(\Cref{sec:hes-sos-cube}), and other future directions~(\Cref{sec:future}).
\end{enumerate}
Various auxiliary technical definitions and proofs used throughout the paper are delegated to the appendices.

\subsection{Contributions}\label{sec:contributions}

We formalize a system of proofs that extends the SoS system to have the capability to reason about the moments of distributions over steps $\sigma_i := v_1 + \dots + v_i$. 
Each step $v_j$ is sampled one at a time, so that, for example, the moments of $v_i$ can be functions of $v_1, \dots, v_{i-1}$. 
The objects of this proof system are recursive polynomials of moments, for example, $(\E_{v_1}v_1(\E_{v_2} \iprod{v_1,v_2})^2)^3$. An inequality $p \ge 0$ can be proved by writing $p = \sum_i q_i^2$ as a sum of squares of polynomial expressions in the proof system~(\Cref{sec:moments-of-filtered-spaces}).

This set of proofs is convex~\cite{barak2014sum} and, therefore, it is the dual proof system to some primal space.
We analyze the SoS meta-algorithm for this proof system, formulated by optimizing some objective function over the primal space and then rounding the result.
Given below is a statement of the key properties that must permit low-degree SoS proofs under HES distributions to allow for certificates with a precise value.

Essentially, SoS can globally search over all low-degree polynomials $\tilde H$ on the ball $\{\sigma: \norm{\sigma}_2^2 \le 1\}$ that certifiably upper-bound the original objective $H$ on the sphere $\{\sigma: \norm{\sigma}_2^2 = 1\}$ to find one so that an (appropriately designed) local iterative algorithm run on $\tilde H$ provably optimizes $H$.
We remark that this is only one way to analyze the certification, and that in general there need not be any pre-existing local algorithm.

\begin{theorem}[Certifiable properties of HES distributions]\label{thm:main}
Let $\cA$ be some set of polynomial axioms that can be satisfied by the moments of a sequence of step variables $\sigma_i = \frac{1}{\sqrt{k}}\left(v_1 + \dots + v_i\right)$.

Let $H$ be a function on the sphere $\cS^{n-1}(1)$ in $\R^n$, for which there exists an extension $\tilde{H}$ defined on the $n$-dimensional ball such that $\nabla \tilde{H}(0) = 0$ and so that there is a SoS proof that $H(\sigma) \le \tilde{H}(\sigma)$ for all $\sigma \in \cS^{n-1}(1)$.
Suppose that $\tilde H$ additionally satisfies the following properties:
\begin{description}
    \item[Certifiably bounded higher-order derivatives]
    There is an integer $k$ and a SoS proof that for every sequence of steps $v_1, \dots, v_k$ satisfying $\cA$ and every $i \in [k]$ and every $r \in \{3, 4, \dots, d_H\}$ and every integer $p \le d_H/\eps$,
    \[ \E_{v_1, \dots, v_{i}}\;\left(\frac{1}{r!\,k^{r/2}}\iprod{\frac{1}{n}\nabla^r\tilde{H}(\sigma_{i-1}),v_i^{\ot r}}\right)^p
    \le \left(\frac{\eps^2}{2k}\right)^p 
    \,, \]
    \item[Certifiable stepwise boundedness]
    There is a SoS proof that for every sequence of steps $v_1, \dots, v_k$ satisfying $\cA$ and every $i \in [k]$,
    \[ \frac{1}{n}\E_{v_1, \dots, v_i} \tilde{H}(\sigma_i) \le (1-\eps + O(\eps^2))\frac{1}{k}\sum_{j=1}^i\mathcal{E}_i\,, \]
    where $\mathcal{E}_i$ only depends on $i$.
    \item[Agreement with $H$ near optima]
    There is a SoS proof that if
    \[ \frac{1}{n}\E_{v_1, \dots, v_{k}} \tilde{H}(\sigma) \ge (1-\eps')\frac{1}{k}\sum_{j=1}^{k}\mathcal{E}_j \]
    then
    \[ \frac{1}{n}\E_{v_1, \dots, v_{k}} \tilde{H}(\sigma) \le \frac{1}{n}\E_{v_1, \dots, v_{k}} H(\sigma) + \eps\eps'\,.\]
\end{description}
Then the HES SoS hierarchy outputs a certificate of value $(1-O(\eps))e^{O(k^2)}(\mathcal{E}_1 + \dots + \mathcal{E}_k)n/k$ and size $\mathsf{poly}\left(n,\exp\left(\frac{1}{\eps}\right)\right)$.
\end{theorem}
In this paper we prove that the spherical spin glass satisfies these properties for some explicit choice of $\cE_i$ with a degree-$O\left(\exp\left(\frac{1}{\eps}\right)\right)$ proof over the HES SoS hierarchy. In the companion paper \cite{ss24rounding}, we provide a rounding algorithm that (with high probability) outputs configurations with performance equivalent to Subag's algorithm \cite{subag2021following}. 

\paragraph{SoS proofs for HES distributions}
The first of the conditions in~\Cref{thm:main} is \emph{a priori} the most difficult to verify, as bounding the higher-order derivatives requires reasoning about injective tensor norms, and it is impossible to give a constant-degree pointwise SoS proof doing so due to known lower bounds~\cite{bhattiprolu2017sum, hopkins2017power}.
However, a fundamental lemma about the moments of HES distributions~(\Cref{lem:nuc-norm-moments-informal}) allows us to circumvent this problem by reasoning about the spectral norm of a matrix reshaping of the tensor instead, which in particular is satisfied by polynomials with independent random coefficients.
We first give a precise definition of high-entropy steps:
\begin{definition}[High-Entropy Step (HES) Distributions]\label[Definition]{def:hes}
    The random vectors $v_1, \dots, v_k$ are a sequence of \emph{high-entropy steps} (HES) and $\sigma := v_1 + \dots + v_k$ is \emph{HES-distributed} if each $v_i$ is Gaussian conditioned on the values of all $v_j$ for $j \in [i-1]$ and
    \[ f_i(v_1, \dots, v_{i-1}) := \E_{v_i}[v_iv_i^T \mid v_1, \dots, v_{i-1}] \,, \]
    (so that $f_i$ is a matrix-valued function), then 
    \[ \E_{v_1, \dots, v_{i-1}} \norm{f_i(v_1, \dots, v_{i-1}) - f_i^{\le d_1}(v_1, \dots, v_{i-1})}_F^{d_2} \le o_n(1)\, ,
    \]
    for some $d_1$ and $d_2$, where ${f_i}^{\le d_1}$ is the degree-at-most-$d_1$ part of $f_i$ (e.g. if $f_i$ is decomposed into Hermite polynomials or any other family of weakly orthogonal polynomials~\cite{chihara2011introduction}), and
    \[ \opnorm{f_i(v_1, \dots, v_{i-1})} \le \frac{1}{\delta n}\, ,\]
    for some $\delta > 0$.
\end{definition}

We now state the bound on the nuclear norm of (polynomials of) joint moments of high-entropy steps.

\begin{lemma}[Nuclear norm property of HES moments (informal)]
\label[Lemma]{lem:nuc-norm-moments-informal}
    Let $p$ be a homogeneous degree-$d_p$ scalar-valued polynomial over $\R^n$ and let $\mat(p)$ be the matrix whose row indices are degree-$\left\lfloor \frac{d_p}{2}\right\rfloor$ monomials and whose column indices are degree-$\left\lceil \frac{d_p}{2}\right\rceil$ monomials and whose entries are
    \[ \mat(p)_{x^{\alpha}, x^{\beta}} =
    \frac{1}{C_{\alpha+\beta}}
    p_{\alpha+\beta}\,, \]
    where $p_{\alpha+\beta}$ is the coefficient of $p(x)$ in front of $x^{\alpha+\beta}$ and $C_{\alpha+\beta}$ is the number of entries of $\mat(p)$ that represent the same coefficient.
    Let $v_1, \dots, v_i$ be a sequence of high-entropy steps with the (conditional) covariance of each step bounded by $\nu$ in spectral norm.
    
    Then, there is a sum-of-squares proof that
    \[ \E [p(v_1+\dots + v_i)]^{\ell} \le e^{O(i^2)} (\nu^{d_p/2}n^{\lfloor d_p/2+1\rfloor/2}\opnorm{\mat(p)})^{\ell}\,,\]
    where the $O(1)$ is hiding dependencies on $i$, $d_p$, $\ell$, and the smoothness parameter of the high-entropy steps.
\end{lemma}
This informal lemma is a special case of a more general lemma (\Cref{thm:high-entropy-nuclear-norm-bound}) where $p$ can be a multi-argument polynomial whose arguments are each high-entropy step vectors.

While the spectral norm of $\mat(p)$ may appear at first to be an opaque quantity, we can interpret it as the maximum over all low-degree polynomials $\phi: \cS^{n-1}(1) \to \R$ of the quantity
\[\frac{1}{n^{d_p/2}}\cdot\frac{|\E_{\sigma \sim S^{n-1}(1)} p(\sigma)\phi(\sigma)^2|}{\E_{\sigma \sim S^{n-1}(1)} \phi(\sigma)^2}\,.\]
In other words, $\norm{\mat(p)}{}_{\mathrm{op}}$ is small if the average value of $p$ is small on every low-degree polynomial reweighing of the uniform distribution over the sphere. 

Note that when $p: \cS^{n-1}(1) \to \R$ is a polynomial where each possible coefficient is random, zero-mean, variance $O(n)$, and  bounded in its higher moments, the value of $\norm{\mat(p)}{}_{\mathrm{op}}$ will be on the order of $n^{\lceil d_p/2+1\rceil/2}$~\cite[Theorem 4.4.5]{vershynin2018high} and we can choose $\nu^{d_p/2}$ to be on the order of $n^{-d_p/2}$, so that the order-$n$ contributions in \Cref{lem:nuc-norm-moments-informal} precisely cancel each other out.
The resulting bound on $\E p(\sigma)$ is $O(n)$, which is of the same order as the true value of $\max_{\sigma} p(\sigma)$ with high probability~\cite[Lemma 58, Eq. B.6]{arous2020geometry}.

\Cref{lem:nuc-norm-moments-informal} complements the graph matrix norm bound framework~\cite{ahn2016graph}; graph matrix norm bounds provide upper bounds on the \emph{spectral norms} of matrices that can be used to represent \emph{polynomials}, and~\Cref{lem:nuc-norm-moments-informal} provides upper bounds on the \emph{nuclear norms} of matrices that represent \emph{moments}.
Consequently, the bounds provided by~\Cref{lem:nuc-norm-moments-informal} are on the dual norm ($1$-norm) in the dual space to those provided by graph matrix norm bounds ($\infty$-norm), allowing us to apply each technique to its corresponding term in a $(1,\infty)$-H\"older's inequality.

This serves as a template for future work bounding the spectral statistics of other random matrices and tensors via low-degree SoS proofs, with the critical idea being to parameterize the set of step-distributions in a way that the nuclear norm of the joint moments can be made to ``fight'' the typical spectral norm of the appropriate reshaping of the tensor, thereby allowing for a smooth transfer of bounds from random matrix theory and free probability to be SoS-compliant. 

\paragraph{Spherical spin glass model}
Finally, we show that the spherical spin glass model (defined in~\Cref{def:spin-glass}) satisfies these sufficient conditions, and therefore the following algorithm exists:
\begin{theorem}
\label{thm:final-result-informal}
For every $\eps > 0$ and $n$, there exists a set of SDP constraints $R_{\eps}$~(\Cref{eq:the-program}) of polynomial size, such that with high probability with respect to a spin glass Hamiltonian $H$ over $n$ variables, there exist $O\left(\mathsf{poly}\left(n,\exp\left(\frac{1}{\eps}\right)\right)\right)$-sized SoS certificates over this SDP as well as actual HES distributions which certify 
\[
    e^{O(1/\eps^2)}\,\mathsf{ALG} \ge \E_{v_1,\dots,v_k}\left[H\left(\frac{1}{\sqrt{k}}\sum_{i=1}^k v_j\right)\right] \ge (1-\eps - O(\eps^2))\,\mathsf{ALG}\,,
\]
where $\mathsf{ALG}$ is a conjectured upper limit on the quality of solution achievable by polynomial-time algorithms, elaborated on in \Cref{sec:alg-threshold}.

Furthermore, in the fRSB regime of the spherical spin glass model~(\Cref{prop:frsb-equivalent}), we have
\[ e^{O(1/\eps^2)}\sup_{\sigma \in \cS^{n-1}} H(\sigma) \ge \E_{v_1,\dots,v_k}\left[H\left(\frac{1}{\sqrt{k}}\sum_{i=1}^k v_j\right)\right] \ge (1-\eps - O(\eps^2))\sup_{\sigma \in \cS^{n-1}} H(\sigma)\,.\]
\end{theorem}

For the spherical spin glass problem, the Crisanti-Sommers formula~\cite{crisanti1992sphericalp} is known to give an explicit expression for the typical optimal value~\cite{talagrand2006free, subag2018free}. The $\mathcal{E}_i$ value achieved for spherical spin glass Hamiltonians in~\Cref{thm:main} (formally shown in~\Cref{lem:hessian-schatten-norm}) and the bounds on the higher-order derivatives provided by~\Cref{lem:nuc-norm-moments-informal} (formally shown in~\Cref{thm:high-entropy-nuclear-norm-bound}) combine to show that the HES SoS program gives \emph{low-degree certificates} on the expected value of the Hamiltonian under any HES distribution. The quality of the bound of this certificate is equivalent (up to constant factors) to that of a ``relaxation'' of the Crisanti-Sommers formula~(\Cref{proposition:huang-alg}), a result of possibly independent interest. As mentioned earlier, the rigorous proofs of the Crisanti-Sommers formula by Talagrand~\cite{talagrand2006free} and Subag~\cite{subag2018free} make heavy use of concentration of measure and the analytical part of the toolkit used to analyze Gaussian processes (such as smooth interpolations and chaining) to characterize the optimum value, neither of which are accessible to low-degree SoS proofs. This provides further evidence for the low-degree hardness of certifying the expected \emph{true} optimal value. The ``relaxed'' value is the conjectured maximum attainable by any polynomial time algorithm~\cite{huang2021tight, jones2022random}.

The constant factors from $\mathsf{ALG}$ which the HES SoS certificates are off by \emph{may} be removed by a more judicious analysis of~\Cref{lem:ydotsypq}. One way to do this is to attempt a proof which uses a particular low-degree representation for the conditional covariance $\fE_i \left[v_iv_i^\sT \mid v_1,\dots,v_{i-1}\right]$ which satisfies the HES constraints, but is also near-optimal. It is not hard to verify that this can be achieved by a matrix Bernstein polynomial that approximates the projector into the top-eigenspace of the Hessian~(\Cref{cor:large-hs-correlation-bernstein-hessian}), or a choice of a sufficiently large power of a ``shifted'' Hessian that tightens the H\"older inequality applied to bound the contribution of the Hessian in the analysis~(\Cref{cor:shifted-hessian-holders}). The remaining technical challenge then becomes to understand the \emph{derivatives} of the near-optimal choice of conditional covariance, which seems to require deeper techniques from the study of asymptotically free tensors~\cite{bonnin2024freeness}. 

\paragraph{Ensemble-agnosticism}
The HES hierarchy makes no assumptions specific to the spherical spin glass model:
its structure is more closely related to the smoothness and correlational structure of low-degree Gaussian processes.
While the HES hierarchy does not inherently encode any local algorithm, we analyze its behavior via one, and in that view, it is able to search over all possible low-degree extensions of the original input problem to find one where a local algorithm works.

In \Cref{sec:robustness}, we provide evidence for this viewpoint by considering anisotropic modifications of the spherical spin glass problem. 

\begin{lemma}[HES hierarchy on ansiotropic spherical spin-glasses]
    There exist ansisotropic modifications to the spherical spin-glass problem, such as those specified in~\Cref{example:degree-scaling},~\Cref{example:direct-sum} and~\Cref{example:shared-sum}, for which there exist HES distributions achieving better value than direct Hessian ascent.      
\end{lemma}

The idea behind these anisotropic modifications is to ``confuse'' the Hessian ascent algorithm by putting differing amounts of weight in different subspaces for some extension of the original function $\tilde{H}$ inside the ball. 
We conjecture that this gives this SoS relaxation a robustness against model misspecification that cannot be attained by any purely local algorithm.

\begin{conjecture}[Ensemble-agnosticism]\label[Conjecture]{conj:distributional-robustness}
    Let $H_n: \cS^{n-1} \to \R$ be a random sequence of polynomials of degree at most some constant $d$. Suppose that $\{H_n\}$ is an asymptotically Gaussian process on $\cS^{n-1}$, in the sense that for every large enough $n \in \N$, for every $r \in \N$, and every fixed sequence $\sigma_1, \dots, \sigma_r \in \cS^{n-1}$, the $\ell$th cumulant tensor (\Cref{def:cumulants}) where $\ell \ge 3$ of the joint distribution on $H_n(\sigma_1), \dots, H_n(\sigma_r)$ is bounded by at most $o_n(r^{r/2}\norm{C}_F^{r/2})$ in 2-norm, where $C \in \R^{r \times r}$ is the covariance matrix between $H_n(\sigma_1), \dots, H_n(\sigma_r)$. 

    Suppose that for every $n$ and every fixed $\sigma \in \cS^{n-1}$,
    \[ \E_{H_n} H_n(\sigma) = 0 \qquad\quad
    \text{and}
    \quad\qquad\E_{H_n} H_n(\sigma)^2 \le O(1) \E_{H_n}\, \norm{H_n}_{L^2(\cS^{n-1})}^2
    \,,\]
    where the expectations are over the randomness of $H_n$ and $\norm{H_n}_{L^2(\cS^{n-1})}^2 := \E_{\sigma \sim \cS^{n-1}} H_n(\sigma)^2$.
    Suppose also that
    \[ \E_{\sigma \sim \cS^{n-1}}\sigma H_n(\sigma) = 0 \]
    (meaning that $H_n$ has no linear part).
    Then there is a rounding algorithm for the HES relaxation that, with high probability, achieves value on $H$ that is an arbitrarily good multiplicative approximation to optimal among polynomial-time algorithms as $n$ grows large.
\end{conjecture}

Two examples of $H_n$ meeting these criteria include:
\begin{itemize}
    \item \emph{Subspace-projected spin-glasses}: A linear combination of a constant number of spherical spin glass models (\Cref{sec:the-problem}) in different (not necessarily mutually orthogonal) subspaces of dimension $\Omega(n)$.
    \item \emph{Random spherical CSPs}: Given a function $f: \R^k \to \R$ with $\E_x f(x_1, \dots, x_k) = 0$ and $\E_x xf(x_1, \dots, x_k) = 0$ when $x \sim \cN(0,\Id_k)$, taking $H_n(\sigma)$ to be the sum over $j \in [m]$ of $f(s_{1,j}\sigma_{i_{1,j}}, \dots, s_{k,j}\sigma_{i_{k,j}})$, where $s_{1,j}$, \dots, $s_{k,j}$ are independent random signs, $i_{1,j}$, \dots, $i_{k,j}$ are independent random elements of $[n]$, and $m$ is superlinear in $n$.
\end{itemize} 
The law of large numbers can be applied to show that this second example is an asymptotically Gaussian process. One example that is not an asymptotically Gaussian process is the random rank-one tensor where $H_n(\sigma) = \iprod{a,\sigma}^3$ for a uniformly random $a \in \cS^{n-1}$.

Proving the conjecture in full generality is a daunting technical challenge.
A potentially more tractable first step is to prove that the Hessians of LDPs that are approximate Gaussian processes have a spectrum whose low degree moments can be specified via a SoS proof under the HES constraints.

\begin{conjecture}[Hessians of smooth Gaussian processes]
    Suppose that $H_n: \cS^{n-1} \to \R$ is a random sequence of polynomials that is an asymptotically Gaussian process satisfying the constraints mentioned in~\Cref{conj:distributional-robustness}. Then, for every low-degree polynomial extension $\tilde{H}: \cB^{n-1}(1) \to \R$ of $H$ satisfying $\tilde{H}(\sigma) = H(\sigma)$ for all $\sigma \in \cS^{n-1}$, there exists some bounded function $f_{\tilde{H}}: \cB^{n-1}(1) \to \R$ which \emph{does not depend on the input randomness}, such that, for every $\sigma_i \in \cB^{n-1}(1)$, it is the case that,
    \[
        \{ \text{HES constraints} \} \proves \fE_{1}\left(\frac{1}{n}\norm{\frac{1}{n}\nabla^2 \tilde{H}(\sigma_i)}_p^p - f_{\tilde{H}}(\sigma_i)\right)^2 \le o_n(1)\, ,
    \]
    for every constant $p \in \Z_{\text{even}}$ (see \Cref{sec:grammar} for notational definitions).
\end{conjecture}

\paragraph{SoS HES hierarchy on the hypercube}
It is likely possible to instantiate a version of the SoS HES hierarchy on the hypercube. This would entail some extra technical work in proving versions of~\Cref{lem:moment-nuclear-norm} and~\Cref{lem:hessian-schatten-norm} under the changed HES axioms.

\begin{openquestion}[HES SoS hierarchy for the hypercube]\label[Open Question]{conj:hes-cube}
    There exists a HES SoS hierarchy over the hypercube which yields low-degree SoS certificates up to the $\mathsf{ALG}$ threshold for mixed spin-glasses on the hypercube~\cite{huang2021tight}. Furthermore, this new hierarchy has exactly the same ``ensemble-agnosticism`` for random low-degree polynomials $H_n: \{\pm 1\}^n \to \R$ which are an asymptotically Gaussian process, and moments that satisfy the conditions stated in~\Cref{conj:distributional-robustness}. 
\end{openquestion}

\paragraph{SoS proofs for analytic inequalities}
We give a collection of SoS proofs for certain inequalities from commutative and non-commutative analysis during the course of analyzing  the relaxation, which may be of independent interest (\Cref{sec:sos-commutative-inequalities},~\Cref{sec:non-commutative-inequalities}). These include a proof for a matrix H\"older's inequality with dyadic norms~(\Cref{lem:matrix-p-q-holders}) and a proof of strong convexity of the map $x \to x^q$ with $q \in (1,2]$.

\subsection{Implications for Sum-of-Squares relaxations}\label{sec:sos-implications}
We briefly describe the immediate implications of our work for SoS relaxations.

\paragraph{Convex relaxations over distributions}
In sampling, a now-popular perspective considers the problem of sampling from a Gibbs distribution as equivalently the problem of optimizing a regularized objective over the set of all distributions \cite{ge2018simulated,moitra2020fast,doi:10.1137/S0036141096303359}.
This view gives primacy to the \emph{probability distribution} over algorithm outputs as itself an object of optimization.
We extend this view to the realm of convex relaxations.

\paragraph{Concentration of measure in SoS} A common difficulty in previous SoS work was that concentration of measure and union bound proofs could not be converted into SoS proofs~\cite[Lec.~1.2, Marley Paradigm]{barak2016proofs}.
For example, bounding each point in an exponentially large epsilon-net with a Chernoff bound and then applying the union bound to add up the probabilities at each point usually requires access to linearly large degree pseudo-distributions.

This feature presented a gap between powerful SoS techniques and relatively simple iterative local-update techniques. 
These local update techniques, in contrast, can be \emph{enabled} by concentration arguments, allowing them to find iterative improvements upon a rough first guess with high probability~\cite{montanari2021optimization, barak2015beating, subag2021following, alaoui2023local}.

This is where it helps to lift to a convex relaxation over distributions.
A good choice of constraints on the rounded distributions drastically reduces the degree needed to show these concentration bounds by drastically curtailing the set of primal objects that the bounds need to apply to.

We believe this is the first time that a low-degree SoS proof of the existence of \emph{many} potential solutions (possessing some regularity in their occurrence) can be provided (while still enabling a roudning scheme) rather than a proof that there are a limited number of solutions~\cite{barak2014rounding,barak2011rounding,ma2016polynomial,karmalkar2019list, raghavendra2020list} or of product distributions of solutions~\cite{barak2011rounding, jain2019mean}.

\paragraph{Robustness in average-case} The SoS proofs-to-algorithms framework has demonstrated remarkable adversarial/outlier robustness in learning settings~\cite{kothari2018robust,hopkins2018mixture,diakonikolas2022robust,bakshi2020outlier,liu2021settling}. We provide evidence that the HES SoS hierarchy exhibits a kind of ``robustness'' to model-misspecification~(\Cref{sec:robustness}) that allows it to be applied in an ensemble-agnostic manner that prior algorithms used for such average-case problems cannot be.
A general SoS approach to average-case problems could extend these robustness properties to these settings as well, as has already been done for certain robustness properties for any problem with a low-degree AMP algorithm~\cite{ivkov2023semidefinite}. 

\paragraph{Entropy-energy variational trade-off} The HES distribution features a tunable parameter $\delta$, which can be understood as a form of entropy constraint: setting $\delta$ to 1 results in the largest possible entropy, the uniform distribution over the sphere, whereas $\delta = 1/n$ does not constrain the solution at all.

Similarly to the entropy-energy trade-off in statistical physics, the HES SoS program searches for the highest-value HES distribution at a given value $\delta$.
Unlike previous uses of pseudo-entropy constraints~\cite{ma2016polynomial,karmalkar2019list,raghavendra2020list}, this program fails to be a relaxation of the original maximization problem, since the HES entropy constraint reduces the achievable value to \emph{below} the true optimum. 
This is in keeping with the thermodynamic analogy, where temperatures above zero give progressively softer soft-minima of the energy. 
This gives the HES distribution family a trade-off between value and entropy, similar to the Gibbs distribution.

\paragraph{Average case optimization in SoS} This work advances understanding toward a program of Barak~\cite[Chapter 9, p.~105]{barak2014op} regarding the universality of SoS programs in average-case problems.
Average-case problems are usually where concentration of measure is the primary characterization of solutions, and so they are resistant to techniques based on low-degree pointwise proofs.

\subsection{Related work}\label{sec:related-work}

\paragraph{Analytic \& geometric properties of spin glasses}
The Crisanti-Sommers formula~\cite{crisanti1992sphericalp} gives a variational principle for the limiting free energy density of spherical spin glasses.
An alternative characterization of the free energy density over the \emph{ball} (and not just the sphere) was given by a generalized TAP free energy formula due to Subag~\cite{subag2018free}.
This approach was then strengthened and adapted to the setting of mean-field spin glasses over the hypercube by Chen, Panchenko and Subag~\cite{chen2018generalized, chen2021generalized}.
In the spherical case, Subag's characterization was independent of the Parisi framework, but in the setting over the hypercube the generalized TAP free energy still relies on the Fenchel dual of a solution to the Parisi PDE on restricted intervals.
In addition to the generalized TAP free energy, Subag~\cite{subag2018free} gave a complete characterization of the free energy landscape of spherical spin glasses by appealing to the study of thin spherical bands induced by the orthogonal subspace of points in the unit ball that maximize the generalized TAP free energy at a specific distance.
This approach inspired Subag's Hessian-based optimization algorithm~(\Cref{sec:subag}) under fRSB~\cite{subag2021following}.

\paragraph{Algorithmic results for optimization of mean-field spin glasses}
Along with the local spectral algorithm of Subag~\cite{subag2021following}, there has been an equivalent class of results that use an AMP algorithm to find near-ground states of mean-field spin glasses over the hypercube~\cite{montanari2021optimization, alaoui2021optimization}.
In a work of Huang and Sellke~\cite{huang2021tight}, it is established that both these families of algorithms achieve a particular threshold of approximation (called $\mathsf{ALG}$) which is given by the infimum of a ``relaxed'' version of the Parisi formula.~In fact, this is the best approximation ratio achievable by any ``overlap-concentrated'' algorithm.
It is known that on spin-glasses, Langevin dynamics and certain families of AMP algorithms are overlap-concentrated~\cite{huang2021tight}.
Certain low-degree polynomials are also conjectured to be overlap-concentrated~\cite{jones2022random}, but this has not yet been formally proved~(\Cref{sec:discussion-hardness}).

\paragraph{Hardness results for optimization of mean-field spin glasses}
Complementing the algorithmic side of the results, it is well established that when the unique minimizer~\cite{auffinger2015parisi} of the Parisi formula (which has value $\mathsf{OPT}$) is not strictly increasing, there is a gap as $\mathsf{ALG} < \mathsf{OPT}$~\cite{huang2021tight}.
This situation corresponds to the existence of a so-called \emph{overlap-gap property} and there is abundant work in the literature defining various variants of this property~\cite{gamarnik2014limits, rahman2017local, chen2019suboptimality, huang2021tight} and then using it to obtain varying strengths of hardness results for many problems~\cite{rahman2017local, chen2019suboptimality, chou2022limitations, wein2022optimal, gamarnik2020low, huang2022computational} including spin-glasses~\cite{gamarnik2020low, gamarnik2021overlap, huang2021tight} against a large collection of algorithms.

\paragraph{SoS certificate complexity for spin glasses} There is now a fairly complete body of work showing that certifying $\mathsf{OPT}$ to $\pm\epsilon$-precision using the standard SoS relaxation requires $\Omega(n^{f(\eps)})$-degree on the sphere~\cite{bhattiprolu2017sum,hopkins2017power} and the hypercube~\cite{kunisky2021tight, ghosh2020sum}.
The certification lower bound on the sphere is for \emph{pure} spherical-spin glasses with $d_H \ge 3$, but can be extended by the same trace-moment method to \emph{mixed} spherical spin glasses as each degree term in the mixed setting makes a $\Theta(n)$ contribution.
The lower bounds on the hypercube are explicitly against the SK model~\cite{sherrington1975solvable}, but can likely also be extended to the mixed setting as well for similar reasons.

\paragraph{SoS algorithms, entropy, and rounding}
Previous work has added various forms of high-entropy constraints to the standard SoS hierarchy~\cite{ma2016polynomial,karmalkar2019list,raghavendra2020list}, serving as a precursor to the full idea of writing an SoS program over distributions rather than individual solutions.
Correlation rounding has been proven optimal for convex relaxations with pointwise proofs optimizing degree-2 spin glass polynomials over the discrete hypercube~\cite{jain2019mean,risteski2016calculate}.
The approximation ratio of the standard SoS hierarchy has been studied on average-case CSPs~\cite{barak2013optimality}.

\section{Technical Overview}\label{sec:technical-overview}

\subsection{The spherical spin glass problem}
\label{sec:the-problem}

\paragraph{The Hamiltonian: Dense random low-degree polynomials}
The goal is to optimize a random polynomial of degree $d_H$ over the hypersphere $\cS^{n-1}(1)$ of radius $1$. A spherical spin glass corresponds to the \emph{low temperature} limit of a material whose Hamiltonian is given by the very same random polynomial.

\begin{definition}[Mixed Spherical $p$-Spin Glass]
\label[Definition]{def:spin-glass}
    The hamiltonian of a mixed spherical $d_H$-spin glass is given as,
    \begin{equation}\label{eq:spherical-ham}
        H_n(\sigma) = \sum_{k=1}^{d_H}\gamma_k\sqrt{n}\!\!\sum_{i_1,\dots,i_k \in [n]^k}g_{i_1,\dots,i_k}\sigma_{i_1}\dots\sigma_{i_k}\, ,
    \end{equation}
    where $g_{i_1,\dots,i_k} \overset{i.i.d.}{\sim} \cN(0, 1)$ for every $k \in [d_H]$. Furthermore, the coefficients $\left\{\gamma_k\right\}_{k = 2}^p$ are non-negative, real valued and independent of $n$.\footnote{Although we do not handle a linear term $k=1$ in this manuscript, such a term can be accommodated with a modification of the SoS-encoded HES distributions in order to allow for a \emph{drift} term. We implement this in forthcoming work addressing the SK model and its generalizations.}
\end{definition}

The goal of the optimization task is to maximize the hamiltonian above over $\cS^{n-1}(1)$,
\[
    \sup_{\sigma \in \cS^{n-1}(1)} H_n(\sigma)\, .
\]
This quantity is called the \emph{ground state energy} of the spherical spin glass and is equivalent to its so-called free energy at zero temperature.

Note that we formulate the problem on the unit sphere, but have adjusted the Hamiltonian to be equivalent to previous formulations on the radius-$\sqrt{n}$ sphere \cite{subag2018free,subag2021following}.

\paragraph{The mixture function}
The solution quality achieved by the algorithm on this optimization problem depends on the structure of the mixture coefficient $\{\gamma_k\}_{k=2}^{d_H}$. Under a specific assumption~(\Cref{prop:frsb-equivalent}) on these, termed ``full-Replica Symmetry Breaking'' (fRSB), our algorithm will output a solution that is arbitrarily close to the true optimal value.

\begin{definition}[Mixture Polynomial]
    The mixture polynomial of a hamiltonian $H_n(\sigma)$ is given as,
    \[
        \nu(q) = \sum_{k = 2}^{d_H} \gamma_k^2 q^k\, ,
    \]
    and is defined for $q \in [0, 1]$.
\end{definition}

\subsection{Pseudo-distributions over filtered distributions}
\label{sec:parameterized}

A SoS relaxation~\cite{barak2014sum} takes an optimization problem of the form
\begin{align*}
\max_{x \in \R^n} \;&p(x)\\
\text{subject to } &{}q_1(x) \ge 0
\\&\qquad\vdots
\\
&q_m(x) \ge 0
\end{align*}
(where $p$, $q_1$, \dots, $q_m$ are polynomials)
and augments the domain of the problem to a higher-dimensional space, so that both the constraint set and the objective function are convex in the new augmented domain.
This convex program is then solvable in polynomial time, and rounding techniques are used to convert the solution found in this augmented domain back to a solution in the original domain without losing too much of the objective value.

The augmentation adds variables ($\pE x^{\alpha}$ for each monomial $x^{\alpha}$ up to some degree limit, where the entire expression ``$\pE x^{\alpha}$'' is a single SDP variable) representing the low-degree moments of a distribution over the original solution domain. This converts the objective into one of optimizing the \emph{pseudo-expectation} $\pE p(x)$ of the original polynomial objective $p(x)$.
Now $\pE p(x)$ is a linear objective in the \emph{pseudo-moments}, since the expectation of a polynomial is simply a linear combination of moments.
Similarly, the original polynomial constraints $q_i(x) \ge 0$ each become a family of linear constraints on the pseudo-moments.

Finally, the SoS relaxation adds consistency constraints  stating that the pseudo-expectation of every square polynomial $\pE r(x)^2$ is non-negative.
The non-negativity of all square polynomials (up to a degree limit) can be realized as a single convex constraint simply requiring a certain matrix of pseudo-moments to be positive semi-definite~\cite{barak2014sum}.

To this basic framework, we introduce symbols $v_1, \dots, v_k$ each representing a random variable taking values in $\R^n$.
The convex program variables are now not pseudo-moments of the $v_i$, but pseudo-moments \emph{of the moments of} $v_i$, for example $\pE (\E (v_i)_j^3)^2$.
For this to be non-trivial as a relaxation, we add non-convex polynomial constraints on those moments, restricting the set of possible distributions represented by $v_i$, before taking the pseudo-moment relaxation.
This turns those polynomial constraints on moments into linear constraints on pseudo-moments of moments.
For pseudo-distributions over filtered distributions such as the HES distributions, the distribution of each $v_i$ depends on all previous $v_1$, \dots, $v_{i-1}$, so that, for example, if $\fE_i$ represents the expectation conditioned on $v_1, \dots, v_{i-1}$, then $\pE (\E_{1}(v_1)_4(\E_2 (v_2)_3)^2)^3$ is a SDP variable in the program\footnote{Note that this term is not expressible as a joint moment, making this set-up strictly stronger than using joint moments as program variables. This is elaborated on, and the equivalence with replicas explained, in~\Cref{sec:replica-equivalance}}.

In an extention to the axiom notation of standard SoS~\cite{barak2014sum}, we also state axiomatic constraints as, for example, $\{p(v_1) \ge 0\}$ and $\{p(v_1,v_2) \ge 0\}$, which intuitively mean that $p(v_1) \ge 0$ and $p(v_1,v_2) \ge 0$ for all $v_1$ and $v_2$ in the support of the filtered distribution.
These may be formally defined in the SDP as meaning $\pE\E_1 [p(v_1)q(v_1)^2] \ge 0$ and $\pE \E_1\left[\E_2[p(v_1,v_2)q(v_2)^2]\cdot r(v_1,\E_2[v_2],\E_2[v_2^2], \dots)^2\right] \ge 0$ for all polynomial expressions $q$ and $r$ for which those pseudo-expectations are defined.
This system is explained in detail in~\Cref{sec:moments-of-filtered-spaces}.

\subsection{The HES SoS hierarchy}\label{sec:hes-sos-program}

We state our program and briefly explain and motivate the constraints.

The actual solution is represented implicitly as a sum of $k$ equally-sized components $\sigma = \sum_{i=1}^k v_i/\sqrt{k}$.
The variables $k$, $\delta$, $d$, and $d'$ are adjustable parameters.
\begingroup
\allowdisplaybreaks
\begin{align}\label{eq:the-program}
    \tag{P.0}
    {}&\;\qquad\qquad\qquad\qquad\max \pE\fE_1\left[ H\left(\sum_{j=1}^{k} \frac{v_j}{\sqrt{k}}\right)\right]\qquad\qquad\qquad\qquad
    \\
    &\notag \qquad\qquad\qquad\qquad\qquad\qquad \text{s.t.}\;& \\
    \label{eqs:the-program-intro-norm}
    \tag{P.1}
    &\qquad\forall i \in [k]\mathrel{.}{} \\
    &\qquad\qquad\qquad\left\{\norm{v_i}_2^2 = 1\right\} \nonumber  \,, \\
    \label{eqs:the-program-intro-orth}
    \tag{P.2}
    &\qquad\forall i \in [k]\mathrel{.}{} \\
    &\qquad\qquad\qquad\left\{\left\langle v_i, \sum_{j=1}^{i-1}v_j \right\rangle = 0\right\}\nonumber\,, \\
    \label{eqs:the-program-intro-cum}
    \tag{P.3}
    &\qquad\forall i \in [k]\mathrel{.}{}\forall j \in [d]{\setminus}\{2\}\mathrel{.}{} \\
    &\qquad\qquad\qquad\left\{\norm{\kappa_{j}\left[v_{i}\mid v_{i-1}\right]}_2^2 \le \frac{o_n(1)}{n^{j/2}}\right\} \nonumber\,, \\
    \label{eqs:the-program-intro-opnorm}
    \tag{P.4}
    &\qquad\forall i \in [k]\mathrel{.}{} \\
    &\qquad\qquad\qquad\left\{\opnorm{\fE_i\left[v_iv_i^\sT\right]} \le \frac{1+o_n(1)}{\delta n}\right\} \nonumber\,, \\
    \label{eqs:the-program-intro-ldp}
    \tag{P.5}
    &\qquad\forall i \in [k]\mathrel{.}{} \\
    &\qquad\qquad\qquad\left\{\norm{\left(\fE_i\left[v_iv_i^\sT\right]^{> d}\right)^{\ot d^{k-i}}}^2_F \le \frac{o_n(1)}{n^{d^{k-i}}}\right\} \nonumber\, .
\end{align}

\begin{enumerate}
    \item[1.] The normalization constraint ensures that steps are of equal length.
    \item[2.] The orthogonality constraint ensures that we can compute the length of the sum of the steps via Pythagorean theorem.\footnote{This orthogonality constraint is approximately implied to high precision by the other constraints, but we keep it for convenience so that the norms of the sums of the steps are always known exactly. Changing the constraint to be an upper and lower bound up to $\pm o_n(1)$ factors allows many computations to become more convenient, but does not affect the quantitative strength of the results. Therefore, in~\Cref{sec:matrix-bernstein-analytic-sparsity} and the companion paper, we invoke arguments that only preserve the constraint up to small fluctuations.}
    \item[3.] The stepwise Gaussianity constraint says that each step is close to a zero-mean Gaussian in its low-degree moments, conditioned on the $(i-1)$th step of the filtration.
    It enforces this by matching the cumulants (and therefore moments) of $v_i$ to those of a Gaussian~(\Cref{sec:gaussians,sec:appendix-sdp}).
    This constraint enables the use of the Weiner-Hermite polynomial basis~(\Cref{sec:prelims-multivar-hermites}), whose weak orthogonality is equivalent to vanishing higher cumulants.
    
    \item[4.] The spectral high-entropy constraint forces the Gaussian to be spread out over at least $\delta n$ directions.
    This ensures that the distribution of $v_i$ does not put too much weight into any single direction, allowing us to avoid pointwise proofs.

    \item[5.] The low-degreeness constraint asserts that the dependence of the distribution of $v_i$ on the previous iterates should be  close to a low-degree polynomial.  Here, $f^{> d'}$ for a function $f$ refers to the projection of $f$ to its degree-more-than-$d'$ part (up to the degree limit of the pseudo-distribution).
    This constraint is defined in terms of the Fourier basis given by the Weiner-Hermite polynomials in~\Cref{sec:appendix-sdp}.
\end{enumerate}

Note that while it is rigorously established that the support of the Gibbs distribution for these models resides in a very complicated approximately ultrametric structure, only a \emph{trivial} approximately ultrametric structure can be inferred from the HES constraints.\footnote{The fact that the support of the Gibbs measure should reside in an ultrametric space was first predicted by Parisi~\cite{parisi1980sequence}.
Ruelle then explicitly characterized the Gibbs via a construction of a random measure supported on the leaves of an ultrametric tree, defined by a family of cascading independent Poisson-Dirichlet processes on the inner nodes - This construction was termed the Ruelle Probability Cascades~\cite{ruelle1987mathematical, panchenko2013ruelle} and it gave an object on which the Parisi formula could be equivalently re-written.
It was later shown by Panchenko~\cite{panchenko2013parisi} that the so-called Ghirlanda-Guerra identities~\cite{ghirlanda1998general} \emph{implied} the ultrametricity of the solution-space.
It is an interesting question as to whether there even exist SoS proofs (under some modified step constraints) that can certify that solutions reside in a hierarchically organized ultrametric space consistent with the RSB framework, and this is discussed further in~\Cref{sec:sos-cert-gg-identities}.}

\paragraph{Rounding higher pseudo-moments} The HES structure was partly motivated by the following concrete question: 
\begin{center}
    \emph{What rounding strategies can round pseudo-distributions by matching higher pseudo-moments?}
\end{center} 
Rounding by matching degree-2 moments is well understood due to the quadratic sampling lemma~(\Cref{lem:quadratic-sampling}), but in general, the question of whether a set of degree-3 pseudo-moments even corresponds to the moments of an actual distribution is NP-hard~\cite{grigoriev2001linear}.

As a Gaussian vector's third moments are a symmetrization of the tensor product of its mean with its covariance, this distribution can only match the order-3 moments of a pseudo-distribution to the extent that it matches the first and second moments.
But consider the sum of two vectors, $u+v$, where $u$ is Gaussian and $v$ is also Gaussian but only when conditioned on $u$, so that the covariance of $v$ may depend on $u$.
Then the third-order moments of this sum includes a $\E_{u,v} u\ot v\ot v$ term, which can be expressed as $\E u \ot f(u)$ for some arbitrary function $f: \R^n \to (\R^{n})^{\ot 2}$.
This is much more expressive: for example, when $f$ is degree-$3$, $\E u \ot f(u)$ may have full tensor rank, whereas pure Gaussian third moments cannot.

Rounding to HES distributions would therefore allow much greater expressivity than rounding to Gaussians. 
While random reweighing~\cite{barak2014rounding} can improve the expressiveness of Gaussian rounding, it is limited by only being able to output a statistical mixture of Gaussians.

\subsection{Technical challenges}

The spectral conditions for the success of the HES hierarchy stated in~\Cref{thm:main} are based on finding algebraic structure compatible with a SoS proof that the dominant certifiable contributions are from the Hessian. 
This amounts to certifying the Taylor expansion of the Hamiltonian under HES distributions, which recovers Subag's Hessian ascent algorithm~\cite{subag2021following} on spherical spin glasses.
We briefly overview, for SoS experts, some of the technical challenges involved in approaching this task\footnote{Readers that are not familiar with average-case SoS lower bounds may freely skip this subsection without any loss of ability to follow along the remaining subsections in~\Cref{sec:technical-overview}.}.

\begin{description}
\item[Reweighed iterative rounding along a path in the ultrametric tree.] It is natural to consider explicitly representing steps in the ultrametric tree as program variables (\Cref{eqs:the-program-intro-norm,eqs:the-program-intro-orth}), as this follows the path of existing iterative algorithms.
It then becomes apparent that previous methods of rounding fail to make progress due to the impossibility of proving that there are few good solutions~\cite{barak2014rounding, ma2016polynomial, karmalkar2019list, raghavendra2020list} or few good product distributions~\cite{barak2011rounding, jain2019mean} of solutions (since, in fact, there are many good product distributions of solutions, corresponding in this case to ``pure states'' in the statistical physics literature~\cite{talagrand2010construction}).

\item[Eigendecompositions and edge statistics for the eigenspectrum.] Progressing in this direction, one runs into further problems if one attempts to follow the iterative spectral algorithm of Subag too closely: the proofs and algorithm of Subag make heavy use of high-degree SoS operations and objects, such as concentration inequalities with exponentially decaying tails~\cite[Lemma 3]{subag2021following} and eigendecompositions~\cite[Eq 3.10 \& Eq. 3.11]{subag2021following}.

\item[Bounding higher-order derivatives via Dudley-style entropy bounds.]  When bounding the contribution of the higher-order derivatives, Subag~\cite[Lemma 11]{subag2021following} invokes a bound~\cite[Lemma 58]{arous2020geometry} on the injective norm of random tensors.
This uses Dudley-style entropy bounds~\cite{latala2006estimates} (or chaining arguments) to upper-bound the expected norm, and then invokes the Borell-TIS inequality~\cite[Theorem 2.1.1]{adler2009random} to guarantee strong concentration around the bounded value. 
These techniques are not amenable to pointwise low-degree SoS proofs~\cite{bhattiprolu2017sum,hopkins2017power}. 

\item[Disorder chaos.] To compound the technical challenge, we believe it is not possible to capture the \emph{true} low-temperature Gibbs distribution with a low-degree SoS program due to \emph{disorder chaos}: small perturbations in the Hamiltonian lead to large fluctuations in the moments of the Gibbs distribution~\cite[Theorem 3]{chen2017parisi}\cite[Eq.~2.11, Eq.~5.16 \& Eq.~5.22]{alaoui2022sampling}, whereas SoS is generally believed to excel when solutions are stable to input perturbations.
\end{description}

The difficulties involving concentration or too many good solutions are averted by introducing pseudo-distributions over HES distributions.

\subsection{Structure of main argument}\label{sec:proof-overview}

The main structure of the argument for the sufficient conditions in~\Cref{thm:main} will be to certify the Taylor expansion of the objective function achieves good value when run on the extension $\tilde H$ to the ball of the original problem $H$ on the sphere.

\subsubsection{Certifiable Taylor expansion}

This argument is based on a Taylor expansion (\Cref{sec:certifiable-taylor}), from which we derive strict upper bounds on the expected objective value achieved by each step.

The gradient term is automatically suppressed by the HES constraint that $\E_{i} v_iv_i^{\sT}$ is bounded in spectral norm.\footnote{As in Subag's analysis~\cite{subag2021following}, it is crucial that we \emph{do not} follow the gradient---briefly, while following the gradient can gain more objective value \emph{per distance travelled} in the step, it will also reduce the \emph{number of steps} we can take, due to the gradients being aligned with the radial direction and thereby reaching the unit sphere more quickly.}
The third- and higher-order derivatives of $\tilde H(\sigma_{i-1})$ are bounded when the step size is small enough, $\sigma_{i-1}$ is drawn from a HES distribution, and $\tilde H$ is well behaved enough, as discussed in the next section.
Therefore the dominant contributions are from the Hessian term at each step.

\subsubsection{SoS proofs about HES distributions}

To show upper bounds on $\E p(\sigma)$ for a degree-$d$ homogeneous polynomial $p$ (suppose $d$ is even for simplicity), we represent both the polynomial $p$ and the moments of $\sigma$ as matrices
\[ \E p(\sigma) = \iprod{\mat(p), \E \sigma^{\ot d/2}(\sigma^{\ot d/2})^{\sT}}\,, \]
where $\mat(p) \in (\R^{n \times n})^{\ot d/2}$ is the matrix representation of the coefficients of $p$ satisfying 
the tensor mode permutation symmetries
\[\mat(p)_{(i_1, \dots, i_{d/2}),(i_{d/2+1}, \dots, i_d)} = \mat(p)_{(i_{\pi(1)}, \dots, i_{\pi(d/2)}),(i_{\pi(d/2+1)}, \dots, i_{\pi(d)})}\]
for every permutation $\pi \in \mathrm{S}_d$ in the symmetric group.
This representation always exists since $\E p(\sigma)$ is a linear function of the degree-$d$ moments of $\sigma$, each of which is represented in at least one entry of $\E \sigma^{\ot d/2}(\sigma^{\ot d/2})^{\sT}$.

Now observe that since $\mat(p)$ is mode-symmetric, for any matrix $V$ such that $\E \sigma^{\ot d/2}(\sigma^{\ot d/2})^{\sT}$ is the mode-symmetrization of $V$, we have
\[ \E p(\sigma) = \iprod{\mat(p), V}\,. \]

For many average-case ensembles of polynomials over the unit sphere, we will have tight concentration on $\opnorm{\mat(p)}$, so that $\opnorm{\mat(p)}$ will be at most $O(n^{d/4-1/2})$ times larger than $\max_{\sigma}p(\sigma)$ with high probability, so we can consider by H\"older's inequality
\[ \E p(\sigma) \le \opnorm{\mat(p)}\norm{V}_1\,. \]

Then, for a given distribution over $\sigma$, all it takes is to find some $V$ whose mode-symmetrization is $\E \sigma^{\ot d/2}(\sigma^{\ot d/2})^{\sT}$ and whose nuclear norm is at most $O(n^{-d/4+1/2})$ in order to bound $\E p(\sigma)$ to within a constant of $\max_{\sigma}p(\sigma)$ with high probability.

To illustrate the broad approach to this, we can look at the simple case where $\sigma$ is drawn from a Gaussian distribution with mean $0$ and covariance $\Sigma \in \R^{n\times n}$.
By Isserlis's theorem (\Cref{lem:isserlis}), $\E \sigma^{\ot d/2}(\sigma^{\ot d/2})^{\sT}$ is equal to a sum of mode-permutations of $\Sigma^{\ot d/2}$.
Letting $\Phi \in (\R^n)^{\ot 2}$ be the tensor whose entries are $\Phi_{(i,j)} = \Sigma_{i,j}$, we can choose $V$ to be one of those mode-permutations
\[ V := (\Phi\Phi^{\sT})^{\ot\lfloor d/4\rfloor}\ot \Sigma^{\ot (d/2 - 2\lfloor d/4\rfloor)}\,. \]
Then the mode-symmetrization of $V$ is equal to $\E \sigma^{\ot d/2}(\sigma^{\ot d/2})^{\sT}$ and $\norm{V}_1 = \norm{\Phi}_2^{2\lfloor d/4\rfloor}\norm{\Sigma}_1 = \norm{\Sigma}_F^{2\lfloor d/4\rfloor}$ since $\norm{\Sigma}_1 = 1$.
If we assume $\norm{\Sigma}_F^2 \le O\left(\frac{1}{n}\right)$, we get $\norm{V}_1 \le O(n^{-d/4+1/2})$ as desired.

To generalize this so that $\sigma$ can be HES-distributed, we do similar manipulations inductively over the steps $\sigma = v_1 + \dots + v_k$, this time using the moment-cumulant formula (\Cref{def:cumulants}), which generalizes Isserlis's theorem to the case of non-Gaussian distributions.
It takes a fair amount of care here to characterize all the cumulants that arise and how to reshape them to matrices with low nuclear norm.

Through a more careful analysis of mode symmetries, this generalizes to the case where $p$ is a power of a polynomial whose matrix representation satisfies a spectral norm bound (similar to the technique in the upper bound proof of \cite{bhattiprolu2017sum}), and then further on to cases where $p$ can be represented as a graph matrix whose components have spectrally bounded matrix representations, as these are also relevant in~\Cref{sec:overview-hessian-cert}.

This also generalizes to the case when $p$ is a polynomial of both $\sigma$ and $v$, where $v$ is a high-entropy step with $\sigma$ the sum of the previous iterates, and captures the case when $p$ is a third- or higher-order derivative.

\subsubsection{Certifiable spectral moments for the Hessian}
\label{sec:overview-hessian-cert}

The only missing certification now is for the value achieved/achievable by the Hessian term of the Taylor expansion, which is done in \Cref{sec:wigner-cert-hessian}.

The Taylor expansion showed that the change $\E \tilde{H}(\sigma_{i}) - \E \tilde{H}(\sigma_{i-1})$ is dominated by the Hessian term $\frac{1}{2k}\E\iprod{ \nabla^2 \tilde{H}(\sigma_{i-1}), \E_i v_iv_i^{\sT}}$.
For $v_i$ a high-entropy step, the value attainable will be determined by the $p$-norms of $\nabla^2 \tilde{H}(\sigma_{i-1})$, by a direct matrix generalization of the analytic sparsity arguments of \cite{arora2015subexponential,barak2012hypercontractivity}.
Therefore, we need to certifiably bound these Schatten $p$-norms for all HES distribution over $\sigma_{i-1}$.

We do this by borrowing the non-backtracking walk argument from random matrix theory.
The $p$th power of the Schatten $p$-norm of a matrix $M$ is equal to $\Tr M^p$, which can be expressed as a sum over length-$p$ closed walks on a graph where the edge $\{i,j\}$ is labelled by $M_{i,j}$ and the value to be summed for each walk is the product of all the edge labels on the walk.

For the spherical spin glass Hamiltonian, it will turn out via the non-backtracking walk argument that the expectation term $\E_g \Tr [\nabla^2 H(\sigma_{i-1})]^p$ (with expectation over randomness of the instance) is essentially independent of $\sigma_{i-1}$, as nearly all of the expectation terms will factor out into a multiple of a power of $\norm{\sigma_{i-1}}_2^2$ (which is a constant) yielding an (upto $o_n(1)$ factors) exact expression.

We then represent the coefficients of the polynomial expressing the non-expectation terms $\Tr [\nabla^2 H(\sigma_{i-1})]^p - \E_g \Tr [\nabla^2 H(\sigma_{i-1})]^p$ as a sum of graph matrices~\cite{ahn2016graph}.
This allows for the invocation of the graph matrix norm bound framework on the spectral norms of the matrix representations~(\Cref{thm:graph-matrix-norm-bound}) of those polynomials.
At this point, by a ($1,\infty$)-H\"older's inequality, one can upper bound the contribution of each non-expectation term,
\[
    \E\left[\Tr [\nabla^2 H(\sigma_{i-1})]^p - \E_g \Tr [\nabla^2 H(\sigma_{i-1})]^p\right] = \sum_{\text{noise}}\E[p_{\text{noise}}] \le \sum_{\text{noise}}\opnorm{\mat{(p)_{\text{noise}}}}\norm{V(\sigma)}_1\, .
\] 
The exact extraction of the expected trace component and bound on the non-expected components gives a certificate for the $p$-Schatten norms of the Hessian consistent with the low-degree moments of a scaled $\mathsf{GOE}(n)$ matrix. An upper bound on $\frac{1}{2k}\E\iprod{ \nabla^2 \tilde{H}(\sigma_{i-1}), \E_i v_iv_i^{\sT}}$ now follows by a $(p,q)$-H\"older's inequality, where the choice of $(p,q)$ conjugates is ascertained by a judicious use of Catalan combinatorics and analysis~(\Cref{cor:shifted-hessian-holders}), and the bound can be invoked with high probability (over the randomness of instance) due to standard implications of Gaussian hypercontractivity~(\Cref{lem:hypercontractivity-gaussians}).

\subsubsection{Ensemble-agnosticism}

We provide a few simple examples of modifications to the spherical spin glass ensemble for which the HES hierarchy would provably outperform the local Hessian ascent algorithm (\Cref{sec:robustness}), \emph{conditioned} on the existence of a rounding scheme mentioned in~\Cref{conj:distributional-robustness}.
This is done in an ensemble-agnostic way, so that the SoS-based algorithm is exactly the same one run for the plain spherical spin glass model, with no parameter tuning needed.

The first one is motivated by considering a homogeneous degree-4 spherical spin glass, except that the entries of the random tensor $g$ that have repeated indices (for example $g_{i,i,j,j}$ or $g_{i,j,i,j}$) have $\alpha n$ times larger variance than the other entries for some constant $\alpha$.
On the sphere, this is equivalent to sum of a degree-2 spin glass $\alpha H_{2}(\sigma)$ and a degree-4 spin glass model $H_4(\sigma)$, but on the ball, it is $\alpha \norm{\sigma}^2 H_{2}(\sigma) + H_4(\sigma)$.
This makes a local algorithm neglect the degree-2 term in the beginning when $\norm{\sigma}^2$ is small.

The second is the vector space direct sum of a pure degree-2 spherical spin glass and a pure degree-4 spherical spin glass, $\alpha H_2(\Pi \sigma) + H_4(\Pi^{\perp} \sigma)$ for $\Pi$ and $\Pi^{\perp}$ the projectors to the two subspaces of the direct sum. Here, the local algorithm will follow the degree-2 term and get stuck in that subspace, even if following the degree-4 term eventually attains a higher value.

The third example adds a random bilinear interaction between the two subspaces of the second example, so that the objective is $\alpha H_2(\Pi \sigma) + H_4(\Pi^{\perp} \sigma) + 2\beta H_{2,\mathrm{off}}(\Pi \sigma,\Pi^{\perp} \sigma)$.
This starts to make the SoS proof for the ``stepwise upper bound'' condition of \Cref{thm:main} fairly non-trivial, requiring a bound for all possible values of the ``state variables'' $\norm{\Pi \sigma_i}_2^2$ and $\norm{\Pi^{\perp} \sigma_i}_2^2$ over all steps $i$.
For this example, we give a proof sketch with an assumption about SoS proofs of a certain type of matrix concentration, rather than a full proof.

\section{Preliminaries}\label{sec:prelims}

\subsection{Sum-of-squares optimization: proofs to algorithms}
\subsubsection{Sum-of-squares certificates and pseudo-expectations}

We first define formally the notion of an ``axiom" $\cA$ which will be a collection of low-degree polynomial equalities and inequalities~\cite{barak2014sum,ma2016polynomial}.

\begin{definition}[SoS Axioms $\cA$]
    A set of \emph{axioms} $\cA$ is a collection of polynomial equalities and inequalities, formally written as,
    \[
        \cA = \left\{f_i = 0\right\}_{i=1}^M \cup\, \left\{g_j \geq 0\right\}_{j=1}^K\, ,
    \]
    where, $f_i, g_j \in \R_{\le d}[x_1\dots,x_n]$ are real-valued polynomials over some domain.
\end{definition}

Critical to the notion of efficient SoS algorithms are low-degree proofs of non-negativity of polynomials over the underlying semi-algebraic set $D$ specified by being exactly the set $x$ of points satisfying the axioms $\cA$.
A proof of non-negativity over $D$ will require the degree to be at least at large as that of the polynomial whose non-negativity is being shown.
\begin{definition}[SoS Non-Negativity Proof of Degree-$k$]
\label[Definition]{def:sos-proof}
    Given polynomials $p(x), q(x)$, a degree-$k$ SoS proof of non-negativity of $p - q$ is given as
    \[
        p(x) - q(x) = \sum_{i=1}^r (r_i(x))^2\prod_{j\in \alpha_i} g_{j}(x) + \sum_{i=1}^M s_i(x)f_i(x)\, ,
    \]
    for some choices of $r_i,s_i \in \R_{\le d}[x_1,\dots,x_n]$ and $\alpha_i \subseteq [K]$ and where the total degree of any term is at most $k$.
\end{definition}
For certain domains, such as over the boolean hypercube ($D = \{\pm 1\}^n$), a low-degree proof implies that $r \leq n^{k/2}$ and this further means that a low-degree SoS proof of negativity is efficiently verifiable, provided the coefficients can be represented with reasonable bit complexity~\cite{raghavendra2017bit,bulatov2022complexity}.
When a low-degree SoS proof of deg-$k$ exists, it is written as,
\[
\proves_k \{p \geq q\}\, .
\]
If the proof exists modulo some SoS axioms $\cA$, then it is written as,
\[
\cA \proves_k \{p \geq q\}\, .
\]

The set of low-degree SoS proofs is a convex set: a convex combination of sums of squares $g_i(x)^2$ is another sum of squares.
This implies the existence of a $\mathsf{poly}(n)$ time algorithm to output such a proof if it exists and has polynomial bit complexity.
This involves solving a particular $\mathsf{poly}(n)$-sized SDP. Specifically, this means that
\[
    \cA \proves_k \{p \geq q\} \implies O(MKn)^{O(k)}\text{ algorithm to output $r_i$, $s_i$, and $\alpha_i$ according to \Cref{def:sos-proof}.}
\]
By strong duality, an important consequence of the above is that whenever there is no low-degree SoS proof that $p \ge q$, there is a linear operator $\pE$ called a \emph{pseudo-expectation} that operates on polynomials, satisfies $\pE p(x) < \pE q(x)$, respects the axioms $\cA$, is normalized, and respects non-negativity of squared polynomials of not too-large degree.
This operator is the output of the underlying semi-definite program (SDP) whenever it fails to find a SoS proof.
\begin{definition}[Pseudo-distribution of degree $k$]
    A pseudo-distribution $\cD$ of degree $k$ with respect to SoS axioms $\cA$ corresponds to a linear functional,
    \[
        \pE_\cD: \R_{\le k}[x_1,\dots,x_n] \to \R\, ,
    \] 
    that satisfies the following propositions,
    \begin{description}
        \item \textbf{Scaling}: $\pE_\cD[1] = 1$.
        \item \textbf{Non-negativity}: $\pE_\cD[g^2] \geq 0\, ,\, \forall $polynomials $g$ with deg($g$)$\,\leq k/2$.
        \item \textbf{Satisfies Axioms}: The following holds for the polynomials in the axioms $\cA$:
        \begin{enumerate}
            \item $\pE_\cD[a_if_i] = 0$, $\forall a_i$ such that deg($a_if_i$)$\, \leq k$.
            \item $\pE_\cD[b^2_ig_i] \geq 0$, $\forall b_i$ such that deg($b_ig^2_i$)$\, \leq k$.
        \end{enumerate}
    \end{description}
We write $\cD \models \cA$ whenever the pseudo-expectation $\pE_{\cD}$ satisfies the axioms given by $\cA$.
\end{definition}
Note that the Riesz representation theorem implies that every pseudo-distribution $\cD$ can be associated with a linear functional $\pE_{\cD}$ that satisfies the above axioms.
The SDP formally outputs a pseudo-moment matrix of $\mathsf{deg}(k)$ (when it is feasible), and the linear functional acts on polynomials in ways that respect certain inequalities, such as Jensen's inequality~(\Cref{lem:pseudo-jensen-even}) and the Cauchy-Schwarz inequality.
\begin{lemma}[Pseudo Cauchy-Schwarz Inequality,~{\cite[21.~Theorem]{barak2016proofs}}]\label[Lemma]{lem:pseudo-cauchy-schwarz}
    Given a degree-$d$ pseudo-distribution $\cD$ and polynomials $p, q$ of degree $\le d/2$, the following holds:
    \begin{equation}
        \bigg(\pE_{\cD} pq\bigg)^2 \le \bigg(\pE_{\cD} p^2\bigg)\bigg(\pE_{\cD} q^2\bigg)\, .
    \end{equation}
\end{lemma}

\subsubsection{Rounding \& quadratic sampling lemma}
Critical to a relaxation of a problem to a space where it is tractable to solve, is the ability to ``round'' back the relaxed solution into the original solution space with (hopefully) minimal loss in solution quality.
For SoS relaxations, since pseudo-moment matrices are positive semi-definite (PSD), a natural candidate for rounding is to use these as moment matrices for some underlying distribution that is chosen appropriate to the problem at hand~\cite{barak2014rounding, barak2016proofs}.
The second-most basic rounding algorithm is to use the degree-1-and-2 pseudo-moments as the moments of a Gaussian.
The fact that this can always be done is known as the \emph{quadratic sampling lemma}.
\begin{lemma}[Quadratic Sampling Lemma, {\cite[23.~Lemma]{barak2016proofs}}]\label[Lemma]{lem:quadratic-sampling}
    For every degree-2 pseudo-distribution $\cD$, there exists a gaussian $\cN(\mu, \Sigma)$ over $\R^n$, such that,
    \begin{equation*}
        \mu = \pE_{x \sim \cD}[x]\, ,
    \end{equation*}
    and,
    \begin{equation*}
        \Sigma + \mu\mu^{\sT} = \pE_{x \sim \cD}[xx^\sT]\, .
    \end{equation*}
\end{lemma}

\subsection{Mixed spherical spin glasses}

\subsubsection{Replica symmetry breaking}
Full-Replica Symmetry Breaking is an imposition on the allowed overlaps in the solution space.
The condition implies that the normalized overlap of two configurations sampled i.i.d.~from the Gibbs distribution can take any value between $0$ and $1$.
More precisely, given $\sigma_1, \sigma_2 \sim G_{n, \beta}^{\otimes 2}$,
\[
    \textsf{supp}\left(\frac{1}{n}|\langle \sigma_1, \sigma_2\rangle |\right) = [0, q_\beta]\, ,
\]
where $q_\beta \to 1$ as $\beta \to \infty$.
In fact, this condition can be used to ascertain finer properties about the input hamiltonian $H_n$ when it is an instance of a Spherical Spin Glass, and the following proposition from Subag~\cite{subag2021following} outlines these (equivalent) properties.

\begin{proposition}[full-RSB Models{~\cite[Proposition 1]{subag2021following}}]\label[Proposition]{prop:frsb-equivalent}
The following conditions are equivalent:
\begin{enumerate}
    \item Given the mixture polynomial $\nu$ for the input hamiltonian $H_n$, $\nu''(q)^{-1/2}$ is concave for $q \in (0, 1]$.
    \item The Parisi Measure $\mu_\beta$ has \emph{full support} $[0, q_\beta]$ for any $\beta > 0$.
    \item The following can be said about the support of the Parisi measure based on the mixture polynomial $\mu$'s behavior at $0$:
    \begin{enumerate}
        \item If $\beta \leq \nu(0)''^{-1/2}$, then $\textsf{supp}(\mu_\beta) = \{0\}$.
        \item If $\beta > \nu(0)''^{-1/2}$, then $\textsf{supp}(\mu_\beta) = [0, q_\beta]$ where $q_\beta$ is the unique solution of
        \[
            \nu''(q)^{-1/2} = \beta(1 - q)\, .
        \]
        Additionally, the CDF of the overlap distribution $\mu_\beta$ has the following form,
        $$
        \mu_\beta(q) = \begin{cases}
            \frac{\mu'''(q)}{2\beta v''(q)^{3/2}}, & \text{ if $q \in [0, q_\beta)$} \\
            1, & \text{ if $q \in [q_\beta, 1]\, .$}
        \end{cases}
        $$
    \end{enumerate}
\end{enumerate}
\end{proposition}
The first condition of~\Cref{prop:frsb-equivalent} asserts the concavity of the negative square-root of the second derivative of the mixture polynomial under the fRSB assumption, and the third one allows us to have an explicit characterization of the CDF of the overlap distribution.

\subsubsection{The generalized TAP free energy}
There are two formulations of the limiting free energy density for typical instances of spherical hamiltonians, where the first is the so-called \emph{Crisanti-Sommers} representation which is the spherical analogue of the infamous Parisi-Variational Principle~\cite{crisanti1992sphericalp, talagrand2006free} and the second is an explicit expression proved by Subag~\cite{subag2018free} via a generalized TAP formula. Under the fRSB assumption, the generalized TAP free energy formula will be computable as an explicit function of the mixture polynomial, and so we work with the this representation as introduced by Subag~\cite{subag2018free}.

A result of Chen and Sen~\cite{chen2017parisi} in conjunction with the fact that the mixture polynomial continues to be concave under re-scaling as,
\[
    s \to \nu(qs)\, ,
\]
gives a clean expression for $E_*(q)$ (which is the ground state energy at radius $q$) as a function of the mixture polynomial in the fRSB regime.

\begin{lemma}[fRSB Ground State Energy, {\cite{chen2017parisi}}]\label[Lemma]{lem:chen-sen}
    The ground state energy of the mixed spherical spin glass under fRSB on the ball $\cB^n$ at radius $q$ is given as,
    \begin{align}
        E_*(q) = E_{\nabla^2}(q) = \int_0^{\ q} \nu''(t)^{1/2}dt\, .
    \end{align}
\end{lemma}

Restating the points above for $q = 1$ corresponds to the case that $\sigma \in \cS^{n-1}(\hspace{-1mm}\sqrt{n})$, and yields that
\begingroup
\allowdisplaybreaks
\begin{align}
    \lim_{n \to \infty}\frac{1}{n}\max_{\sigma \in \cS^{n-1}(\hspace{-0.5mm}\sqrt{n})} H_n(\sigma) \overset{a.s.}{=} E_*(1) = \int_{0}^1 \nu''(t)^{1/2}dt\, ,
\end{align}
\endgroup
provided the coefficients $\{\gamma_k\}_{k = 2}^{d_H}$ satisfy~\Cref{prop:frsb-equivalent}.

\subsubsection{Eigenvalue distribution of the Hessian}

As is shown by Subag~\cite{subag2021following}, under the fRSB condition~(\Cref{prop:frsb-equivalent}) it suffices to look at the top-$\delta n$ eigenvectors of the hessian ($\nabla^2 H$) at every point in the unit ball and follow iteratively perpendicular directions with infinitesmal updates.
The reason for the perpendicular updates corresponds to the ultrametric structure of the solution space~\cite{panchenko2013parisi}, but the reason for following the top eigenspace is built upon two key observations:
\begin{enumerate}
    \item The eigenvalues of the hessian ($\nabla^2 H$) are distributed as a mixture polynomial dependent scaling of the Wigner Law for $\textsf{GOE}(n-1)$ matrices,
    \[\nabla^2 H \sim \sqrt{n/(n-1)}\cdot \nu''(q)^{1/2}\textsf{GOE}(n-1)\, .\]
    Furthermore, with uniform control over every point $\sigma \in \cB^n$ that maximizes the generalized TAP Free Energy at distance $\norm{\sigma}^2_2$, there are (with high probability) at least $\delta n$ eigenvectors with eigenvalue at least $2\nu''(q)^{1/2} - \epsilon$~\cite[Lemma 3]{subag2021following}.
    \item It then follows by a Taylor expansion argument (given small contributions from the higher-order derivatives) and some algebra that the last iterate will have energy given by a truncated sum of the Chen and Sen formula (\Cref{lem:chen-sen}) and, under the fRSB assumption, this will be arbitrarily close to optimal. More specifically, the algorithm will output a value,
    \[
        \sum_{i=0}^{k}(1 - \epsilon(q_i))\cdot\nu''(q_i)^{1/2} \approx \int_{0}^{1}\nu''(q)^{1/2}dq\, ,
    \]
    where $q_i$ is the $\ell^2$-norm of the $i$-th iterate and $\epsilon(q_i)$ is some arbitrarily small number that depends only on choosing a sufficiently large (but constant) number of iterates.
\end{enumerate}
The most critical observation in the argumentation above is to notice that the maximum eigenvalue of $\nabla^2 H$ will be a scaling of the Wigner law by a factor that will exactly correspond to the optimum in the fRSB setting.
One can then construct an algorithm that starts at the origin, iteratively makes infinitesmal perpendicular updates that also maximize $\nabla^2 H$, and then outputs the last iterate.

\subsubsection{Hessian ascent}
\label{sec:subag}
We state here, for completeness, the pseudo-code of the algorithm proposed by Subag~\cite{subag2021following}. The $\perp$ sign denotes that the gradient $\nabla_\perp H(v_{i-1})$ and hessian $\nabla^2_\perp H(v_{i-1})$ are projected to be orthogonal to $v_{i-1}$ itself.
\begingroup
\allowdisplaybreaks
\begin{algorithm}
\begin{algorithmic}
\State $v_0 = (\frac{1}{\sqrt{k}},0,\dots,0)$.
\State $i = 1$.
\While{$i \in [k]$}
    \State Choose $x \perp v_{i-1}$ with $\norm{x} = 1$, such that,
    \begin{itemize}
        \item $\langle x, \nabla_{\perp} H(v_{i-1}) \rangle = 0$.
        \item $\langle x, \nabla^2_{\perp} H(v_{i-1})x \rangle \geq (1-\epsilon)2\nu''(v_i)$.
    \end{itemize}
    \State $v_i = v_{i-1} + \frac{1}{\sqrt{k}}x$.
    \State $i \leftarrow i + 1$.
\EndWhile
\State Output $v_k$.
\end{algorithmic}
\caption{Hessian-Based Optimization for Spherical Spin Glasses~\cite[Pg 12]{subag2021following}}\label{alg:subag}
\end{algorithm}
\endgroup

\subsubsection{The algorithmic threshold of Huang and Sellke~\texorpdfstring{\cite{huang2021tight}}{}}
\label{sec:alg-threshold}
The value achieved by the HES SoS relaxation is certifiably $\eps$-close to $\mathsf{ALG}_{Sp}$ for any $\eps > 0$. $\mathsf{ALG}_{Sp}$ is precisely characterized by~\cite[Proposition 2.2 and Equation 2.5]{huang2021tight} and given by a \emph{relaxed} version of the Crisanti-Sommers formula~\cite{crisanti1992sphericalp}.
When $h = 0$ (which is the case for the hamiltonian given in~\Cref{eq:spherical-ham}) the infimum is given by~\Cref{eq:alg-threshold}.

\begin{proposition}[Variational Algorithmic Threshold for Spherical Spin Glasses,~{\cite[Equations 1.3-2.5]{huang2021tight}}]\label[Proposition]{proposition:huang-alg}
    The value achieved by Subag's algorithm is given by the following variational principle,
    \begin{equation}\label{eq:alg-threshold}
        \mathsf{ALG}_{Sp} = \inf_{(B, \mu) \in \cV(\nu)} \left(\frac{1}{2}\left[\frac{h^2}{B_\mu(0)} + \int_{0}^1\left(\frac{\nu''(t)}{B_\mu(t)} + B_\mu(t)\right)dt\right]\right)\, ,
    \end{equation}
    where, $B_\mu(t) = B - \int_{t}^1 \nu''(q)\mu(q)dq$ and the optimizing set is the following collection of tuples,
    \begin{equation}\label{eq:functional-optimization-tuple}
        \cV(\nu) = \left\{(B, \mu) \in \R_{> 0}\,\times\cL([0,1])\,:\,B > \int_{0}^1 \nu''(t)\mu(t)dt\right\}\, ,
    \end{equation}
    with the set $\cL$ denoting a relaxation of the set of all CDFs over the interval $[0,1]$ specifically defined as,
    \begin{equation}\label{eq:relaxed-cdfs}
        \cL([0,1]) = \left\{\mu: [0,1) \to \R_{\ge 0}\,:\,\mu\,\text{is right-continuous},\int_{0}^1 \nu''(t)\mu(t)dt < \infty \right\}\, .
    \end{equation}
\end{proposition}

By~\cite[Corollary 1.8]{huang2021tight} the above variational representation can be simplified to an explicit formula. We restate the proposition, restricting to the case that $h = 0$ and the hamiltonian only has terms with with degree $\in [2,d_H]$.

\begin{proposition}[Explicit Algorithmic Threshold,~{\cite[Corollary 1.8]{huang2021tight}}]\label[Proposition]{prop:alg-explicit}
    The algorithmic threshold is given as a function of the mixture polynomial as,
    \begin{equation}\label{eq:alg-threshold-explicit}
        \mathsf{ALG}_{Sp} = q_1 \nu''(q_1)^{1/2} + \int_{q_1}^1 \nu''(q)^{1/2}dq\, ,
    \end{equation}
    where $q_1 \in [0,1)$ is the unique solution to,
    \begin{equation}\label{eq:q-condition}
        \nu'(q_1) = q_1\nu''(q_1)\, .
    \end{equation}
\end{proposition}
When $h= 0$ and $\gamma_1 = 0$, the only way to satisfy~\Cref{eq:q-condition} is when $q_1 = 0$, which immediately implies that for the family of hamiltonians detailed in~\Cref{eq:spherical-ham},
\begin{equation}\label{eq:frsb-alg-threshold}
    \mathsf{ALG}_{Sp} = \int_0^1 \nu''(q)^{1/2}dq\, .
\end{equation}
For the solution to~\Cref{eq:q-condition} in the case that $\gamma_1 = 0$ and $h = 0$, observe that
\begingroup
\allowdisplaybreaks
\begin{align*}
    \nu'(q_1) =  q_\nu''(q_1) &\implies \sum_{i=2}^{d_H}i\cdot\gamma^2_iq^{i-1}_1 = q_1\left(\sum_{i=2}^{d_H}i(i-1)\cdot\gamma^2_i q^{i-2}_1\right) \\
    &\implies \sum_{i=2}^{d_H} i\cdot\gamma^2_iq_1^{i-1} = \sum_{i=2}^{d_H}i(i-1)\cdot \gamma^2_i q^{i-1}_1 \\
    &\implies \sum_{i=2}^{d_H}i(2-i)\cdot\gamma^2_i q^{i-1}_1 = 0  \\
    &\implies \sum_{i=3}^{i}(2-i)\gamma^2_i q_1^{i-1} = 0\, ,
\end{align*}
\endgroup
and since $(2-i) < 0$ for $i > 2$, the above term can \emph{only} be made $0$ at $q_1 = 0$. This immediately justifies the algorithmic threshold in~\Cref{eq:frsb-alg-threshold} and one notices that this threshold coincides with the ground state energy under fRSB in the Chen and Sen formula~(\Cref{lem:chen-sen}). \\

\subsection{High-dimensional Gaussian distributions}
\subsubsection{Moments of Gaussian distributions}
We state below two elementary facts about the moments of Gaussian distributions.
\begin{lemma}[Moments of $\cN(0,1)$]\label[Lemma]{lem:standard-gaussian-moments}
    Given a random variable $X \sim \cN(0,1)$, its moments are given as follows,
    \[
        \E[X^m] =\begin{cases}
        0 & \text{if }m\text{ is odd}\, , \\
        \frac{m!}{(m/2)!}\frac{1}{2^m} & \text{if }m\text{ is even}\, .
        \end{cases}
    \]
\end{lemma}
The above lemma can be proved by a straightforward integration-by-parts argument.
\begin{lemma}[Isserlis' Lemma]\label[Lemma]{lem:isserlis}
    Given $(x_1,\dots,x_m) \sim \cN(0,\, \Sigma_m)$, the following holds,
    \begin{align}
        \E\left[\prod_{i=1}^m x_i\right] = \sum_{p \in P^2_m}\prod_{\{i, j\} \in p}\mathsf{Cov}[x_i, x_j]\, ,
    \end{align}
    where $P^2_m$ is the set of partitions of $\{1,\dots,m\}$ into pairs of indices.
\end{lemma}

\subsubsection{Hypercontractivity of low-degree polynomials of Gaussians}
It is a standard result by now that low-degree polynomial functions of gaussians are ``stable'' (in the $\ell^2$-norm) to small perturbations in the input.
This critically relies on a standard hyper-contractive inequality that applies to smooth functions of gaussians with bounded-degree.

\begin{lemma}[Hypercontractivity of LDPs of Gaussians,~{\cite[Proposition 3.3]{gamarnik2020low}}]\label[Lemma]{lem:hypercontractivity-gaussians}
    Given a function $f: \R^n \to \R$ with $f \in L_2(\R^n, \mu^{\ot n})$ where $\mu \sim \cN(0,1)$ and $f$ is a multi-variate polynomial with $\mathsf{deg}(f) \le d$, the following statement holds,
    \begin{equation}\label{eq:hypercontractivity-gaussians}
        \Pr_{x \sim \mu^{\ot n}}\left[f(x)^2 \ge t\E_{x \sim \mu^{\ot n}}\left[f(x)^2\right]\right] \le \exp\left(-\frac{dt^{1/d}}{3e}\right)\, ,
    \end{equation}
    for any $t \ge (6e)^d$.
\end{lemma}

The lemma above is used in the SoS certification of the Schatten norms of the Hessian ($\nabla^2 H$)~(\Cref{lem:hessian-schatten-norm}) to obtain strong concentration statements when bounding the so-called ``noise'' term.

\subsection{Real polynomials: linear representations and smooth approximation}\label{sec:matrix-reps}

\subsubsection{Multi-index notation}

A \emph{multi-index} $\alpha$ over $\{0\}\cup [n]$ is a tuple $(\alpha_1, \dots, \alpha_n)$ of non-negative integers, representing a multi-set of indices, such that
\[ x^{\alpha} = \prod_{i\in[n]} x_i^{\alpha_i} \]
for $x \in \R^n$.

We have the following notation:
\[ \alpha! = \prod_{i \in [n]} \alpha_i! \]
\[ \alpha!! = \prod_{i \in [n]} \alpha_i!! \]
\[|\alpha| = \sum_{i \in [n]} \alpha_i\]
\[\norm{\alpha}_{\infty} = \max_{i \in [n]} \alpha_i\]
\[ \alpha + \beta = (\alpha_1+\beta_1, \dots, \alpha_n + \beta_n) \]
\[ \alpha + k = (\alpha_1+k, \dots, \alpha_n + k)
\,.\]

We often treat multi-indices as if they were multi-sets containing those indices and vice-versa.

\subsubsection{Polynomials represented as coefficient matrices}
\label{sec:matrix-representations}

In order to leverage the full power of matrix Holder's inequality on the polynomials that appear in our proofs, we represent those polynomials and the corresponding pseudo-moments as matrices.
Recall that actual moments are a special case of pseudo-moments, so all of the definitions here can be applied using the moments of an actual distribution.

\begin{definition}
\label[Definition]{def:mat-representations}
    A matrix $M$ is a \emph{matrix representation} of a homogeneous polynomial $p(x)$ if there are $a$ and $b$ such that
    \[p(x) \equiv (x^{\ot a})^{\sT}\,M\,x^{\ot b}\,.\]

    Given a pseudo-expectation operator $\pE$ over a vector-valued variable $x$, a matrix $V$ is a \emph{matrix representation} of its degree-$d$ pseudo-moments if
    \[\sum_{\{i_1,\dots,i_a,j_1,\dots,j_b\} = S} V_{(i_1,\dots,i_a),(j_1,\dots,j_b)} = \pE_x \prod_{s \in S}x_{s}\,,\]
    for all multisets $S$ of indices.
\end{definition}

These definitions also apply to polynomials and pseudo-moments of more than one variable, albeit with an additional requirement to keep track of which tensor modes belong to which variables.

\begin{definition}
\label[Definition]{def:mat-representations-multivar-poly}
    A matrix $M$ is a \emph{matrix representation} of a homogeneous multivariate polynomial $p(u_1, \dots, u_k)$ \emph{with shape} $(A,B)$ if
    \[p(u_1, \dots, u_k) \equiv \left(\bigotimes_{x \in A}x\right)^{\sT}\,M\,\left(\bigotimes_{x \in B}x\right)\,,\]
    where $A$ and $B$ are tuples of the variables $u_1, \dots, u_k$ so that the number of times a particular variable $u_i$ shows up in $A$ and $B$ is equal to the degree of $p$ in $u_i$.
    Note that the tensor products here are taken in order of the tuples.
    \end{definition}
    
\begin{definition}
\label[Definition]{def:mat-representations-multivar-moments}
    Given a pseudo-expectation operator $\pE$ over several vector-valued variables $u_1, \dots, u_k$ and a multiset $S$ of those variables, a matrix $V$ is a \emph{matrix representation} of its $S$-pseudo-moments \emph{with shape} $(A,B)$ where $A$ and $B$ are tuples of the variables $u_1, \dots, u_k$ if, for all sequences $\gamma_1, \dots, \gamma_k$ of multi-indices with $|\gamma_{m}|$ equal to the multiplicity of $u_{m}$ in $S$,
    \[\sum_{\substack{(i_1,\dots,i_a) \in \operatorname{ind}(A,\alpha_1, \dots, \alpha_k)\\(j_1,\dots,j_b)\in\operatorname{ind}(B,\beta_1,\dots,\beta_k)\\\alpha_m + \beta_m = \gamma_m\,\forall m\in[k]}} V_{(i_1,\dots,i_a),(j_1,\dots,j_b)} = \pE_{u_1, \dots, u_k} \prod_{m \in [k]} u_m^{\gamma_m}\,,\]
    where $a$ is the number of elements in $A$, $b$ is the number of elements in $B$, $\operatorname{ind}(A, \alpha_1, \dots, \alpha_k)$ is the set of all tuples $(i_1,\dots,i_a)$ so that $i_k$ is a valid index into the $k$th variable in $A$ and $\alpha_m$ counts the number of occurrences of each index associated with $u_m$ among $(i_1,\dots,i_a)$, and $S$ is the multiset of all elements of $A$ and $B$ with multiplicity.
\end{definition}

A \emph{re-representation} of a matrix is another matrix representation of the same polynomial or pseudo-moments.

\begin{definition}
\label[Definition]{def:mode-symmetric-matrix}
    A matrix representation $M$ of a polynomial or a set of pseudo-moments is \emph{mode-symmetric} if
    \[ M_{(i_1, \dots, i_a),(i_{a+1}, \dots, i_{a+b})} = M_{(i_{\sigma(1)}, \dots, i_{\sigma(a)}),(i_{\sigma(a+1)}, \dots, i_{\sigma(a+b)})} \]
    for every permutation $\sigma \in S_{a+b}$.
\end{definition}

\begin{definition}
\label[Definition]{def:mode-symmetric-moments}
    A matrix representation $M$ of a multivariate polynomial or a set of $S$-pseudo-moments with shape $(A,B)$ is \emph{variable-mode-symmetric} if
    \[ M_{(i_1, \dots, i_a),(i_{a+1}, \dots, i_{a+b})} = M_{(i_{\sigma(1)}, \dots, i_{\sigma(a)}),(i_{\sigma(a+1)}, \dots, i_{\sigma(a+b)})} \]
    for every permutation $\sigma \in S_{a+b}$ that fixes the concatenation $C$ of $A$ and $B$.
    To be more precise,
    \[ C := (A_1, \dots, A_a, B_1, \dots, B_b)\]
    and the permutation $\sigma$ is required to satisfy
    \[ (C_1, \dots, C_{a+b}) = (C_{\sigma(1)}, \dots, C_{\sigma(a+b)})\,. \]
\end{definition}

We can now specify the pseudo-expectation of a polynomial as a Hilbert-Schmidt inner product between two matrix representations.

\begin{proposition}
\ifnum\fastmode=0
\label[Proposition]{prop:representation-product}
\fi
\ifnum\fastmode=1
\label[Proposition]{prop:representation-product}
\fi
    If $M$ is a matrix representation of a multivariate polynomial $p(u_1, \dots, u_k)$ with shape $(A,B)$ and $V$ is a matrix representation of the $S$-pseudo-moments of $\pE$ also with shape $(A,B)$, and \emph{at least one} of $M$ or $V$ is variable-mode-symmetric, then
    \[ \iprod{M,V} = \pE_{u_1, \dots, u_k} p(u_1, \dots, u_k) \,. \]
\end{proposition}

\begin{definition}
    A \emph{reshaping} of a matrix $M: (U_1 \ot \dots \ot U_a) \to (U_{a+1} \ot \dots \ot U_{a+b})$ is another matrix $M': (V_1 \ot \dots \ot V_{a'}) \to (V_{a'+1} \ot \dots \ot V_{a'+b'})$ with $a+b = a'+b'$ so that there is a permutation $\pi: [a+b] \to [a'+b']$ so that $U_k = V_{\pi(k)}$ for all $k \in [a+b]$ and
    \[ M'_{(i_{1},\dots,i_{a'}),(i_{a'+1},\dots,i_{a'+b'})}
    = M_{(i_{\pi(1)},\dots,i_{\pi(a)}),(i_{\pi(a+1)},\dots,i_{\pi(a+b)})}
    \, \]
    for all $i_1 \in V_{1}, \dots, i_{a'+b'} \in V_{a'+b'}$.

    A matrix representation is \emph{squaremost} if there is no reshaping of it whose row dimension is closer to its column dimension.
\end{definition}

\begin{definition}[Canonical Coefficient Matrices of Polynomials]\label[Definition]{def:canonical-representation-polynomial}
    Given a homogeneous degree-$d$ polynomial $p \in \R[x]$ in a single vector-valued variable $x \in \R^n$ written in the monomial basis as
    \[
        p(x) = \sum_{\alpha, |\alpha| = d}C_{\alpha} x^{\alpha}\, ,
    \]
    the following matrix $M_p:
(\R^{n})^{\ot \lfloor d/2\rfloor} \to (\R^{n})^{\ot \lceil d/2\rceil}$ constitutes its canonical coefficient representation:
    \begin{align}\label{eq:coeff-matrix-polynomial}
        M_p(u, v) := \frac{|S(u) + S(v)|}{(S(u) + S(v))!}C_{S(u) + S(v)}\, ,
    \end{align}
    where $u \in [n]^{\lfloor d/2\rfloor}$ and $v \in [n]^{\lceil d/2\rceil}$ and $S(u)$ is the multi-index which represents the same collection of indices as the tuple $u$.

    This is the unique square-most mode-symmetric matrix representation of $p$.
\end{definition}

\subsubsection{Gaussian cumulants: symmetric invariant representations and partition lattices}
\label{sec:gaussians}
Due to the structure of the coefficient matrix representation of polynomials of finite degree over $\R$, we will invoke projectors into the symmetric subspace ($\vee^d \R^n$).
As it turns out, this is also the natural class of subspaces to restrict to when evaluating the eigendecomposition of moment matrices of the multi-variate normal distribution.

\begin{definition}[Symmetric Subspace $\vee^d \R^n$]\label[Definition]{def:sym-subspace}
    The symmetric subspace $\vee^d \R^n$ is a subspace of $(\R^n)^{\ot d}$ which is stabilized under the action of the linear representation of the symmetric group $S_d$ as,
    \begin{equation}
        \vee^d \R^n := \left\{v \in (\R^n)^{\ot d}\,|\, P_d(\pi)\cdot v = v\,, \forall \pi \in S_d\right\}\, ,
    \end{equation}
    where the linear representation $P_d(\pi)$ of $\pi \in S_d$ is given by,
    \begin{equation}
        P_d(\pi) = \sum_{i_1,\dots,i_d \in [n-1]} \left(e_{i_{\pi^{-1}(1)}}\ot \dots \ot e_{i_{\pi^{-1}(d)}}\right)(e_{i_1}\ot \dots \ot e_{i_d})^\sT\, .
    \end{equation}
\end{definition}

The symmetric projector, as defined below, simply ``averages'' over all permutations of the multi-index of degrees.

\begin{definition}[Orthogonal projector $\mathsf{sym}$ into $\vee^d \R^n$]\label[Definition]{def:sym-proj}
    The projector $\mathsf{sym}: (\R^n)^{\ot d} \to \vee^d \R^n$ is a linear map defined as follows,
    \begin{equation}
        \mathsf{sym} := \frac{1}{|S_d|}\sum_{\pi \in S_d} P_d(\pi)\, .
    \end{equation}
\end{definition}
Proving orthogonality relies on showing that $\mathsf{sym}^\sT\mathsf{sym} = \mathsf{sym}$.
A standard argument for this can be found in, for instance,~\cite[Proposition 1]{harrow2013church}.
One can use Isserlis' lemma~(\Cref{lem:isserlis}) in conjunction with elementary properties of~\Cref{def:sym-proj} to observe the following fact about moments of tensor powers of multi-variate random vectors,
\begin{proposition}[Isserlis's formula for tensor moments,~{\cite[Eq.~14]{harrow2013church}}]\label[Proposition]{prop:tensor-vector-normal-moments}
    Given a multi-variate normal vector $x\sim \cN\left(0,\Id_n\right)$, the following holds for every $d \ge 2$,
    \begin{equation}\label{eq:tensor-vector-normal-moments}
        \E_{x\sim\cN(0,\Id)}\left[x^{\ot d}\right] = \frac{1}{n^{d}}\sum_{M \in \cM_{2d}} \left(\sum_{I = (i_1,\dots,i_{2d})\,, I\in S(M)}\left(e_{i_1}\ot\dots\ot e_{i_d}\right)\left(e_{i_{d+1}}\ot\dots\ot e_{i_{2d}}\right)^\sT\right)\, ,
    \end{equation}
    where $\cM_{2d}$ denotes the set of all possible perfect matchings over $2d$ vertices, $M$ is a perfect matching in$\cM_{2d}$, and $S(M)$ denotes a choice of indices $i_1,\dots,i_{2d} \in [n]^{2d}$ where $i_{j} = i_k$ for every $(j,k) \in M$.
\end{proposition}
The representation above can be used to compute an explicit eigendecomposition of $\E_{x\sim\cN(0,\Id_n)}\left[(xx^\sT)^{\ot d}\right]$.

Extending the above idea further, the actual entropy constraints constrain the cumulants in \emph{any} $\delta$n-dimensional subspace to fall off akin to a specific sub-gaussian distribution. This necessitates the introduction of a cumulant tensor.

\begin{definition}[Cumulant Tensor $\kappa^d$ of $X_1,\dots,X_n$]
\label[Definition]{def:cumulants}
    Given a degree parameter $d > 0$ and a jointly distributed family of random vectors $X_1,\dots,X_n$, the $d$-th cumulant tensor $\kappa^d(X_1,\dots,X_n) \in \vee^d \R^n$ is described as,
    \begin{equation}\label{eq:cumulant-degree-d}
        \kappa^d(X_1,\dots,X_n)_{i_1,\dots,i_d} := \sum_{\pi \in \cP(i_1,\dots,i_d)}(|\pi|-1)!(-1)^{|\pi|-1}\prod_{B \in b(\pi)}\E\left[\prod_{i \in B}X_i\right]\, ,
    \end{equation}
    where $\cP(i_1,\dots,i_d)$ denotes the lattice of partitions of the set $\{i_1,\dots,i_d\}$, $\pi$ denotes a particular partition from the lattice, $b(\pi)$ denotes the collection of disjoint subsets of $\pi$, and $|\pi|$ denotes the size of the partition.
\end{definition}
We will be interested in the regime where $d = O_\eps(1)$ and, therefore, the number of summands for each entry in the degree-$d$ cumulant matrix will be $O(d^d) = O_\eps(1)$. The number of entries in the tensor will be $O(n^d) = n^{O(1/\eps)} = \mathsf{poly}(n)$. We would like the cumulants (of large but constant degree) of the conditional pseudo-distributions in the program to match those of a multi-variate gaussian in $\delta$n-dimensions (scaled appropriately). The cumulants of a multi-variate normal distribution $\cN(\mathbf{0}, \Id_n)$ beyond the second degree are zero.

\begin{fact}[Cumulants of $\cN(\mathbf{0}, \Id_n)$]\label[Fact]{fact:gaussian-cums}
    The first and second cumulants of $\cN(\mathbf{0}, \Id_n)$ are,
    \begin{align}\label{eqs:first-second-cum-multi-normal}
        &\kappa^1(\cN(\mathbf{0}, \Id_n)) = \mathbf{0}\, ,\\
        &\kappa^2(\cN(\mathbf{0}, \Id_n)) = \Id_n\, ,
    \end{align}
    and the higher cumulants are zero.
\end{fact}

\begin{definition}
    Let $x^{(1)}$, \dots, $x^{(m)}$ be a sequence of vector-valued random variables.
    The $m$th \emph{joint cumulant tensor} $\kappa_m(x^{(1)}, \dots, x^{(m)})$ of these variables is the tensor whose entries are given by
    \[ \left(\kappa_m(x^{(1)}, \dots, x^{(m)})\right)_{i_1, \dots, i_m} :=
    \kappa_m(x^{(1)}_{i_1}, \dots, x^{(m)}_{i_m})\,. \]
    Suppose that $v_1, \dots, v_k$ is a sequence of vectors and that $\alpha: [m] \to [k]$ is a function so that $x^{(j)} = v_{\alpha(j)}$ for all $j \in [m]$.
    Then we may also denote 
    \[ \kappa_m(v^{\ot \alpha}) := \kappa_m(x^{(1)}, \dots, x^{(m)})\,. \]
\end{definition}

\subsubsection{Weierstrass approximation and Bernstein polynomials}
We briefly introduce the Bernstein polynomial and state a version of the Weierstrass Approximation theorem~\cite{de1959stone} using the Bernstein polynomials~\cite{bernstein_1912}.
Specifically, we state the precise quantification of the desired degree of the Bernstein polynomial required to approximate any Lipschitz continuous function $f$ on some compact interval $[a, b]$.

\begin{definition}[Lipschitz Continuity]
    A function $f: [a,b] \to \R$ is $C$-Lipschitz continuous if,
    \begin{align}\label{eq:C-lipschitz}
        \frac{|f(x) - f(y)|}{|x-y|} \leq C\, ,\, \forall x,\,y \in [a,b]\, .
    \end{align}
\end{definition}

\begin{definition}[Bernstein Polynomials]\label[Definition]{def:scalar-bernstein}
    Given a continuous function $f \in C([a,b], \R)$, its $k$-th Bernstein Polynomial $B_k$ is defined as,
    \begin{align}\label{eq:Bernstein-polynomial}
        B_k = \sum_{i=0}^k f\left(a +\frac{i(b-a)}{k}\right)\binom{k}{i}\frac{(x-a)^i(b-x)^{k-i}}{(b-a)^k}\, .
    \end{align}
\end{definition}

Bernstein polynomials can be used to approximate any $C$-Lipschitz continuous function with maximum error at most $\epsilon$

\begin{theorem}[Weierstrass approximation Theorem via Bernstein Polynomials]\label{thm:bernstein-approx}
    Given any continuous function $f: [a,b] \to \R$ that is $C$-lipschitz, it can be approximated with maximum error $\epsilon$ using a Bernstein polynomial $B_k(x)$ with degree $k = \frac{C(b-a)^3}{2\epsilon^3}$.
\end{theorem}

The ability of a low-degree (univariate) polynomial to approximate the step-function (a continuous version of which we approximate in~\Cref{cor:large-hs-correlation-bernstein-hessian}) is well established by, for instance,~\cite[Theorem 4.4 and Lemma 4.5]{diakonikolas2010bounded}.
Furthermore, this low-degree polynomial approximation permits a low-degree SoS proof, provided the inputs are over some compact interval $[a,b]$ where $a$ and $b$ are chosen independent of $n$~\cite[Theorem 3.2]{bafna2021playing}.
It is permissible to use the Bernstein polynomials as they are the \emph{unique} minimizers of the maximum error, and therefore, the low-degree SoS proof can be made to apply to them.
We require a version of~\Cref{thm:bernstein-approx} that is applied to \emph{matrix} (univariate) polynomials with respect to their Frobenius norm, and such a theorem can be proved to hold for matrices in a compact subspace of the underlying manifold of matrices on the Loewner order~(\Cref{prop:matrix-bernstein-linear-ramp}).

\subsection{Fourier analysis of Gaussian \texorpdfstring{$L^2$}{L2}-spaces}\label{sec:harmonic-analysis}
The design of the cumulant constraints (\Cref{eqs:the-program-intro-cum}) in the relaxation critically uses the insight that such a constraint implies an orthogonal (Fourier) decomposition (in the low-degree moments) for a current iterate as a function of the prior iterate. As the cumulant constraints are consistent with the low-degree cumulants of a gaussian, they allow the use of the Hermite polynomials as the Fourier basis. This decomposition, in conjunction with gaussian cumulant combinatorics and matrix reshaping arguments, permits a sum-of-squares proof to control the nuclear norm of the joint moments of the HES distribution. Introduced below are key definitions and properties of Hermite polynomials, including combinatorial formulae for algebraic manipulation of the Hermite polynomials.

\subsubsection{Univariate Hermite polynomials}
The standard gaussian measure on $(\R, \cB(\R))$ will be denoted as $dg = e^{-x^2/2}dx$ leading to a probability space $(\R, \cB(\R), dg)$.

\begin{definition}[Univariate Hermite polynomials]
    The univariate hermite polynomials are defined in the following equivalent ways:
    \begin{enumerate}
        \item \textbf{Appell Sequence} \begin{align*}
            &\He_0(x) = 1\, , \\
            &\partial_x \He_d(x) = d \He_{d-1}(x)\, .
            \end{align*}
        \item \textbf{Rodrigues' formula} \begin{align*}
            \He_d(x) e^{-x^2/2} = (-1)^d \partial^d_x e^{-x^2}\, .
        \end{align*}
        \item \textbf{Generating Function} \begin{align*}
            \sum_{d=0}^\infty \frac{a^d}{d!}\He_d(x) = e^{ax - a^2/2}\, .
        \end{align*}
    \end{enumerate}
\end{definition}
Over the set of all measurable functions of $X \sim \cN(0,1)$, the vector space denoted $L^2(X)$ denotes those that are square-integrable. With the choice of inner-product,
\[
    \langle f, g \rangle = \int_{-\infty}^\infty f(x)g(x) dg(x) \le_{\text{Cauchy-Schwarz}} \norm{f}_2\norm{g}_2\, ,
\]
under the gaussian measure, this space becomes a Hilbert space. It is not hard to check that,
\[
    \langle \He_d(x), \He_e(x) \rangle = d!\, \delta_{de}\, ,
\]
which implies that Hermites form an orthogonal basis for the set of all functions that are square-integrable. Since the Hermites form an orthogonal basis, any function $f \in L^2(dg)$ can be expressed as,
\[
    f(x) = \sum_{d=0}^\infty \frac{f_d}{d!} \He_d(x)\, ,
\]
where,
\[
    f_d = \langle f, \He_d \rangle = \E_{x \sim \cN(0,1)}[f(x)\He_d(x)] = \mathsf{Cov}(f, \He_d)\, .
\]
Since the vector space is complete, one can always increase the degree of the Hermites till one actually converges to the desired function (in the $L^2$-sense),
\[
    \lim_{d \to \infty} \left\langle f - \sum_{k=0}^d \frac{f_k}{k!} \He_k, f - \sum_{k=0}^d \frac{f_k}{k!} \He_k(x) \right\rangle = \lim_{d \to \infty}\int_{-\infty}^\infty \left(f(x) - \sum_{k=0}^d\frac{f_k}{k!}\He_k(x)\right)^2 dg(x)= 0\, .
\]

\subsubsection{Multivariate Hermite polynomials}
\label{sec:prelims-multivar-hermites}
Given a \emph{centered} gaussian process $\cX = \{ X_k \mid 1 \le k \le n\}$ with some covariance structure, the Hermite basis is given by the multivariate Hermite polynomials. A straightforward generalization of the generating function in the univariate case using the covariance matrix defined as,
\[
    \left\{C_{ij} = \E[X_iX_j]\right\}_{i,j=1}^n\, ,
\]
suffices to obtain the required definition.
\begin{definition}[Multivariate Hermite Polynomials~{\cite[Section 4.2]{terdik2021multivariate}}]
    The generating function for the multivariate Hermite polynomials of the space $(L^2(\cX), \cB(\R^n), d\vec{g})$ is defined as,
    \[
        \Psi(\cX, a_1,\dots,a_n) = e^{\langle \vec{a}, \cX \rangle - \frac{1}{2}\langle \cX, C \cX \rangle}\, ,
    \]
    where we are abusing notation so that $\cX$ is interpreted as a vector. The multivariate Hermite polynomials of $\cX$ are then,
    \begin{align}\label{eq:mv-hermite}
        &\He_{m_1,\dots,m_n}(X_1,\dots,X_n) = \He_m(Y_1,\dots,Y_m) \vert_{(Y_1,\dots,Y_{m_1}) \to X_1,\dots,(Y_{m_{n-1} + 1},\dots,Y_{m_n}) \to X_n} \\ \nonumber
        &= \partial_{a_1}\dots\partial_{a_m} \Psi(Y_1,\dots,Y_m,a_1,\dots,a_m)\, ,
    \end{align}
    where $m = m_1 + \dots + m_n$ and $Y$ is a multivariate Gaussian independent of $X$ with covariances $\mathsf{Cov}(Y_i,Y_j) = \mathsf{Cov}(X_{\alpha(i)}, X_{\alpha(j)})$ with $\alpha: [m] \to [n]$ mapping the indices as shown in~\autoref{eq:mv-hermite}.

    We may sometimes write $\He_{\beta}(X)$ with $\beta$ a multi-index over the coordinates of $X$ to mean $\He_{\beta_1, \dots, \beta_n}(X_1, \dots, X_n)$.
    We may also write $\He_{\beta}(X \mid C)$ to make explicit the covariance matrix that generates the polynomial.
\end{definition}

Listed below are some properties of multivariate Hermite polynomials. As the above mapping trick shows, we need not worry about non-distinct indices as we can do a simple substitution after taking the derivatives. Therefore, we only define the multivariate Hermite polynomials in the regime that the number of variables is equivalent to the degree.

\begin{proposition}[Properties of multivariate Hermite polynomials~{\cite[Section 4.2]{terdik2021multivariate}}]
    The multivariate Hermite polynomials $\{\He_{m_1,\dots,m_n}(X_1,\dots,X_n)\}_{m_1,\dots,m_n}$ satisfy the following: 
    \begin{enumerate}
        \item \textbf{Appell Sequence} \begin{align*}
            &\He_0(X_1,\dots,X_n) = 1\, ,\\
            &\partial_{X_k} \He_n(X_1,\dots,X_n) = \He_{n-1}(X_1,\dots,X_{k-1},X_{k+1},\dots,X_n)\, .
        \end{align*}
        \item \textbf{Recurrence Formula}
            \begin{align*}
                &\He_0(X_1,\dots,X_n) = 1\, , \He_1(X_i) = X_i\, \forall i \in [n]\, , \\
                &\He_n(X_1,\dots,X_n) = X_n \He_{n-1}(X_1,\dots,X_{n-1}) - \sum_{j=1}^{n-1}C_{jn}\He_{n-2}(X_1,\dots,X_{j-1},X_{j+1},\dots,X_{n-1})\, .\
            \end{align*}
            This formula makes it clear that the coefficients of the Hermites are themselves polynomials in the covariance matrix $C$, enabling the use of Hermite polynomials over indeterminate covariances in sum-of-square proofs.
        \item \textbf{Conditional Expectation}
        Let $Y$ denote a multivariate Gaussian with the same covariance as $X$ but independent of it. Then
            \begin{align*}
                \He_n(X_1,\dots,X_n) = \E\left[\prod_{j=1}^n (x_j + \iota Y_j)\right] \vert_{x_1 = X_1,\dots,x_n = X_n}\, .
            \end{align*}
        \item \textbf{Independent Factorization}
             \begin{align*}
                &\text{If } X_1,\dots, X_k \perp X_{k+1},\dots,X_n \text{, then,} \\
                &\He_n(X_1,\dots,X_n) = \He_k(X_1,\dots,X_k) \He_{n-k}(X_{k+1},\dots,X_n)\, .
            \end{align*}
        \item \textbf{Invariance under $S_n$}\\
        If $C = c\,I$ for some $c \in \R_+$, then
            \begin{align*}
                \forall \pi \in S_n\,,\, \He_n(X_1,\dots,X_n) = \He_n(X_{\pi(1)},\dots,X_{\pi(n)})\, .
            \end{align*}
        \item \textbf{Multilinearity}
            \begin{align*}
                \forall a, b \in \R,\,\He_{n+1}(X_1,\dots,X_n, aY + bZ) = a\He_{n+1}(X_1,\dots,X_n,Y) + b\He_{n+1}(X_1,\dots,X_n,Z)\, .
            \end{align*}
    \end{enumerate}
\end{proposition}
The proof for the properties is omitted, but can be derived relatively straightforwardly from the generating function definition of the Hermites using basic properties about the derivative operator and some elementary combinatorics.

In the multivariate case with arbitrary covariance structure, the Hermites are no longer necessarily orthogonal to each other under the $L^2(C)$ inner product $\iprod{f,g}_{L^2(C)} := \E_{x \sim \cN(0,C)} f(x)g(x)$.
However, they are still \emph{weakly orthogonal}, meaning that $\iprod{\He_{\alpha},\He_{\beta}}_{L^2(C)} = 0$ if $|\alpha| \ne |\beta|$, and it is still possible to decompose functions in $L^2(C)$ into what is known as a \emph{generalized polynomial chaos expansion}~\cite{rahman2017wiener}:
\begin{definition}
\label[Definition]{def:wiener-hermite-expansion}
    Let $f \in (L^2(C))^n$.
    Then the \emph{generalized Wiener-Hermite polynomial expansion} is given by $\hat{f}(\alpha)$ over the set of multi-indices $\alpha$, so that
\[f(X) = \sum_{\alpha} \hat{f}(\alpha)\, \He_{\alpha}(X \mid C)\,, \]
where $\hat{f}(\alpha)$ satisfy the linear system of equations: for each $\beta$,
\[ \sum_{\substack{\alpha}} \hat{f}(\alpha)\E_{X \sim \cN(0,C)}\He_{\alpha}(X \mid C)\,\He_{\beta}(X \mid C) = \E_{X \sim \cN(0,C)}f(X)\,\He_{\beta}(X \mid C) \,.\]
\end{definition}

We then have Bessel's inequality (Plancherel's theorem in the limit as the degree of the expansion goes to infinity):
\begin{theorem}[Bessel's inequality and Plancherel's theorem~{\cite[Proof of Theorem 14]{rahman2017wiener}}]~\linebreak
\label{thm:hermite-plancherel}
For every $d \in \N$,
\[ \E_{X \sim \cN(0,C)} \sum_{k \in \{0\}\cup [d]}\norm{\sum_{\substack{\alpha\\|\alpha| = k}}\hat{f}(\alpha)\,\He_{\alpha}(X \mid C)}_2^2 \le \E_{X \sim \cN(0,C)}\norm{f(X)}_2^2 \,.\]
Furthermore, for every $f \in L^2(C)$,
\[ \E_{X \sim \cN(0,C)} \sum_{k \in \N}\norm{\sum_{\substack{\alpha\\|\alpha| = k}}\hat{f}(\alpha)\,\He_{\alpha}(X \mid C)}_2^2 = \E_{X \sim \cN(0,C)}\norm{f(X)}_2^2 \,.\]
\end{theorem}

To observe how to encode these statements (as well as constraints that utilize these polynomials) as valid sum-of-squares statements in the underlying polynomial grammar~(\Cref{sec:sos-poly-basis}), refer to~\Cref{sec:hermites-and-cums-in-sos} and~\Cref{sec:appendix-sdp}.

Finally, there is an explicit expression for the inner products of the Hermites
\begin{lemma}[{\cite[Equation (4.11)]{terdik2021multivariate}}]
\label[Lemma]{lem:hermite-inner-prods}
Let $\alpha$ and $\beta$ be multi-indices over $[n]$.
Let $S(\gamma)$ for $\gamma$ a multi-index be the multi-set which contains each $i \in [n]$ with multiplicity $\gamma_i$.
Let $L(\gamma)$ be an arbitrary ordering of $S(\gamma)$ accessed through square brackets so that $L(\gamma)[1], \dots, L(\gamma)[|\gamma|]$ is a list of all elements (with multiplicity) of $S(\gamma)$.
If $|\alpha| = |\beta| = \ell$, then
    \[ \E_{X \sim \cN(0,C)}\He_{\alpha}(X \mid C)\,\He_{\beta}(X \mid C)
    = \sum_{\pi \in S_{\ell}} \prod_{i \in [\ell]} C_{L(\alpha)[i],L(\beta)[\pi(i)]} \,.\]
    If $|\alpha| \ne |\beta|$ then the inner product is $0$.
\end{lemma}

The \emph{cumulants} of a distribution are polynomial transforms of the moments that possess convenient additive properties. For instance, the cumulants of the sum of two independent random variables add, whereas this is not true for the moments (and knowing the moments would typically require doing some complicated Fourier convolution).

\subsection{Random walks and the semi-circle law}

We are interested in random walks over a complete hypergraph with i.i.d.~gaussian weights.
Specifically, we critically use the the set of ``Dyck-shaped'' closed walks to compute a ``signal'' term when certifying even trace powers of the Hessian ($\nabla^2 H$).
As is proven in~\Cref{lem:hessian-schatten-norm}, the signal term will be stronger than the corresponding ``noise'' term, yielding sums-of-squares proofs for both a lower bound and an upper bound for the Schatten norm of $\nabla^2 H$.

\subsubsection{Dyck-shaped walks}
\label{sec:dyck}

We introduce the notions of a Dyck-shaped walk, followed by the notation used to indicate Dyck paths and indices into specific nodes involved in the walk.
We borrow notation from \cite[Section 4.2]{chen2022cut}, and redefine it here for completeness.

\begin{definition}[$q$-Dyck Paths]
    The set $\cT_q$ of $q$-Dyck paths is the set of all zero-indexed length-$2q$ \emph{integer} walks $\tau$ that are:
    \begin{description}
        \item[Closed] $\tau_0 = \tau_{2q} = 0$.
        \item[Non-negative] $\tau_i \ge 0$ for all $i \in [2q]$.
        \item[Unit Monotone Indices] $|\tau_{i} - \tau_{i-1}| = 1$ for all $i \in [2q]$.
    \end{description}
\end{definition}
We think of a $q$-Dyck path $\tau = (\tau_0,\dots,\tau_{2q})$ as a set of indices that represent a valid Dyck sequence that can be \emph{instantiated} by a sequence $\omega$ of $(q+1)$ vertices in $[n]$.
\begin{definition}[$2q$-walk instantiation of a $q$-Dyck path]
    Given a $q$-Dyck path $\tau \in \cT_q$, a walk $w \in [n]^{2q}$ is an instantiation of $\tau$ specified by a zero-indexed vector $\omega \in [n]^{q+1}$ if:
    \begin{description}
        \item[Root Vertex] $w_0 = \omega_0$.
        \item[Ascending Transition] If $\tau_i - \tau_{i-1}$ is equal to 1, and this is the $j$th time it has been equal to $1$, then $w_i = \omega_j$.
        \item[Descending Transition] If $\tau_i - \tau_{i-1}$ is equal to $-1$, and $j$ is the largest value less than $i$ such that $\tau_i = \tau_j$, then $w_i = w_j$.
        \item[Unique Shape] $\omega_{i} \ne \omega_{i-1}$ for any $i$.
    \end{description}
\end{definition}

\begin{fact}
\label[Fact]{fact:dyck-catalan}
The number of $q$-Dyck paths is $|\cT_q| = C_q$, where $C_q$ is the $q$th Catalan number.
The number of instantiations of a path $\tau \in \cT_q$ is $n(n-1)^{q}$.
\end{fact}

\subsubsection{Catalan numbers and the semi-circle law}
It is a standard result in random matrix theory that the limiting spectral distribution of the $\mathsf{GOE}(n)$ ensemble is given by the semi-circle law.
Since we give low-degree SoS certificates of the Schatten norms of the hessian $\nabla^2 H(\sigma)$ converging to moments of this distribution, certain facts about the integral representation and asymptotic approximation are abundantly useful.

\begin{definition}[Catalan Numbers]
    For any $n > 0$, the $n$-th Catalan number is given by the following formula,
    \begin{equation}\label{eq:catalan-number-def}
        C_n = \frac{\binom{2n}{n}}{n+1}\, ,
    \end{equation}
    which represents the number of ways to have $n$ balanced parentheses in a sequence of $2n$ spots.
\end{definition}

We will use the fact that the $n$-th Catalan number $C_n$ is the $2n$-th moment of the semi-circle law to give a low-degree matrix SoS certifying closeness of a particular matrix Bernstein polynomial to the appropriate moment of the semi-circle law.

\begin{fact}[Integral Representation for Catalan Numbers]
    For any $n > 0$,
    \begin{equation}\label{eq:integral-rep-catalan}
        C_n = 4^n \int_{-1}^1 x^{2n}\frac{2}{\pi}\sqrt{1 - x^2}dx = \E_{x \sim \mu_{sc}}[x^{2n}]\, ,
    \end{equation}
    where $\mu_{sc} = \frac{2}{\pi}\sqrt{1 - x^2}$ represents the density function for the semi-circle law.
\end{fact}

\subsubsection{Truncated moments of the semi-circle law}
We require explicit representations for integrals that evaluate various quantities that are related to the $\mathsf{CDF}$ of the semi-circle law.

\begin{proposition}[Integrals related to $\mu_{sc}$]\label[Proposition]{prop:semi-circle-integrals}
    The following (real) integral identities hold (up to constants of integration):
    \begingroup
    \allowdisplaybreaks
    \begin{align*}
        &\int \frac{2}{\pi}\sqrt{1-x^2}dx = \mathsf{CDF}(x) = \frac{x\sqrt{1-x^2}}{\pi} + \frac{\arcsin{x}}{\pi} \, , \\
        &\int x\cdot\frac{2}{\pi}\sqrt{1-x^2}dx = -\frac{2(1-x^2)^{3/2}}{3\pi}\, .
    \end{align*}
    \endgroup
\end{proposition}

When computing the expected correlation between a Bernstein polynomial that approximates the ``ramp function'' and the value of an element chosen under the semi-circle measure, it will be crucial to evaluate multiple truncated moments on different compact intervals.
The integrals summarized above are used extensively therein~(\Cref{cor:large-hs-correlation-bernstein-hessian}).

\subsection{Graph matrices and Schatten norm bounds}

\subsubsection{Graph matrices}\label{sec:prelims-graph-matrices}
Invoking the trace power method on a \emph{random} matrix $M$ to bound its Schatten norm $\norm{M}_p$ is by now a mainstay in random matrix theory as well as proofs for certifying these norms for moments of certain random matrices~\cite{anderson2010introduction, ahn2016graph, rajendran2023concentration}. As such, we use the machinery of the ``graph matrix'' framework~\cite{ahn2016graph, rajendran2023concentration}. The definitions below capture the minimal concepts used in the building-block of arbitrary graph matrices. For a comprehensive treatment, see~\cite[Section 2 \& Section 7]{ahn2016graph}. This framework will prove critical in bounding the ``noise'' term to be of order $o($``signal''$)$ when certifying the $p$-Schatten norm of the Hessian. \\
Let $G = (V,E)$ denote some (possibly random) graph. We first introduce the notion of an ``index'' shape which is just a $m$-tuple of distinct variables that can be concretely fixed by a ``realization'' to be any $m$-tuple of vertices in $V(G)$.
\begin{definition}[Index Shape]
    An index shape $U = (u_1,\dots,u_m)$ is a tuple of \emph{distinct} variables, which induces a set of (unrealized) vertices $V(U) = \{u_1,\dots,u_m\}$.
\end{definition}

The notion of an index shape is sufficient to describe a shape $\alpha$, which can be seen as a representation of a subgraph of $G$ with a specific edge structure.

\begin{definition}[Shapes]\label[Definition]{def:shape}
    Given some graph $G = (V,E)$, a shape $\alpha$ is represented by a $4$-tuple $(U_\alpha, V_\alpha, W_\alpha, E(\alpha))$ where,
    \begin{itemize}
        \item $U_\alpha \subset V(G)$ is the set of ``left'' index shapes.
        \item $V_\alpha \subset V(G)$ is the set of ``right'' index shapes.
        \item $W_\alpha \subset V(G)$ satisfies $W_\alpha \cap (U_\alpha \cup V_\alpha) = \emptyset$, and is termed the set of ``middle'' index shapes.
        \item $E(\alpha)$ is a set of (hyper)edges that are induced by the index shapes in $U_\alpha \cup V_\alpha \cup W_\alpha$ where the vertices are \emph{distinct}.
    \end{itemize}
\end{definition}

Given any shape $\alpha$, we can get a ``realization'' of the shape by associating some fixed set of vertices to every index shape in $\alpha$. A realization of a shape $\alpha$ is termed a ``ribbon''.

\begin{definition}[Ribbons]\label[Definition]{def:realization}
    Given some graph $G = ([n],E)$, a shape $\alpha = (U_\alpha, V_\alpha, W_\alpha, E(\alpha))$ and a realization $\sigma: V(\alpha) = U_\alpha \cup V_\alpha \cup W_\alpha \hookrightarrow [n]$, a ribbon $R = \sigma(\alpha)$ is given as,
    \begin{equation}\label{eq:ribbon-from-shape}
        R = \sigma(\alpha) := (\sigma(U_\alpha), \sigma(V_\alpha), \sigma(W_\alpha), \sigma(E(\alpha)))\, ,
    \end{equation}
    where, $\sigma(U_\alpha) = (\sigma(u_1),\dots,\sigma(u_m))$ and $\sigma(V_\alpha)$ and $\sigma(W_\alpha)$ are defined similarly. The action on the edge set $\sigma(E(\alpha))$ is obtained by applying $\sigma$ to every $(u, v) \in E(\alpha)$ as ($\sigma(u), \sigma(v)$).
\end{definition}

A ``realization'' of the left index shape $U_\alpha$ and right index shape $V_\alpha$ indexes into a matrix $M_\alpha$~(\Cref{def:graph-matrix}) whose value is then given by a generalized Fourier coefficient that also depends on the middle index shape $W_\alpha$. This intuitively implies that a graph matrix for some shape $\alpha$ consists of specifying the entries of $M_\alpha$ at every possible ribbon induced by $\alpha$.

\begin{definition}[Graph Matrix]\label[Definition]{def:graph-matrix}
    Given a graph $G = (V = [n], E)$ and a shape $\alpha$, the graph matrix associated $M_\alpha$ is defined for any entry by indices $A, B \subset [n]$ that satisfy $|A| = |U_\alpha|$ and $|B| = |V_\alpha|$ as,
    \begin{align}\label{eq:graph-marix-entry}
        M_\alpha(A,B) := \sum_{\sigma, \sigma(U_\alpha) = A, \sigma(V_\alpha) = B} \chi_{\sigma(E(\alpha))}\, ,
    \end{align}
    where,
    \begin{align}
        \chi_{\sigma(E(\alpha))} := \prod_{e\hspace{.5mm}\in\hspace{.5mm} \sigma(E(\alpha))}\chi(e)\, ,
    \end{align}
    and $\chi(e) = 1$ if $e \in E(G)$ and $-1$ otherwise.
\end{definition}
Note that the coefficients $\chi(e): E(G) \to \{\pm 1\}$. Therefore, for a random graph $G = (V,E)$ chosen according to some product measure $\mu^{\otimes \tbinom{|V(G)|}{2}}$, they can be chosen to be the Fourier coefficients of an appropriate Fourier basis. \\

\begin{definition}[Generalized Matrix Index Pieces]
    A matrix index piece $A = ((a_1,\dots,a_m), p)$ corresponds to the monomial
    \begin{equation}\label{eq:gen-matrix-index-monomial}
        p_A := 
        \prod_{j=1}^m x^p_{
        a_j}\, ,
    \end{equation}
    where $x_{a_j}$ is the $a_j$-th variable
    and $p$ is the degree of every element in the monomial.
\end{definition}
Once again, for the purposes of this work, $p_A$ will be a monomial over $\{\sigma_i\}_{i=1}^n$ of degree $p \in [d_H - 2]$, where we implicitly use the convention that $\sigma_0^{p} = 1$. The monomial itself will be now be associated with a vertex in a hypergraph (along with the indices $\{i, j\}$ that represent the action of the derivative operator). Therefore, a matrix index $A$ will correspond to a set of \emph{distinct} monomials $\{p_{A_i}\}_{i \in [m]}$.
\begin{definition}[Generalized Matrix Index]
    A generalized matrix index $A = \{A_i\}_{i=1}^m$ is a collection of \emph{distinct} matrix index pieces, where distinctness is enforced by requiring that $p_i < p_j$.
\end{definition}
The size of the generalized matrix index $|A| = m$ is the number of index pieces it consists of, and the vertex set induced by it is $V(A) = \cup_{i \in [m]} V(A_i) = \cup_{i \in [m]}\cup_{j \in |A_i|} \{(a_j, i)\}$.
Together with a Fourier coefficient $\chi(e)$ defined for every hyper-edge $e \in E$, this yields a natural definition for generalized graph matrices $M_\alpha$ with the fourier basis for $\chi(e)$ being the rescaled hermite polynomials~\cite[Section 7.2.2.]{ahn2016graph}.

\subsubsection{Minimal weight vertex separators and Schatten norm bounds}
Before giving an explicit bound on the Schatten norms of generalized graph matrices $M_\alpha$, we need to introduce the notion of a \emph{minimum weight vertex separator}. To do this, we must first associate a weight $w(v)$ to every vertex $v \in V(\alpha)$ of a shape $\alpha$.
\begin{definition}[Vertex weight]
    Given a shape $\alpha$, the weight of a vertex $v \in V(\alpha)$ is
    $1$.
\end{definition}
Note that this immediately implies that $w(\sigma_0) = 0$ and $w(\sigma) = 1$.

\begin{definition}[Weight of a shape]
    For any $q > 0$ and shape $\alpha$, the weight of any subset of vertices $V \subseteq V(\alpha)$ is given by,
    \begin{align*}
        w(V) = \sum_{v \in V}w(v)\, .
    \end{align*}
\end{definition}

The above definitions suffice to precisely state the notion of a minimal weight vertex separator $S_{\min}$, which corresponds to a subset of vertices of $V(\alpha)$ that separate $U_\alpha$ and $V_\alpha$ with minimal weight.
\begin{definition}[Minimal weight vertex separator]
    A minimal weight vertex separator $S_{min}$ for a shape $\alpha$ is a subset $S_{min} \subset V(\alpha)$ that separates $U_\alpha$ and $V_\alpha$ with the condition that $w(S_{min}) \le w(S)$ for all other separators $S$.
\end{definition}

The bound below is critical in bounding the contribution of the ``noise'' term in the proof of~\Cref{lem:hessian-schatten-norm} in conjunction with \Cref{thm:high-entropy-nuclear-norm-bound}.

\begin{theorem}[Spectral Norm Bounds on Shapes, {\cite[Lemma 8.6]{ahn2016graph}}]
\label{thm:graph-matrix-norm-bound}
    Given a shape $\alpha$, the minimal weight vertex separator $S_{min}$ of $\alpha$, the following holds with probability at least $1-\eps$ when $M$ is a matrix of independent Gaussian entries and $M_{\alpha}$ denotes a generalized graph matrix where each entry is a sum over Fourier coefficients of ribbons induced by realizations $\sigma$:
    \begin{equation}\label{eq:bound-spectral-norm-partition-graph-matrix}
        \opnorm{M_{\alpha}} \le 2m^{m}n^{[w(V(\alpha)) - w(S_{min}) + w(S_{iso})]/2}\left(6e\left\lceil\frac{w(S_{min})\log n - \log \eps}{6(|V(\alpha) \setminus (U_{\alpha}\cap V_{\alpha})| + l(\alpha))}\right\rceil\right)^{l(\alpha)+|V(\alpha) \setminus (U_{\alpha}\cap V_{\alpha})|} ,
    \end{equation}
    where $l(\alpha)$ is the sum of labels of all hyperedges in $\alpha$, which we might also denote as $\deg(\alpha) := l(\alpha)$ and $m$ denotes the number of vertices that are not in $U_{\alpha} \cap V_{\alpha}$.
\end{theorem}

\section{High-Entropy Step Distributions as Solution Concepts in Sum-of-Squares Programs}\label{sec:sdp-proof-system}

\subsection{Matrix constraints: Operator norms of the conditional covariance}\label{sec:matrix-sos-system}

\newcommand{\dyad}[1]{{#1}{#1}^{\sT}}
\subsubsection{Sum-of-Squares proofs for matrix polynomials}

A HES spectral entropy constraint (\Cref{eqs:the-program-intro-opnorm}) will be of the form $\left\{\opnorm{\fE_i v_iv_i^{\sT}} \le c\right\}$ for some $c \in \R$.
We briefly introduce the framework for matrix polynomial constraints on pseudo-distributions and matrix sum-of-squares proofs, so as to reason about the encoding and consequences of such a constraint.

Matrix sum-of-squares proofs generalize those of scalar sum-of-squares~\cite{CIMPRIC201289}:
a scalar polynomial constraint $\{f \ge 0\}$ in the matrix sum-of-squares setting is notational shorthand for an $1\times 1$ matrix polynomial constraint $\{[f] \ge 0\}$.
Matrix polynomials used in constraints will always be symmetric as matrices.

\paragraph{Square positivity of matrices under pseudodistributions} We write $D \models_\ell \{ M_1(x) \succeq 0, \dots, M_m(x) \succeq 0\}$ if $\pE_D H\left(\bigotimes_{i \in S} M_i\right)H^\sT \succeq 0$, for every multiset $S$ consisting of elements of $[m]$ and every matrix polynomial $H$ that validly left-multiplies $\bigotimes_{i \in S} M_i$ such that $|S|\ell + 2\deg(H) \le d$.
The tensor product may be taken in any order since the different orderings are related to each other via conjugation by a permutation matrix.

\paragraph{Matrix sum-of-squares proof} For a $d\times d$ matrix $A(x)$ with entries that are polynomials in $x$,
we write $\{M_1(x) \succeq 0, \ldots, M_m(x) \succeq 0\} \vdash_{\ell} \{A\succeq 0\}$ if
there exists some $k$ and a sequence of matrix polynomials $(P_{S,j}(x))_{S \subseteq [r], j \in [m]}$ for multisets $S$
such that
\[A = \sum_{S \subseteq [m],j\in [k]} P_{S,j}\left(\bigotimes_{i \in S} M_i\right){P_{S,j}}^\sT\mcom\]
where each of the summands have degree at most $\ell$,
and where the empty tensor product is defined as the one-dimensional identity matrix.
Some special cases are $\vdash_{0} \{\Id \succeq 0\}$, $\{A \succeq 0\} \vdash_{\deg(PAP^\sT)} \{PAP^\sT \succeq 0\}$ for every matrix $P$ that validly conjugates $A$, and
$\{f \ge 0, M \succeq 0\} \vdash_{\deg(fM)} \{fM \succeq 0\}$.
The valid form of a matrix SoS proof is used in~\Cref{sec:moments-of-filtered-spaces}. 

Matrix sum-of-squares proofs satisfy all of the same inference rules as those of polynomial sum-of-squares proofs, except that the multiplication rule is replaced by tensor products:
\begin{align}
  &\frac{\cA \vdash_{\ell} \{A\succeq 0\},\qquad \cA \vdash_{\ell'} \{B\succeq 0\}}{\cA \vdash_{\ell+\ell'} \{A\otimes B \succeq 0\}} \mper
  \tag{tensor product}
\end{align}
This rule follows from the observation that $A = MFM^\sT$ and $B = NGN^\sT$ implies $A\otimes B = (M\otimes N)(F \otimes G)(M\otimes N)^\sT$.

The same complexity and soundness properties hold for the matrix case as in the scalar case, and furthermore the proofs are of identical substance.

\begin{theorem}[$\mathsf{poly}(n)$-time SDPs for bounded matrix SoS proofs]\label{thm:polyn-matrix-sdp}
  There exists a $(n+\card{\cA})^{O(d)}$-time algorithm that, given any explicitly bounded and satisfiable system $\cA$ of matrix polynomial constraints in $n$ variables, outputs (up to arbitrary accuracy) a degree-$d$ pseudo-distribution that satisfies $\cA$.
\end{theorem}
\Cref{thm:polyn-matrix-sdp} immediately implies that provided the degree of the scalar polynomials in the SoS HES hierarchy is constant and the matrices are polynomial in size, there exists a $\mathsf{poly}(n)$-sized SDP for the feasibility problem, just as in the scalar case.

\begin{lemma}[Soundness]
  If $D \models_{\ell} \cA$ for $D$ a pseudo-distribution and $\cA$ a system of matrix polynomial constraints, and there exists a matrix sum-of-squares proof $\cA \vdash_{\ell'} \cB$, then $D \models_{\ell\cdot\ell'} \cB$.
\end{lemma}

\subsubsection{Spectral norms of matrix formal variables}

We may sometimes write a constraint as $\left\{ \|M\| \le \alpha \right\}$, which is understood as equivalent to the pair of constraints $\left\{M \preceq \alpha \Id, -M\preceq \alpha \Id\right\}$. This gives the encoding for the constraint,
\[
    \opnorm{\E_i[v_iv_i^\sT]} \le \frac{1}{\delta n}\, ,
\]
in the SoS HES hierarchy, and this is explained further in~\Cref{sec:appendix-sdp}.

\subsection{Pseudo-expectations over moments of filtered probability spaces}\label{sec:moments-of-filtered-spaces}

We formalize filtered probability distributions in SoS-compatible language, laying the groundwork to define an SoS program to optimize over them.

The formal variables of the pseudo-distribution will be the \emph{moments of a filtered distribution} over $v_1, \dots, v_k$ rather than the variables $v_1, \dots, v_k$ themselves.
In this way, the program is a relaxation of optimizing a function over some family of distributions.

This allows us to refer to polynomials of moments in a sum-of-squares proof, rather than just polynomials of the variables (e.g. $\pE (\E x)^2(\E y)$ instead of $\pE x^2y$).
In this way, we can impose the constraint that the distribution over a pseudo-distribution variable $v_1$ is Gaussian without imposing any specific Gaussian (by constraining the pseudo-moments of the cumulants of $v_1$).
We impose a hierarchical structure on $v_1$, \dots, $v_k$, with a filtration so that $v_1$ is sampled first, then $v_2$, and so on.
Although we follow the structure of a filtration, the definitions in the following sections will be self-contained and will not require understanding the formalism of filtered probability spaces. Regardless, the definitions of the polynomial grammar provided here are consistent with the rules of conditional expectations over \emph{filtrated} $\sigma$-algebras, and this can be easily verified~\cite[Chapter 1.1]{ambrosio2012introduction}.

The raw variables of the SDP are the pseudo-expectations of ``monomials'' of this system, such as $\pE\E_1[v_1]$, $\pE\E_1[\E_2[v_2]]$, $\pE\E_1[v_1^3\E_2[v_2^2]]$, and $\pE\E_1[v_1^2(\E_2[v_2])^3(\E_2[v_2^4])^2]$.
The default constraints in the SDP are the usual non-negativity of square polynomials and the scaling $\pE 1 = 1$.

In a later subsection, we will introduce a syntactic extension of this proof system that allows us to reason about ``nearly pointwise'' properties of the random variables themselves, by reducing those properties to statements about their moments. 

We also introduce notation for \emph{conditional constraints}.
Let $p$ be any polynomial in $v_1$, \dots, $v_k$ and their expectations under $\E_{i}$, $\E_{i+1}$, \dots, $\E_{k}$.
Then $\{p \ge 0\}$ means (intuitively) that $p\ge 0$ for every value of $v_1, \dots, v_{i-1}$ in the support of the distribution. More formally (and axiomatically for sum-of-squares), it asserts that
\[
    \fE_1\dots\fE_{i-1}pq^2 \ge 0\, ,
\]
for any $q$ which is a polynomial expression in $v_1$, \dots, $v_{i-1}$ and/or their moments under $\fE_1, \dots, \fE_k$ which falls under the degree bound. This is more formally stated in~\Cref{sec:conditional-constraints}.

\subsubsection{Properties of the filtered expectation operator}
\label{sec:filtered-properties}

The expectation at time 1 of the expectation at time 2 of some quantity is the same thing as the expectation at time 1 of the same quantity.
Symbolically, this can be expressed as $\fE_1\fE_2 = \fE_1$.
When reasoning about the expectation of some expression at a particular level of the filtration, we will sometimes want that expression to include a reference to another expectation at the same or a preceding level of the filtration (for a canonical example of this, see \Cref{prop:var-nonneg}).
In this case, the preceding-level expectation is actually a constant with respect to the level we are reasoning about.
Concretely in terms of expectation operators, this means that $\fE_2\fE_1 = \fE_1$.
Both of these situations can be summed up in general as
\[ \fE_j\fE_i = \fE_{\min(i,j)}\,. \]

Each expectation operator is associative with itself in multiplication:
\[\fE_j (p \fE_j q) = (\fE_j p) (\fE_j q) \,.\]

Finally, the expectation operator $\fE_i$ commutes with multiplication by any polynomial of $v_1, \dots, v_{i-1}$:
\[\fE_i q(v_1,\dots,v_{i-1})p = q(v_1,\dots,v_{i-1})\fE_i p\qquad \forall q\in\R[v_1,\dots,v_{i-1}] \,.\]

\subsubsection{Grammar}
\label{sec:grammar}
\newcommand{\cfg}{\mathsf{cfg}}
\newcommand{\ecfg}{\mathsf{ecfg}}
\newcommand{\mcfg}{\mathsf{mcfg}}

The set of allowable expressions is captured by a (redundant) context-free grammar. The language generated is then \emph{truncated} up to a degree bound, and \emph{quotiented} by a set of inequalities that split it into an equivalence class. These expressions can be expressed in a monomial basis, with the equivalences being encoded as inequalities in the SDP.

\begin{definition}[Context-free grammar]\label{def:cfg}
    The context-free grammar of allowable expressions is defined as:
    \begin{align*}
        \cfg[0] &:= \poly(v_1, \dots, v_k)
        \\\ecfg[i] &:= \fE_1 \cfg[i] \mid \dots \mid \fE_k \cfg[i] \mid \cfg[i]
        \\\cfg[i] &:= \poly(\ecfg[i-1])\,.
    \end{align*}
\end{definition}
Here, $\poly(v_1, \dots, v_k)$ refers to the set of all scalar multi-variate polynomials in the variables $v_1, \dots, v_k$, where $v_i \in \R^n$ for every $i \in [k]$.
The set of expressions $\ecfg[i]$ allows you to optionally put any of the $k$ expectation operators in front of any expression in $\cfg[i]$. 
The set of expressions $\poly(\ecfg[i-1])$ is the set of polynomials in any of the expressions in $\ecfg[i-1]$.

The term $\cfg$ without a number refers to $\cfg[i]$ for some large enough $i$ depending on the context (usually the degree of the pseudo-distribution).

We define formally the notion of matrix polynomial used in the proofs.
\begin{definition}[Matrix polynomial]
    In compliance with the matrix SoS proof system introduce in~\Cref{sec:matrix-sos-system}, a matrix polynomial at level $i \in [k]$ is given by a matrix $M_i \in \R^{n^{m_1} \times n^{m_2}}$ such that,
    \[
        (M_i)_{j,k} = p_{j,k}\, ,
    \]
    where $p_{j,k} \in \cfg[i]$.
\end{definition}  

The set of all valid matrix polynomials for $\cfg[i]$ is denoted as $\mcfg[i]$. If a matrix $M_i \in \mcfg[i]$ has entries that are filtrated expectations of polynomials, such that it can be decomposed into a sum of rank-1 tensors, then such matrices may be equivalently ``flattened'' into tensors in the obvious way. This fact is often used during the analysis in~\Cref{sec:hods}, and is consistent with an equivalent representation followed under the equivalences implied by~\Cref{def:equivalences}.

The matrix polynomials (up to a degree bound) are, as a result, the set of expressions that (up to a degree bound) can be reasoned about in the SoS proof system.
\subsubsection{Equivalences}
\label{sec:equivalences}
The set of polynomials generated by $\cfg[i]$ is to be treated as an equivalence class of polynomials under certain algebraic structure. This algebraic structure is reflected as a mixture of the structure generated by a \emph{filtration} and an \emph{algebra} (over a vector space).
\begin{definition}[Equivalence classes of the $\cfg$-algebra]\label[Definition]{def:equivalences}
    Let $i \in [k]$ be a level of filtration. Then, the filtered expectation $\fE_i$ obeys the following equivalences when acting on scalar multi-variate polynomials:
    \begin{description}
        \item [Linearity of expectation] For any $p, q \in \cfg[k]$, $\fE_i\left(p + q\right) = \fE_i p + \fE_i q$.
        \item [Constant expectations are idempotent] For any $p \in \mcfg$ and $p \in \R^{m_1 \times m_2}$, $\fE_i p = p$.
        \item [Ring laws for polynomials] For any $p, q, r \in \cfg[k]$, the following is true:
            \[
                \fE_i\left(p(q + r)\right) = \fE_i(pq) + \fE_i(pr)\,,\, \fE_i((p+q)r) = \fE_i (pr) + \fE_i(qr)\, .
            \]
            \[
                \fE_i\left(pq\right) = \fE_i(qp)\, .
            \]
            \[
                \fE_i[(pq)r] = \fE_i[p(qr)]\, .
            \]
        \item [Towering property of filtrations] Given two filtrated expectations $\fE_i$ and $\fE_j$, the following holds,
            \[
                \fE_i\fE_j(p) = \fE_j\fE_i(p) = \fE_{\min(i,j)}(p)\, ,
            \]
            for any $p \in \cfg[l]$.        
        \item [Filtrated associativity] Given two $p, q \in \cfg[i]$, the following holds,
            \[
                \fE_j (p \fE_j q) = (\fE_j p) (\fE_j q) \,.
            \]
        \item [Filtrated commutativity] Given a polynomial $p \in \cfg[i]$ and $q \in \cfg[i]$, such that $q \in \R[v_1,\dots,v_{i-1}]$, the following holds,
            \[
                \fE_i q(v_1,\dots,v_{i-1})p = q(v_1,\dots,v_{i-1})\fE_i p\, .
            \]
    \end{description}
        
\end{definition}

We will also allow vector, matrix, and tensor variables in the grammar, equipping it with a tensor product $\otimes$ with the usual properties of a tensor algebra.

It will always be possible to simplify any expression in $\E_1 \cfg[i]$ for any $i$ into a sum of ``monomials''.
The SoS program will optimize the pseudo-expectations of these monomials.

\subsubsection{Monomial basis \& truncated filtered-polynomial algebras}\label{sec:sos-poly-basis}

\begin{definition}[Monomial basis (modulo equivalences)]
\label[Definition]{def:monomial-basis}
Let $v_1$, ..., $v_k$ be vector-valued variables in $\R^n$.
\begin{itemize}
    \item Let
    \[\mon{k} := \{1, v_k, v_k^{\ot 2}, \dots\}\,.\]
    \item For $i \in [k-1]$, let 
    \[\mon{i} := \left\{1, v_i, v_i^{\ot 2}, \dots\right\} \otimes \left(\{1\} \cup \left(\fE_{i+1}\mon{i+1}\right) \cup \left(\fE_{i+1}\mon{i+1}\right)^{\ot 2} \cup \dots\right)\,,\]
    where the operations $\ot$, $\fE_{i+1}$, and $(\cdot)^{\ot 2}$ are applied elementwise when operating on sets.
    For example, $\{a,b\} \ot \{x,y\} = \{a\ot x,a\ot y,b\ot x,b\ot y\}$ and $(\fE_{i+1} \{x,y\})^{\ot 2} = \{(\fE_{i+1}x)^{\ot 2}, (\fE_{i+1}y)^{\ot 2}, (\fE_{i+1}x)\ot(\fE_{i+1}y), (\fE_{i+1}y)\ot(\fE_{i+1}x)\}$.
    \item Let 
    \[\mon{0} := \fE_1\mon{1}\,.\]
    \item Let $\mathtt{closure}(\mon{0})$ be the formal real vector space generated by all terms in $\mon{0}$.
\end{itemize}
\end{definition}

\begin{definition}
    \label[Definition]{def:expression-degree}
     The \emph{degree} of an expression satisfies $\deg (\fE_{i+1} p) = \deg p$, and otherwise accumulates as usual with respect to addition, multiplication, and exponentiation.
    Let $\cfg[i]^{\le d}$ and $\mon{0}^{\le d}$ be the intersections of $\cfg[i]$ and $\mon{0}$ with the set of expressions of degree at most $d$.
\end{definition}

\begin{proposition}
    Every expression in $\fE_1\cfg[i]$ is equivalent under the relations in \Cref{sec:equivalences} to some element of $\mathtt{closure}(\mon{0})$ of the same degree.
\end{proposition}

This reduction to the monomial basis allows us to define a pseudo-expectation operator.

\begin{definition}[Pseudo-Expectation of Filtered Expectations]\label[Definition]{def:filtered-pseudo}
    A degree-$d$ pseudo-expectation of filtered distributions $\pE$ is a linear functional on $\fE_1\cfg[d]^{\le d}$,
    \[
        \pE: \fE_1\cfg[d]^{\le d} \to \R\, , 
    \]
    that satisfies the following constraints:
    \begin{description}
        \item[Scaling]: $\pE \fE_1 1 = 1$.
        \item[Equivalances]:
        $\pE \fE_1 p = \pE \fE_1 q$ for all expressions $p$ and $q$ that are equivalent to each other under the equivalence relations in \Cref{sec:equivalences}.
        \item[Non-negativity of squares]: $\pE \fE_1 p^2 \ge 0$ for all $i \in [k]$ and $p \in \cfg[d/2]^{\le d/2}$.
    \end{description}
\end{definition}

These pseudo-expectations can be optimized by semi-definite programs since the non-negativity of squares condition can be stated as the positive semi-definiteness of the block matrix whose blocks are $\pE pq^{\sT}$ for each $p,q \in \mon{0}^{\le d}$ up to the degree bound.

\subsubsection{Filtered conditional inequalities in SoS proofs}\label{sec:conditional-constraints}

For $p \in \cfg[d]$ (or $\mcfg[d]$), we write 
\[ \left\{p \succeq 0\right\} \]
to mean 
\[ \fE_1 qpq^{\sT} \succeq 0  \]
and furthermore 
\[ \pE\fE_1 qpq^{\sT} \succeq 0  \]
for all expressions $q \in \cfg[d/2]$ fitting under the degree bound of the pseudo-expectation.

This can be implemented as a constraint in a SDP that searches for a pseudo-distribution by writing the block matrix with blocks $\pE rps^{\sT}$ for each $r,s \in \mon{0}^{\le d}$ up to the degree bound.

\begin{proposition}[Idempotence of $\E_i$ for conditional constraints]
\label[Proposition]{prop:add-conditioning}
For every expression $p \in \cfg[k]$ and every $i$,
\[\left\{p \succeq 0\right\} \;\proves\; 
\left\{\fE_i p \succeq 0\right\}\,.
\]
\end{proposition}
\begin{proof}
    Note that,
    \begin{align*}
        &\left\{p \succeq 0\right\} \equiv \fE_1 pqp^\sT \succeq 0\, .
    \end{align*}   
    Similarly,
    \[  
        \left\{\fE_i p \succeq 0\right\} \equiv \fE_1\fE_i qpq^\sT \succeq 0\, ,
    \]
    which is, by the towering property of filtrations, equivalent to
    \[
        \fE_1 \fE_i qpq^\sT = \fE_{\min(1,i)} qpq^\sT = \fE_1 qpq^\sT \succeq 0\, . 
    \]
\end{proof}

\subsubsection{Jensen's inequality and non-negativity of variance}
One of the core proofs that is expressible in this system is the fact that the variance is always nonnegative:
\begin{proposition}
\label[Proposition]{prop:var-nonneg}
\[ \proves \left\{\fE_j v_jv_j^{\sT} - (\fE_j v_j)(\fE_j v_j)^{\sT} \succeq 0 \right\}\,. \]
\end{proposition}
\begin{proof}
In the proof below, all inequalities are to be interpreted conditionally (as in \Cref{sec:conditional-constraints}) and we omit the $\{ \}$ notation for brevity. \\
By \Cref{fact:sos-proof}, we have $(v_j - \fE_j v_j)(v_j - \fE_j v_j)^{\sT} \succeq 0$ and by \Cref{prop:add-conditioning}, $\fE_j (v_j - \fE_j v_j)(v_j - \fE_j v_j)^{\sT} \succeq 0$, so
\begin{align*}
0 &\preceq \fE_j [(v_j - \fE_j v_j)(v_j - \fE_j v_j)^{\sT}]
\\&= \fE_j [v_j(v_j - \fE_j v_j)^{\sT} - (\fE_j v_j)(v_j - \fE_j v_j)^{\sT}]
\\&= \fE_j v_j(v_j - \fE_j v_j)^{\sT} - \fE_j(\fE_j v_j)(v_j - \fE_j v_j)^{\sT}
\\&= \fE_j v_j(v_j - \fE_j v_j)^{\sT} - (\fE_j v_j)(\fE_j(v_j - \fE_j v_j))^{\sT}
\\&= \fE_j v_j(v_j - \fE_j v_j)^{\sT} - (\fE_j v_j)(0)^{\sT}
\\&= \fE_j [v_jv_j^{\sT} - v_j(\fE_j v_j)^{\sT}]
\\&= \fE_j v_jv_j^{\sT} - \fE_j v_j(\fE_j v_j)^{\sT}
\\&= \fE_j v_jv_j^{\sT} - (\fE_j v_j)(\fE_j v_j)^{\sT}
\,.
\end{align*}
\end{proof}

\begin{proposition}[Jensen's inequality for expectations of squares]
\label[Proposition]{prop:expectation-jensens}
\[ \proves\left\{(\fE_j p)^2 \le \fE_j p^2\right\}\,.  \]
\end{proposition}
\begin{proof}
The same proof as for \Cref{prop:var-nonneg} works, just substituting $p$ for $v_j$.
\end{proof}

\subsubsection{Hermite coefficients}
\label{sec:encoding-hermite-decomps}

We will also introduce SoS variables $ \mathscr{W}_i[A](\alpha)$ for the Wiener-Hermite expansions (\Cref{def:wiener-hermite-expansion}) of each atom in $A \in \fE_{i+1}\mon{i+1}$ with respect to the variable $v_i$, so that for all $\beta$ up to the degree limit of the program,
\[\left\{\,\fE_i\left(A - \sum_{\alpha: |\alpha| = |\beta|} \mathscr{W}_i[A](\alpha)\, \He_{\alpha}(v_i \mid \fE_{i}v_iv_i^{\sT})\right)\,\He_{\beta}(v_i \mid \fE_{i}v_iv_i^{\sT}) = 0\,\right\}\]
is a constraint of the SoS program.
Said Wiener-Hermite expansions also have Wiener-Hermite expansions of their own at one level lower of the filtration: $\mathscr{W}_{i-1}[\mathscr{W}_i[A](\beta)](\alpha)$, for example.

\subsubsection{Sum-of-squares proofs over filtered polynomials in the HES hierarchy}
The proof system which uses the monomial basis defined in~\Cref{def:monomial-basis} (up to a certain degree) is a sum-of-squares proof system. This proof system certifies the non-negativity of polynomials in $\cfg[i]$ via the familiar notion of a sum-of-squares proof, but where the polynomials may not just be over the variables $\{v_i\}_{i=1}^k$, but also over moments up to various levels in the filtration. 

\begin{fact}[Sum-of-squares proofs over $\cfg$ and $\mcfg$]\label[Fact]{fact:sos-proof}
    Let,
    \[
    	\cA = \left\{ \{ p_i \ge 0\}, \{ q_j = 0\} \right\}_{i, j=1}^{m_1,m_2}\, ,
    \]
    with $p_i \in \cfg[f_i]$, $q_j \in \cfg[g_j]$ for every $i, j \in [m_1] \times [m_2]$ and $f_i, g_j \in [k]^2$. A degree-$d$ sum-of-squares proof certifying the non-negativity of $h \in \cfg[i]$ is given as,
    \[
        h = \sum_i s^2_i + \sum_j a_jq_j + \sum_k b_k^2 p_k\, ,     
    \]
    where $\{s_i\}_{i \in \mathsf{poly}(n)}$ have degree $\le d/2$, $\mathsf{deg}(a_jq_j) \le d$ and $\mathsf{deg}(b_k^2p_k) \le d$.
    More generally, for matrix polynomials $h \in \mcfg[i]$, a sum-of-squares proof is given as,
    \[
        h = \sum_i s_is_i^\sT + \sum_j r_jq_j + \sum_k t_k \bigotimes_{j=1}^r p_j t_k^\sT\, ,   
    \]
    with degree bounds as defined in~\Cref{sec:matrix-sos-system} and a semi-algebraic set $\cA'$ where $\{p_j \succeq 0\}_j$ and $\{q_j = 0\}_j$.
\end{fact}

A critical fact to note is that while the monomial basis technically allows for proofs \emph{pointwise} over the variables $\{ v_i\}_{i=1}^k$, the HES SoS hierarchy introduced in~\Cref{eq:the-program} certifies a polynomial that \emph{expects} over the entire filtration. Furthermore, all the constraints in the hierarchy are over \emph{expected} polynomials at various levels of the filtration.

\subsubsection{Insufficiency of joint expectation operator and equivalence with replicas}
\label{sec:replica-equivalance}

It may be tempting to attempt to simplify the above system and express everything in terms of the joint moments of $v_1$, \dots, $v_k$, so that there is only one expectation operator to deal with.

However, doing so discards the power to express what would otherwise be low-degree information: for example, $\E_1(v_1(\E_2 v_2)^2)$ cannot be written as any product of bounded-degree joint moments of $v_1$ and $v_2$.

One may see this by introducing a \emph{replica} of $v_2$: another variable $u$ which has the same joint moments with $v_1$ as $v_2$ does.
Then one can in fact write these expectations-of-polynomials-of-expectations as joint expectations in the replicas: $\E_1(v_1(\E_2 v_2)^2) = \E_1\E_2\E_u v_1v_2u$.
Then the dependency graph between $v_1$, $v_2$, and $u$ is a triangle, whereas if we expressed a nontrivial product of joint moments of $v_1$ and $v_2$ in terms of replicas, the resulting dependency graph would always be disconnected.

\subsubsection{Soundness for actual distributions}
\label{sec:actual-conditional}

Suppose that $\hat{v}_1, \dots, \hat{v}_k$ are drawn from an actual distribution of high-entropy steps that satisfies
\[ \forall \hat{v}_1, \dots, \hat{v}_k \in \supp(\cD) \mathrel{.} \hat{p} \succeq 0\,, \]
where $\hat{p}$ is an expression from $\widehat{\cfg}[k]$, which is the same as $\cfg[k]$ but with the atoms $v_1, \dots, v_k$ replaced with $\hat{v}_1, \dots, \hat{v}_k$ and $\fE_1$, $\fE_2$, \dots, $\fE_k$ replaced with $\E_{\hat{v}_1,\dots,\hat{v}_k}$, $\E_{\hat{v}_2,\dots,\hat{v}_k}$, \dots, $\E_{\hat{v}_k}$ respectively.

Then the pseudo-distribution defined by $\pE\fE_1 q := \E_{\hat{v}_1,\dots,\hat{v}_k} \hat{q}$ for all $q \in \cfg[k]$ satisfies the corresponding pseudo-distribution constraint
\[ \left\{p \succeq 0\right\}\,, \]
since for all $q \in \cfg[k]$\,
\[ 
\pE\fE_1 [qpq^{\sT}]
= \E_{\hat{v}_1,\dots,\hat{v}_k} \hat{q}\hat{p}\hat{q}^{\sT}
= \int \hat{q}\hat{p}\hat{q}^{\sT} d\mu(\hat{v}_1,\dots,\hat{v}_k)
\succeq 0\,,
\]
since the integrand is always PSD.

\subsubsection{Notational variants}
We may sometimes write $\E_{\sigma}$ or $\E_{\sigma_i}$ to mean $\fE_1$ with $\sigma = \sigma_i = (v_1 + \dots + v_{i})/\sqrt{k}$, with the intended meaning that we average over all past iterates.

We may also write $\E_{v}$ or $\E_{v_i}$ to mean $\fE_i$, meaning that we average over only the $i$th iterate, conditioned on all the preceding ones.

Sometimes we write conditional expectation notation $\E[y \mid v_i]\,x$ to mean $(\fE_{i+1} y)x$.

\subsection{Almost-pointwise sum-of-squares proofs}
An important concept in this proof system that is used extensively in~\Cref{sec:hods} is the existence of SoS proofs that work only up to some \emph{bounded} degree. This degree depends on the problem being analyzed, and is reflected in the degree of the largest monomial in the grammar~$\cfg$[i].

\subsubsection{Sub-Weibull algebras}
The HES SoS hierarchy creates joint moments for distributions at a particular distance from the sphere that are compliant with those of sub-Weibull distributions. These families of distributions form an algebra~\cite{bakhshizadeh2023algebra}. 

\begin{definition}[Vector-valued sub-Weibull expressions]\label{def:nu-theta-subweibull}
A vector-valued expression $ p(\sigma_{i})$ is $(\nu,\theta)$-sub-Weibull with respect to $\sigma_i$ if

\[  \fE_{1}\norm{(p(\sigma_{i}))^{\ot k}}_2^2 \le \nu^{k}k^{2\theta k}\, ,\]
for all $k$.
\end{definition}

To prove properties about the closure of sub-Weibull algebras requires access to a SoS friendly \emph{functional} Cauchy-Schwarz inequality.
\begin{proposition}[SoS functional Cauchy-Schwarz inequality]\label[Proposition]{prop:sos-functional-cs}
    Given two polynomials $p(\sigma_i),\,q(\sigma_i) \in \cfg[k]$ for any $i \in [k]$, the following holds,
    \[
        \proves_{4\mathsf{deg}(p)\mathsf{deg}(q)} \left(\fE_1 p(\sigma_i)q(\sigma_i)\right)^2 \le \fE_1(p(\sigma_i)^{2})\fE_1(q(\sigma_i)^{2})\, .
    \] 
\end{proposition}
\begin{proof}
    The proof follows by using the Hermite decomposition.
    \begin{align*}
        \left(\fE_1 p(\sigma_i) q(\sigma_i)\right)^2&=_{\text{by~\Cref{sec:encoding-hermite-decomps}}} \left(\fE_1 \sum_{\alpha, \beta}\He_{\alpha}\He_{\beta} \hat{f}(\alpha)\hat{g}(\beta)\right)^2 \\
        &=_{\text{\Cref{prop:sos-orthgonality-hermites}}} \left(\sum_{j=0}^d \sum_{|\alpha| = |\beta| = j}d_j\hat{f}(\alpha)\hat{g}(\beta)\right)^2\\
        &\le_{\text{Cauchy-Schwarz}} \left(\sum_{i=0}^{d}d_i\sum_{\alpha, \beta, |\alpha| = |\beta| = j}\hat{f}(\alpha)\hat{f}(\beta)\right)\left(\sum_{j=0}^dd_j\sum_{\alpha',\beta',|\alpha'| = |\beta'|}\hat{g}(\alpha')\hat{g}(\beta')\right) \\
        &=_{\text{\Cref{prop:sos-orthgonality-hermites}}} \left(\sum_{\alpha, \beta} \hat{f}(\alpha)\hat{f}(\beta)\fE_1 \He_\alpha \He_\beta\right)\left(\sum_{\alpha', \beta'} \hat{g}(\alpha')\hat{g}(\beta')\fE_1 \He_\alpha' \He_\beta'\right)\\
        &= \fE_1 p(\sigma_i)^2 \fE_1 q(\sigma_i)^2\, .     
    \end{align*}
\end{proof}
The above proof can be straightforwardly adapted to the case that $p$ and $q$ are matrix polynomials by using the fact that the tensor product is multiplicative in any norm, along with the use of a \emph{matrix} Cauchy-Schwarz inequality~(\Cref{lem:matrix-cauchy-schwarz}) as opposed to the scalar one used above.

Stated below is the definition of what it means to give sum-of-squares proofs over the HES hierarchy of the fact that some polynomial in the grammar has sub-Weibull tails (up to some fixed degree).
 
\begin{definition}[Sum-of-squares proofs for $(\nu,\theta)$-sub-Weibull-ity]
A vector-valued expression $ p(\sigma_{i})$ is $(d, \nu,\theta,k_{\max})$-provably-sub-Weibull with respect to $\sigma_i$ if there exists a degree-$d$ SoS proof that
\[  \{\text{HES constraints}\} \proves_d \left\{\fE_{1}\norm{(p(\sigma_{i}))^{\ot k}}_2^2 \le \nu^{k}k^{2\theta k}\right\}\]
for all $k \le k_{\max}$ and $d \le \deg(\pE)$.
\end{definition}

Since sub-Weibull distributions form an algebra, it is possible to parameterize their closure under multiplication and addition. We give SoS proofs of this fact under the HES hierarchy, yielding the new sub-Weibull parameters (and up to what degree) this proof holds.
   
\begin{lemma}[SoS proofs of sub-Weibull closure~\cite{bakhshizadeh2023algebra}]
    Let $p(\sigma_i)$ and $q(\sigma_i)$ be two vector valued expressions that are $(\nu_1,\theta_1)$ and $(\nu_2, \theta_2)$ sub-Weibull with respect to $\sigma_i$. Then,
    \[
        \{\text{HES constraints}\}\proves_{4k\mathsf{deg(p)}\mathsf{deg}(q)}\fE_1 \norm{\left(p(\sigma_i) \ot q(\sigma_i)\right)^{\ot k}}_2^2 \le (\nu_1\nu_2)^{k}(k^{2(\theta_1 + \theta_2)k})    \, ,
    \]
    and,
    \[
        \{\text{HES constraints}\}\proves_{4k\mathsf{deg(p)}\mathsf{deg}(q)}\fE_1 \norm{\left(p(\sigma_i) + q(\sigma_i)\right)^{\ot k}}_2^2 \le \max(\nu_1^{k}, \nu_2^{k})k^{\max(2(10 + \theta_1)k,\,2(10 + \theta_2)k)}\, ,
    \]
    for $k \le k_{\max}$ and $4k\mathsf{deg}(p)\mathsf{deg(q)} \le \mathsf{deg}(\pE)$.
    This implies that the vector valued expressions $\norm{\left(p(\sigma_i) \ot q(\sigma_i)\right)^{\ot k}}^2_2$ and $\norm{\left(p(\sigma_i) + q(\sigma_i)\right)^{\ot k}}_2^2$ are also sub-Weibull with parameters that only depend on $\nu_1, \nu_2, \theta_1, \theta_2$ and absolute constants, for any choice of $k \in \Z$.
\end{lemma}
\begin{proof}
    We show both parts separately. The first part follows with small modifications to the proofs given in~\cite[Lemma 2]{bakhshizadeh2023algebra} but the second requires a different approach that is SoS compliant. Throughout the proofs, the equivalences mentioned in~\Cref{sec:equivalences} are implicitly used.
    \begin{enumerate}
        \item The proof for the first part follows directly by an application of the functional Cauchy-Schwarz inequality~(\Cref{prop:sos-functional-cs}).
        \begin{align*}
            &\left(\fE_1 \norm{\left(p(\sigma_i) \ot q(\sigma_i)\right)^{\ot k}}_2^2\right)^2 = \left(\fE_1\left( \norm{p(\sigma_i) \ot q(\sigma_i)}_2^2\right)^k\right)^2 \\
            &= \left(\fE_1\left( \norm{p(\sigma_i)}_2^2\norm{q(\sigma_i)}_2^2\right)^k\right)^2 = \left(\fE_1 \norm{p(\sigma_i)}_2^{2k}\norm{q(\sigma_i)}_2^{2k} \right)^2 \\
            &\le_{\text{\Cref{prop:sos-functional-cs}}} \fE_1\norm{p(\sigma_i)}_2^{4k}\fE_1\norm{q(\sigma_i)}
            ^{4k}_2 = \left(\fE_1 \norm{p(\sigma_i)^{\ot 2k}}_2^2\right)\left(\fE_1\norm{q(\sigma_i)^{\ot 2k}}_2^2\right) \\
            &\le ((\nu_1^22^{4\theta_1})(\nu_2^22^{4\theta_2}))^{k/2} = (\nu_1\nu_2)^k(2^{2(\theta_1 + \theta_2)k})
        \end{align*}
        \item This part follows by a standard triangle inequality~(\Cref{lem:l1-to-lt}) in conjunction with the binomial expansion and functional CS inequality~(\Cref{prop:sos-functional-cs}).
        \begin{align*}
            &\left(\fE_1 \norm{\left(p(\sigma_i) + q(\sigma_i)\right)^{\ot k}}_2^2\right)^2 = \left(\fE_1 \norm{p(\sigma_i) + q(\sigma_i)}_2^{2k}\right)^2\\
            &\le_{\text{\Cref{lem:l1-to-lt}}} 2^{2k}\left(\fE_1\left(\norm{p(\sigma_i)}_2^2 + \norm{q(\sigma_i)}_2^2\right)^k\right)^2 \\
            &= 2^{2k}\left(\sum_{j=0}^k \binom{k}{j}\fE_1\norm{p(\sigma_i)}_2^{2j}\norm{q(\sigma_i)}_2^{2(k-j)}\right)^2\\
            &\le_{\text{\Cref{lem:l1-to-lt}}} k2^{2k}\left(\sum_{j=0}^k \binom{k}{j}^2\left(\fE_1\norm{p(\sigma_i)}_2^{2j}\norm{q(\sigma_i)}_2^{2(k-j)}\right)^2\right) \\
            &\le_{\text{\Cref{prop:sos-functional-cs}}} k2^{2k}\left(\sum_{j=0}^k \binom{k}{j}^2 \fE_1 \norm{p(\sigma_i)}_2^{4j}\fE_1 \norm{q(\sigma_i)}_2^{4(k-j)}\right) \\
            &= k2^{2k}\left(\sum_{j=0}^k \binom{k}{j}^2\fE_1\norm{p(\sigma_i)^{\ot 2j}}_2^2 \fE_1 \norm{q(\sigma_i)^{\ot 2(k-j)}}_2^{2}\right) \\ 
            &\le k2^{2k}\left(\sum_{j=0}^k \binom{k}{j}^2 (\nu_1^{2j}(2j)^{2\theta_12j})(\nu_2^{2(k-j)}(2(k-j))^{2\theta_2 2(k-j)})\right) \\
            &\le_{\text{Assume WLOG $\nu_1 \ge \nu_2$}} \sim  k2^{2k} \cdot k \cdot \binom{k}{k/2}^2 \nu_1^{2k} k^{2(\theta_1 + \theta_2)k} \le \nu_1^{2k}k^{2(\theta_1 + \theta_2 + 20)k} \\
            &\le_{\text{Assume WLOG $\theta_1 \ge \theta_2$}} \nu_1^{2k}k^{2(2\theta_1 + 20)k}\, ,
        \end{align*}
        where it is used that the terms that contribute the most to the binomial sum are in the interval $j \in [k/2 \pm O(\sqrt{k})]$. Then, using absolute constants (which can be sufficiently loose without any consequence) and some algebra, the desired bound follows.

        Note that both bounds include squaring the polynomial, but by degree reduction~(\Cref{sec:comm-semi-algebraic-degree-reduction}) these can be made to hold on the original polynomial.
    \end{enumerate}
\end{proof}

\subsubsection{Near-pointwise filtered SoS proofs for Subweibull algebras}

\begin{lemma}[``Almost pointwise'' proof]
If a vector-valued expression $ p(\sigma_{i})$ is $(d,\nu,\theta,d_{\max})$-provably-sub-Weibull with respect to $\sigma_i$, then the following conditional constraint is true of it up to a total degree $d_{\max}$ of proof:
\[ \left\{ \norm{p(\sigma_i)}_2^2 \le d_{\max}^{2\theta}\nu \right\}_{\le d_{\max}} \]
    
\end{lemma}

\section{Joint Moments of the High-Entropy Step Constraints}\label{sec:hods}

We now give the key SoS proof for high-entropy step distributions: that there exist matrix representations of their moments that have low nuclear norm.
The main theorem of this section is \Cref{thm:high-entropy-nuclear-norm-bound}, with a generalization for graph matrices in \Cref{thm:high-entropy-nuclear-norm-bound-graph-matrices}.

The most basic case of this is if a random variable $x$ is drawn from the single-step distribution which is simply a Gaussian with variance $\Sigma$.
In this case, Isserlis's lemma (\Cref{lem:isserlis}) shows that the degree-$2\eta$ moment matrix is equal to the sum of all mode-permutations of $\Sigma^{\ot \eta}$.
If we represent a degree-$2\eta$ homogeneous polynomial $p$ canonically square-most as a mode-symmetric matrix $M$, then $\pE p = \iprod{M,V}$ when $V$ is any representation of the degree-$2\eta$ moments of $x$.
One of these representations is $(2\eta-1)!!\Phi^{\ot \lfloor\eta/2\rfloor}\Phi^{\ot \lceil\eta/2\rceil}$, where $\Phi \in \R^{n^2}$ is defined by $\Phi_{(i,j)} := \Sigma_{i,j}$, and this representation has nuclear norm $(2\eta-1)!! \norm{\Phi}_2^{\eta} = (2\eta-1)!! \norm{\Sigma}_F^{\eta}$.
The upshot of this is that by applying Holder's inequality~(\Cref{lem:holders-one-inf}) we bound $\pE p \le (2\eta-1)!! \norm{\Sigma}_F^{\eta}\opnorm{M}$ by the \emph{spectral norm} of $M$.

This analysis generalizes to multiple conditional steps through the use of the \emph{moment-cumulant} formula~(\Cref{eq:cumulant-degree-d}), which expresses the joint moments of multiple variables as a sum of tensor products of their cumulants\footnote{The authors believe there exists a simpler formulation of the analysis in terms of Stein's lemma instead of the moment-cumulant formula; we leave elaboration of this to possible future work.}.
The covariance matrix $\Sigma$ is an example of such a cumulant.
We bound the 2-norm of each cumulant in its vector reshaping (\Cref{sec:cum-nuc-norm-bounds}) using concentration/isoperimetric properties exposed via their Hermite decompositions (\Cref{sec:hermite-properties}) so as to show that the outer product of tensor products of cumulants has small nuclear norm.
Then we argue that even if we cannot construct a matrix of the correct dimensions as an outer product of tensor products, reshaping it into a matrix of the right dimensions will generally not increase the nuclear norm too much (\Cref{sec:reshaping-tools}).

Finally, we put it together into an induction over the conditional structure of high-entropy steps (\Cref{sec:conditional-induction}).
In this induction, there is a tricky bad case when one attempts to add odd numbers of modes to both sides of the moment matrix: for example, in the single-step Gaussian case if we are considering second-order moments, in order to minimize nuclear norm, we would like to add a factor of $\Phi$ instead of its matrix reshaping $\Sigma$.
It is always possible to reshape the maximal number of factors of $\Sigma$ into $\Phi$ if the matrix desired at the end has an even number of modes on at least one side, but if an $n^{2k+1} \times n^{2k+1}$ matrix is desired, at least one factor of $\Sigma$ must remain as a $\Sigma$.
This case is responsible for the extra factor of $n^{k/2}$ that shows up in the power of $n$ in \Cref{thm:high-entropy-nuclear-norm-bound}, but we only have to incur this penalty once because after the first level of induction, we can maintain nuclear norm bounds on two different shapes of moment matrices---one that is square-most and one that is almost square-most---and in each case of the induction, pick whichever one would not force us to add an odd number of modes to both sides in order to get the next desired matrix shape.
The avoidance of this bad case is done in \Cref{lem:moment-nuclear-norm}.

\subsection{Matrix representations}

We extend the matrix representation definitions in \Cref{sec:matrix-representations} to delineate different degrees of freedom in reshaping the coefficient matrices and the moment matrices that jointly maintain the inner product between the two matrices as an invariant.

 In a bit of an unprincipled but convenient maneuver, we will sometimes allow the matrix representations of pseudo-moments to interchange tensor modes that correspond to different variables.
In order to preserve the correspondence between pseudo-expectation and the Hilbert-Schmidt inner product when we do this, we will have to add more symmetry restrictions to the matrix representations of the polynomials.

\begin{definition}
    A matrix representation of a multivariate polynomial is \emph{fully symmetric} if the matrix is mode-symmetric across all modes, including those modes that represent different variables (note that such a representation can only exist if the polynomial is itself symmetric with respect to interchange of its variables).
    \end{definition}

    \begin{definition}
    Given a pseudo-expectation operator $\pE$ over several vector-valued variables $u_1, \dots, u_k$ and a multiset $S$ of those variables, a matrix $V$ is a \emph{fully permutable matrix representation} of its $S$-pseudo-moments if, for some $a$ and $b$ with $a+b = d$, there is a list of shapes $(A_1,B_1), \dots, (A_{\ell},B_{\ell})$, each satisfying the requirement that $S$ is the multiset of elements of $A_{m'}$ and $B_{m'}$ counted with multiplicity, and a list of matrices $V^{(1)}$, \dots, $V^{(\ell)}$, so that $V = V^{(1)} + \dots + V^{(\ell)}$ and for all sequences $\gamma_1, \dots, \gamma_k$ of multi-indices with $|\gamma_{m}|$ equal to the multiplicity of $u_{m}$ in $S$,
    \[\sum_{m' \in [\ell]}\;\sum_{\substack{(i_1,\dots,i_a) \in \operatorname{ind}(A_{m'},\alpha_1, \dots, \alpha_k)\\(j_1,\dots,j_b)\in\operatorname{ind}(B_{m'},\beta_1,\dots,\beta_k)\\\alpha_m + \beta_m = \gamma_m\,\forall m\in[k]}} V^{(m')}_{(i_1,\dots,i_a),(j_1,\dots,j_b)} = \pE_{u_1, \dots, u_k} \prod_{m \in [k]} u_m^{\gamma_m}\,,\]
    where $a$ is the number of elements in each $A_{m'}$, $b$ is the number of elements in each $B_{m'}$ and $\operatorname{ind}(A, \alpha_1, \dots, \alpha_k)$ is the set of all tuples $(i_1,\dots,i_a)$ so that $i_k$ is a valid index into the $k$th variable in $A$ and $\alpha_m$ counts the number of occurrences of each index associated with $u_m$ among $(i_1,\dots,i_a)$.
\end{definition}

\begin{proposition}
\ifnum\fastmode=0
\label[Proposition]{prop:representation-product-fully}
\fi
\ifnum\fastmode=1
\label[Proposition]{prop:representation-product-fully}
\fi
    If $M$ is a fully symmetric matrix representation of an $S$-homogeneous multivariate polynomial $p(u_1, \dots, u_k)$ and $V$ is a fully permutable matrix representation of the $S$-pseudo-moments of $\pE$, and $M$ and $V$ have the same dimensions,
    then
    \[ \iprod{M,V} = \pE_{u_1, \dots, u_k} p(u_1, \dots, u_k) \,. \]
\end{proposition}

\subsubsection{For powers of polynomials}

\begin{definition}
    \label[Definition]{def:c-stacked}
    A \emph{$c$-stacked permutable} matrix representation of the $S$-pseudo-moments of a pseudo-distribution is a fully permutable matrix representation of the $\underbrace{S \cup \dots \cup S}_{c}$-pseudo-moments of the same pseudo-distribution, with the additional restriction that $a$ and $b$ are both divisible by $c$ and each shape $(A_{m'},B_{m'})$ can be written as $(A_{m'},B_{m'}) = (A_{(m', 1)} \circ \dots \circ A_{(m', c)}, B_{(m', 1)} \circ \dots \circ B_{(m', c)})$, where $\circ$ denotes tuple concatenation and for every $z \in [c]$, one of $A_{(m',z)}$ or $B_{(m',z)}$ contains $a/c$ elements and the other contains $b/c$ elements.
\end{definition}

\begin{proposition}
\ifnum\fastmode=0
\label[Proposition]{prop:representation-product-stacked}
\fi
\ifnum\fastmode=1
\label[Proposition]{prop:representation-product-stacked}
\fi
    If $M$ is a fully symmetric matrix representation of an $\underbrace{S \cup \dots \cup S}_{c}$-homogeneous multivariate polynomial $p(u_1, \dots, u_k)$ and $V$ is a $c$-stacked permutable matrix representation of the $S$-pseudo-moments of $\pE$, and $M^{\ot c}$ and $V$ have the same dimensions,
    then
    \[ \iprod{M^{\ot c},V} = \pE_{u_1, \dots, u_k} p(u_1, \dots, u_k)^c \,. \]
\end{proposition}

\subsubsection{For graph matrices}

\begin{definition}
    \label[Definition]{def:single-flip-stacked}
    Let $S_1$, \dots, $S_c$ be multi-sets of formal variables with $a_1$, \dots, $a_c$ elements respectively.
    Let $S = S_1 \cup \dots \cup S_c$.

    An \emph{$(a_1, \dots, a_c)$-stacked single-flip-permutable} matrix representation of the $S$-pseudo-moments of a pseudo-distribution is a fully permutable matrix representation of its $S$-pseudo-moments,
    with the additional restriction that each shape $(A_{m'},B_{m'})$ can be written as $(A_{m'},B_{m'}) = (A_{(m', 1)} \circ \dots \circ A_{(m', c)}, B_{(m', 1)} \circ \dots \circ B_{(m', c)})$, where $\circ$ denotes tuple concatenation and for every $z \in [c]$, 
    \[|A_{(m',z)}| \in \{\lceil a_z/2+1/2\rceil, \lfloor a_z/2\rfloor, \lfloor a_z/2-1/2\rfloor\}\]
    and $|B_{(m',z)}| = a_z - |A_{(m',z)}|$ and
    if $|A_{(m',z)}| > a_z/2$, then letting $\ell$ be the smallest number such that $|A_{(m',z+\ell )}| \ne a_z/2$, if $\ell$ exists we will have $|A_{(m',z+\ell)}| < a_z/2$, and similarly, if $|A_{(m',z)}| <  a_z/2$, then letting $\ell$ be the smallest number such that $|A_{(m',z+\ell)}| \ne a_z/2$, if $\ell$ exists we will have $|A_{(m',z+\ell )}| > a_z/2$.
\end{definition}

\begin{definition}
    \label[Definition]{def:single-flip-graph-matrix}
    Let $(V,E)$ be a hypergraph, $W \subseteq V$ be a set of middle vertices, and $(e_1, \dots, e_c)$ be an ordering of all the hyperedges in $E$.
    We suppose that each hyperedge $e_i$ is a multi-set which contains exactly $a_i$ vertices (with multiplicity) that are not in $W$.
    
    The set of \emph{single-flipped graph matrix shapes} corresponding to $(V,E)$ and $W$ and the ordering $(e_1, \dots, e_c)$ is the set of all graph matrix shapes $\alpha = (U_{\alpha}, V_{\alpha}, E, W)$ with the property that for each hyperedge $e_i \in E$ we have 
    \[|e_i \cap U_{\alpha}| \in \{\lceil a_i/2+1/2\rceil,\lfloor a_i/2\rfloor, \lfloor a_i/2-1/2\rfloor\}\]
    and if $|e_i \cap U_{\alpha}| > a_i/2$, then letting $\ell$ be the smallest number such that $|e_{i+\ell} \cap U_{\alpha}| \ne a_i/2$, if $\ell$ exists we will have $|e_{i+\ell} \cap U_{\alpha}| < a_i/2$, and similarly, if $|e_i \cap U_{\alpha}| <  a_i/2$, then letting $\ell$ be the smallest number such that $|e_{i+\ell} \cap U_{\alpha}| \ne a_i/2$, if $\ell$ exists we will have $|e_{i+\ell} \cap U_{\alpha}| > a_i/2$.
\end{definition}

\subsection{Nuclear norm bound on moments of sum of high-entropy steps}

\begin{theorem}
\label{thm:high-entropy-nuclear-norm-bound}
Let $\sigma = v_1 + \dots + v_i \in \R^n$ be a sum of unit-norm high-entropy steps, such that each $v_j$ is a high-entropy step conditioned on $v_1, \dots, v_{j-1}$, with covariance matrices bounded by $\nu$ in operator norm and having degree at most $d$ in the previous iterates.
Let $u$ be a unit-norm high-entropy step conditioned on $v_1, \dots, v_i$ and also with covariance matrix bounded by $\nu$ in operator norm and having degree at most $d$ in $v_1, \dots, v_i$.

Let $p(\sigma, u)$ be a homogeneous polynomial of degree $d_{\sigma}$ in $\sigma$ and degree $d_u$ in $u$, with $d_p := d_{\sigma}+d_u$.
Let $P$ be its fully symmetric matrix representation.
Then 
\[\E_{\sigma,u} p(\sigma, u)^k \le i^{kd_{\sigma}}2^{kd(i+1)^2/2}(kd_pi)^{kd_p/2+2ki+2k}\nu^{kd_p/2}
n^{k\lfloor d_p/2+1\rfloor/2}\opnorm{P}^k
 \]
 by a low-degree sum-of-squares proof.

\end{theorem}
\begin{proof}
    Let $T$ be the tensor reshaping of $P$.
    First, let's suppose that $k=1$.
    Then, for every $M$ a squaremost matrix representation of the moments $\E \sigma^{\ot d_{\sigma}} \ot u^{\ot d_u}$,
        \begin{align*}\E_{\sigma,u} p\left(\sigma,u\right)  
        &= \E \iprod{T, \sigma^{\ot d_{\sigma}} \ot u^{\ot d_u}}
        \\&= \iprod{P, M}
        \\&\le \opnorm{P}\norm{M}_1
        \,,
        \end{align*}
        where the last inequality was by \Cref{lem:holders-one-inf}.

    By a multinomial expansion and \Cref{lem:moment-nuclear-norm}, there exists such a matrix representation $M$ satisfying
    \[ \norm{M}_1 \le i^{d_{\sigma}}2^{d(i+1)^2/2}d_p^{d_p/2+2i+2}\nu^{d_p/2}n^{\lfloor d_p/2\rfloor + 1/2}\,. \]
    This completes the argument.
        
    Next, we discuss $k>1$.
 First, we apply a multinomial expansion with $\sigma = \sum_{j\in[i-1]}v_j$ and summing over multi-indices $\alpha$ on $i$ options, and then applying a weighted version of \Cref{lem:l1-to-lt}:
        \begin{align*}\E_{\sigma,u} p\left(\sigma,u\right)^k  
        &= \E \iprod{T, \sigma^{\ot d_{\sigma}} \ot u^{\ot d_u}}^k
    \\&= \E \left(\sum_{|\alpha| = d_{\sigma}} \binom{d_{\sigma}}{\alpha} \iprod{T,v^{\ot \alpha} \ot u^{\ot d_u}}\right)^k
    \\&\le i^{(k-1)d_{\sigma}}\sum_{|\alpha| = d_{\sigma}} \binom{d_{\sigma}}{\alpha} \E \iprod{T,v^{\ot \alpha} \ot u^{\ot d_u}}^k
    \\&= i^{(k-1)d_{\sigma}}\sum_{|\alpha| = d_{\sigma}} \binom{d_{\sigma}}{\alpha} \iprod{T^{\ot k},\E(v^{\ot \alpha} \ot u^{\ot d_u})^{\ot k}}
    \,, \end{align*}
    where the last equality is by the bilinearity of the tensor product.
    
By reshaping and then \Cref{lem:holders-one-inf} and \Cref{prop:representation-product-stacked}, this quantity is
\[ p\left(\sigma,u\right)^k \le 2^{(k-1)d_{\sigma}}\sum_{|\alpha| = d_{\sigma}} \binom{d_{\sigma}}{\alpha} \iprod{P^{\ot k},M_{\alpha}}
\le i^{(k-1)d_{\sigma}}\sum_{|\alpha| = d_{\sigma}} \binom{d_{\sigma}}{\alpha} \opnorm{P}^{ k}\norm{M_{\alpha}}_1
\]
where $M_{\alpha}$ is any squaremost $k$-stacked permutable matrix representation (\Cref{def:c-stacked}) of $\E(v^{\ot \alpha} \ot u^{\ot d_u})^{\ot k}$.

By \Cref{lem:moment-nuclear-norm-stacked}, there exist these representations $M_{\alpha}$ satisfying
\[\norm{M_{\alpha}}_1 \le 2^{kdi^2/2}(c\eta)^{k\eta/2+2ki}\,\nu^{k/2}\, n^{k\lfloor\eta/2+1\rfloor/2}\,,\]
so
\begin{align*}\E_{\sigma,u} p\left(\sigma,u\right)^k  
    &\le i^{(k-1)d_{\sigma}}\sum_{|\alpha| = d_{\sigma}} \binom{d_{\sigma}}{\alpha} 2^{kdi^2/2}(k\eta)^{k\eta/2+2ki}\,\nu^{k/2}\, n^{k\lfloor\eta/2+1\rfloor/2}\opnorm{P}^{k}
    \\&= i^{(k-1)d_{\sigma}} 2^{kdi^2/2}(k\eta)^{k\eta/2+2ki}\,\nu^{k/2}\, n^{k\lfloor\eta/2+1\rfloor/2}\opnorm{P}^{k}\sum_{|\alpha| = d_{\sigma}} \binom{d_{\sigma}}{\alpha}
    \\&=i^{(k-1)d_{\sigma}} 2^{kdi^2/2}(k\eta)^{k\eta/2+2ki}\,\nu^{k/2}\, n^{k\lfloor\eta/2+1\rfloor/2}\opnorm{P}^{k} \cdot i^{d_{\sigma}}
    \,.
\end{align*}

\end{proof}

\subsubsection{Generalization for graph matrices}

\begin{theorem}
\label{thm:high-entropy-nuclear-norm-bound-graph-matrices}
Let $\sigma = v_1 + \dots + v_i \in \R^n$ be a sum of high-entropy steps, such that each $v_j$ is a high-entropy step conditioned on $v_1, \dots, v_{j-1}$, with covariance matrices bounded by $\nu$ in operator norm.

Let $(V,E)$ be a hypergraph, $W$ be a set of middle vertices, and $(e_1, \dots, e_c)$ be an ordering of all the hyperedges in $E$.
Let $q(\sigma)$ be the polynomial represented by the graph matrix $P_{(U',V',E,W)}$ for any choice of multi-sets $U'$ and $V'$ satisfying $U' \cup V' \cup W = V$, so that consequently,
\[q(\sigma) = \iprod{P_{(U',V',E,W)}, \sigma^{\ot |U'|}(\sigma^{\ot |V'|})^{\sT}}\,.\]
Let $a_i$ be the number of elements of $e_i$ that are not in $W$ and let $a = a_1 + \dots + a_c$.
Then for constant $k$,
\[ \E q(\sigma)^k \le O(1)\nu^{ka}n^{k\lfloor a/2+1\rfloor/2}\sum_{\alpha \in G}\opnorm{P_{\alpha}}^k\,, \]
where $G$ is the set of all single-flipped graph matrix shapes corresponding to $(V,E)$ and $W$ and $(e_1, \dots, e_c)$ and $P_{\alpha}$ is the graph matrix with shape $\alpha$.
\end{theorem}
\begin{proof}
    Assume $k = 1$ for now.
    We have
    \[ \E q(\sigma) = \iprod{P, \E \sigma^{\ot \lfloor a/2\rfloor}\sigma^{\ot \lceil a/2\rceil}} \]
    where $P$ is any $n^{\lfloor a/2\rfloor}$-by-$n^{\lceil a/2\rceil}$ matrix that represents $q$.
    It follows that for any set of matrices $M_1, \dots, M_m$ satisfying
    \[ \E \sigma^{\ot \lfloor a/2\rfloor}\sigma^{\ot \lceil a/2\rceil} = M_1 + \dots + M_m\,,\]
    it is true that 
    \[ \E q(\sigma) = \sum_{j \in [m]} \iprod{P, M_j}\,. \]
    It then also follows that for any sequence $\pi_1, \dots, \pi_m$ of mode permutations,
    \[ \E q(\sigma) = \sum_{j \in [m]} \iprod{\pi_j(P), \pi_j(M_j)} \le \sum_{j \in [m]} \opnorm{\pi_j(P)}\norm{\pi_j(M_j)}_1 \,. \]
    By taking $P = P_{\alpha_0}$ for some $\alpha_0 \in G$ satisfying $|e_i \cap U_{\alpha_0}| \in \{\lceil a_i/2\rceil, \lfloor a_i/2\rfloor\}$, and restricting each $\pi_{j}$ to be a permutation that takes $P_{\alpha_0}$ to some $P_{\alpha}$ for some other $\alpha \in G$, the sum $\sum_j \pi_j(M_j)$ becomes allowed to be any $(a_1, \dots, a_c)$-stacked single-flip permutable representation of the $a$th pseudo-moment of $\sigma$.

    By taking a multinomial expansion and then applying \Cref{lem:moment-nuclear-norm-stacked-graph-matrix}, there exists such a $(a_1, \dots, a_c)$-stacked single-flip-permutable representation which satisfies
    \[ \sum_j \norm{\pi_j(M_j)}_1 \le O(1)\nu^{a}n^{\lfloor a/2+1\rfloor/2}\,.\]

    This completes the argument for $k=1$.
    To handle $k > 1$, we use the same technique as in the $k > 1$ part of \Cref{thm:high-entropy-nuclear-norm-bound}, adding a layer of $k$-stacking on top of the $(a_1, \dots, a_c)$-stacking so as to apply a version of \Cref{lem:moment-nuclear-norm-stacked} that has a layer of \Cref{lem:moment-nuclear-norm-stacked-graph-matrix} nested inside of it.
\end{proof}

\subsection{Nuclear norm bounds on joint moments of high-entropy steps}
\label{sec:high-entropy-set-nuc-norm-bounds}

\begin{definition}
    Given a multiset $S = \{x_1, \dots, x_k\} \subseteq \R^n$ of random variables, its \emph{$S$-moments} are the collection of degree-$|S|$ joint moments given by the tensor $\E \bigotimes_{i \in k} x_i$.
\end{definition}

\begin{definition}
\label[Definition]{def:mom-nuc-norm}
    The \emph{$(r,s)$-moment-nuclear norm} of a multi-set of vector-valued random variables $S = \{x_1, \dots, x_k\} \subseteq \R^n$ (where $k = r+s$) is the smallest nuclear norm of any $\R^{n^r \times n^s}$ matrix representation of the $S$-moments.

    More precisely, it is equal to
    \[ \norm{\E S^{(r,s)}}_1 := \inf_{M \in \cM} \norm{M}_1\,, \]
    where $\norm{\cdot}_1$ denotes the matrix nuclear norm and $\cM$ is the set of matrix representations (\Cref{def:mat-representations-multivar-moments}) of the $S$-moments with co-domain $(\R^n)^{\ot r}$ and domain $(\R^n)^{\ot s}$.

    We will also write $\norm{\E [S^{(r,s)} \mid y_1, \dots, y_{k'}]}_1$ to denote $\norm{\E S^{(r,s)}}_1$ where the variables in $S$ are interpreted as conditioned on the values of the variables $y_1, \dots, y_{k'}$.
\end{definition}

\begin{definition}
\label[Definition]{def:rank-decomposing}
    A \emph{rank-decomposing} SoS proof that $\norm{ M}_1 \le C$ conditioned on $v_1, \dots, v_j$ where $M$ is a matrix indeterminate is an explicit decomposition $M = \sum_i \E_j p_iq_i^{\sT}$ where each $p_i$ and $q_i$ are vector-valued indeterminates, such that $\sum_i \frac{1}{2}\E_j(\norm{p_i}_2^2 + \norm{q_i}_2^2) \le C$ by a SoS proof.
\end{definition}

\begin{lemma}
\label[Lemma]{lem:moment-nuclear-norm}
Suppose that $S$ is a multiset containing formal variables $u_1, \dots, u_k \in \R^n$ with multiplicities $\eta_1, \dots, \eta_k$ respectively.
Let $\eta = \eta_1 + \dots + \eta_k$.

Suppose further that conditioned on the values of $u_1, \dots, u_{i-1}$, each $u_i$ is a Gaussian with covariance matrix $\Sigma_i$ satisfying $\Sigma_i \preceq \nu_i \cdot \Id$ and furthermore that the covariance matrix of $u_i$ is a polynomial of degree at most $d$ in $u_1, \dots, u_{i-1}$, and let $\nu = \prod_{i \in [k]} \nu_i^{\eta_i}$.

Then
\[ \norm{\E S^{(\lfloor\eta/2\rfloor, \lceil\eta/2\rceil)}}_1 \le 2^{dk^2/2}\eta^{\eta/2+2k}\,\nu^{1/2}\, n^{\lfloor\eta/2+1\rfloor/2} \]
by a rank-decomposing sum-of-squares proof.

\end{lemma}
\begin{proof}
    This will be a proof by induction, although the inductive hypothesis will be slightly stronger than the lemma statement.

    Let $S_{(j)}$ be the multiset containing $u_{j}, \dots, u_k$ with multiplicities $\eta_j, \dots, \eta_k$, let $\eta^{(j)} = \eta_j + \dots + \eta_k$, and let $\nu^{(j)} = \prod_{i \in [k] \setminus [j-1]} \nu_i^{\eta_i}$.
    Furthermore, let $r^{(j)} = \lfloor\eta^{(j)}/2\rfloor$ and $s^{(j)} = \lceil\eta^{(j)}/2\rceil$.

    The inductive hypothesis is that for every $j+1$, we have, by a rank-decomposing sum-of-squares proof,
\begin{align*}
&\max\left(\norm{\E\left[ S_{(j+1)}^{(r^{(j+1)}, s^{(j+1)})}\mid u_1, \dots, u_j\right]}_1, \norm{\E\left[ S_{(j+1)}^{(r^{(j+1)}-1, s^{(j+1)}+1)}\mid u_1, \dots, u_j\right]}_1\right)
\\&\qquad\qquad\qquad\le 2^{d(k-j)}\eta^{\eta^{(j+1)}+2(k-j)}(\nu^{(j+1)})^{1/2}n^{\lfloor\eta^{(j+1)}/2+1\rfloor/2}
\,, \end{align*}
where we count $\norm{\E S_{(j+1)}^{\left(r, s\right)}}_1$ as $0$ if $r$ or $s$ is negative.

In the base case, assuming without loss of generality that $\eta_k \ne 0$, then we have $S_{(k)}$ is the multiset containing only $u_k$ with multiplicity $\eta_k$.
We can compute the moments $\E\left[u_k^{\ot \eta_k} \mid u_1, \dots, u_{k-1}\right]$ directly using Isserlis's lemma~(\Cref{prop:sos-isserlis}).
If $\eta_k$ is odd, then all the degree-$\eta_k$ moments are $0$ and we are done.
If $\eta_k$ is even, then Isserlis's theorem tells us that $\E\left[u_k^{\ot \eta_k} \mid u_1, \dots, u_{k-1}\right]$ is a sum over the $(\eta_k-1)!!$ different mode permutations of $\Sigma_k^{\ot \eta_k/2}$.
Letting $\Phi$ be the vector reshaping of $\Sigma_k$, for the $\norm{\E\left[ S_{(k)}^{(r^{(k)}, s^{(k)})}\mid u_1, \dots, u_j\right]}_1$ case we can choose the matrix representation that represents each summand as
\[ (\Phi\Phi^{\sT})^{\ot \lfloor\eta_k/4\rfloor} \ot \Sigma_k^{\ot (\eta_k/2 - 2\lfloor\eta_k/4\rfloor)}\,. \]
The nuclear norm of each summand is then, using~\Cref{lem:holders-one-inf} to see that $\Tr \Sigma_k \le \nu_kn$ and $\norm{\Sigma_k}_F^2 \le \opnorm{\Sigma_k}\Tr \Sigma_k \le \nu_k^2n $, 
\[\norm{\Sigma_k}_F^{2\lfloor\eta_k/4\rfloor}(\Tr \Sigma_k)^{\eta_k/2 - 2\lfloor\eta_k/4\rfloor}
\le \nu_k^{\eta_k/2}n^{\lfloor\eta_k/4\rfloor}n^{\eta_k/2 - 2\lfloor\eta_k/4\rfloor}
\le \nu_k^{\eta_k/2}n^{\lfloor\eta_k/2+1\rfloor/2}
\,.\]
Meanwhile, in the $\norm{\E\left[ S_{(k)}^{(r^{(k)}-1, s^{(k)}+1)}\mid u_1, \dots, u_j\right]}_1$ case, we can choose the matrix representations as
\[ (\Phi\Phi^{\sT})^{\ot \lfloor\eta_k/4-1/2\rfloor} \ot \Sigma_k^{\ot (\eta_k/2 - 1 - 2\lfloor\eta_k/4 - 1/2\rfloor)} \ot \Phi\,. \]
The nuclear norm of each summand is then 
\begin{align*}
\norm{\Sigma_k}_F^{2\lfloor\eta_k/4-1/2\rfloor+1}(\Tr \Sigma_k)^{\eta_k/2 - 1 - 2\lfloor\eta_k/4 - 1/2\rfloor}
&\le \nu_k^{\eta_k/2}n^{\lfloor\eta_k/4-1/2\rfloor+1/2}n^{\eta_k/2 - 1 - 2\lfloor\eta_k/4 - 1/2\rfloor}
\\&\le \nu_k^{\eta_k/2}n^{\eta_k/2-\lfloor\eta_k/4-1/2\rfloor-1/2}
\\&\le \nu_k^{\eta_k/2}n^{\lfloor\eta_k/2+1\rfloor/2}
\,.\end{align*}

In both cases, there are $(\eta_k-1)!! \le \eta_k^{\eta_k/2}$ summands and these were both rank-decomposing proofs, so the inductive hypothesis is satisfied in the base case.

In the inductive case, we mainly use \Cref{lem:rank-decomposing-induction}.
We break it into cases depending on the value of $\eta_j$.
If $\eta_j \equiv 0 \pmod 4$, then in both the $S_{(j)}^{(r^{(j)}, s^{(j)})}$ and $S_{(j)}^{(r^{(j)}-1, s^{(j)})+1}$ cases, we invoke \Cref{lem:rank-decomposing-induction} with $\Delta r = \Delta s = \eta_j/2$ on the respective cases of the inductive hypothesis.
Then the nuclear norm is at most $2^d\eta_{j}^{\eta_j/2+2}\nu_j^{\eta_j/2}n^{\lfloor\eta_j/2\rfloor/2}$ times as large as the previous step, which since $\eta_k \equiv 0 \pmod 2$, is at most as large as the desired ratio
\begin{align*}
    \frac{2^{d(k-j-1)}\eta^{\eta^{(j)}/2+2(k-j-1)}(\nu^{(j)})^{1/2}n^{\lfloor\eta^{(j)}/2+1\rfloor/2}}{2^{d(k-j)}\eta^{\eta^{(j+1)}/2+2(k-j)}(\nu^{(j+1)})^{1/2}n^{\lfloor\eta^{(j+1)}/2+1\rfloor/2}}
& = 2^d\eta^{\eta_j/2+2}\nu_j^{\eta_j/2}\frac{n^{\lfloor(\eta^{(j+1)} + \eta_j)/2+1\rfloor/2}}{n^{\lfloor\eta^{(j+1)}/2+1\rfloor/2}} \\
& = 2^d\eta^{\eta_j/2+2}\nu_j^{\eta_j/2}n^{\eta_j/4}\,.
\end{align*}

If $\eta_j \equiv 2 \pmod 4$, the argument works much the same way, except that in the $S_{(j)}^{(r^{(j)}, s^{(j)})}$ case we invoke \Cref{lem:rank-decomposing-induction} with $\Delta r = \eta_j/2 + 1$ and  $\Delta s = \eta_j/2 - 1$ on the inductive hypothesis for $S_{(j+1)}^{(r^{(j+1)}-1, s^{(j+1)}+1)}$, and in the $S_{(j)}^{(r^{(j)}-1, s^{(j)}+1)}$ case, we invoke \Cref{lem:rank-decomposing-induction} with $\Delta r = \eta_j/2 - 1$ and  $\Delta s = \eta_j/2 + 1$ on the inductive hypothesis for $S_{(j+1)}^{(r^{(j+1)}, s^{(j+1)})}$.
Then, compared with the $0 \pmod 4$ case, we get the same expression for the bound on the new nuclear norm and also for the desired ratio, so the desired bound is still satisfied.

In the case that $\eta_j$ is odd, if $\eta^{(j)}$ is even then in both the $S_{(j)}^{(r^{(j)}, s^{(j)})}$ and $S_{(j)}^{(r^{(j)}-1, s^{(j)}+1)}$ cases, we invoke \Cref{lem:rank-decomposing-induction} with $\Delta r = \lceil\eta_j/2\rceil$ and  $\Delta s = \lfloor\eta_j/2\rfloor$ on the respective cases of the inductive hypothesis.
If $\eta^{(j)}$ is odd then in both the $S_{(j)}^{(r^{(j)}, s^{(j)})}$ and $S_{(j)}^{(r^{(j)}-1, s^{(j)}+1)}$ cases, we invoke \Cref{lem:rank-decomposing-induction} with $\Delta r = \lfloor\eta_j/2\rfloor$ and  $\Delta s = \lceil\eta_j/2\rceil$ on the respective cases of the inductive hypothesis.
In both cases when $\eta_j$ is odd, the new nuclear norm is at most $2^d\eta_j^{\eta_j/2+2}\nu_j^{\eta_j/2}n^{\lfloor \eta_j/2 \rfloor/2}$ times as large as the old one, and the desired bound on the ratio was for it to be at most
\begin{align*}\frac{2^{d(k-j-1)}\eta^{\eta^{(j)}/2+2(k-j-1)}(\nu^{(j)})^{1/2}n^{\lfloor\eta^{(j)}/2+1\rfloor/2}}{2^{d(k-j)}\eta^{\eta^{(j+1)}/2+2(k-j)}(\nu^{(j+1)})^{1/2}n^{\lfloor\eta^{(j+1)}/2+1\rfloor/2}}
&= 2^d\eta^{\eta_j/2+2}\nu_j^{\eta_j/2}\frac{n^{\lfloor(\eta^{(j+1)} + \eta_j)/2+1\rfloor/2}}{n^{\lfloor(\eta^{(j+1)} )/2+1\rfloor/2}}
\\&\ge 2^d\eta_j^{\eta_j/2+2}\nu_j^{\eta_j/2}n^{\lfloor\eta_j/2\rfloor/2}\,,
\end{align*}
remembering that $\lfloor\eta_j/2\rfloor = (\eta_j-1)/2$ since $\eta_j$ is odd.
And that covers all the cases.
\end{proof}

\begin{lemma}
\label[Lemma]{lem:moment-nuclear-norm-stacked}
Suppose that $S$ is a multiset containing formal variables $u_1, \dots, u_k \in \R^n$ with multiplicities $\eta_1, \dots, \eta_k$ respectively.
Let $\eta = \eta_1 + \dots + \eta_k$.

Suppose further that conditioned on the values of $u_1, \dots, u_{i-1}$, each $u_i$ is a Gaussian with covariance matrix $\Sigma_i$ satisfying $\Sigma_i \preceq \nu_i \cdot \Id$ and furthermore that the covariance matrix of $u_i$ is a polynomial of degree at most $d$ in $u_1, \dots, u_{i-1}$, and let $\nu = \prod_{i \in [k]} \nu_i^{\eta_i}$.

Then there exists a $c$-stacked permutable matrix representation (\Cref{def:c-stacked}) of $S$ with domain $(\R^n)^{\ot c\lfloor\eta/2\rfloor}$ and co-domain $(\R^n)^{\ot c\lceil\eta/2\rceil)}$, denoted $(S^{(\lfloor\eta/2\rfloor, \lceil\eta/2\rceil)})^{\mathrm{stack}(c)}$, so that
\[ \norm{\E (S^{(\lfloor\eta/2\rfloor, \lceil\eta/2\rceil)})^{\mathrm{stack}(c)}}_1 \le 2^{cdk^2/2}(c\eta)^{c\eta/2+2ck}\,\nu^{c/2}\, n^{c\lfloor\eta/2+1\rfloor/2} \]
by a rank-decomposing sum-of-squares proof.
\end{lemma}
\begin{proof}
    This works very similarly to \Cref{lem:moment-nuclear-norm}.
    The inductive case will follow the same argument, but with all $c$ stacked copies of a mode being grouped together, so that each formal variable in $S$ actually represents a vector in $(\R^n)^{\ot c}$ rather than $\R^n$. Then each invocation of \Cref{lem:rank-decomposing-induction} also applies to $n^{c}$-dimensional vector-valued variables instead of $n$-dimensional.
    This ensures that whenever  \Cref{lem:rank-decomposing-induction} reshapes a matrix, it keeps all of the $c$ copies of a particular mode on the same side of the matrix.

    In the base case, it will no longer be possible to convert all but one copy of $\Sigma$ into $\Phi$, since there may be some copies of $\Sigma$ that tie together two modes from different blocks.
    It may be possible to convert this $\Sigma$ into a $\Phi$ by moving one of the modes from being a row mode to a being a column mode in the same block, but this will not be possible if we've already used up all the row modes or column modes of the block.
    In the end, we will get at most $c$ copies of $\Sigma$, for a total nuclear norm of $(c\eta_k -1)!!\nu^{c\eta_k/2}n^{c\lfloor \eta_k/2+1\rfloor/2}$.
    
    To formalize this, we may consider the $c$ blocks as vertices with degree at most $\eta$ in a graph, with $\Sigma$s and $\Phi$s being edges.
    Each endpoint of each edge is labelled from the set of labels $\{\mathsf{row},\mathsf{col}\}$, and we are allowed to choose the labels for each endpoint.
    Each vertex must be incident to $r^{(k)}$ edge endpoints labelled $\mathsf{row}$ and $s^{(k)}$ endpoints labelled $\mathsf{col}$ (or vice-versa) in the first case, or $r^{(k)}-1$ and $s^{(k)}+1$ (or vice-versa) in the second case.
    An edge where both endpoints has the same label is an instance of $\Phi$ and one where the two labels are different is an instance of $\Sigma$.
    We add a ``phantom edge'' between pairs of vertices that both have odd degree, increasing the degree of each vertex by at most $1$ to make them all even degree.
    The graph is now a union of Eulerian graphs, and each component has an Eulerian cycle.
    
    In the $(r^{(k)},s^{(k)})$ case, we can follow each cycle, alternating between assigning $(\mathsf{row},\mathsf{row})$ labels and $(\mathsf{col}, \mathsf{col})$ labels to each edge in the cycle.
    Since each time the cycle enters a vertex, it must also leave it, this ensures that all but one vertex in the cycle gets an equal number of $\mathsf{row}$ labels and $\mathsf{col}$ labels.
    The only edge that might be forced to have mismatched endpoints is the last edge in our cycle traversal.
    There can only be one of these per connected component of the graph, so there are at most $c$ of these $\Sigma$ edges.

    In the $(r^{(k)}-1,s^{(k)}+1)$ case, we do the same thing but instead of always alternating between $(\mathsf{row},\mathsf{row})$ and $(\mathsf{col}, \mathsf{col})$, the first time we enter each vertex, we choose to exit using the same label that we entered it on.
\end{proof}

\begin{lemma}
\label[Lemma]{lem:moment-nuclear-norm-stacked-graph-matrix}
Suppose that $S$ is a multiset containing formal variables $u_1, \dots, u_k \in \R^n$ with multiplicities $\eta_1, \dots, \eta_k$ respectively.
Let $\eta = \eta_1 + \dots + \eta_k$.

Suppose further that conditioned on the values of $u_1, \dots, u_{i-1}$, each $u_i$ is a Gaussian with covariance matrix $\Sigma_i$ satisfying $\Sigma_i \preceq \nu_i \cdot \Id$ and furthermore that the covariance matrix of $u_i$ is a polynomial of degree at most $d$ in $u_1, \dots, u_{i-1}$, and let $\nu = \prod_{i \in [k]} \nu_i^{\eta_i}$.

Then there exists a $(a_1, \dots, a_c)$-stacked single-flip-permutable matrix representation (\Cref{def:single-flip-stacked}) of $S$ with domain $(\R^n)^{\ot \sum_{i\in [c]}\lfloor a_i/2\rfloor}$ and co-domain $(\R^n)^{\ot \sum_{i\in [c]}\lceil a_i/2\rceil}$, denoted $(S^{(\sum_{i\in [c]}\lfloor a_i/2\rfloor, \sum_{i\in [c]}\lceil a_i/2\rceil)})^{\mathrm{stack}(a_1, \dots, a_c)}$, so that
\[ \norm{\E (S^{(\sum_{i\in [c]}\lfloor a_i/2\rfloor, \sum_{i\in [c]}\lceil a_i/2\rceil)})^{\mathrm{stack}(a_1, \dots, a_c)}}_1 \le O(1)\,\nu^{\sum_{i\in [c]} a_i/2}\, n^{\lfloor1+\sum_{i \in [c]} a_i/2\rfloor/2} \]
by a rank-decomposing sum-of-squares proof.
\end{lemma}
\begin{proof}
    The argument will be similar to \Cref{lem:moment-nuclear-norm-stacked}.
    Again, modes in different block will be grouped together, though this time not every grouping will result in a vector of the same dimension since unlike in \Cref{lem:moment-nuclear-norm-stacked}, we can have $a_1 \ne \dots \ne a_c$.

    For the base case, we use almost the same construction as in \Cref{lem:moment-nuclear-norm-stacked}.
    The construction in \Cref{lem:moment-nuclear-norm-stacked} results in at most one $\Sigma$ per connected component of the graph.
    We use the extra degree of freedom in \Cref{def:single-flip-stacked} to flip a single mode of each block from a row mode to a column mode or vice-versa to convert those $\Sigma$s into $\Phi$s.

    First, let us consider the case when all the entries of $\eta_k$ are even.
    
    We start by taking the graph from \Cref{lem:moment-nuclear-norm-stacked} and labelling the vertices with labels in $[c]$ according to the order imposed by \Cref{def:single-flip-stacked}.
    Let $I$ be the set of vertices whose labels are the smallest among all the other vertices in the same connected component.

    Note that we are able to convert one $\Sigma$ into a $\Phi$ at the cost of converting an adjacent $\Phi$ into a $\Sigma$ by swapping the $\mathsf{row}$/$\mathsf{col}$ labels of those two edges at the two edge endpoints where they are adjacent.
    This procedure alternates between creating a $\mathsf{row} \to \mathsf{col}$ instance of $\Sigma$ and a $\mathsf{col} \to \mathsf{row}$ instance if you do it iteratively along a path.

    For each connected component, using this procedure iteratively along the Eulerian cycle, we continue until we create a $\Sigma$ that is incident to the vertex in $I$.
    Note that the construction in \Cref{lem:moment-nuclear-norm-stacked} only creates a $\Sigma$ in a connected component that has an odd-length Eulerian cycle.
    So if the $\Sigma$ that we got was incident to the vertex in $I$ via an endpoint labelled $\mathsf{col}$ when we wanted one labelled $\mathsf{row}$ for instance, we can instead create a $\Sigma$ that is incident to $I$ via a $\mathsf{row}$ endpoint by iterating the same procedure though the entire Euclidean cycle again.
    Then when it arrives back at the vertex in $I$, since the cycle was odd-length, by the alternating property of the procedure, it will now have the opposite label at the endpoint incident to $I$.

    Therefore, we can configure the graph to have the $\Sigma$s all be incident to $I$ and alternate so that, if $i_j$ is the vertex with the $j$th smallest label in $I$, the $\Sigma$ that is incident to $i_j$ is incident via a $\mathsf{col}$ endpoint if $j$ is odd and a $\mathsf{row}$ endpoint if $j$ is even.
    Up to this point, we have only permuted modes within blocks, which is always allowed by \Cref{def:single-flip-stacked}.

    Finally, we flip the endpoints of the $\Sigma$s incident to the vertices in $I$ from being $\mathsf{col}$ modes to $\mathsf{row}$ modes or vice-versa, except for the vertex in $I$ with the largest label if $|I|$ is odd.
    This turns all but one instance of $\Sigma$ into a $\Phi$, and because it makes each block imbalanced by at most one mode and it alternates the direction of the imbalance, it satisfies the constraints of \Cref{def:single-flip-stacked}.
    
    Since at most one instance of $\Sigma$ is left after this, this gives a nuclear norm of at most $(\norm{\eta_k}_1-1)!!\nu_k^{\norm{\eta_k}_1/2}n_{\vphantom{k}}^{\lfloor \norm{\eta_k}_1/2+1\rfloor/2}$ in the base case.

    When some of the entries of $\eta_k$ are odd, we use the same label-swapping procedure along the Eulerian cycles that contain the odd-degree vertices until the $\Sigma$ has been moved to one of the ``phantom edges'' from the proof of \Cref{lem:moment-nuclear-norm-stacked}, at which point the $\Sigma$ disappears from the actual moment matrix we construct.
    
\end{proof}

\subsection{Properties of Hermite polynomials}
\label{sec:hermite-properties}

We make heavy use of the existence of weakly orthogonal Hermite polynomials, whose orthogonality is implied by the constraints that make the low-degree moments/cumulants of the high-entropy step consistent with those of a (normalized) Gaussian (\Cref{eqs:the-program-intro-cum}). 

\begin{definition}
    A \emph{diagram} over a sequence of multi-indices $\alpha^{(1)}, \dots, \alpha^{(k)}$ is a graph where the $k$ vertices correspond to the $k$ multi-indices, where each edge between vertices $\alpha^{(i)}$ and $\alpha^{(j)}$ is labelled with $(a,b)$ where $a$ is an index in $\alpha^{(i)}$ and $b$ is an index in $\alpha^{(j)}$, and where each element of each multi-index is used exactly once in an edge label.

    The diagram is \emph{closed} if the graph is connected.

    The diagram is \emph{without loop} if there are no self-loops in the edge set.
\end{definition}

\begin{lemma}
    \label[Lemma]{lem:kappay1dotsykp}
    Let $\kappa_k$ denote the $k$th cumulant with respect to the the vector-valued random variable $y$.
    Suppose that $\alpha$ is a multi-index with $|\alpha| > 0$.
    Let $\beta$ be the multi-index with indices given by the multi-set $\{i_1, \dots, i_k\}$.
    Then
    \[ \kappa_{k+1}(y_{i_1}, \dots, y_{i_k},\He_{\alpha}(y \mid \Sigma))
    \;=\; \kappa_2(\He_{\beta}(y \mid \Sigma), \He_{\alpha}(y \mid \Sigma))
    \;=\; \E_y \He_{\beta}(y \mid \Sigma)\, \He_{\alpha}(y \mid \Sigma)
    \,.
    \]
    These quantities are all equal to zero if $|\alpha| \ne k$.

    In T-cumulant notation~\cite{terdik2021multivariate},
    \[ \kappa_{k+1}(y, \dots, y,\He_{\alpha}(y \mid \Sigma))
    \;=\; \kappa_2(\He_k(y \mid \Sigma), \He_{\alpha}(y \mid \Sigma))\,, \]
    where $\He_k(y \mid \Sigma) = \sum_{i_1, \dots, i_k \in [n]} \He_{i_1, \dots, i_k}(y \mid \Sigma)e_{i_1} \ot \dots \ot e_{i_k}$.

    By multilinearity of the joint cumulant and a Wiener-Hermite decomposition,
    \[ \kappa_{k+1}(y, \dots, y,p(y))
    \;=\; \kappa_2(\He_k(y \mid \Sigma), p(y)) \]
    for any vector-valued polynomial $p$.
\end{lemma}
\begin{proof}
    By \cite[Proposition 4.3]{terdik2021multivariate},
    \[ \kappa_{k+1}(y_{i_1}, \dots, y_{i_k},\He_{\alpha}(y \mid \Sigma)) = \sum_{(V,E) \in \cG}\prod_{(u,v,(i,j)) \in E} \Sigma_{i,j}\,, \]
    where $\cG$ is the set of closed diagrams without loop over $\{i_1\}$, \dots, $\{i_k\}$, and $\alpha$.

    By \cite[Proposition 4.3]{terdik2021multivariate},
    \[ \kappa_{2}(\He_{\beta}(y \mid \Sigma),\He_{\alpha}(y \mid \Sigma)) = \sum_{(V,E) \in \cG'}\prod_{(u,v,(i,j)) \in E} \Sigma_{i,j}\,, \]
    where $\cG'$ is the set of closed diagrams without loop over $\beta$ and $\alpha$.

    Since every diagram in $\cG$ is closed and each of the vertices corresponding to the singleton sets $\{i_1\}$, \dots, $\{i_k\}$ must have degree 1, they must all have exactly one edge with the vertex corresponding to $\alpha$.
    Therefore, there is exactly one diagram for each way to pair the elements of $\{i_1, \dots, i_k\}$ with the elements of $\alpha$ (treating duplicate elements as distinct, so that there are $k!$ diagrams), and the summand corresponding to each diagram is the product of $\Sigma_{i,j}$ when $i \in \{i_1, \dots, i_k\}$ is paired with $j \in \alpha$.
    
    Since $\cG'$ is the set of diagrams without loop, every $i \in \beta$ must be paired with some unique $j \in \alpha$.
    So again, there is one diagram for every way to pair the elements of $\{i_1, \dots, i_k\}$ with the elements of $\alpha$, with the same value of summands.
    
    Since these two cumulants evaluate to the same sum over the same summands, they must be equal.
    In both cases, the result is non-zero only if $|\alpha| = k$.

    Finally, $\kappa_2(\He_{\beta}(y \mid \Sigma), \He_{\alpha}(y \mid \Sigma))
    = \E_y \He_{\beta}(y \mid \Sigma)\, \He_{\alpha}(y \mid \Sigma)$ because every Hermite of positive degree is zero-mean.
\end{proof}

\begin{lemma}
    \label[Lemma]{lem:kappay1dotsykpq}
    Let $\kappa$ denote cumulants with respect to the the vector-valued random variable $y$.
    Suppose that $\alpha$ is a multi-index with $|\alpha| > 0$.
    Let $\gamma$ be the multi-index with indices given by the multi-set $\{i_1, \dots, i_k\}$.
    Then
    \[ \kappa_{k+1}(y_{i_1}, \dots, y_{i_k},\He_{\alpha}(y \mid \Sigma),\He_{\beta}(y \mid \Sigma))
    \;=\; \kappa_3(\He_{\gamma}(y \mid \Sigma), \He_{\alpha}(y \mid \Sigma), \He_{\beta}(y \mid \Sigma))
    \,.
    \]
    These quantities are all equal to zero if $|\alpha| + |\beta| + k$ is not even.
\end{lemma}
\begin{proof}
    The argument is the same as in \Cref{lem:kappay1dotsykp}, just with one more vertex in the diagram.
\end{proof}

\begin{lemma}
\label[Lemma]{lem:hermite-quadraticization}
    Let $y \in \R^n$ be a Gaussian vector $y \sim \cN(0,\Sigma)$ with $\Sigma\preceq\nu\cdot\Id_n$.
    Let $\gamma$, $\alpha$, and $\beta$ be multi-indices over $[n]$.
    Then
    \begin{align*}
        \kappa_3(\He_{\gamma}(y \mid \Sigma), \He_{\alpha}(y \mid \Sigma), \He_{\beta}(y \mid \Sigma))
        = \sum_{\substack{\gamma_1,\gamma_2\\\!\!\!\!\!\!\!\!\text{partition of }\gamma\!\!\!}} \sum_{\alpha' \le \alpha} \sum_{\beta' \le \beta} \kappa_2(\He_{\gamma_1}, \He_{\alpha'})\kappa_2(\He_{\gamma_2}, \He_{\beta'})\kappa_2(\He_{\alpha-\alpha'}, \He_{\beta - \beta'})\,,
    \end{align*}
    where we have elided some repetitive notation on the right.
\end{lemma}
\begin{proof}
    By \cite[Proposition 4.3]{terdik2021multivariate},
    \[ \kappa_{3}(\He_{\gamma}, \He_{\alpha}, \He_{\beta}) = \sum_{(V,E) \in \cG}\prod_{(u,v,(i,j)) \in E} \Sigma_{i,j}\,, \]
    where $\cG$ is the set of closed diagrams without loop over $\gamma$, $\alpha$, and $\beta$.
    This set of diagrams can be partitioned by which of the edges coming out of $\gamma$ go to $\alpha$ and which go to $\beta$: that gives us the sum over partitions $(\gamma_1, \gamma_2)$ of $\gamma$.
    It can be further partitioned by which set of indices in $\alpha$ are linked to $\gamma_1$ and which set of indices in $\beta$ are linked to $\gamma_2$: this gives us the sum over $\alpha'$ and $\beta'$ respectively.

    Finally, once we have determined these values of $\gamma_1$, $\gamma_2$, $\alpha'$, and $\beta'$, to enumerate all the diagrams consistent with these choices, it suffices to enumerate all the ways that the elements of $\gamma_1$ can be paired with the indices in $\alpha'$ and then take the product of $\Sigma_{i,j}$ for each corresponding pair of $i \in \gamma_1$ and $j \in \alpha'$.
    But this is the same as enumerating all closed diagrams without loop between $\gamma_1$ and $\alpha'$, hence this factor is equal to $\kappa_2(\He_{\gamma_1}, \He_{\alpha'})$.
    The same applies to the other two pairings, hence the final formula.
\end{proof}

\begin{lemma}
\label[Lemma]{lem:hermite-cum-matrix-opnorm}
Fix a $\Sigma \in \R^{n \times n}$ PSD and $k \ge 0$.

Consider the vector space spanned by $\{\He_{\alpha} : |\alpha| = k \}$ where $\alpha$ ranges over multi-indices in $[n]$, whose dual vectors are spanned by $\{\He_{\alpha}^{\sT} : |\alpha| = k \}$, and equipped with the $L^2(\Sigma)$ inner product, the linear extension of
\[\iprod{\He_{\alpha}, \He_{\beta}}_{L^2(\Sigma)} := \E_{x \sim \cN(0,\Sigma)} \He_{\alpha}(x \mid {\Sigma})\, \He_{\beta}(x \mid {\Sigma}) \,. \]

Let $\mathbf{H}^{(k)}$ be the linear operator given by
\[ \mathbf{H}^{(k)} = \sum_{\substack{I\in[n]^k}} \He_{I}\He_{I}^{\sT}\,, \]
where $\He_{I} := \He_{\alpha}$ when $I$ is a tuple of indices and $\alpha$ is the multi-index constituted of the same indices. 
Then
\[ \opnorm{\mathbf{H}^{(k)}} \le k!\,\opnorm{\Sigma}^k\,. \]
\end{lemma}
\begin{proof}
    Let $\mathbf{G}^{(k)}$ be the linear operator given entrywise by
\[ \mathbf{G}^{(k)}_{I,J} := \iprod{ \He_{I},\He_{J}}_{L^2(\Sigma)}\,. \]
By weak orthonormality (\Cref{sec:prelims-multivar-hermites}),
\[ \opnorm{\mathbf{H}^{(k)}} = \opnorm{\mathbf{G}^{(k)}}\,, \]
so we can just bound the latter.

    By \Cref{lem:hermite-inner-prods}, we have
    \[ \mathbf{G}^{(k)}_{I,J} = \sum_{\pi \in S_k} \prod_{i \in [k]} \Sigma_{I_i, J_{\pi(i)}}
    \,. \]
    We can decompose this matrix further as $\mathbf{G}^{(k)} = \sum_{\pi \in S_k} \mathbf{G}^{(k,\pi)}$ where
        \[ \mathbf{G}^{(k,\pi)}_{I,J} = \prod_{i \in [k]} \Sigma_{I_i,J_{\pi(i)}}
    \,. \]
    Now we can invert the permutation $\pi$ to get
        \[ \mathbf{G}^{(k,\pi)}_{I,\pi^{-1}J} = \prod_{i \in [k]} \Sigma_{I_i,J_i}
        \,,
     \]
     where $\pi^{-1}(j_1, \dots, j_k) = (j_{\pi^{-1}(1)},\dots,j_{\pi^{-1}(k)})$, and see therefore that
        \[ \mathbf{G}^{(k,\pi)}P_{\pi^{-1}} =  \Sigma^{\ot k}\,,
     \]
     where $P_{\pi^{-1}}e_{J} = e_{\pi^{-1}J}$ is a mode permutation matrix.
     Since $P_{\pi^{-1}}$ is an orthogonal matrix, we have
     \[  \opnorm{\mathbf{G}^{(k,\pi)}}
     = \opnorm{\mathbf{G}^{(k,\pi)}P_{\pi^{-1}}}
     = \opnorm{\Sigma^{\ot k}}
     = \opnorm{\Sigma}^k
     \,.\]
     And since $\mathbf{G}^{(k)}$ is the sum of $k!$ operators all whose operator norms are equal to $\opnorm{\Sigma}^k$, its own operator norm is no more than $k!\,\opnorm{\Sigma}^k$.
\end{proof}

\begin{corollary}
    \label[Corollary]{cor:hermite-operator-norm}
    Let $y \in \R^n$ be a Gaussian vector $y \sim \cN(0,\Sigma)$ with a SoS proof that $\Sigma\preceq\nu\cdot\Id_n$ and let $p(y)$ be a vector-valued polynomial.
    Then
    \[ \sum_{I \in [n]^k}  \norm{\E_y\He_I(y)p(y)}_2^2
    \le k!\opnorm{\Sigma}^k\E_y \norm{p(y)}_2^2\,.\]
    by a SoS proof. As an immediate consequence,
    \[ \norm{\kappa_2\left(\He_k(y),p(y)\right)}_2^2
    \le k!\opnorm{\Sigma}^k\E_y \norm{p(y)}_2^2\,, \]
    where $\He_k(y) = \sum_{i_1, \dots, i_k \in [n]} \He_{i_1, \dots, i_k}(y)e_{i_1} \ot \dots \ot e_{i_k}$ and $\kappa_2(p,q)$ is the second tensor cumulant (or T-cumulant~\cite{terdik2021multivariate}; equivalently, the covariance) between vector-valued polynomials $p$ and $q$ with respect to the randomness of $y$.
\end{corollary}
\begin{proof}
    Apply \Cref{lem:hermite-cum-matrix-opnorm}, writing $\norm{\kappa_2(p(y),\He_I(y))}_2^2$ as 
    \[\sum_j \kappa_2(p_j(y),\He_I(y))^2 = \sum_j \iprod{p_j,\He_I}_{L^2(\Sigma)}^2
    \le  k!\opnorm{\Sigma}^k\sum_j\norm{p_j}_{L^2(\Sigma)}^2
    \,,\]
    where $j$ is summed over all indices of $p$ as a vector and we applied \Cref{lem:holders-2-inf-2} as a SoS proof.
\end{proof}

\begin{lemma}
    \label[Lemma]{lem:hermite-gram-tensor}
    Consider the same vector space as in \Cref{lem:hermite-cum-matrix-opnorm}.
    
    Let
    \[ \mathbf{G}^{(k)}_{\alpha,\beta} := \iprod{ \He_{\alpha},\He_{\beta}}_{\ell^2(\Sigma)}\]
    be the Gram matrix indexed by $\alpha$ and $\beta$ both multi-indices satisfying $|\alpha|=|\beta|=k$.
    
    Let also
    \[ \mathbf{G}^{(k,\ell)}_{(\alpha,\alpha'),(\beta,\beta')} := \iprod{ \He_{\alpha+\alpha'},\He_{\beta+\beta'}}_{\ell^2(\Sigma)}\]
    when $|\alpha| = |\beta| = k$ and $|\alpha'| = |\beta'| = \ell$.

    Let $\sym^{(k,\ell)}$ be the symmetric matrix operating on the space whose coordinates are indexed by two multi-indices $(\alpha,\alpha')$ satisfying $|\alpha| = k$ and $|\alpha'| = \ell$, such that $\sym^{(k,\ell)}$ symmetrizes between all pairs of lists whose concantenations contain the same multi-set of elements.
    In other words,
    \[\sym^{(k,\ell)}e_{\alpha}\ot e_{\alpha'} = \frac{1}{|\{(\beta,\beta'): \alpha+\alpha'=\beta+\beta', |\beta|=k,|\beta'|=\ell\}|}\sum_{\substack{|\beta| = k, |\beta'|=\ell\\\beta +\beta' = \alpha + \alpha'}} e_{\beta}\ot e_{\beta'}\,, \]
    
    Then
    \[\sym^{(k,\ell)}\left(\mathbf{G}^{(k)} \ot \mathbf{G}^{(\ell)}\right)\sym^{(k,\ell)} = \binom{k+\ell}{k}^{-1} \mathbf{G}^{(k,\ell)} \,.\]
\end{lemma}
\begin{proof}
    By \Cref{lem:hermite-inner-prods}, we have
    \[ \mathbf{G}^{(k)}_{\alpha,\beta} = \sum_{\pi \in S_k} \prod_{i \in [k]} \Sigma_{\alpha[i],\beta[\pi(i)]}
    \,, \]
    where $\alpha[\cdot]$ and $\beta[\cdot]$ index into some arbitrary fixed ordering of the indices.

    So 
    \[ \left(\mathbf{G}^{(k)}\ot\mathbf{G}^{(\ell)}\right)_{(\alpha,\alpha'),(\beta,\beta')} = \sum_{\pi \in S_k}\sum_{\sigma \in S_{\ell}} \left(\prod_{i \in [k]} \Sigma_{\alpha[i],\beta[\pi(i)]}\right)\left(\prod_{i \in [\ell]} \Sigma_{\alpha'[i],\beta'[\sigma(i)]}\right)
    \,, \]

    And
    \[\mathbf{G}^{(k,\ell)}_{(\alpha,\alpha'),(\beta,\beta')} = \sum_{\pi \in S_{k+\ell}} \left(\prod_{i \in [k+\ell]} \Sigma_{(\alpha+\alpha')[i],(\beta+\beta')[\pi(i)]}\right)\,. \]

    Symmetrizing the former expression leads to $\binom{k+\ell}{k}^{-1}$ times the latter expression.
\end{proof}

\begin{lemma}
\label[Lemma]{lem:hermite-fourier-tensor}
    Let $p$ be a degree-$d$ vector-valued polynomial with Wiener-Hermite coefficients $\hat{p}(\alpha)$.
    Let $B$ be the matrix whose $\alpha$th column is $\He_{\alpha}$ defined as a vector as in the previous lemmata and let $P$ be the matrix whose $\alpha$th column is $\hat{p}(\alpha)$ when $\alpha$ is a multi-index.
    Let $P'$ be the matrix whose $(\alpha',\alpha'')$th column is $\hat{p}(\alpha' + \alpha'')$ when both $\alpha'$ and $\alpha''$ are multi-indices and neither are empty.
    Then
    \[  \norm{(B\ot B){P'}^{\sT}}_F^2 \le (2^{d}-1)\,\norm{BP^{\sT}}_F^2\,.\]
\end{lemma}
\begin{proof}
    Let $B^{(k)}$ and $P^{(k)}$ be defined the same way as $B$ and $P$ except the multi-index $\alpha$ must satisfy $|\alpha| = k$.
    Similarly, $P^{(k,\ell)}$ is defined the same way as $P'$, except that $|\alpha'| = k$ and $|\alpha''| = \ell$.

    Then
    \[(B \ot B){P'}^{\sT} = \bigoplus_{k,\ell \ge 1} (B^{(k)} \ot B^{(\ell)}){P^{(k,\ell)}}^{\sT}\,,\]
    where $\oplus$ here denotes a matrix concatenation.
    By the weak orthogonality of the Hermites and the definition of the inner product involved in $B$ (\Cref{lem:hermite-cum-matrix-opnorm}), 
    \[\norm{(B\ot B){P'}^{\sT}}_F^2 = \sum_{k,\ell \ge 1}\norm{(B^{(k)} \ot B^{(\ell)}){P^{(k,\ell)}}^{\sT}}_F^2\,.\]
    By \Cref{lem:hermite-gram-tensor} and the fact that $P^{(k,\ell)}_{\alpha',\alpha''}$ depends only on the sum $\alpha'+\alpha''$,
    \[\norm{(B\ot B){P'}^{\sT}}_F^2 = \sum_{k,\ell \ge 1}\binom{k+\ell}{k}^{-1}\norm{B^{(k,\ell)}{P^{(k,\ell)}}^{\sT}}_F^2\,,\]
    where $B^{(k,\ell)}$ is the matrix whose $(\alpha',\alpha'')$th column is $\He_{\alpha'+\alpha''}$ when $|\alpha'|=k$ and $|\alpha''|=\ell$.
    Reparameterizing, 
    \[\norm{(B\ot B){P'}^{\sT}}_F^2 = \sum_{k \ge 1}\sum_{\ell = 1}^{k}\binom{k}{\ell}^{-1}\norm{B^{(\ell,k-\ell)}{P^{(\ell,k-\ell)}}^{\sT}}_F^2\,.\]
    Note that $B^{(\ell,k-\ell)}$ and $P^{(\ell,k-\ell)}$ contain the same columns as $B^{(k)}$ and $P^{(k)}$, but the former may contain repeats of the same column up to $\binom{k}{\ell}$ times.
    Therefore,
    \begin{align*}\norm{(B\ot B){P'}^{\sT}}_F^2 &\le \sum_{k \ge 1}\sum_{\ell = 1}^{k}\binom{k}{\ell}^{-1}\binom{k}{\ell}^{2}\norm{B^{(k)}{P^{(k)}}^{\sT}}_F^2
    \\&= \sum_{k \ge 1}\norm{B^{(k)}{P^{(k)}}^{\sT}}_F^2\sum_{\ell = 1}^{k}\binom{k}{\ell}
    \\&= \sum_{k \ge 1}(2^k-1)\norm{B^{(k)}{P^{(k)}}^{\sT}}_F^2
    \,.\end{align*}
    By $1-\infty$ Holder's inequality and the weak orthogonality of the Hermite polynomials,
    \[\norm{(B\ot B){P'}^{\sT}}_F^2 \le (2^d-1)\,\norm{BP^{\sT}}_F^2\,.\]
\end{proof}

\subsection{Matrix reshaping tools}
\label{sec:reshaping-tools}

\begin{lemma}
\label[Lemma]{lem:vector-folding-nuc-norm}
    Suppose $v \in \R^{\ot n} \ot \R^{\ot m}$ and $C > 0$ is a constant.
    Then there exists a set $I$ and a decomposition
    \[v = \sum_{i\in I} w_i \ot z_i\]
    such that $w_i \in \R^n$ and $z_i \in \R^m$ and
    \[
    \sum_{i \in I} \norm{w_i}_2^2 = C\sqrt{n} \qquad\text{and}\qquad \sum_{i\in I}\norm{z_i}_2^2 = \frac{\sqrt{n}}{C}\norm{v}_2^2\,.\]
    In particular,
    \[\sum_{i \in I} \frac{1}{2}(\norm{w_i}_2^2 + \norm{z_i}_2^2) = \frac{C\sqrt{n}}{2} + \frac{\sqrt{n}}{2C}\norm{v}_2^2\,. \]
    
    Therefore if $\norm{v}_2^2 \le C^2$ by a degree-$d$ sum-of-squares proof, then there is a degree-$d$ rank-decomposing sum-of-squares proof that $\norm{V}_1 \le C\sqrt{n}$, where $V \in \R^{n \times m}$ is the matrix reshaping of $v$.
\end{lemma}
\begin{proof}
    We take $I = [n]$ and $w_i = \sqrt[4]{C^2/n}\,e_i$ and $z_i = \sqrt[4]{n/C^2}\,(e_i^{\sT} \ot \Id_m)v$.
    Then
    \begin{align*}
         \sum_{i \in [n]} w_i \ot z_i &= \sum_{i \in [n]} \sqrt[4]{C^2/n}\;e_i \ot \left[\sqrt[4]{n/C^2}(e_i^{\sT} \ot \Id_m)v\right]
    \\&= \sum_{i \in [n]}e_i \ot \left[(e_i^{\sT} \ot \Id_m)v\right]
    \\&=  \left(\sum_{i \in [n]}(e_i^{\vphantom{T}}e_i^{\sT}) \ot \Id_m\right)\,v
    \\&= (\Id_n \ot \Id_m)\,v
    \\&= v
    \end{align*}
    and by using the multiplicativity of norms over tensor products to introduce a tensor factor of $e_i$ without changing the norm, and then the orthonormality of $\{e_i\}_i$ to equate the sum of the square norms to the square norm of the sums,
    \begin{align*}
        \sum_{i \in [n]} \frac{1}{2}(\norm{w_i}_2^2 + \norm{z_i}_2^2)
    &= \sum_{i \in [n]}\frac{1}{2}\left(\frac{C}{\sqrt{n}} + \frac{\sqrt{n}}{C} \norm{(e_i^{\sT} \ot \Id_m)v}_2^2\right)
    \\&= \frac{C\sqrt{n}}{2} + \frac{\sqrt{n}}{2C}\sum_{i \in [n]}\norm{(e_i^{\sT} \ot \Id_m)v}_2^2
    \\&= \frac{C\sqrt{n}}{2} + \frac{\sqrt{n}}{2C}\sum_{i \in [n]}\norm{e_i\ot \left[(e_i^{\sT} \ot \Id_m)v\right]}_2^2
    \\&= \frac{C\sqrt{n}}{2} + \frac{\sqrt{n}}{2C}\norm{\sum_{i \in [n]}e_i\ot \left[(e_i^{\sT} \ot \Id_m)v\right]}_2^2
    \\&= \frac{C\sqrt{n}}{2} + \frac{\sqrt{n}}{2C}\norm{v}_2^2
    \,.
    \end{align*}
\end{proof}

\subsection{Cumulant nuclear norm bounds}
\label{sec:cum-nuc-norm-bounds}

Up to constants, \Cref{lem:ydotsyp-vec} is a special case of \Cref{lem:ydotsypq}, but we leave the simpler argument in for readability.

\begin{lemma}
\label[Lemma]{lem:ydotsyp-vec}
    Let $y \in \R^n$ be a Gaussian vector $y \sim \cN(0,\Sigma)$ with $\Sigma\preceq\nu\cdot\Id_n$ and let $p \in \R[y]^{m}$ be a vector-valued function of $y$.

    Then 
    \[ \norm{\kappa_{k+1}(\underbrace{y, \dots, y}_k, p)}_2^2 \le k!\,\nu^{k}\E_y \norm{p(y)}_2^2 \]
    by a sum-of-squares proof of degree $2k$ in $\Sigma$ and also degree $2k$ times the largest degree of a $({\le}k)$th order Wiener-Hermite coefficient of $p$ in any other indeterminates.
\end{lemma}
\begin{proof}
By \Cref{lem:kappay1dotsykp} then \Cref{cor:hermite-operator-norm},
    \begin{align*}
    \norm{\kappa_{k+1}(y,\dots,y, p)}_2^2
    = \norm{\kappa_{2}(\He_k(y), p)}_2^2
    \le k!\nu^k\E_y\norm{p(y)}_2^2\,.
    \end{align*}
\end{proof}

\begin{lemma}
\label[Lemma]{lem:ydotsypq}
    Let $y \in \R^n$ be a Gaussian vector $y \sim \cN(0,\Sigma)$ with $\Sigma \preceq \nu\cdot\Id_n$ and let $p \in (\R[y]^{n})^{\ot \ell_1}$ and $q \in (\R[y]^{n})^{\ot \ell_2}$ be vector-valued polynomials in $y$ with degree at most $d$.

    Then for any $r \in \{0, \dots, \lfloor(k+|\ell_2 - \ell_1|)/2\rfloor\}$,
    \[ \kappa_{k+2}(\underbrace{y, \dots, y}_k, p, q) \]
    can be re-represented as a sum of terms of the form 
    \[\sum_{j\in J} w_jz_j^T \]
    for some index set $J$ so that $w_j \in (\R^n)^{\ot (\max(\ell_1,\ell_2) + k-r)}$ and $z_j \in (\R^n)^{\ot (\min(\ell_1,\ell_2) + r)}$ and both are polynomials of degree at most $\lfloor(d+k)/2\rfloor$ in the covariance $\Sigma$ and
    \[ \sum_{j\in J} \E_y \frac{1}{2}(\norm{w_j}_2^2 + \norm{z_j}_2^2) \le (2^d-1)k^{k/2+1}\nu^{k/2}n^{\lfloor k/2\rfloor/2}
    \cdot\frac{1}{2}\E_y(\norm{p}_2^2 + \norm{q}_2^2) \,.\]
    This inequality has a sum-of-squares proof.
\end{lemma}
\begin{proof}
    By the multilinearity of the joint cumulant and by taking Wiener-Hermite decompositions,
    \begin{align*} \kappa_{k+2}(y,\dots,y, p, q)
    &=
    \sum_{i_1, \dots, i_k \in [n]}\sum_{\alpha, \beta} \kappa_{k+2}(e_{i_1}y_{i_1},\, \dots,\, e_{i_k}y_{i_k},\, \hat{p}(\alpha)\He_{\alpha}(y),\, \hat{q}({\beta})\He_{\beta}(y))
    \\&= 
    \sum_{i_1, \dots, i_k \in [n]}\sum_{\alpha, \beta} \kappa_{k+2}(y_{i_1}, \dots, y_{i_k}, \He_{\alpha}(y), \He_{\beta}(y))\;\cdot\; e_{i_1} \ot \dots \ot e_{i_k} \ot \hat{p}(\alpha) \ot \hat{q}({\beta})
    \\&= 
    \sum_{I \in [n]^k}\sum_{\alpha, \beta} \kappa_{3}(\He_{I},\He_{\alpha},\He_{\beta})\;\cdot\; e^{\ot I} \ot \hat{p}(\alpha) \ot \hat{q}({\beta})
    \,,
    \end{align*}
    where in the last equality we have used \Cref{lem:kappay1dotsykpq} and introduced the notational convenience $e^{\ot I} := e_{i_1} \ot \dots \ot e_{i_k}$ when $I = (i_1, \dots, i_k)$.

    At this point we will decompose the cumulant $\kappa_{3}(\He_{I},\He_{\alpha},\He_{\beta})$ into terms capturing the different interactions between $I$, $\alpha$, and $\beta$, then group those terms together based on how many variables in each of the blocks are interacting with each other block.
    This will allow us to find hidden rank-1 terms with smaller norm than if you had simply added up the norms of their components.
    
    By \Cref{lem:hermite-quadraticization},
    \begin{align*}
         &\kappa_{k+2}(y,\dots,y, p, q) \\&= \sum_{I \in [n]^k}\sum_{\alpha,\beta}\sum_{\substack{(I_1, I_2)\\\text{partition of } I}} \sum_{\alpha' \le \alpha} \sum_{\beta' \le \beta} \kappa_2(\He_{I_1}, \He_{\alpha'})\,\kappa_2(\He_{I_2}, \He_{\beta'})\,\kappa_2(\He_{\alpha-\alpha'}, \He_{\beta - \beta'})\, e^{\ot I} \ot \hat{p}(\alpha) \ot \hat{q}({\beta})
         \\&=  \sum_{\ell \in \{0,\dots,k\}}\binom{k}{\ell}\sum_{\alpha'',\beta'',\alpha',\beta'}\sum_{I_1 \in [n]^{\ell}}\sum_{I_2 \in [n]^{k-\ell}} \kappa_2(\He_{I_1}, \He_{\alpha'})\,\kappa_2(\He_{I_2}, \He_{\beta'})\,\kappa_2(\He_{\alpha''}, \He_{\beta''})\,e^{\ot (I_1+I_2)} \ot \hat{p}(\alpha' + \alpha'') \ot \hat{q}({\beta' + \beta''}) \,.
    \end{align*}
    For each $\ell$, the above sum over $\alpha', \beta', \alpha'', \beta'', I_1, I_2$ can be re-represented as the matrix
    \begin{align}
    &\sum_{\alpha'',\beta''} \kappa_2(\He_{\alpha''}, \He_{\beta''})\,\left(\sum_{I_1 \in [n]^{\ell}}\sum_{\alpha'}\kappa_2(\He_{I_1}, \He_{\alpha'})\,e^{\ot I_1}\ot\hat{p}(\alpha' + \alpha'') \right)\left(\sum_{I_2 \in [n]^{k-\ell}}\sum_{\beta'}\kappa_2(\He_{I_2}, \He_{\beta'})\,e^{\ot I_2} \ot \hat{q}({\beta' + \beta''})\right)^{\sT}
    \nonumber
    \\&:=\sum_{\alpha'',\beta''} \kappa_2(\He_{\alpha''}, \He_{\beta''})\,\left(L^{(\ell,\alpha'')} \right)\left(R^{(\ell,\beta'')}\right)^{\sT}
    \nonumber
    \\&= \E_y \left(\sum_{\alpha''} L^{(\ell,\alpha'')}\He_{\alpha''}(y) \right)\left(\sum_{\beta''} R^{(\ell,\beta'')}\He_{\beta''}(y)\right)^{\sT}
    \,,\nonumber\label{eq:L-R-decomp} \end{align}
    such that $ L^{(\ell,\alpha'')} \in (\R^n)^{\ot (\ell_1 + \ell)}$ and $R^{(\ell,\beta'')} \in (\R^n)^{\ot (\ell_2 + k - \ell)}$.

    Assume without loss of generality that $\ell_1 \le \ell_2$ (otherwise we can exchange $p$ with $q$).
    Assume also without loss of generality that $\ell_1+\ell \ge \max(\ell_1,\ell_2) + k-r$ (otherwise we have $\ell_2 + (k-\ell) \ge \min(\ell_1,\ell_2)+r$ and we can apply the next step to $\sum_{\beta''} R^{(\ell,\beta'')}\He_{\beta''}(y)$ rather than $\sum_{\alpha''} L^{(\ell,\alpha'')}\He_{\alpha''}(y)$) and let 
    \[\Delta \ell := \ell_1+\ell - (\max(\ell_1,\ell_2) + k-r) \ge 0\,.\]
    Note that since $k-r \ge k/2-|\ell_1-\ell_2|/2 = k/2-(\ell_2-\ell_1)/2$ by definition of $r$, we have $\Delta \ell = (\ell_1 - \ell_2) + \ell - (k-r) \le \ell - k/2 - (\ell_2-\ell_1)/2 \le k/2 - (\ell_2-\ell_1)/2 \le k/2$.
    
    Now by \Cref{lem:vector-folding-nuc-norm} applied to $\sum_{\alpha''} L^{(\ell,\alpha'')}\He_{\alpha''}(y)$, there is a decomposition
    \[ \sum_{j \in J^{(\ell)}} x^{(\ell)}_{j} \ot y^{(\ell)}_{j} = \sum_{\alpha''} L^{(\ell,\alpha'')}\He_{\alpha''}(y)
    \,.\]
    such that $x^{(\ell)}_j \in (\R^n)^{\ot (\max(\ell_1,\ell_2) + k - r)}$ and $y^{(\ell)}_{j} \in (\R^n)^{\ot \Delta \ell}$ and
    \[\sum_{j\in J^{(\ell)}}\norm{x^{(\ell)}_{j}}_2^2 = \sqrt{n^{\Delta \ell}}\norm{\sum_{\alpha''} L^{(\ell,\alpha'')}\He_{\alpha''}(y)}_2^2,\qquad
    \text{and}\qquad
    \sum_{j\in J^{(\ell)}}\norm{y^{(\ell)}_{j}}_2^2 = \sqrt{n^{\Delta \ell}}\,.\]
    So we will take our final decomposition as
    \[ w_{(\ell,j)} := (k\nu)^{(k-2\ell)/4}\,x^{(\ell)}_{j}\qquad\text{and}\qquad z_{(\ell,j)} := (k\nu)^{(2\ell-k)/4}\,y^{(\ell)}_{j} \ot \sum_{\beta''} R^{(\ell,\beta'')}\He_{\beta''}(y) \]
    and we have
    \begin{align*} \sum_{\ell,j}  \norm{w_{(\ell,j)}}_2^2 
    &= \sum_{\ell,j}  \sqrt{(k\nu)^{k-2\ell}}\norm{x^{(\ell)}_{j}}_2^2
    \\&= \sum_{\ell}\sqrt{(k\nu)^{k-2\ell}n^{\Delta \ell}}\norm{\sum_{\alpha''} L^{(\ell,\alpha'')}\He_{\alpha''}(y)}_2^2
    \\&= \sum_{\ell}\sqrt{(k\nu)^{k-2\ell}n^{\Delta \ell}}\norm{\sum_{\alpha''}\sum_{I_1 \in [n]^{\ell}}\sum_{\alpha'}\kappa_2(\He_{I_1}, \He_{\alpha'})\left(e^{\ot I_1}\ot\hat{p}(\alpha' + \alpha'') \right)\He_{\alpha''}(y)}_2^2
    \,.
    \end{align*}
    By \Cref{cor:hermite-operator-norm}, introducing the replica $y^{(1)}$ to represent the copy of $y$ being expected over in the cumulant (as in \Cref{sec:replica-equivalance}),
    \begin{align*}
    \sum_{\ell,j} \norm{w_{(\ell,j)}}_2^2 
    &\le \sum_{\ell}\sqrt{(k\nu)^{k-2\ell}n^{\Delta \ell}}\ell!\nu^{\ell}\E_{y^{(1)}}\norm{\sum_{\alpha''}\sum_{\alpha'} \He_{\alpha'}(y^{(1)})\hat{p}(\alpha' + \alpha'') \He_{\alpha''}(y)}_2^2
    \,.
    \end{align*}

    At this point, we have a lot of summations that we would like to interpret as matrix multiplications.
    To this end, let $\bar{B}_{\ell}$ be the matrix whose $I_1$th column is $\He_{I_1}$ when $I_1 \in [n^{\ell}]$, let $B$ be the matrix whose $\alpha$th column is $\He_{\alpha}$ when $\alpha$ is a multi-index, and let $P^{(2)}$ be the matrix whose $(\alpha',\alpha'')$th column is $\hat{p}(\alpha' + \alpha'')$ when both $\alpha'$ and $\alpha''$ are non-empty multi-indices.
    Then, recalling that the second joint cumulant between two Hermite polynomials is equal their the inner product, we can write the previous expression as
    \begin{align*} \sum_{\ell,j} \E_y \norm{w_{(\ell,j)}}_2^2 
    &= \sum_{\ell}\sqrt{(k\nu)^{k-2\ell}n^{\Delta \ell}}\norm{\sum_{\alpha''}\sum_{I_1 \in [n]^{\ell}}\sum_{\alpha'}\iprod{\bar{B}_{\ell}e^{\ot I_1}, Be_{\alpha'}}\left(e^{\ot I_1}\ot \left(P^{(2)}(e_{\alpha'}\ot e_{\alpha''}) \right)\right)(Be_{\alpha''})^{\sT}}_F^2
    \\&= \sum_{\ell}\sqrt{(k\nu)^{k-2\ell}n^{\Delta \ell}}\norm{\sum_{\alpha''}\sum_{I_1 \in [n]^{\ell}}\sum_{\alpha'}\left(\left(\left(e^{\ot I_1}{e^{\ot I_1}}^{\sT}\right)\bar{B}_{\ell}^{\sT}Be_{\alpha'}\right)\ot \left(P^{(2)}(e_{\alpha'}\ot e_{\alpha''}) \right)\right)(Be_{\alpha''})^{\sT}}_F^2
    \,,\end{align*}
    where we commuted the inner product through since it is a scalar, and rewrote the inner product using a transpose.
    Next, we recognize that the sum over $I_1$ is an identity matrix, and then use the multilinearity of tensor products and the linearity of matrix multiplication to group together the elementary basis vectors in order to find another identity matrix from the sum over $\alpha''$:
    \begin{align*}
    \sum_{\ell,j} \E_y \norm{w_{(\ell,j)}}_2^2
    &= \sum_{\ell}\sqrt{(k\nu)^{k-2\ell}n^{\Delta \ell}}\norm{\sum_{\alpha''}\sum_{\alpha'}\left(\left(\bar{B}_{\ell}^{\sT}Be_{\alpha'}\right)\ot \left(P^{(2)}(e_{\alpha'}\ot e_{\alpha''}) \right)\right)(Be_{\alpha''})^{\sT}}_F^2
    \\&= \sum_{\ell}\sqrt{(k\nu)^{k-2\ell}n^{\Delta \ell}}\norm{\sum_{\alpha''}\sum_{\alpha'}\left(\left(\bar{B}_{\ell}^{\sT}B\right)\ot P^{(2)}\right)\left(e_{\alpha'} \ot e_{\alpha'}\ot e_{\alpha''} \right)e_{\alpha''}^{\sT}B^{\sT}}_F^2
    \\&= \sum_{\ell}\sqrt{(k\nu)^{k-2\ell}n^{\Delta \ell}}\norm{\sum_{\alpha'}\left(\left(\bar{B}_{\ell}^{\sT}B\right)\ot P^{(2)}\right)\left(e_{\alpha'} \ot e_{\alpha'}\ot B^{\sT} \right)}_F^2
    \,.
    \end{align*}
    Now we use the fact that the Frobenius norm is invariant under tensor/matrix reshapings, and bring the first mode on the left over to the right, and find our last identity matrix .
    \begin{align*} \sum_{\ell,j} \E_y \norm{w_{(\ell,j)}}_2^2 
    &= \sum_{\ell}\sqrt{(k\nu)^{k-2\ell}n^{\Delta \ell}}\norm{\sum_{\alpha'} P^{(2)}\left( e_{\alpha'}^{\sT}(e_{\alpha'}B^{\sT}\bar{B}_{\ell})\ot B^{\sT} \right)}_F^2 
    \\&= \sum_{\ell}\sqrt{(k\nu)^{k-2\ell}n^{\Delta \ell}}\norm{ P^{(2)}\left( (B^{\sT}\bar{B}_{\ell})\ot B^{\sT} \right)}_F^2 
    \\&= \sum_{\ell}\sqrt{(k\nu)^{k-2\ell}n^{\Delta \ell}}\,\Tr P^{(2)} (B\ot B)^{\sT}\left(\bar{B}_{\ell}\bar{B}_{\ell}^{\sT}\ot \Id\right)(B\ot B) {P^{(2)}}^{\sT}
    \end{align*}
    By \Cref{lem:hermite-cum-matrix-opnorm}, $\opnorm{\bar{B}_{\ell}\bar{B}_{\ell}^{\sT}} \le \ell!\opnorm{\Sigma}^{\ell} \le \ell!\nu^{\ell}$, so applying \Cref{lem:holders-2-inf-2}, we find    
    \begin{align*} \sum_{\ell,j} \E_y \norm{w_{(\ell,j)}}_2^2 
    &\le \sum_{\ell}\sqrt{(k\nu)^{k-2\ell}n^{\Delta \ell}}\;\ell!\nu^{\ell} \norm{(B\ot B) {P^{(2)}}^{\sT}}_F^2
    \,.
    \end{align*}
    Finally, by \Cref{lem:hermite-fourier-tensor}, $\norm{(B \ot B){P^{(2)}}^{\sT}}_F^2 \le (2^d-1)\,\norm{BP^{\sT}}_F^2$, where $P$ is the matrix whose $\alpha$th column is $\hat{p}(\alpha)$, and by \Cref{thm:hermite-plancherel}, $\norm{BP^{\sT}}_F^2 \le \E \norm{p}_2^2$, so
    \begin{align*} \sum_{\ell,j} \E_y \norm{w_{(\ell,j)}}_2^2 
    &\le \sum_{\ell}\sqrt{(k\nu)^{k-2\ell}n^{\Delta \ell}}\,(2^d-1)\ell!\,\nu^{\ell} \norm{BP}_F^2
    \\&\le \sum_{\ell}\sqrt{(k\nu)^{k-2\ell}n^{\Delta \ell}}\,(2^d-1)k^{\ell}\nu^{\ell}\E \norm{p}_2^2
    \\&\le \sum_{\ell}\sqrt{(k\nu)^{k}n^{\Delta \ell}}\,(2^d-1)\E \norm{p}_2^2
    \,.
    \end{align*}
    A very similar calculation with the $z$ vectors yields
    \begin{align*} \sum_{\ell,j} \E_y \norm{z_{(\ell,j)}}_2^2 
    &= \sum_{\ell,j} \E_y \sqrt{(k\nu)^{2\ell-k}}\norm{y^{(\ell)}_{j}}_2^2\norm{\sum_{\beta''} R^{(\ell,\beta'')}\He_{\beta''}(y)}_2^2
    \\&= \sum_{\ell} \E_y \sqrt{(k\nu)^{2\ell-k}n^{\Delta \ell}}\norm{\sum_{\beta''} R^{(\ell,\beta'')}\He_{\beta''}(y)}_2^2
    \\&= \sum_{\ell} \E_y \sqrt{(k\nu)^{2\ell-k}n^{\Delta \ell}}\norm{\sum_{\beta''}\sum_{I_2 \in [n]^{k-\ell}}\sum_{\beta'}\kappa_2(\He_{I_2}, \He_{\beta'})\left(e^{\ot I_2}\ot\hat{q}(\beta' + \beta'') \right)\He_{\beta''}(y)}_2^2
    \,.
    \end{align*}
    Now we apply the same matrix interpretation and manipulations as for the $w$ vectors to find
    \begin{align*} \sum_{\ell,j} \E_y \norm{z_{(\ell,j)}}_2^2 
    &\le \sum_{\ell} \sqrt{(k\nu)^{2\ell-k}n^{\Delta \ell}}k^{(k-\ell)}\nu^{k-\ell}(2^d-1)\E\norm{q}_2^2 
    \\&\le \sum_{\ell} \sqrt{(k\nu)^kn^{\Delta \ell}}\,(2^d-1)\E\norm{q}_2^2
    \,.
    \end{align*}
    Therefore,
    \[\sum_{\ell,j}\E_y\frac{1}{2}\left(\norm{w_{(\ell,j)}}_2^2+\norm{z_{(\ell,j)}}_2^2\right) \le (2^d-1)k\sqrt{(k\nu)^kn^{\Delta \ell}}\cdot\frac{1}{2}\E(\norm{p}_2^2 + \norm{q}_2^2)\,.\]
    Since $\Delta \ell \le  k/2$ and $\Delta \ell$ is an integer, we have the result.
    \end{proof}

\subsection{Induction over steps}
\label{sec:conditional-induction}

\begin{lemma}
\label[Lemma]{lem:rank-decomposing-induction}
Suppose that there exists a rank-decomposing SoS proof that 
\[\norm{\E [S^{(r,s)} \mid y]}_1 \le C\]
for some multiset $S$ of random variables and some Gaussian vector variable $y \sim \cN(0, \Sigma)$ with $\Sigma \preceq \nu \cdot \Id$.

Suppose that the joint moments of $S$ conditioned on $y$ are degree-$d$ polynomials in $y$.

Then for every $k$ and every $(\Delta r, \Delta s)$ satisfying $\Delta r + \Delta s = k$ where at least one of $\Delta r$ and $\Delta s$ is even, there is a rank-decomposing SoS proof that 
\[\norm{\E_y (S \cup \{\underbrace{y, \dots, y}_k\})^{(r+\Delta r,s+\Delta s)} }_1 \le 2^{d}k^{k/2+2}\nu^{k/2}n^{\lfloor k/2\rfloor/2}C \,.\]
\end{lemma}
\begin{proof}    
    Without loss of generality, assume that there is only one term $M = pq^{\sT}$ with $p \in (\R[y]^n)^{\ot r}$ and $q \in (\R[y]^n)^{\ot  s}$ in the rank-decomposing SoS proof establishing the upper bound $\norm{\E [S^{(r,s)} \mid y]}_1 \le C$, so that $\frac{1}{2}(\norm{p}_2^2 + \norm{q}_2^2) \le C$.
    Then sublinearity of the nuclear norm is enough to extend the argument to an arbitrary number of terms.

    By the law of conditional expectation,
    \[ \E_y S \cup \{\underbrace{y, \dots, y}_k\} = \E_y\; (\E[S \mid y])\ot \{\underbrace{y, \dots, y}_k\} \,. \]

    Then by our assumption,
    \[ \E_y S \cup \{y, \dots, y\} = \E_y\; p \ot q\ot \{y, \dots, y\} \,. \]

    By the moment-cumulant formula,
    \[\E_y\; p \ot q\ot \{y, \dots, y\}
    \propto
    \sum_{\substack{\pi \\ \text{partition of $\{p,q,y,\dots, y\}$}}}
    \prod_{B \in \pi} \kappa_{|B|}(B)
    \,, \]
    where here $\propto$ means that they have equal mode symmetrizations.
    Since $\kappa(y,y) = \E_y (y \ot y) $ and the cumulant of any other number of $y$s by themselves is 0, for all the non-zero terms of that sum, every block of the partition that doesn't contain $p$ or $q$ must be exactly $\{y,y\}$.
    Therefore,
    \begin{align*}
    \E_y\; p \ot q\ot \{y, \dots, y\}
    \propto&
    \left[\sum_{\substack{a, b \\ k-a-b \in 2\N}} (\E_y y \ot y)^{\ot (k-a-b)/2} \ot 
    \kappa_{a+1}(\underbrace{y, \dots, y}_a, p) \ot \kappa_{b+1}(\underbrace{y, \dots, y}_b, q) \right]
    \\&{}+ \left[\sum_{\substack{a \\ k-a \in 2\N}} (\E_y y \ot y)^{\ot (k-a)/2} \ot 
    \kappa_{a+2}(\underbrace{y, \dots, y}_a, p, q) \right]
    \,. \end{align*}
    Now we just have to find re-representations of these terms that satisfy the nuclear norm requirement.
    We will do so one summand at a time.

    Let's start with the terms where $p$ and $q$ are in the same cumulant block.
    Let $\Delta r' = 2\lfloor\min(\Delta r, k-a)/2\rfloor$ and $\Delta s' = k-a-\Delta r'$, noting that $\Delta s'$ must be even since $k-a$ and $\Delta r'$ are both even and that $(\Delta r- \Delta r') + (\Delta s - \Delta s') = a$.
    By \Cref{lem:ydotsypq}, there is a representation of $\kappa_{a+2}(\underbrace{y, \dots, y}_a, p, q)$ as a matrix $\sum_j w_jz_j^{\sT}$ with $w_j \in (\R^n)^{\ot r} \ot (\R^n)^{\ot (\Delta r - \Delta r')}$ and $z_j \in (\R^{n})^{\ot s} \ot (\R^n)^{\ot (\Delta s - \Delta s')}$ and
    \[ \sum_j \E_y \frac{1}{2}\left(\norm{w_j}_2^2 + \norm{z_j}_2^2\right) \le (2^d-1)a^{a/2+1}\nu^{a/2}n^{\lfloor a/2\rfloor/2}C \,.\]
    Then our matrix representation for $(\E_y y \ot y)^{\ot (k-a)/2} \ot 
    \kappa_{a+2}(\underbrace{y, \dots, y}_a, p, q)$ is, in rank-decomposing form, 
    \[
    \sum_j \left(\norm{\E_y y \ot y}_2^{(\Delta s' - \Delta r')/4}\left(\E_y y \ot y\right)^{\ot \Delta r'/2} \ot w_j\right)\left(\norm{\E_y y \ot y}_2^{(\Delta r' - \Delta s')/4}\left(\E_y y \ot y\right)^{\ot \Delta s'/2} \ot z_j\right)^{\sT}
    \]
    \[ := \sum_j p'_jq'_j{}^{\sT}\,.\]
    Then we have, using the fact that $\norm{\E_y y\ot y}_2^2 \le \nu^2 n$,
    \begin{align*} \sum_j \E_y\frac{1}{2}\left(\norm{p'_j}_2^2 + \norm{q'_j}_2^2\right) &\le (\nu^2n)^{(\Delta r' + \Delta s')/4}(2^d-1)a^{a/2+1}\nu^{a/2}n^{\lfloor a/2\rfloor/2}C 
    \\& = (\nu^2 n)^{(k-a)/4}(2^d-1)a^{a/2+1}\,\nu^{a/2}n^{\lfloor a/2\rfloor/2}C 
    \\& \le (2^d-1)k^{k/2+1}\,\nu^{k/2}n^{\lfloor k/2\rfloor/2}C 
    \,.\end{align*}
Since there are at most $k$ different values of $a$, in total the terms with $p$ and $q$ in the same cumulant block make a nuclear norm contribution of $(2^d-1)k^{k/2+2}\,\nu^{k/2}n^{\lfloor k/2\rfloor/2}C$.

Now consider the terms where $p$ and $q$ are in different cumulants.
Without loss of generality, assume $\Delta r-a \ge \Delta s - b$ (otherwise we just swap the roles of $(r,a)$ and $(s,b)$ in this argument), noting that this implies $a \le \Delta r$ since $a+b \le k = \Delta r + \Delta s$.
Let $\Delta r' = 2\lfloor\min(\Delta r-a, k-a-b)/2\rfloor$ and $\Delta s' = k-a-b-\Delta r'$, noting that $\Delta s'$ must be even since $k-a-b$ and $\Delta r'$ are both even and that $\Delta s' \le \Delta s - b$ since $\Delta s' = k-a-b-\Delta r' \le k-b-\Delta r = \Delta s - b$.
By \Cref{lem:ydotsyp-vec}, there is a representation of $\kappa_{a+1}(\underbrace{y, \dots, y}_a, p)$ as a vector $v \in (\R^n)^{\ot (r + a)}$ and by a combination of \Cref{lem:ydotsyp-vec} and \Cref{lem:vector-folding-nuc-norm}, there is a representation of $\kappa_{b+1}(\underbrace{y, \dots, y}_b, q)$ as a matrix $\sum_j w_jz_j^{\sT}$ such that $w_j \in (\R^n)^{\ot (b - \Delta s + \Delta s')}$ and $z_j \in (\R^n)^{\ot s} \ot (\R^n)^{\ot (\Delta s - \Delta s')}$ and
    \[ \norm{v}_2^2 \le a! \nu^{a}\E_y\norm{p}_2^2\,. \]
    \[ \sum_j \frac{1}{2}(\norm{w_j}_2^2 + \norm{z_j}_2^2) \le \sqrt{b!}\,\nu^{b/2}n^{(b - \Delta s + \Delta s')/2}\sqrt{2C}, \,\]

Then our matrix representation for $(\E_y y \ot y)^{\ot (k-a-b)/2} \ot 
    \kappa_{a+1}(\underbrace{y, \dots, y}_a, p) \ot \kappa_{b+1}(\underbrace{y, \dots, y}_b, q)$ is, in rank-decomposing form, $\sum_j p'_jq'_j{}^{\sT}$ with
    \begin{align*}
    p'_j &:= \frac{1}{\sqrt[4]{a!\nu^{a}2C}}(\nu^2 n)^{(\Delta s' - \Delta r')/8}\left(\E_y y \ot y\right)^{\ot \Delta r'/2} \ot w_j \ot v
    \\
    q'_j &:= \sqrt[4]{a!\nu^{a}2C}(\nu^2 n)^{(\Delta r' - \Delta s')/8}\left(\E_y y \ot y\right)^{\ot \Delta s'/2} \ot z_j
    \,.\end{align*}
    Then we have, using the fact that $\norm{\E_y y\ot y}_2^2 \le \nu^2 n$,
    \begin{align*}
    \sum_j \frac{1}{2}\left(\norm{p'_j}_2^2 + \norm{q'_j}_2^2\right)
    &\le \sum_j\sqrt{a!}\,\nu^{a/2}(\nu^2 n)^{(\Delta r' + \Delta s')/4}\frac{1}{2}\left(\norm{w_j}_2^2\sqrt{(2C)^{-1}}\E_y\norm{p}_2^2 + \norm{z_j}_2^2{\sqrt{2C}}\right) 
    \\&\le \sum_j\sqrt{a!}\,\nu^{a/2}\sqrt{2C}(\nu^2 n)^{(\Delta r' + \Delta s')/4}\frac{1}{2}\left(\norm{w_j}_2^2 + \norm{z_j}_2^2\right) 
    \\&= \sqrt{a!}\,\nu^{a/2}\sqrt{2C}(\nu^2 n)^{(k-a-b)/4}\sum_j\frac{1}{2}\left(\norm{w_j}_2^2 + \norm{z_j}_2^2\right) 
    \\&\le \sqrt{a!}\,\nu^{a/2}\sqrt{2C}(\nu^2 n)^{(k-a-b)/4}\cdot \sqrt{b!}\,\nu^{b/2}n^{(b - \Delta s + \Delta s')/2}\sqrt{2C}
    \\& = 2\sqrt{a!b!}\,\nu^{k/2}n^{(k-a+b-2\Delta s + 2\Delta s')/4}C 
    \\& \le 2k^{k/2}\nu^{k/2}n^{(k-a+b-2\Delta s + 2\Delta s-2b)/4}C
    \\&= 2k^{k/2}\nu^{k/2}n^{(k-a-b)/4}C
    \\&\le 2k^{k/2}\nu^{k/2}n^{2\lfloor k/2\rfloor/4}C
    \,,\end{align*}
    where the last inequality is because $k-a-b$ is an even integer that is at most $k$.
Since there are at most $k^2/2$ different values of $a$ and $b$, in total the terms with $p$ and $q$ in the separate cumulant block make a nuclear norm contribution of $k^{k/2+2}\nu^{k/2}n^{\lfloor k/2\rfloor/2}C$.
    
Therefore, by adding together the nuclear norm contributions of the two types of terms, we have a matrix representation for 
$\E (S \cup \{\underbrace{y, \dots, y}_k\})^{(r+\Delta r,s+\Delta s)}$ with a rank-decomposing nuclear norm bound of $2^{d}k^{k/2+2}\nu^{k/2}n^{\lfloor k/2\rfloor/2}C$.
\end{proof}

\subsection{Upper bound for \texorpdfstring{$\norm{\E[v_iv_i^\sT\,\mid\,v_{i-1}]}^q_q$}{conditional pseudo-covariance} under high-entropy}

We show that the operator norm HES constraint in conjunction with the fact that $\Tr\left[\E\left[v_iv_i^\sT\,|\,\mid v_{i-1}\right]\right] = 1$ implies that for $q > 1$, the Schatten $q$-norm of the conditional covariance is bounded appropriately.

\begin{proposition}[Bounded Schatten $q$-norm of conditional covariance]\label[Proposition]{prop:q-norm-cond-cov}
    Let $V$ be a symmetric matrix-valued variable.
    Then,
    \[ \left\{\vphantom{\frac{1}{1}}\Tr V = 1\,,\;0\preceq V \preceq c\,\Id\right\} \proves\left\{\vphantom{\frac{1}{1}}\norm{V}_q^q \le c^{q-1}\right\}\,. \]
\end{proposition}
\begin{proof}
    This follows by an application of a $(1,\infty)$-Matrix H\"older's inequality~(\Cref{lem:holders-one-inf}) as follows,
    \begin{align*}
        \norm{V}^q_q &= \Tr\left[V^q\right] = \left\langle V^{q-1}, V \right\rangle \\
        &\le \opnorm{V^{q-1}}\norm{V}_1 \\
        &= \opnorm{V}^{q-1}\norm{V}_1 \\
        &\le c^{q-1}\cdot 1 = c^{q-1}\, .
    \end{align*}
    \textbf{Note:} While $q-1$ is \emph{fractional}, one can use the standard trick of converting to integer powers in conjunction with the encoding of the operator norm constraint via a slack-matrix to make the assertion that
    \[
        \opnorm{V^{q-1}} = \opnorm{V}^{q-1}\, ,
    \]
    has a sum-of-squares proof using with the radical $X = V^{1/(p-1)}$ and the fact that $q-1 = \frac{1}{p-1}$~(\Cref{sec:appendix-sdp}, Radicals of the polynomial grammar).
\end{proof}

\section{Low-Degree Polynomial Approximation of the Top Eigenspace of the Hessian}\label{sec:hessian-taylor}
In this section we analyze the spectrum of the Hessian of the spherical spin-glass Hamiltonian on points that are generated by HES processes. In particular, we provide low-degree sum-of-squares certificates that certify the value of the low-degree moments of the empirical distribution of the eigenvalues of the Hessian up to lower order corrections (\Cref{lem:hessian-schatten-norm}). We then use these sum-of-squares certificates about the low-degree spectral moments of this Hessian along with the bounds provided in \Cref{sec:hods} to certify the amount the value of the Hamiltonian can increase in a single step under \emph{any} HES process (\Cref{lem:taylor}).

Having certified the maximum energy increase of the Hamiltonian under any HES process, we also provide a specific low-degree matrix polynomial which arbitrarily approximates a projector into the top part of the bulk spectrum of the Hessian (\Cref{cor:large-hs-correlation-bernstein-hessian}). This low-degree polynomial also happens to be positive semi-definite with bounded operator norm, and is used extensively in the feasibility argument provided in \Cref{sec:feasibility-main}.

A crucial part of the argumentation that certifies how much the energy can increase under any HES process is estimating the rate at which the low-degree spectral moments of the Hessian converge to the desired value (which is the operator norm). Using some involved Catalan combinatorics in conjunction with the basic analytic properties of the Lambert-$W$ function allows for a fairly precise estimate as a function of $\eps$ (\Cref{cor:shifted-hessian-holders}). This rate of convergence governs the degree of the underlying sum-of-squares certificates over the HES SoS hierarchy that are used in the proofs in this section. 

\subsection{Certifiable approximate Wignerianity of the Hessian}\label{sec:wigner-cert-hessian}

    Subag \cite{subag2018free} showed that the eigenspectrum of the Hessian converges exponentially quickly to the semicircle law, with a scaling factor of $\sqrt{\nu''(\norm{x}_2^2)}$.

    Here we give a SoS-certifiable approximate version of the same statement and argument. To make this statement fit in the low-degree sums-of-squares framework, we relax the requirement somewhat so that we certify only the \emph{low-degree moments} of the semicircle law by bounding the Schatten norms of the Hessian.
    Although we do not do this, the bounds are strong enough to be interpreted as equivalent to conditional pointwise bounds \[
    \left\{\,|\,\|\frac{1}{n}\nabla^2H(\sigma)\|_{2q}^{2q} - nC_q(\nu''(\norm{\sigma}_2^2))^q| \le Cn^{0.5+2qd_H\lambda}\right\}\] for any large enough $C$ as long as the downstream proofs using it do not exceed a certain degree limit depending on $C$ and the degree of the pseudo-distribution. 

    This involves a careful combinatorial calculation to compute the expectation of $\norm{\frac{1}{n}\nabla^2 H(\sigma)}_{2q}^{2q}$ as a polynomial in $\sigma$ over the randomness of the instance.
    This turns out to be provably close to a function of $\norm{\sigma}_2^2$, so in effect, the expectation does not depend on $\sigma$ when $\sigma$ is known to fall on a sphere.
    The calculation counts the number of non-backtracking walks that come out of the trace power calculation for $\norm{\frac{1}{n}\nabla^2 H(\sigma)}_{2q}^{2q}$.
    
    Due to heavy concentration of the trace powers of a random matrix~\cite{forrester2023review,collins2007second}, there is a lot of slack between the bound we need and the bound that's actually true for the fluctuations of $\norm{\frac{1}{n}\nabla^2 H(\sigma)}_{2q}^{2q}$ around its expectation.
    Graph matrix norm bounds~\cite{ahn2016graph} together with \Cref{thm:high-entropy-nuclear-norm-bound-graph-matrices} suffice to certifiably control these fluctuations with subexponential probability.

\begin{lemma}[Certifiable Approximate Wignerianity]
\label[Lemma]{lem:hessian-schatten-norm}
    Let $C_q$ refer to the $q$th Catalan number. Then for each $p$ and $q$ positive integer, for every $\lambda>0$, with probability $1-O(\exp(-n^{\lambda}))$ over the randomness of the instance, there are SoS proofs whose degree depends only on $p$ and $q$ and $d_H$ demonstrating that for all high-entropy step distributions over $\sigma$ satisfying the constraint $\norm{\sigma}^2 = \zeta$ for some $\zeta$ with stepwise covariances bounded by $c$,
    \[
    \E_{\sigma} \left(\norm{\tfrac{1}{n}\nabla^2 H(\sigma)}_{2q}^{2q} - nC_q(\nu''(\norm{\sigma}_2^2))^{q}\right)^{2p}
    \le O(n^{2p(0.5 + 2qd_H\lambda)})
    \,,\]    \[
    \E_{\sigma} \left(\Tr (\tfrac{1}{n}\nabla^2 H(\sigma))^{2q-1}\right)^{2p}
    \le O(n^{2p(0.5 + 2qd_H\lambda)})
    \,.\]
\end{lemma}
\begin{proof}
In this argument, the big-$O$ notation $O(\cdot)$ will refer to asymptotics with respect to $n$, with $q$, $p$, $d_H$, and $\nu$ treated as constants.

    We begin by writing the expectation (over the randomness of the instance) of the $2q$-th Schatten norm of $\nabla^2 H(\sigma)$ as a polynomial in $\sigma$, and show that this polynomial is close to $nC_q(\nu''(\norm{\sigma}_2^2))^{q}$.
    Later we will write the $2p$th moment of $2q$th Schatten norm of $\tfrac{1}{n}\nabla^2 H(\sigma)$ as a polynomial in both $\sigma$ and the gaussian tensors $g$ (as in \Cref{def:spin-glass}), to show concentration (in both $\sigma$ and $g$) around the expectation.

    For convenience, we define
    \[T_{\alpha} = \frac{\gamma_{|\alpha|}}{\alpha!}\sum_{\pi \in S_{|\alpha|}}g_{\alpha_{\pi(1)},\dots,\alpha_{\pi(|\alpha|)}}\,,
    \]
    where in an abuse of notation, $\alpha$ is used as both a multi-index and an (arbitrarily) ordered tuple of its constituent indices.

    Then
    \[\tfrac{1}{n}\nabla^2 H(\sigma)_{i,j} = \sum_{\substack{|\alpha| \le d_H-2}}\frac{1}{\sqrt{n}}T_{\alpha \cup \{i,j\}}\sigma^{\alpha}\,,\]
    where, again in another abuse of notation, we treat $\alpha$ as both a multi-index and a multi-set.

    The trace of the $2q$-th power of the Hessian can now be written as, letting $\cC_{2q}$ be the set of closed walks of length $2q$ over $[n]$, where $w(i)$ is the $i$th vertex in the walk when $w \in \cC_{2q}$ and we take $w(0) := w(2q)$, as defined in \Cref{sec:dyck},
    \begingroup
    \allowdisplaybreaks
    \begin{align}\label{eq:schatten-norm-hessian-poly}
        \Tr[(\tfrac{1}{n}\nabla^2 H(\sigma))^{2q}] &= \sum_{w \in \cC_{2q}}\prod_{i=1}^{2q}\left(\tfrac{1}{n}\nabla^2 H(\sigma)\right)_{w(i-1), w(i)} \\
        &= \frac{1}{n^{q}}\sum_{w \in \cC_{2q}} \prod_{i=1}^{2q}\sum_{\substack{|\alpha| \le d_H-2}}T_{\alpha \cup \{w(i-1),w(i)\}}\sigma^{\alpha} \\
        &= \frac{1}{n^{q}}\sum_{w \in \cC_{2q}} \sum_{\substack{\;\alpha_1, \dots, \alpha_{2q} \\ |\alpha_i| \le d_H-2 \;}}\prod_{i=1}^{2q} T_{\alpha_i\cup\{w(i-1),w(i)\}}\sigma^{\alpha_i}\; \label{eq:schatten-norm-hessian-poly-3} \\
        &:= \mathsf{signal} + \mathsf{noise}\,,
    \end{align}
    \endgroup
    where $\mathsf{signal}$ is the part of that polynomial that is simultaneously even in $\sigma$ and degree-$0$ when expressed in the Hermite basis (equivalently, having non-zero expected value) as a polynomial in $g$, and $\mathsf{noise}$ is everything else.

    \paragraph{Signal term}
    First, let us define, using notation from \Cref{sec:dyck} to represent Dyck-shaped walks,
    \begingroup
    \allowdisplaybreaks
    \begin{align*}
        \mathsf{dyck} &:=
        \frac{1}{n^{q}}\sum_{\tau \in \mathcal{T}_{q}}\sum_{\omega \in \Omega_{\tau}}\sum_{\substack{\alpha_1, \dots, \alpha_{q} \\ |\alpha_i| \le d_H-2 \\ \alpha_1 \ne \dots \ne \alpha_q }}\prod_{i=1}^{q} T_{\alpha_i\cup\{W(\tau,\omega)_{i-1},W(\tau,\omega)_{i}\}}^2\sigma^{2\alpha_i}
        \,,
    \end{align*}
    \endgroup
    where the notation $\alpha_1 \ne \dots \ne \alpha_q$ is used to indicate that none of the $\alpha_i$ are equal to each other.

    Let $\mathsf{supersignal}$ denote the degree-0 portion of $\mathsf{dyck}$ when expressed in the basis in $g$.
    We will show that
    \[\mathsf{signal} = (1 + O(1/n))\mathsf{supersignal} = (1 \pm O(1/n))C_q(\nu''(\norm{\sigma}_2^2))^{q}\,.\]

    We start by showing that
    \[\mathsf{signal} = (1 + O(1/n))\mathsf{supersignal}\,.\]
    First, note that by linearity of the Hermite decomposition (in this case equivalent to linearity of expectation), each of the summands has its own degree-0 term, additive with each other, all of which are non-negative and within a $O(1)$ factor of each other since the magnitude of the degree-0 part of each summand is determined only by $q$, $\nu$, and the cardinalities of the indices in $\alpha_1, \dots, \alpha_q$ and $\omega$ or $w$.
    And the summands in $\mathsf{dyck}$ are a subset of those that contribute to $\mathsf{signal}$, specifically the subset where the walk $w$ is Dyck-shaped and $\alpha_i = \alpha_j$ whenever the $j$ is the falling step corresponding to $i$ which is a rising step in the Dyck-shaped walk and all $\alpha$'s are otherwise distinct.

    Therefore it suffices to show that the number of summands in $\mathsf{dyck}$ corresponding to a certain even monomial $\sigma^{2\beta}$ is within a factor of $(1 - O(1/n))$ of the total number of summands that contribute nontrivially to the same monomial in $\mathsf{signal}$.

    First, we argue that most of the nonzero terms in the expectation are where $\{\alpha_1, \dots, \alpha_{2q}\}$ is a multiset where all elements are even cardinality.
    Since we fixed $2\beta = \alpha_1 + \dots + \alpha_{2q}$, there is no factor of $n$ in the combinatorics of how that $2\beta$ is divided into the terms $\alpha_1 + \dots + \alpha_{2q}$, so any factor of $n$ difference will come from the choice of the walk $w$.
    But the condition that the different factors of $T_{\alpha_i\cup\{w(i-1),w(i)\}}$ must multiply together into an even polynomial is a strong one. Consider any specific pairing of values of $i$ such that paired up values have the same index in $T$: if $i$ and $j$ are paired up so that $\alpha_i\cup\{w(i-1),w(i)\} = \alpha_j\cup\{w(j-1),w(j)\}$ while $\alpha_i \ne \alpha_j$, then that implies $(\{w(i-1),w(i)\} \triangle \{w(j-1),w(j)\}) \triangle (\alpha_i \triangle \alpha_j) = \emptyset$, which fixes at last two of the choices of $w(i-1)$, $w(i)$, $w(j-1)$, and $w(j)$. Meanwhile the equality condition means that if two of those choices remain unfixed, then they must be equal to each other.
    In the end, this is at least one factor of $n$ fewer choices than if $\alpha_i = \alpha_j$ instead.

    So now it suffices to show that, for a specific choice of $\alpha_1', \dots, \alpha_q'$ such that $\beta = \alpha_1' + \dots + \alpha_q'$, it is the case that most of the terms where $\{\alpha_1, \dots, \alpha_{2q}\} = 2\{\alpha_1', \dots, \alpha_{q}'\}$ as a multiset equality are in the situation where $w$ is a Dyck-shaped walk whose pairs of rising and falling steps correspond to the pairs of $\{\alpha_1, \dots, \alpha_{2q}\}$.

    To see this, note that for a walk to be closed and traverse each edge an even number of times, it must necessarily have at most $q+1$ distinct vertices since it has at most $q$ distinct edges.
    But in order for it to not be a Dyck-shaped walk, it must repeat at least one edge without retracing its steps to in order to get back to that edge, meaning that the graph of traversed edges contains a non-trivial cycle, meaning that there are at most $q$ distinct vertices in the walk.

    Therefore there are $O(1/n)$ as many non-Dyck-shaped closed walks as there are Dyck-shaped ones, since enumerating the Dyck-shaped walks involves making $q+1$ arbitrary choices of vertices from $[n]$ instead of $q$ choices.

    And we cannot have mismatches in the pairings between $\alpha_i$s and the pairings between rising and falling steps in the Dyck-shaped walk without introducing additional constraints where at least two of the vertices of the walk are forced to be equal to some of the indices in the $\alpha_i$s.

    Finally we show that
    \[\mathsf{supersignal} = (1 \pm O(1/n))C_q(\nu''(\norm{\sigma}_2^2))^{q}\,.\]

    Evaluating the Hermite basis with respect to $g$ (equivalently, the Gaussian moments),
    \begingroup
    \allowdisplaybreaks
    \begin{align*}
        \E_g \mathsf{supersignal} &= \frac{1}{n^{q}}\sum_{\tau}\sum_{\omega}\sum_{\;\alpha_1, \dots, \alpha_{q} \;}\zeta_{\alpha_1, \dots,\alpha_q,\omega}\prod_{i=1}^{q}\sigma^{2\alpha_i}\;\frac{(|\alpha_i| + 2)!}{(\alpha_i \cup \{W(\tau,\omega)_{i-1},W(\tau,\omega)_{i}\})!}\gamma_{|\alpha_i|+2}^2\,,
    \end{align*}
    \endgroup
    where $\zeta_{\alpha_1, \dots,\alpha_q,\omega}$ is the product of the double factorials of one less than double the number of occurrences of each index $i \in [n]$ among $\alpha_1, \dots,\alpha_q$ and $\omega$ (this factor coming from the Gaussian moments or Hermite basis).

    Now, since this will necessarily be an even symmetric polynomial, it decomposes into a sum of products of even power sums, by Newton's identities.
    We will argue that we only need to care about the terms that are powers of $\norm{\sigma}_2^2$, because the other terms will be asympotically smaller by a SoS proof.
    This follows from the canonical matrix representation of the polynomial $\norm{\sigma}_{2j}^{2j}$ having spectral norm 1, which implies by \Cref{thm:high-entropy-nuclear-norm-bound} that $\E_{\sigma} \norm{\sigma}_{2j}^{2j} \le O(n^{-j/2}) < O(\norm{\sigma}_2^{2j}/n)$ for $j\ge 2$, and from an argument that the coefficient in front of any term with a factor of $\norm{\sigma}_{2j}^{2j}$ is within a constant factor of the coefficient in front of the same term with that factor replaced by $\norm{\sigma}_{2}^{2j}$.
    And this is true because the number of summands corresponding to each choice of $\alpha_1, \dots, \alpha_q$ doesn't depend on what those choices are and the summands themselves do not depend on $n$.

    This allows us to conclude that we only lose a multiplicative factor of $(1-1/n)$ by redefining $\zeta_{\alpha_1, \dots,\alpha_q,\omega} = 1$ regardless of the choices of indices involved.
    And a similar argument allows us to substitute $(\alpha_i \cup \{W(\tau,\omega)_{i-1},W(\tau,\omega)_{i}\})!$ with $\alpha_i!$ again with only a loss of a factor of $(1-1/n)$.

    Finally, we can carry through the calculation, using \Cref{fact:dyck-catalan} in the second equation after the summands lose dependence on $\tau$ and $\omega$:
    \begingroup
    \allowdisplaybreaks
    \begin{align*}
        \E_g \mathsf{supersignal} &= \frac{(1+O(1/n))}{n^{q}}\sum_{\tau}\sum_{\omega}\sum_{\;\alpha_1, \dots, \alpha_{q} \;}\prod_{i=1}^{q}\sigma^{2\alpha_i}\frac{(|\alpha_i| + 2)!}{\alpha_i!}\gamma_{|\alpha_i|+2}^2
        \\&=\frac{(1+O(1/n))}{n^{q}}C_qn(n-1)^{q}\sum_{\;\alpha_1, \dots, \alpha_{q} \;}\prod_{i=1}^{q}\sigma^{2\alpha_i}\frac{(|\alpha_i| + 2)!}{\alpha_i!}\gamma_{|\alpha_i|+2}^2
        \\&=\frac{(1+O(1/n))}{n^{q}}C_qn(n-1)^{q}\left(\sum_{\alpha}\sigma^{2\alpha}\frac{(|\alpha| + 2)!}{\alpha!}\gamma_{|\alpha|+2}^2\right)^q
        \\&=\frac{(1+O(1/n))}{n^{q}}C_qn(n-1)^{q}\left(\sum_{r=2}^{d_H}\sum_{|\alpha|=r-2}\sigma^{2\alpha}\;\frac{r!}{\alpha!}\gamma_{r}^2\right)^q
        \\&=\frac{(1+O(1/n))}{n^{q}}C_qn(n-1)^{q}\left(\sum_{r=2}^{d_H}r(r-1)\gamma_r^2\norm{\sigma}_2^{2(r-2)}\;\right)^q
        \\&=\frac{(1+O(1/n))}{n^{q}}C_qn(n-1)^{q}\left(\sqrt{\nu''\left(\norm{\sigma}_2^2\right)}\right)^{2q}
        \\&= n\left(1\pm O\left(\tfrac{1}{n}\right)\right)C_q\left(\sqrt{\nu''\left(\norm{\sigma}_2^2\right)}\right)^{2q}
        \,,
    \end{align*}
    \endgroup
    where the third-to-last equality uses the fact that $\norm{x}_2^{2m} = \sum_{|\alpha| = m} \frac{m!}{\alpha!} x^{2\alpha}$.

    \paragraph{Odd moments}
    Because all terms of $\Tr (\tfrac{1}{n}\nabla^2H(\sigma))^{2q-1}$ are odd in $g$, its expectation over $g$ is $0$ and therefore all terms are noise terms.
    The argument follows the same outline as the noise term for the even moments.

    \paragraph{Noise terms}
    We proceed by applying graph matrix norm bounds (\Cref{def:realization} and \Cref{def:graph-matrix}) to the polynomial $\mathsf{noise}$.

    The broad idea is that $\Tr (\frac{1}{n}\nabla^2H(\sigma))^{2q}$ is a sum over $2q$-length walks $w_1, \dots, w_{2q} = w_0$ over the indices of $\nabla^2H(\sigma)$ of the entries of $\frac{1}{n}\nabla^2H(\sigma)$ indexed by edges of the walk $(w_i,w_{i+1})$. The noise terms are formed of those walks where at least one entry of $\nabla^2H(\sigma)$ appears an odd number of times, or the higher-order Hermite terms of the walks where each entry appears an even number of times.
    Each entry of $\frac{1}{n}\nabla^2H(\sigma)$ is a polynomial $p_{i,j}$ in $\sigma$, so to expand the term associated with a single walk into the monomial basis of $\sigma$, we end up with a sum over all different possible choices of $2q$ multi-indices (one per edge in the walk).
    So just like in the argument for the signal term, we label the $i$th multi-index of the walk as $\alpha_{i}$ for $i \in [2q]$, and we label the $i$th node of the walk as $w_{i}$.
    
    We write the evaluation of this polynomial as an inner product between a coefficient matrix and the moment matrix of $\sigma$.
    We split up the coefficient matrix into a sum over ``shape matrices'' in order to use graph matrix norm bounds on the spectral norms of those shape matrices, and use \Cref{thm:high-entropy-nuclear-norm-bound-graph-matrices} to control the nuclear norm of the moment matrix.
    Each shape matrix contains only some of the terms of the expansion of $\mathsf{noise}$ into the monomial basis, specifically those where the $\alpha_{i}$ and the $w_{i}$ have a particular ``index shape'' (pattern of repeated indices and degrees).
    Each shape matrix is further specified by which indices end up as row indices (to be denoted $U_{\kappa}$ for a shape $\kappa$) of the final coefficient matrix and which are column indices ($V_{\kappa}$).

    We split $\mathsf{noise}$ into a sum of graph matrices as
    \[\mathsf{noise}(\sigma) = \sum_{\kappa} \frac{c_{\kappa}}{n^{q}}\iprod{M_{\kappa},\left(\sigma^{\ot \deg(U_{\kappa})}\right)\left(\sigma^{\ot \deg(V_{\kappa})}\right)^{\sT}}\,, \]
    where the sum is over generalized shapes $\kappa$, $U_{\kappa}$ and $V_{\kappa}$ are the generalized index shapes of $\kappa$, and $\deg(U_{\kappa})$ is the sum of the degrees (the total degree) of $U_{\kappa}$, for some constants $c_{\kappa}$ which do not depend on $n$.
    The relevant shapes can be constructed as follows, using definitions from \Cref{sec:prelims-graph-matrices}:
        \begin{enumerate}
        \item All variables of the same type, with ground set $[n]$.
        \item Generalized index shapes $U_{\kappa}$ consisting of variables of total degree up to $2q(d_H-2)$, and $V_{\kappa}$ similarly consisting variables of total degree up to $2q(d_H-2)$, with possibly some variables shared with $U_{\kappa}$ (although the same variable might have different degrees in $U_{\kappa}$ and $V_{\kappa}$).
        \item Middle index shapes $W_{\kappa}$ consisting of variables of total degree up to $4q$.
        \item The shape $\kappa$ is further specified by parameters $\widehat{\alpha}_{1},\dots,\widehat{\alpha}_{2q}$, each of which is a multiset of up to $d_H-2$ elements of $V(V_{\kappa}) \cup V(U_{\kappa})$ (possibly with repeats), and $\widehat{w}_{1}, \dots, \widehat{w}_{2q}$, each of which is an element of $V(\kappa)$ (again with possible repeats).
        \item The hyperedges $E(\kappa)$ consist of $\widehat{\alpha}_{i} \cup \{\widehat{w}_{i-1},\widehat{w}_{i}\}$ for each $i \in [2q]$, where $\widehat{w}_{0}$ is taken to be another name for $\widehat{w}_{2q}$, subject to the constraint that the total number of occurences in $E(\kappa)$ of any given vertex in $V(\kappa)$ is equal to the degree of that variable.
        The random variable associated with this edge is $T^{\sym}(\alpha_{i},w_{i-1},w_{i})$, and it is a Gaussian scaled by some factor which is bounded by a number that depends only on $q$ and $d_H$ and $\nu$ (in particular, not on $n$).
        \label[Condition]{cond:hyperedges}
        If the same edge appears more than once over the different values of $i$, then that edge is labelled with the number of times $\ell$ it appears, and its corresponding value is $H_{\ell}(T^{\sym}(\alpha_{i},w_{i-1},w_{i}))$, where $H_{\ell}$ is the $\ell$th Hermite polynomial.
        \item Due to the decomposition into Hermite polynomials, for every shape containing an edge with cardinality larger than $1$, the shape with the same indices and edges except with that edge having cardinality decremented by $2$ (even if that cardinality becomes $0$ as a result) is also a shape included in our consideration (and doing this edge-removing operation increases the factor $c_{\kappa}$ of the result by a factor of 1 plus the new cardinality of that edge after 2 was substracted).
        \item In order to avoid including any $\mathsf{signal}$ terms, we exclude the shape that has no hyperedges at all. \label[Condition]{cond:non-empty-graph}
    \end{enumerate}

    Now by
    \Cref{thm:high-entropy-nuclear-norm-bound-graph-matrices},
    \begin{align}
    \label{eq:noise-bound-nuclear-norm}
    \E_{\sigma} \mathsf{noise}^{2p}
    &\le O(1) \sum_{\kappa \in K}n^{p-2p\lceil\sum_{i \in [2q]} a_i(\kappa)/2\rceil/2}\opnorm{\frac{1}{n^{q}} M_{\kappa}}^{2p}
    \,,
    \end{align}
    where $K$ is the set of graph matrix shapes in our expansion that are single-flipped (\Cref{def:single-flip-graph-matrix}) and $a_i(\kappa)$ is the number of non-middle vertices incident to the $i$th hyperedge of $\kappa$ so that $a_i(\kappa) = |\hat{\alpha}_i|$, the hyperedges ordered according to \Cref{cond:hyperedges}. 
    
    To obtain the bound we wanted, specifically $\E_{\sigma} \mathsf{noise}^{2p} < o(n^{2p})$, it now suffices to show that $\opnorm{\tilde{M}_{\kappa}} < o(n^{q+\lceil 1+\sum_{i \in [2q]} a_i(\kappa)/2\rceil/2})$ for all $\kappa \in K$ with high probability.

    To apply a graph matrix norm bound, first we understand all of the possible hypergraphs that could make up our shape $\kappa$.
    As in \Cref{sec:prelims-graph-matrices}, let $S_{\min}$ be the minimum-size set of vertices that separates $U_{\kappa}$ from $V_{\kappa}$ (since all our vertices are the same type, they all have the same weight), including any vertices in $U_{\kappa} \cap V_{\kappa}$.
    Let also $W_{\mathrm{iso}}$ be the set of isolated vertices (vertices not incident to any hyperedges).
    We argue that
    $|V(\kappa)| - |S_{\min}| + |W_{\mathrm{iso}}| < 2q+1+\lceil\sum_{i \in [2q]} a_i(\kappa)/2\rceil$.

    To understand this, we can start from the extreme case with a shape $\kappa_0$ where all vertices in $\widehat{\alpha}_1, \dots, \widehat{\alpha}_{2q}$ and $\widehat{w}_1, \dots, \widehat{w}_{2q}$ are distinct and included in each multiset with cardinality at most 1.
    Then $|V(\kappa_0)| = 2q+\sum_{i=1}^{2q} a_i(\kappa_0)$ where the $2q$ is from the middle vertices $\widehat{w}_1, \dots, \widehat{w}_{2q}$ and $S_{\mathrm{iso}}$ is empty.
    By the definition of single-flipped graph matrix shape (and taking $a_i := a_i(\kappa_0)$), the $i$th hyperedge has $\lceil a_i/2+1/2\rceil$ or $\lfloor a_i/2\rfloor$ or $\lfloor a_i/2-1/2\rfloor$ vertices in $U_{\kappa_0}$ and $a_i - |e_i \cap U_{\kappa_0}|$ in $V_{\kappa_0}$.
    Therefore $\min(|e_i \cap U_{\kappa_0}|,|e_i \cap V_{\kappa_0}|) \ge \lfloor a_i/2-1/2\rfloor$ and so since a vertex separator must separate $e_i \cap U_{\kappa_0}$ from $e_i \cap V_{\kappa_0}$, at least $\lfloor a_i/2-1/2\rfloor$ of the vertices in $e_i$ are part of any minimum vertex separator $S_{\min}$.
    Furthermore, if $|e_i \cap U_{\kappa_0}| \ne |e_i \cap V_{\kappa_0}| \ne a_i/2$, then the next hyperedge $e_j$ in the cycle that has an imbalance in $|e_i \cap U_{\kappa_0}|$ vs $|e_i \cap V_{\kappa_0}|$ has that imbalance in the opposite direction, forcing either a middle vertex in $W_{\kappa_0}$ to also be part of this minimum vertex separator or for the vertex separator to stay on the same side (either all $U_{\kappa_0}$ or all $V_{\kappa_0}$).
    The latter will be cheaper, at the average cost of up to $1/2$ of a vertex per hyperedge (except for the last one, if the total number of modes is odd). 
    Therefore, $|S_{\min}| \ge \lfloor\sum_{i}  a_i/2\rfloor$.
    Overall, $|V(\kappa_0)| - |S_{\min}| + |W_{\mathrm{iso}}| \le 2q +  \lceil\sum_{i \in [2q]} a_i/2\rceil$, which is the bound we wanted.

    The remaining shapes are formed from these extreme $\kappa_0$ by merging vertices together and removing even numbers of duplicate copies of identical hyperedges.
    Merging two vertices together has the effect of reducing $|V(\kappa)|$ by 1, and it cannot reduce $S_{\min}$ by more than 1, so the merging operation itself cannot cause $|V(\kappa)| - |S_{\min}| + |W_{\mathrm{iso}}|$ to increase.
    The identical-hyperedge-removing operation can cause $|W_{\mathrm{iso}}|$ to increase without affecting $|V(\kappa)|$ or  $|S_{\min}|$, but we argue that for each element of $W_{\mathrm{iso}}$, there must have been at least one merging operation that decreased $|V(\kappa)|$ without changing $|S_{\min}|$.
    Since our construction of a minimum vertex separator did not use any vertices in $W$, any middle vertex that is merged would decrease $|V(\kappa)|$ but not decrease $|S_{\min}|$.
    
    We can lower-bound the number of vertices in $W$ that were merged by looking at the connected components of the graph when it is restricted to only the vertices in $W$.
    Any merging/removal of hyperedges in the original graph matrix shape imply the merging/removing of the corresponding edges in this restricted graph $(V,E)$ with $V = W$.
    Each removal of a pair of identical edges can only increase the number $C$ of connected components in this restricted graph by $1$, so $C - 1 \le q - |E|/2$.
    Unless there are no edges at all in the graph (which is impossible by \Cref{cond:non-empty-graph}), $|W_{\iso}| \le C-1$.
    And finally, since there is at least one component with edges and every vertex has even degree so at least one component is not a tree, $|E| \ge |V| - C + 1$.
    Therefore, $|V| \le |E| + C - 1 \le 2q - C + 1 \le 2q - |W_{\iso}|$.
    Therefore, at least one vertex of $W$ must have been merged for every element of $|W_{\iso}|$ and the inequality $|V(\kappa)| - |S_{\min}| + |W_{\mathrm{iso}}| \le 2q+\lceil\sum_{i \in [2q]} a_i/2\rceil$ holds for all shapes $\kappa \in K$.

    By \Cref{thm:graph-matrix-norm-bound}, we then have with probability at least $1-\eps'$ for each $\kappa$,
    \[ \opnorm{M_{\kappa}} \le O(1) n^{q + \lceil\sum_{i \in [2q]} a_i(\kappa)/2\rceil/2} \left(O(1)\left(|S_{\min}|\log n - \log \eps'\right)\right)^{2qd_H}
    \,.
    \]

    By a union bound, \Cref{eq:noise-bound-nuclear-norm}, and taking $\eps' = e^{-n^{\lambda}}$, we have with probability $1-O(e^{n^{\lambda}})$
    \begin{align*}
    \E_{\sigma} \mathsf{noise}^{2p}
    &\le O(1) \sum_{\kappa \in K}n^{p-2p\lceil\sum_{i \in [2q]} a_i(\kappa)/2\rceil/2}\opnorm{\frac{1}{n^{q}} M_{\kappa}}^{2p}
    \\&\le O(n^{p + 4pqd_H\lambda})
    \,.
    \end{align*}
\end{proof}

\subsection{Certifiable Taylor expansion}
\label{sec:certifiable-taylor}
We state an elementary combinatorial lemma using H\"older's inequality that transparently shows that~\Cref{lem:hessian-schatten-norm} implies an upper bound on the certified value of the hessian under the maximum entropy constraint.
\begin{lemma}[H\"older's inequality for Hessian certificate]\label[Lemma]{lem:hessian-holders}
    Given the certificate in~\Cref{lem:hessian-schatten-norm} the following holds by a sum-of-squares proof:
    \begingroup
    \allowdisplaybreaks
    \begin{align*}
        \E_{\sigma}\left(\frac{1}{k}\left\langle\frac{1}{n} \nabla^2 H(\sigma), \E_{v}\left[vv^\sT\mid\sigma\right]\right\rangle\right)^{2pq(p-1)} \le \left((1-\eps + O(\eps^2))2\nu''(\norm{\sigma}^2_2)^{1/2}\frac{1}{k}\right)^{2pq(p-1)} + o_n(1)\, .
    \end{align*}
    \endgroup
\end{lemma}
\begin{proof}
    The proof follows by a simple combinatorial argument to choose $p$ large enough, so that $C_p \approx (1-\eps +O(\eps^2))$, followed by an application of H\"older's inequality.
    First, note that,
    \begingroup
    \allowdisplaybreaks
    \begin{align*}
        &(C_{p/2}\nu''(\norm{\sigma}^2_2)^{p/2})^{2q(p-1)} = \left(\nu''(\norm{\sigma}^2_2)C_{p/2}^{1/p}\right)^{2pq(p-1)}
        \\&\sim_{\text{\Cref{lem:catalan-stirling}}} \nu''(\norm{\sigma}^2_2)^{2pq(p-1)/2}\left(\left(\frac{4^{p/2}}{\sqrt{\pi}\left(\frac{p}{2}\right)^{3/2}}\right)^{1/p}\right)^{2pq(p-1)} \\
        &= \left(2\nu''(\norm{\sigma}^2_2)^{1/2}\right)^{2pq(p-1)}\left(\left(\frac{2^3}{p^3(\pi)}\right)^{1/2p}\right)^{2pq(p-1)}
    \end{align*}
    \endgroup
    By the maximum entropy constraint,
    \[
        \norm{\E_v[vv^\sT\mid \sigma]}^q_q \le \left(\frac{1+o_n(1)}{\delta n}\right)^{q-1}\, ,
    \]
    and using H\"older's inequality in conjunction with~\Cref{lem:induction-ideal-reduction} then yields,
    \begingroup
    \allowdisplaybreaks
    \begin{align*}
        &\E_{\sigma}\left(\frac{1}{k}\left\langle\frac{1}{n} \nabla^2 H(\sigma), \E_{v}\left[vv^\sT\mid\sigma\right]\right\rangle\right)^{2pq(p-1)} = \frac{1}{k^{2pq(p-1)}}\E_\sigma\left(\left[\left\langle \frac{1}{n}\nabla^2 H(\sigma), \E_v\left[vv^\sT\mid \sigma\right]\right\rangle\right]\right)^{2pq(p-1)} \\
        \le& \frac{1}{k^{2pq(p-1)}} \E_\sigma\norm{\frac{1}{n}\nabla^2 H(\sigma)}^{2pq(p-1)}_p\norm{\E_v\left[vv^\sT\mid \sigma\right]}^{2pq(p-1)}_q \\
        \le& \frac{1}{k^{2pq(p-1)}} \E_\sigma\norm{\frac{1}{n}\nabla^2 H(\sigma)}^{2pq(p-1)}_p\left(\frac{1+o_n(1)}{\delta n}\right)^{2p(q-1)(p-1)} \\
        \le& \frac{1}{k^{2pq(p-1)}}\left((n\nu''(\norm{\sigma}^2_2)^{p/2}C_{p/2})^{2q(p-1)} + o(n^{2q(p-1)})\right)\left(\frac{1+o_n(1)}{\delta n}\right)^{2p(q-1)(p-1)} \\
        \le& \frac{1}{k^{2pq(p-1)}}\frac{n^{2q + 2p - 2q}}{n^{2pq - 2q}}\left(\frac{\nu''(\norm{\sigma}^2_2)^{pq(p-1)}(C_{p/2})^{2q(p-1)}}{\delta^{2p(q-1)(p-1)}}\right) + o_n(1) \\
        =&_{\frac{1}{p}+\frac{1}{q} = 1} \frac{1}{k^{2pq(p-1)}}\left(2\nu''(\norm{\sigma}^2_2)^{1/2}\right)^{2pq(p-1)}\left(\frac{1}{(\delta\sqrt{\pi})^{1/p}(p/2)^{3/(2p)}}\right)^{2pq(p-1)} + o_n(1)
    \end{align*}
    \endgroup
    We first invoke a series argument (ignoring contributions second-order and higher as $p$ will be large) to expand the denominator terms as,
    \begingroup
    \allowdisplaybreaks
    \begin{align*}
        (\delta\sqrt{\pi})^{1/p} = 1 + \frac{\log(\delta\sqrt{\pi})}{p} + O\left(
        \frac{1}{p^2}\right)\, , \\
        (p/2)^{3/2p} = 1 + \frac{3\log(p/2)}{2p} + O\left(\frac{1}{p^2}\right)\, .
    \end{align*}
    \endgroup
    We would like the approximation term to be $1-\eps + O(\eps)^2 \approx \frac{1}{1+\eps+O(\eps^2)}$. Equating the denominator (and ignoring contributions larger than second-order),
    \begingroup
    \allowdisplaybreaks
    \begin{align*}
        1 + \eps + O\left(\frac{1}{p^2}\right) = 1 + \frac{3\log(p/2)}{2p} + \frac{\log(\delta\sqrt{\pi})}{p} + O\left(\frac{\log(p/2)\log(\delta\sqrt{\pi})}{p^2}\right)\, .
    \end{align*}
    \endgroup
    The first-order solution to the above equation is given by the so-called Lambert-W function~\cite{weisstein2002lambert} and is,
    \[
        p(\eps,\delta) = \frac{-3W\left(\frac{-4}{3\pi^{1/3}}\cdot\frac{\eps}{\delta^{2/3}}\right)}{2\eps}\, .
    \]
    The choice of $\delta = A(\eps) = O(\eps)$ comes from~\Cref{cor:large-hs-correlation-bernstein-hessian} yielding a final value for $p$ as,
    \[
        p(\eps) = -\frac{3W\left(\frac{-4}{3\pi^{1/3}}\cdot O(\eps^{1/3})\right)}{2\eps}\, .
    \]
    Note that as $\eps << 1$ we can do a Taylor expansion of the Lambert-W function around $0$,
    \[
        W_0(z) = \sum_{n=1}^\infty \frac{(-n)^{n-1}}{n!}z^n = z - z^2 + \frac{3}{2}z^3 - \frac{8}{3}z^4 + \dots\, ,
    \]
    yielding,
    \[
        p(\eps) = -\frac{3}{2\eps}W\left(\frac{-4}{3\pi^{1/3}}\cdot O(\eps^{1/3})\right) = \frac{2}{\pi^{1/3}}O(\eps^{-2/3}) - O(\eps^{-1/3}) + O(1)\,.
    \]
    Close to $0^-$, the Lambert function is linear in the first-order and, therefore, monotone. This implies that one can scale by a constant factor (smaller than $O(\eps^{-2/3})$) so that $p(\eps)$ is an integer and assert that,
    \[
        \frac{2}{\pi^{1/3}}O(\eps^{-2/3}) - O(\eps^{-1/3}) + O(1) = p(\eps) \ge 2\, ,
    \]
    for a sufficiently small $\eps$. This allows for the conclusion that the H\"older conjugate $q \ge 1$, and that,
    \begingroup
    \allowdisplaybreaks
    \begin{align*}
        \left(\frac{1}{(\delta\sqrt{q})^{1/p}(p/2)^{3/2p}}\right)^{2pq} &= 1 - \eps + O(\eps^2) \implies \\
        &\E_{\sigma}\left(\frac{1}{k}\left\langle\frac{1}{n} \nabla^2 H(\sigma), \E_{v}\left[vv^\sT\mid\sigma\right]\right\rangle\right)^{2pq}\le \left((1-O(\eps))2\nu''(\norm{\sigma}^2_2)^{1/2}\frac{1}{k}\right)^{2pq} + o_n(1)\, .
    \end{align*}
    \endgroup
\end{proof}

While it is possible to achieve the upper bound above with \emph{some} choice of $\E[vv^\sT \mid \sigma]$, as our constraint system restricts this to be a PSD matrix, we must write a small modification of the lemma above to work for an ``adjustment''  of the hessian (where it is shifted to be PSD with high probability).

\begin{corollary}[Large correlation of conditional covariance with ``shifted'' Hessian]\label[Corollary]{cor:shifted-hessian-holders}
    Given the certificate in~\Cref{lem:hessian-schatten-norm}, the following holds by a sum-of-squares proof:
    \[
        \E_\sigma \left(\frac{1}{k}\left\langle \left(2\nu''(\norm{\sigma}^2_2)^{1/2}\Id\right) + \frac{1}{n}\nabla^2 H(\sigma), \E[vv^\sT \mid \sigma] \right\rangle\right) \le \frac{1}{k}2\nu''(\norm{\sigma}^2_2)^{1/2}\left(2 - \eps + O(\eps^2)\right) + o_n(1)\, .
    \]
\end{corollary}
\begin{proof}
    We do a binomial expansion on the shifted Hessian and use one more application of Stirling's approximation of the factorial to get a uniform upper bound on the Catalan numbers, leading to a \emph{slightly} different parameterization of the Lambert-W function to obtain the right choice of $(p(\eps), q(\eps))$-H\"older conjugates that correspond to the desired approximation.

    We simplify the proof used above by using the radical $A = \E[vv^\sT \mid \sigma]^{1/(p-1)}$ and using it with an application of~\Cref{lem:matrix-p-q-holders} (with $B = \frac{1}{n}\nabla^2 H(\sigma) + 2\nu'(\norm{\sigma}^2_2)^{1/2}\Id_n$) and obtain,
    \begin{align*}
        &\left(\E_\sigma \left(\left\langle \left(2\nu''(\norm{\sigma}^2_2)^{1/2}\Id\right) + \frac{1}{n}\nabla^2 H(\sigma), \E[vv^\sT \mid \sigma] \right\rangle\right)\right)^p \\
        &\qquad\qquad\le_{\text{Jensen with $p$=even}} \E_\sigma\left(\Tr[BA^{p-1}]\right)^p \le_{\text{\Cref{lem:matrix-p-q-holders}}} \left(\Tr[A^p]\right)^{p-1}\E_\sigma\Tr[B^p] \\
        &= (\Tr[\E[vv^\sT \mid \sigma]^{p/(p-1)}])^{p-1}\E_\sigma\Tr[B^p] =_{q = p/(p-1)}\left(\Tr[\E[vv^\sT \mid \sigma]^q]\right)^{p-1}\E_\sigma \Tr[B^p] \\
        &= \left(\norm{\E[vv^\sT \mid \sigma]}_q^q\right)^{p-1}\E_\sigma\Tr[B^p] \le_{\text{\Cref{prop:q-norm-cond-cov}}} \left(\frac{1}{(\delta n)^{q-1}}\right)^{p-1} \E_\sigma \Tr[B^p] =_{1/p + 1/q = 1}\frac{1}{\delta n}\E_\sigma \Tr[B^p] \\
        &\le_{\text{Binomial expansion}} \frac{1}{\delta n} \left(\sum_{k=2, \text{even}}^p \binom{p}{k} (2\nu''(\norm{\sigma}^2_2)^{1/2})^{p-k}\E_\sigma\Tr\left[\left(\frac{1}{n}\nabla^2 H(\sigma)\right)^k\right] + (2\nu''(\norm{\sigma}^2_2)^{1/2})^p n + o(n)\right) \\
        &\le_{\text{\Cref{lem:hessian-schatten-norm}}} \frac{1}{\delta n}\left(\sum_{k \in \{2,4,\dots,p\}} \binom{p}{k}(2\nu''(\norm{\sigma}^2_2)^{1/2})^{p-k}\left(\nu''(\norm{\sigma}^2_2)^{k/2}C_{k/2} + o(n)\right) + (2\nu''(\norm{\sigma}^2_2)^{1/2})^p n + o(n)\right)\\
        &= (2\nu''(\norm{\sigma}^2_2)^{1/2})^p\frac{1}{\delta}\left(\left\{ \sum_{k \in \{2,4,\dots,p\}} \binom{p}{k} 2^{-k} C_{k/2}\right\} + 1\right) + o_n(1) \\
        &=_{\text{\Cref{lem:catalan-stirling}}} (2\nu''(\norm{\sigma}^2_2)^{1/2})^p \frac{1}{\delta }\left(\left\{{\frac{2\sqrt{2}}{\sqrt{\pi}}} \sum_{k \in \{2,4,\dots,p\}} \binom{p}{k}\frac{1}{k^{3/2}} \left(1 + O\left(\frac{1}{k}\right)\right)\right\} + 1\right) + o_n(1)\, , 
    \end{align*}
    where we have used semi-algebraic degree reduction (one-sided) of the form as shown in~\Cref{lem:ideal-reduction-2p} in the calculation above. \\
    By a calculation similar to~\cite[Section 5.4]{spencer2014asymptopia}, using estimates similar to those in~\cite[Eq. 5.45 Eq. 5.49 \& Eq. 5.50]{spencer2014asymptopia} we have that, for $k' = \frac{p + 100\sqrt{\ln(p)}}{2}$,
    \[
        \sum_{s \in \{k', p\}} \binom{p}{s}\frac{1}{s^{3/2}} \le \Theta\left(\frac{2^p}{p^{5001}}\right)\, ,
    \]
    using the fact that for the range of $s$ above, $\frac{1}{s^{3/2}} \le O\left(\frac{1}{p^{3/2}}\right)$. Similarly, by symmetry,
    \[
        \sum_{s \in \{0, p - k'\}}\binom{p}{s}\frac{1}{s^{3/2}} \le \Theta\left(\frac{2^p}{p^{4999}}\right)\, ,
    \]
    with the bound of $\frac{1}{s^{3/2}} \le 1$ for the desired range. Since the largest contributions come from the middle terms in the binomial sum, those are bounded as,
    \begin{align*}
        \sum_{s \in \{p-k', k'\}} \binom{p}{s}\frac{1}{s^{3/2}} &\le \left(\frac{2\sqrt{2}}{p^{3/2}} + o(1/p^{3/2})\right)\sum_{s \in \{p-k', k'\}}  \binom{p}{k} \sim \frac{2\sqrt{2}}{p^{3/2}}\binom{p}{p/2}\sum_{s = \frac{p + c\sqrt{\ln(p)}}{2} \in \{k', p-k'\}}e^{-(c^2/2)} \\
        &\le \frac{2\sqrt{2}}{p^{3/2}}\binom{p}{p/2}\sqrt{p}\int_{-\infty}^{\infty}\frac{e^{-c^2/2}}{2}dc = \frac{2\sqrt{2}}{p^{3/2}}\binom{p}{p/2}\sqrt{p}\sqrt{\frac{\pi}{2}}\, 
    \end{align*}
    where we used the fact that for $s \in \{p-k', k'\}$ it is the case that,
    \begin{align*}
        \frac{1}{s^{3/2}} &\le \frac{1}{\left(\frac{p - 100\sqrt{\ln(p)}}{2}\right)^{3/2}} = \frac{2\sqrt{2}}{(p - 100\sqrt{\ln(p)})^{3/2}} \\
        &= \frac{2\sqrt{2}}{p^{3/2}\left(1 - \frac{100\sqrt{\ln(p)}}{p^{3/2}}\right)} = \frac{2\sqrt{2}}{p^{3/2}}\left(1 + O\left(\frac{\sqrt{\ln(p)}}{p^{3/2}}\right)\right) \\
        &= \frac{2\sqrt{2}}{p^{3/2}} + O\left(\frac{\sqrt{\ln(p)}}{p^3}\right)\, ,
    \end{align*}
    for every $p$. Using the Stirling approximation for $\binom{p}{p/2}$ again leads to,
    \[
        \binom{p}{p/2}\sqrt{p}\sqrt{\frac{\pi}{2}} \sim \frac{2^p}{\sqrt{p}}\sqrt{\frac{2}{\pi}}\sqrt{p}\sqrt{\frac{\pi}{2}} = 2^p\, .
    \]
    A similar analysis can be used to show that the terms above strictly dominate the remaining $o\left(\frac{1}{k^{3/2}}\right)$ terms, which continue to contribute on the same order and $\sim$ is used to denote precise equality \emph{up to first order} terms in the approximation. \\
    Putting the above bounds altogether leads to the final bounds on the inner-product as,
    \begin{align*}
        &\left(\E_\sigma \left(\left\langle \left(2\nu''(\norm{\sigma}^2_2)^{1/2}\Id\right) + \frac{1}{n}\nabla^2 H(\sigma), \E[vv^\sT \mid \sigma] \right\rangle\right)\right)^p \le (2\nu''(\norm{\sigma}^2_2)^{1/2})^p \frac{1}{\delta }\left(\frac{8}{\sqrt{\pi}}\frac{2^p}{p^{3/2}} + 1\right) + o_n(1) \\
        &\implies_{\text{\Cref{lem:ideal-reduction-2p}}} \E_\sigma \left\langle \left(2\nu''(\norm{\sigma}^2_2)^{1/2}\Id\right) + \frac{1}{n}\nabla^2 H(\sigma), \E[vv^\sT \mid \sigma] \right\rangle \\
        &\qquad\qquad\qquad\le 2\nu''(\norm{\sigma}^2_2)^{1/2}\left(\frac{1}{\delta }\left(\frac{8}{\sqrt{\pi}}\frac{2^p}{p^{3/2}} + 1\right) + o_n(1)\right)^{1/p}\\
        &\qquad\qquad\qquad\le  2\nu''(\norm{\sigma}^2_2)^{1/2}\left(\frac{1}{\delta}\left(\frac{8}{\sqrt{\pi}}\frac{2^p}{p^{3/2}} + 1\right)\right)^{1/p} + o_n(1)\, ,
    \end{align*}
    and the goal is to set $p(\eps)$ so that,
    \[
        \left(\frac{1}{\delta}\left(\frac{8}{\sqrt{\pi}}\frac{2^p}{p^{3/2}} + 1\right)\right)^{1/p} = (2-\eps + O(\eps)^2)\, .
    \]
    Note that,
    \begin{align*}
        \left(\frac{1}{\delta}\right)^{1/p}\left(1 + \left(\frac{2C^{1/p}}{p^{3/2p}}\right)^p\right)^{1/p} \le \left(\frac{1}{\delta}\right)^{1/p}\left(2\left(\frac{2C^{1/p}}{p^{3/2p}}\right)^p\right)^{1/p} = 2\left(\frac{2}{\delta}\right)^{1/p}\left(\frac{C}{p^{3/2}}\right)^{1/p} = 2\left(\frac{(2C)^{1/p}}{\delta^{1/p} p^{3/2p}}\right)\, ,
    \end{align*}
    and so the final desired ratio is,
    \[
        1 + \frac{\eps}{2} + O(\eps^2) = (\delta)^{1/p}\left(\frac{p}{2C}\right)^{3/2p}\, ,  
    \]
    at which point an analysis for the Lambert-W function identical to the one in~\Cref{lem:hessian-holders} gives the desired choice of $p(\eps) = O(\eps^{-2/3})$.
\end{proof}

\begin{lemma}[Taylor Series expansion under maximum-entropy]\label[Lemma]{lem:taylor}
    Given a program iterate $\sigma$ and increment $v$ from a polynomial grammar satisfying the axioms of the HES SoS relaxation (up to an appropriate constant degree), the following statement holds by a SoS proof:
    \begin{equation}
        \E_{v, \sigma}\left(\frac{1}{n}\left(H\left(\sigma + \frac{v}{\sqrt{k}}\right) - H(\sigma)\right)\right) \le (1-O(\eps) - o_n(1))\nu''(\norm{\sigma}^2_2)^{1/2}\frac{1}{k} + O\left(\frac{1}{k^3}\right) + o_n(1)\, ,
    \end{equation}
    with high probability over the randomness of the input.
\end{lemma}
\begin{proof}
    The proof for this follows by combining the following three facts:
    \begin{enumerate}
        \item The low-degree $p$-Schatten norm SoS certificate of the hessian $\nabla^2 H(\sigma)$ with an appropriately large choice of $p$~(\Cref{lem:hessian-schatten-norm}). This conjoined with a Holder's inequality yields the desired bound.
        \item The gradient term is zero since the first cumulant is zero.
        \item The low-degree SoS certificates that, for all $p \in [3,d_H]$, the contributions of $\frac{1}{j! k^{j/2} n}\langle \nabla^r H(\sigma_j), v^{\ot r}\rangle $~(\Cref{thm:high-entropy-nuclear-norm-bound-graph-matrices}) are $\frac{O(\eps^2)}{k}$. In particular, when $\opnorm{P} = n^{\lceil d_p/2 + 1\rceil/2}$ (as is the case when each of its coefficients/entries is random, independent, and on the order of $\sqrt{n}$), $k=1$ and $\nu = \frac{1}{\delta n}$, we have
        \[
            \E_{\sigma,u} p(\sigma, u) \le O(d_p)^{O(d_p+i+d)}(\nu n)^{d_p/2}n \,.
        \]
        This immediately implies, with the appropriate choice of $\nu$ (and a choice of step precision $\eps$ that is \emph{exponentially} small in the desired global precision) that, $\frac{1}{n}\E_{\sigma,u}[p(\sigma, u)] \le \frac{O(\eps^2)}{k}$.

    \end{enumerate}
    The Taylor expansion of the $H(\sigma+\frac{v}{\sqrt{k}})$ around $\sigma$ is given as,
    \begin{align*}
        \frac{1}{n}H\left(\sigma + \frac{v}{\sqrt{k}}\right) = \frac{1}{n}H(\sigma) + \frac{1}{\sqrt{k}n}\left\langle \nabla H(\sigma), v \right\rangle + \frac{1}{2kn}\left\langle\nabla^2 H(\sigma), vv^\sT\right\rangle + \sum_{r=3}^{d_H}\frac{1}{k^{r/2}r!n}\langle \nabla^r H(\sigma), v^{\ot r} \rangle\, .
    \end{align*}
    Applying the conditional expectation operators as well,
    \begingroup
    \allowdisplaybreaks
    \begin{align*}
        &\E_{\sigma, v}\left(\frac{1}{n}H\left(\sigma + \frac{v}{\sqrt{k}}\right) - \frac{1}{n}H(\sigma)\right)
        = \sum_{j=2}^{d_H}\frac{1}{j!k^{j/2}n}\pE_{\sigma, v}\langle \nabla^r H(\sigma), v^{\ot r}\rangle\, , \\
    \end{align*}
    \endgroup
    where we used the fact that the gradient is zero. To isolate the contribution of the derivatives, we will apply~\Cref{lem:ideal-reduction-2p} in conjunction with a pseudo-Jensen inequality for even powers. Since there are only $k$ steps, only the hessian term will survive while the others will contribute no more than $O(1/k^{3/2})$ per step.
    \begingroup
    \allowdisplaybreaks
    \begin{align*}
        \sum_{j=2}^{d_H}\frac{1}{j!k^{j/2}n}\E_{\sigma, v}\langle \nabla^r H(\sigma), v^{\ot r}\rangle &= \frac{1}{2k}\E_{\sigma}\left\langle\frac{1}{n}\nabla^2 H(\sigma), \E_v\left[vv^\sT\mid \sigma\right]\right\rangle + O\left(\frac{1}{k^{3/2}}\right) + o_n(1)\, .
    \end{align*}
    \endgroup
    We now bound the hessian term.
    \begingroup
    \begin{align*}
        \left(\frac{1}{k}\E_\sigma\left[\left\langle\frac{1}{n}\nabla^2 H(\sigma), \E_v\left[vv^\sT \mid \sigma\right]\right\rangle\right]\right)^{2pq(p-1)} &\le_{\text{\Cref{lem:pseudo-jensen-even}}} \E_{\sigma}\left(\frac{1}{k}\left\langle\frac{1}{n} \nabla^2 H(\sigma), \E_{v}\left[vv^\sT\mid\sigma\right]\right\rangle\right)^{2pq(p-1)} \\
        &\le \left((1-O(\eps))2\nu''(\norm{\sigma}^2_2)^{1/2}\frac{1}{k}\right)^{2pq(p-1)} + o_n(1)\, .
    \end{align*}
    \endgroup
    Applying~\Cref{lem:ideal-reduction-2p} yields,
    \begingroup
    \allowdisplaybreaks
    \begin{align*}
        &\frac{1}{k}\left(\E_\sigma\left[\left\langle\frac{1}{n}\nabla^2 H(\sigma), \E_v\left[vv^\sT \mid \sigma\right]\right\rangle\right]\right) \le \left(\left((1-O(\eps))2\nu''(\norm{\sigma}^2_2)^{1/2}\frac{1}{k}\right)^{2pq(p-1)} + o_n(1)\right)^{1/(2pq(p-1))} \\
        &= (1-O(\eps))2\nu''(\norm{\sigma}^2_2)^{1/2}\frac{1}{k}\left(1 + \frac{o_n(1)}{\left((1-O(\eps))2\nu''(\norm{\sigma}^2_2)^{1/2}\frac{1}{k}\right)^{2pq(p-1)}}\right)^{1/(2pq(p-1))} \\
        &=_{\text{Taylor expansion}} (1-O(\eps))2\nu''(\norm{\sigma}^2_2)^{1/2}\frac{1}{k}\left(1 + \frac{o_n(1)}{2pq(-1)\left((1-O(\eps))2\nu''(\norm{\sigma}^2_2)^{1/2}\frac{1}{k}\right)} + o_n(1)\right) \\
        &= (1-O(\eps))2\nu''(\norm{\sigma}^2_2)^{1/2}\frac{1}{k}\left(1 + o_n(1)\right) = \left(1- O(\eps) - o_n(1)\right)2\nu''(\norm{\sigma}^2_2)^{1/2}\frac{1}{k}\,.
    \end{align*}
    \endgroup
    A similar calculation as above for the higher-order derivative certificates provided by~\Cref{thm:high-entropy-nuclear-norm-bound-graph-matrices} concludes the argument.
\end{proof}

We demonstrate that the global optimization stated in the axioms of the program (which is feasible as shown in~\Cref{sec:feasibility}) in conjunction with the upper bound on the (pseudo-expected) increase of the energy at every step concludes a lower bound as well.

\begin{lemma}[Local maximization via maximum-entropy and global maximization]\label[Lemma]{lem:local-max-lower-bound}
    Given the maximum entropy constraints and the (feasible) constraint of global (hamiltonian) maximization, the following lower bound holds by a SoS proof:
    \[
        \E_{k,\dots,i}\left[H\left(\sum_{j=1}^i \frac{v_j}{\sqrt{k}}\right) - H\left(\sum_{j=1}^{i-1}\frac{v_j}{\sqrt{k}}\right)\right] \ge \frac{1}{k}(1-\eps)\nu''(i/k)^{1/2} - O(\eps^2) - O\left(\frac{1}{k^{3/2}}\right) -o_n(1)\, . 
    \]
\end{lemma}
\begin{proof}
    Note that the bounds provided in~\Cref{prop:odd-bounds-from-even-bounds} are dual sided, and the certificate for the hessian in~\Cref{lem:hessian-schatten-norm} controls the deviation. Therefore, combining~\Cref{prop:odd-bounds-from-even-bounds} and~\Cref{lem:hessian-schatten-norm} immediately yields that,
    \[
        n\nu''(\norm{\sigma}^2_2)^{k/2}C_{k/2} - o(n) \le \E_\sigma\norm{\frac{1}{n}\nabla^2 H(\sigma)}^k_k \le n\nu''(\norm{\sigma}^2_2)^{k/2}C_{k/2} + o(n)\, .
    \]
    Given the assumption of global optimality (under moments) which is made possible because of feasibility, that is,
    \[
        \E\left[H\left(\sum_{i=1}^k \frac{v_i}{\sqrt{k}}\right)\right] \ge (1-\eps)\max_{\sigma \in \cS^{n-1}(1)}  H(\sigma) \sim_{\text{by~\Cref{lem:taylor}}} (1-\eps)\int_0^1 \nu''(q)^{1/2}dq +O(\eps)^2 + o_n(1)\, ,
    \]
    and the bound on the higher-order derivatives certified by sum-of-squares to be $O(1/k^{3/2}) + o_n(1)$~(\Cref{thm:high-entropy-nuclear-norm-bound-graph-matrices}) allows one to conclude by simple linear manipulation that,
    \[
        \E\left[H\left(\sum_{j=1}^k \frac{v_j}{\sqrt{k}}\right)\right] - \E\left[H\left(\sum_{j=1}^{k-1}\frac{v_j}{\sqrt{k}}\right)\right] \ge \frac{1}{k}(1-\eps)\left(\nu''(1)^{1/2}\right) - O(\eps^2) - O\left(\frac{1}{k^{3/2}}\right) - o_n(1)\, ,
    \] 
    where the contribution is as is because one subtracts all steps but the last one to get the relevant contribution from the Hessian. 
\end{proof}

\subsection{Matrix Bernstein approximation bound via analytic sparsity}\label{sec:matrix-bernstein-analytic-sparsity}
Note that the \emph{true} Hessian is distributed as $\nu''(\norm{\sigma})^{1/2}\mathsf{GOE}(n)$ with high probability. Consequently, it is traceless. However, in order to analyze the rounding scheme introduced in the companion paper
, it is critical show that the Bernstein polynomial has trace $\approx 1$. To do this, we need to take the Bernstein polynomial of the \emph{shifted} Hessian $2\nu''(\norm{\sigma}^2_2)^{1/2}\Id_n + \frac{1}{n}\nabla^2 H(\sigma) = p\left(\frac{1}{n}\nabla^2 H(\sigma)\right)$ and invoke~\Cref{lem:bernstein-is-semi-circle} (proved below). This allows one to conclude that, for an appropriate choice of the ramp (fixed in~\Cref{cor:large-hs-correlation-bernstein-hessian}), the trace of the (appropriately) normalized (divided by $1/(\delta n)$) Bernstein polynomial can be made $\approx 1$.

\begin{lemma}[Equivalence with semicircle moments]\label[Lemma]
{lem:bernstein-is-semi-circle}
    Let $X(\sigma) = M(\sigma)/(2\nu''(\norm{\sigma}_2^2)^{1/2})$ for a symmetric matrix $M(\sigma)$ with entries that are low-degree polynomials in $\sigma$. Let $\cD$ be a pseudo-distribution over high-entropy step distributions over $\sigma$ satisfying the conclusions of \Cref{lem:hessian-schatten-norm} and let $p$ be a degree-$d$ polynomial.
    Then, for all positive integer $r$,
    \[ \{ \text{HES Constraints}\} \proves\E_{\sigma}\left(\left[\Tr \frac{1}{n} p(X(\sigma))\right] - n\E_{x \sim \mu_{sc}}\left[p(x)\right]\right)^{2r} \le  o(n^{2r})\,\norm{p}_{L^2}^{2r} \]
    where $\mu_{sc}$ is the Wigner semicircle law given by $\mu_{sc}(x) = \frac{2}{\pi}\sqrt{1-x^2}$.
\end{lemma}
\begin{proof}
    We write $p(\cdot)$ in the monomial basis, and then apply a Cauchy-Schwarz inequality on the summands to separate the constants from the SoS program variables (which are functions of $\sigma$). This is followed by an application of~\Cref{lem:hessian-schatten-norm} in concurrence with~\Cref{lem:l1-to-lt}. We will crucially use the facts that the moment variable $r = O_\eps(1)$ and that the degree of the polynomial $d = O_\eps(1)$.

    First, we define
    \begin{align*}
         y_i(\sigma) := \Tr\left[\left(\frac{1}{n}X(\sigma)\right)^i\right] - n\E_{x \sim \mu_{sc}}\left[x^i\right]\, .
    \end{align*}
    Then, for all bounded $i$ and bounded even $ j$,
    \begin{equation}\label{eq:hessian-centered-bound}
        \E\left[y_{i}(\sigma)^j\right] \le o(n^j)\,,
    \end{equation}
    which follows directly from~\Cref{lem:hessian-schatten-norm}.

    Expressing $p(x) = \sum_{i=0}^d c_ix^i$ and using linearity of the trace and expectation immediately gives
    \begin{align*}
        \E\left[\left(\frac{1}{n}\Tr[p(X(\sigma))] - \E_{x \sim \mu_{sc}}[p(x)]\right)^{2r}\right] &=  \E\left[\left(\sum_{i=0}^dc_i\Tr\left[(\frac{1}{n}X(\sigma))^i\right] - \sum_{i=0}^d c_i\E_{x \sim \mu_{sc}}[x^i]\right)^{2r}\right] \\
        &= \E\left[\left(\sum_i c_i y_i(\sigma)\right)^{2r}\right]\, .
    \end{align*}

    Applying the strategy of first isolating the terms containing the program variables $\left(\{y_i(\sigma))\}_{i=0}^d\right)$ from the constant terms that come from the coefficients $\{c_i\}_{i=0}^d$ via an application of the Cauchy-Schwarz (CS) inequality yields,
    \begingroup
    \allowdisplaybreaks
    \begin{align}\label{eq:cs-decorrelate-hessian}
        &\E_{\sigma \sim \cD}\left[\left(\sum_{i=0}^d c_iy_i(\sigma)\right)^{2r}\right] = \E_{\sigma \sim \cD}\left[\left(\left\langle \mathbf{c}, \mathbf{y(\sigma)}
        \right\rangle^2\right)^r\right] \overset{\mathsf{CS}}{\le} \pE_{\sigma \sim \cD}\left[\left(\sum_{i=0}^d c_i^2\right)^r\left(\sum_{i=0}^d y^2_i(\sigma)\right)^r\right] \nonumber\\
        &= \left(\sum_{i=0}^d c_i^2\right)^r \E_{\sigma \sim \cD}\left[\left(\sum_{i=0}^d y_i^2 (\sigma)\right)^r\right] = \norm{p}^{2r}_2 \E_{\sigma \sim \cD}\left[\sum_{k_0+\dots+k_d = r}\binom{r}{k_0,\dots,k_d}y^{2k_0}_0(\sigma)\cdots y^{2k_d}_d(\sigma)\right] \nonumber \\
        &\le_{\text{\Cref{lem:l1-to-lt}}} \norm{p}^{2r}_2\E_{\sigma \sim \cD}\left[d^{r-1}\sum_{i=0}^d y_i^{2r}\right] \le \norm{p}_2^{2r}d^r o(n^{2r}) = o(n^{2r})\norm{p}_2^{2r} \, .
    \end{align}
    \endgroup
    Since $d, r = O_\eps(1)$, this yields that,
    \begin{align*}
        \E\left[\left(\frac{1}{n}\Tr[p(X(\sigma))] - \E_{x \sim \mu_{sc}}[p(x)]\right)^{2r}\right] \le \norm{p}^{2r}_2 f(\eps) o(n^{2r})\, ,
    \end{align*}
    for some bounded function $f(\eps)$, concluding the proof.
\end{proof}

Alternatively, a better dependency on the constants can be obtained using~\Cref{lem:scalar-am-gm}, however, it is irrelevant to the use case.

\begin{corollary}[Large Hilbert-Schmidt Correlation between Bernstein Polynomial and Hessian]\label[Corollary]{cor:large-hs-correlation-bernstein-hessian}
    Given the Bernstein matrix polynomial $B^{L^\alpha_\gamma}$ as defined in~\Cref{prop:matrix-bernstein-linear-ramp}, it is the case that for any high-entropy step distribution $\cD$,
    \begin{equation}\label{eq:large-hilbert-schmidt-correlation}
        \{\text{HES Constraints}\} \proves \E_{\sigma\sim\cD}\left(\frac{1}{n}\left[\langle B^{L^\alpha_\gamma}(X), X \rangle\right]-  \psi\right)^{2p} \le o_n(1)\, ,
    \end{equation}
    where $a \Id_n \preceq X \preceq b\Id_n$ and $\psi \ge 1-\epsilon$.
\end{corollary}
\begin{proof}
    To compute the correlation with the matrix Bernstein polynomial that approximates the ramp function, we will use the definition of the Bernstein polynomial followed by linearity (again) in conjunction with the bounds above.
    \begingroup
    \allowdisplaybreaks
    \begin{align*}
        \frac{1}{n}\left[\left\langle B^{L_{\alpha,\gamma}}_d(X),X\right\rangle\right] &=  \frac{1}{n}\Tr \left(X\cdot B^{L^{\alpha}_{\gamma}}_d(X)\right),
    \end{align*}
    \endgroup
    By \Cref{lem:bernstein-is-semi-circle},
    \begingroup
    \allowdisplaybreaks
    \begin{align*}
        \left(\frac{1}{\E_{x\sim\mu_{sc}}\left[B^{L^\alpha_\gamma}(x)\right]}\right)^{2p}\E_{\sigma\sim\cD}\left(\frac{1}{n}\left[\langle B^{L^\alpha_\gamma}(X), X \rangle\right]-  \E_{x \sim \mu_{sc}}\left[x\cdot B^{L^{\alpha}_\gamma}(x)\right]\right)^{2p} \le \left(\frac{1}{\E_{x\sim\mu_{sc}}\left[B^{L^\alpha_\gamma}(x)\right]}\right)^{2p}o_n(1) \le o_n(1)\,,
    \end{align*}
    \endgroup
    provided that $\E_{x\sim\mu_{sc}}\left[B^{L^\alpha_\gamma}(x)\right] = C(\eps)$ and $p = O_\eps(1)$, where the former will follow from~\Cref{lem:bernstein-normalize-sc} with the appropriate choice of constants made later. \\
    This reduces the goal to showing that,
    \begin{equation}
        \frac{\E_{x\sim\mu_{sc}}\left[x\cdot B^{L^\alpha_\gamma}(x)\right]}{\E_{x\sim\mu_{sc}}\left[B^{L^\alpha_\gamma}(x)\right]} \ge 1 - \eps\, ,
    \end{equation}
    for an appropriate choice of $d(\eps)$. Since this does not involve any SoS program variable, it need not be a SoS proof.\\
    By the definition in~\Cref{sec:bernsteiny-stuff} and appropriately rescaling the interval $\left[-2\nu''(\norm{\sigma}^2_2)^{1/2}, 2\nu''(\norm{\sigma}^2_2)^{1/2}\right] \to [-1,1]$ and denoting the re-scaled area of the semi-circle containing the valid fraction of eigenvectors as
    \begin{align*}
        \phi := \frac{A(\eps)}{2\nu''(\norm{\sigma}^2_2)^{1/2}}\, ,
    \end{align*}
    we have the following,
    \begingroup
    \allowdisplaybreaks
    \begin{align}\label{eq:l_alpha_gamma}
        L^\alpha_\gamma(x) = \left\{\begin{array}{l l}
            0\,, & x\in \left[-1,\, 1 - \phi\right), \\
            \frac{1}{\phi^2}(x-(1-\phi))\,, & x\in \left[1 - \phi,\, 1 - \phi + \phi^2\right], \\
            1\,, & x\in\left(1 - \phi + \phi^2,\, 1\right],
        \end{array}\right.
    \end{align}
    \endgroup
    where $A(\eps)$ is chosen later. The Weierstrass approximation (\Cref{thm:bernstein-approx}) gives,
    \begingroup
    \allowdisplaybreaks
    \begin{align*}
        \|B^{L^\alpha_\gamma}(x) -  L^\alpha_\gamma(x)\| \le \eps\, , \forall x \in [-1,1]\, ,
    \end{align*}
    \endgroup
    provided $\mathsf{deg}(B^{L^\alpha_\gamma}(\cdot)) \ge O\left(\frac{1}{\eps^3}\right)$.

    \paragraph{Average correlation between $x$ and $B^{L^\alpha_\gamma}(x)$ under $\mu_{sc}$} We now compute the scalar proxy for the correlation between the hessian and the matrix Bernstein polynomial under the semi-circle measure. This will be normalized appropriately using~\Cref{lem:bernstein-normalize-sc}.
    \begingroup
    \allowdisplaybreaks
    \begin{align*}
        &\E_{x\sim\mu_{sc}}\left[x\cdot B^{L^\alpha_\gamma}(x)\right] = \int_{-1}^1 x\cdot B^{L^\alpha_\gamma}(x) \frac{2}{\pi}\sqrt{1-x^2}dx \\
        &\ge \int_{-1}^{1-\phi}x\left(L^\alpha_\gamma + \eps\right)\frac{2}{\pi}\sqrt{1-x^2}dx +  \int_{1-\phi}^1 x\left(L^\alpha_\gamma(x) - \eps\right)\frac{2}{\pi}\sqrt{1-x^2}dx \\
        &= \eps\int_{-1}^{1-\phi}xd\mu_{sc}(x) - \eps\int_{1-\phi}^1 xd\mu_{sc}(x) + \int_{1-\phi}^{1-\phi+\phi^2}x\cdot \left(\frac{1}{\phi^2}(x - (1-\phi))\right)d\mu_{sc}(x) + \int_{1-\phi+\phi^2}^1 xd\mu_{sc}(x)\\
        &= \underbrace{\eps\left(\int_{-1}^{1-\phi}xd\mu_{sc}(x) - \int_{1-\phi}^1xd\mu_{sc}(x)\right)}_{\mathsf{I}(\eps,\phi)} + \underbrace{\int_{1-\phi}^{1-\phi+\phi^2}x\cdot \left(\frac{1}{\phi^2}(x - (1-\phi))\right)d\mu_{sc}(x) }_{\mathsf{II}(\phi)} + \underbrace{\int_{1-\phi+\phi^2}^1 xd\mu_{sc}(x)}_{\mathsf{III}(\phi)}\\
        &\ge \frac{-2\eps}{3\pi}\left((1-(1-\phi)^2)^{3/2} + (1-(1-\phi)^2)^{3/2}\right) + \frac{2}{3\pi}\left(1-(1-\phi+\phi^2)^2\right)^{3/2} \\
        &= \frac{2}{3\pi}\left[(1-(1-\phi+\phi^2)^2)^{3/2} - 2\eps\phi^{3/2}(2-\phi)^{3/2}\right]
    \end{align*}
    \endgroup
    where the first inequality follows from using the closeness of the Bernstein polynomial to the $L^\alpha_\gamma$ function, and the final inequality follows from the observation that $\mathsf{II}(\phi) \ge 0$ and the use of~\Cref{prop:semi-circle-integrals} to evaluate $\mathsf{I}(\eps,\phi)$ and $\mathsf{III}(\phi)$.

    \paragraph{The final lower bound} Combining the upper bound on $\E_{x \sim \mu_{sc}}\left[B^{L^\alpha_\gamma}(x)\right]$ via~\Cref{lem:bernstein-normalize-sc} with the lower bound on the correlation term computed above yields,
    \begingroup
    \allowdisplaybreaks
    \begin{align*}
        \frac{\E_{x\sim\mu_{sc}}\left[x\cdot B^{L^\alpha_\gamma}(x)\right]}{\E_{x\sim\mu_{sc}}\left[B^{L^\alpha_\gamma}(x)\right]} &\ge \frac{\frac{2}{3\pi}\left[(1-(1-\phi+\phi^2)^2)^{3/2} - 2\eps\phi^{3/2}(2-\phi)^{3/2}\right]}{1+\eps - \mathsf{CDF}(1-\phi)} \\
        &= \frac{\frac{2}{3\pi}\left[(2\phi + 2\phi^3 - 3\phi^2 - \phi^4)^{3/2} - 2\eps\phi^{3/2}(2-\phi)^{3/2}\right]}{\frac{4\sqrt{2}}{3\pi}\phi^{3/2} + o(\phi^{3/2})} \\
        &\ge \frac{\frac{2}{3\pi}2^{3/2}\phi^{3/2}\left((1 - 4\phi)^{3/2} - \frac{2}{2^{3/2}}\eps(2-\phi)^{3/2}\right)}{\frac{4\sqrt{2}}{3\pi}\phi^{3/2} + o(\phi^{3/2})} \\
        &\ge \frac{\frac{4\sqrt{2}}{3\pi}\phi^{3/2}\left(\left(1-\frac{3}{2}4\phi + O(\phi^2)\right) - 2\eps\right)}{\frac{4\sqrt{2}}{3\pi}\phi^{3/2} + o(\phi^{3/2})} \\
        &= \frac{1 - 20\phi - 2\eps}{1 + o(\phi)}\, .
    \end{align*}
    \endgroup

    For completeness, we also demonstrate an upper bound of (roughly) the same quantitative strength on the normalized correlation between $x$ and $B^{L^\alpha_\gamma}(x)$ under $\mu_{sc}$. To do this, we utilize a lower bound on the normalization term given in~\Cref{lem:bernstein-normalize-sc} and compute an upper bound on $x\cdot B^{L^\alpha_\gamma}(x)$ using the uniform upper bound of $\phi^2$ on the ramp part.

    \paragraph{The final upper bound} A similar computation as above yields,
    \begingroup
    \allowdisplaybreaks
    \begin{align*}
        &\frac{\E_{x\sim\mu_{sc}}\left[x\cdot B^{L^\alpha_\gamma}(x)\right]}{\E_{x\sim\mu_{sc}}\left[B^{L^\alpha_\gamma}(x)\right]} \le \frac{\int_{-1}^1x\cdot\left(L^\alpha_\gamma(x) + \eps\right) d\mu_{sc}(x)}{(1+2\eps)\frac{4\sqrt{2}}{3\pi}\phi^{3/2} + O(\phi^{5/2}) + o(\phi^{5/2}) + \eps} \\
        &= \frac{\int_{-1}^{1-\phi}x\cdot 0 d\mu_{sc}(x) + \int_{1-\phi}^{1-\phi+\phi^2}\frac{x}{\phi^2}\cdot\left(x - (1-\phi)\right)d\mu_{sc}(x) + \int_{1-\phi+\phi^2}^1xd\mu_{sc}(x) + \eps\E_{x\sim\mu_{sc}}\left[x\right]}{(1+2\eps)\frac{4\sqrt{2}}{3\pi}\phi^{3/2} + O(\phi^{5/2}) + o(\phi^{5/2}) + \eps} \\
        &\le \frac{\int_{1-\phi}^{1-\phi+\phi^2}xd\mu_{sc}(x) + \int_{1-\phi+\phi^2}^1 xd\mu_{sc}(x)}{(1+2\eps)\frac{4\sqrt{2}}{3\pi}\phi^{3/2} + O(\phi^{5/2}) + o(\phi^{5/2}) + \eps} = \frac{\int_{1-\phi}^1 xd\mu_{sc}(x)}{(1+2\eps)\frac{4\sqrt{2}}{3\pi}\phi^{3/2} + O(\phi^{5/2}) + o(\phi^{5/2}) + \eps} \\
        &= \frac{\frac{2}{3\pi}\left(1- (1-\phi)^2\right)^{3/2}}{(1+2\eps)\frac{4\sqrt{2}}{3\pi}\phi^{3/2} + O(\phi^{5/2}) + o(\phi^{5/2}) + \eps} = \frac{\frac{2}{3\pi}\left(\phi(2-\phi)\right)^{3/2}}{(1+2\eps)\frac{4\sqrt{2}}{3\pi}\phi^{3/2} + O(\phi^{5/2}) + o(\phi^{5/2}) + \eps} \\
        &\le \frac{\frac{2}{3\pi}2^{3/2}\phi^{3/2}}{(1+2\eps)\frac{4\sqrt{2}}{3\pi}\phi^{3/2} + O(\phi^{5/2}) + o(\phi^{5/2}) + \eps} = \frac{1}{(1+3\eps) + O(\phi) + o(\phi)}\, .
    \end{align*}
    \endgroup
    A choice of $A(\eps) = O(\eps)$ suffices for the desired approximation ratios, yielding $\phi = \frac{O(\eps)}{2\nu''(\norm{\sigma}^2_2)^{1/2}}$.
\end{proof}

\section{Feasibility}\label{sec:feasibility-main}
In this section we will demonstrate that a randomized version of Subag's Hessian ascent algorithm (\Cref{alg:subag-2}) can be cast as a HES process (\Cref{lem:feasibility}). Doing this demands showing that the conditional covariance matrix, which basically projects in the $\delta$n-dimensional subspace corresponding to the top part of the eigenspectrum of the Hessian, has bounded operator norm and is close to a matrix-low degree polynomial. This low-degree matrix polynomial will be the Bernstein matrix polynomial introduced in \Cref{sec:matrix-bernstein-analytic-sparsity}.

As a minor technical detail, since the HES process on the unit sphere is described as an appropriately scaled set of small Gaussian steps, we will need a lemma that converts between the moments of the real Haar measure on the unit sphere and scaled Gaussian moments (\Cref{lem:denom-conc}).  

\subsection{Feasibility via a near-optimal randomized algorithm}\label{sec:feasibility}

To show feasibility we will work with~\Cref{alg:subag-2}, which is a variant of Subag's algorithm~(\Cref{alg:subag}). Fix $ 0 < \epsilon < 1$ and $A(\eps)$.

\begin{algorithm}[H]
\begin{algorithmic}

\State $v_1 \sim \cS^{A(\eps)n-1}(1)$.
\State $w_1 = \frac{1}{\sqrt{k}}v_1$.
\While{$i \in \{2,\dots,k\}$}
    \State $v_i \sim \cS^{\mathsf{eigen_{A(\eps)n - 1}}(\frac{1}{n}\nabla^2 H(w_{i-1}))}(1)$
    \State $w_i = w_{i-1} + \frac{v_i}{\sqrt{k}}$.
    \State $i \leftarrow i + 1$.
\EndWhile
\State Output $v_k$.
\end{algorithmic}
\caption{Randomized Hessian-Ascent~\cite{huang2022computational}}\label{alg:subag-2}
\end{algorithm}

The only difference between~\Cref{alg:subag} and~\Cref{alg:subag-2} is that $x$ is chosen uniformly at random from the unit sphere on the $\delta n$-dimensional subspace corresponding to the top eigenvectors of the Hessian projected orthogonal to $v_{i-1}$ as $\nabla^2_{\perp} H(v_{i-1})$. Critically,~\Cref{alg:subag-2} continues to be nearly optimal for the same reason as the original algorithm is~\cite[Lemma 3 \& Theorem 4]{subag2021following} in conjunction with the fact that the first step can be arbitrary. A more detailed explanation for the equivalence between these two algorithms can be found in~\cite[Section 3.2.7]{huang2022computational}.

Consequently, showing feasibility amounts to showing that~\Cref{alg:subag-2} satisfies, at every iteration $i \in [k]$, the constraints of the SoS program at the corresponding cumulant SoS variable. Demonstrating this is done by showing that:
\begin{enumerate}
    \item The uniform distribution on the unit-sphere satisfies the conditional maximum-entropy constraints imposed on valid pseudo-distributions at every iterate $i \in [k]$.
    \item The cumulative norm increments have the same norms \emph{and} are incrementally orthogonal. This is trivially true by construction and is omitted in the formal proof.
    \item The conditional pseudo-expectation of the quadratic form of the hessian under $v_i$ is at least $(1-\eps)2\nu''(q_i)$.
\end{enumerate}

The statements above are proved formally in~\Cref{lem:feasibility} below.
\begin{lemma}[Feasibility of the SoS program]\label[Lemma]{lem:feasibility}
    Given $0 < \eps < 1$ and $A(\eps)$, the following facts hold:
    \begin{enumerate}
        \item For $i = 1,\dots,k$,
        \begin{align*}
            &\left\{\pcum{j}{v_i\,\mid\,v_{i-1}} = o_n(1)\right\}_{1,\dots,i-2}\, , j \ne 2\, ,\\
            &\left\{\opnorm{\E[v_iv_i^\sT\,\mid\, v_{i-1}]} \le \frac{1 + o_n(1)}{A(\eps)n}\right\}_{\mid 1,\dots,i-2} \, .
        \end{align*}
        \item For $i = 1,\dots,k$,
        \begin{align*}
            &\E_{x \sim \cS^{A(\eps) n - 1}(\mathsf{eigen}\frac{1}{n}(\nabla^2_\perp H(w_{i-1})))}[x^\sT\frac{1}{n}\nabla^2_\perp H(w_{i-1}) x] \ge 2\nu''(q)^{1/2} - \epsilon\, .
        \end{align*}
    \end{enumerate}
\end{lemma}
\begin{proof}
    Note that every probability distribution is a pseudo-distribution, since $\E_x[1] = 1$ (by normalization of probability) and $\E_x[p^2(x)] \ge 0$ for any polynomial $p$ since it can be viewed as an inner product between non-negative quantities. Therefore, we set $\pE = \E_x$ to show that the maximum-entropy constraints bounds are satisfied.~\cite[Lemma 3]{subag2021following} immediately implies that the expected quadratic form of the hessian is arbitrarily close to the largest eigenvalue under the uniform distribution on the sphere of the top $A(\eps) n$-dimensional subspace of the Hessian orthogonal to the current iterate.

    \paragraph{The $\opnorm{\cdot}$ constraint} Proving this constraint amounts to showing that the operator norm of the covariance of $v_i$ conditioned on $v_1,\dots,v_{i-1}$ is bounded from above by $\frac{1}{A(\eps)n}$.
    Since $v_i \sim \cS^{A(\eps) n - 1}(\mathsf{eigen}(\frac{1}{n}\nabla^2_\perp H(w_{i-1})))$, note that conditioned on $v_1,\dots,v_{i-2}$,
    \begin{align*}
        \opnorm{\E[v_iv_i^\sT\,\mid\, v_{i-1}]} &= \opnorm{\E_{v \sim \cS^{A(\eps) n - 1}(\mathsf{eigen}(\frac{1}{n}\nabla^2_\perp H(w_{i-1})))}\left[vv^\sT\right]} \\
        &=_{\text{Rotational Invariance of Haar-Measure}} \opnorm{\E_{v' \sim \cS^{n-1}(1^{A(\eps)n} \oplus\,0^{(1-A(\eps))n})}\left[v'v'^\sT\right]} \\
        &=_{\text{\Cref{lem:denom-conc} with }q=1} \opnorm{\frac{1 - o_n(1)}{A(\eps) n}\E_{x \sim \cN(0, \Id_{A(\eps)n} \oplus0_{(1-A(\eps))n})}\left[xx^\sT\right]} \\
        &= \frac{1 - o_n(1)}{A(\eps) n}\opnorm{\begin{bmatrix} 1 & \\ & \ddots & \\ & & 1 \\ & & & 0 & \\ & & & & \ddots & \\ & & & & & 0 \end{bmatrix}} \le \frac{1 + o_n(1)}{A(\eps) n}\, ,
    \end{align*}
    where in the above calculation $v' \sim \cS^{n-1}(1^{A(\eps)n} \oplus 0^{(1-A(\eps))n})$ refers to a unit vector chosen uniformly at random from the $A(\eps)n$-dimensional sphere where the last $(1-A(\eps))n$-coordinates are $0$.

    \paragraph{The $\kappa_{j}[v_i\mid v_{i-1}]$ constraint} To show this constraint is satisfied, we follow a similar strategy. Note that, $\forall j \ne 2$, the following holds for every entry of the $j$-th conditional cumulant tensor,
    \begin{align*}
        \kappa_j[v_i\mid v_{i-1}] 
        &=_{\text{by moment-cumulant formula}} \sum_{\pi \in \cP(k_1,\dots,k_j)} (-1)^{|\pi|-1}(|\pi|-1)!\prod_{b \in \pi}\E\left[\prod_{a \in b}v_{i_a}\right] \\
        &=_{\text{by \Cref{lem:denom-conc} and rotational invariance}} \sum_{\pi \in \cP(k_1,\dots,k_j)}(-1)^{|\pi|-1}(|\pi|-1)!\prod_{b \in \pi}\frac{1-o_n(1)}{\left(A(\eps)n\right)^{|b|}}\E\left[\prod_{a \in b}v'_{i_a}\right] \\
        &= \sum_{\pi \in \cP(k_1,\dots,k_j)} \frac{(-1)^{|\pi|-1}(|\pi|-1)!(1-o_n(1))}{(A(\eps)n)^{j}}\prod_{b \in \pi}\E\left[\prod_{a \in b}v'_{i_a}\right]\, ,
    \end{align*}
    where in the above calculation we again used~\Cref{fact:gaussian-cums}. Now, since this is a sum of the form,
    \[
        \sum_\pi \frac{(1 - o_n(1))}{(\delta n)^{j}}a_\pi\, , 
    \]
    we can conclude that,
    \[
        \norm{\kappa_j[v_i \mid v_{i-1}] - \kappa_j\left[\cN(0, \Id_{A(\eps)n \oplus0_{(1-A(\eps))n}})\right]}_F^2 \le o_n(1)\, . 
    \]

    \paragraph{The approximate-LDP constraint} We now show that the original algorithm's updates satisfy the approximate-LDP constraint in the program by demonstrating that the choice of polynomial corresponding to a Bernstein polynomial of sufficiently large (but constant) degree has the following properties:
    \begin{enumerate}
      \item It permits a (weakly) orthogonal decomposition using the Hermite-Weiner fourier basis, by virtue of being square integrable on $\cS^{n-1}(1)$.
      \item By~\Cref{prop:matrix-bernstein-linear-ramp}, it is arbitrarily close in Frobenius norm to the projector into the top $A(\eps)n$-dimensional eigenspace of $\nabla^2H(\sigma)$.
    \end{enumerate}
    Note that any $f \in \ell_2(\Sigma)$ can be expressed as,
    \[
        f(x) =\sum_{\alpha \in \N^n}\left\langle f, \He_\alpha(x\mid \Sigma)\right\rangle_{\ell^2(\Sigma)} \He_\alpha(x\mid \Sigma)\, .
    \]
    Note that for $x = \sum_{j=1}^{i-1} v_j$ with $\norm{v_j}^2 = 1$ and $f = p_d := B^d_{L^{\alpha}_\gamma}\left(c_i\Id_n + \frac{1}{n}\nabla^2 H(\sigma)\right)$ with the Fourier decomposition taken element-wise in the matrix polynomial, it is the case that,
    \begin{align*}
        &\norm{f_{ij}}^2_2 := \langle f_{ij}(x), f_{ij}(x) \rangle_{\ell^2(\Sigma)} = \int_{x \sim \cN(0, \Sigma)}\frac{1}{\delta n}B^d_{L^\alpha_\gamma}\left(c_i(\Id_n)_{ij} + \frac{1}{n}\nabla^2 H(x)_{ij}\right) dg(x) < \infty\,.
    \end{align*}
    Since the algorithm samples uniformly at random from the $A(\eps)n$-dimensional sphere corresponding to the top eigenspace of $\nabla^2 H(\sigma)$, we can then apply~\Cref{lem:denom-conc} to conclude that the appropriately scaled conditional covariance of this algorithm is basically a projector into this subspace times the identity matrix, and a triangle inequality then immediately implies feasibility as demonstrated below:
    \begin{align*}
        &\norm{\E\left[v_iv_i^\sT\mid w_{i-1}\right] - p_d(w_{i-1})}_F^2 = \norm{\E\left[v_iv_i^\sT\mid w_{i-1}\right] - \sum_{\alpha, |\alpha| \le d} \langle \E\left[v_iv_i^\sT\mid w_{i-1}\right], \He_d\rangle_{\ell^2(\Pi_{\delta n})} }_F^2 \\
        &= \norm{\E\left[v_iv_i^\sT\mid w_{i-1}\right] - \frac{1}{\delta n}B^d(c_i\Id_n + \frac{1}{n}\nabla^2 H(\sigma))}_F^2 \\
        &\le_{\text{triangle inequality}} \norm{\E\left[v_iv_i^\sT\mid w_{i-1}\right] - \frac{1}{\delta n}\Pi_{\delta n}}^2_F + \frac{1}{\delta n}\norm{\Pi_{\delta n} - \frac{1}{\delta n}B^d(c_i\Id_n + \frac{1}{n}\nabla^2 H(\sigma))}^2_F \\
        &=_{\text{by \Cref{prop:matrix-bernstein-linear-ramp}}}  \norm{\E\left[v_iv_i^\sT\mid w_{i-1}\right] - \Pi_{\delta n}}^2_F + o_n(1) \\
        &=_{\text{by \Cref{lem:denom-conc}}} \norm{\frac{1 - o_n(1)}{\delta n}\E_{x \sim \cN\left(0,\, \Pi_{\mathsf{eigen}\left(\nabla^2_\perp H(w_{i-1})\right)}\right)}[(xx^\sT)] - \frac{1}{\delta n} \Pi_{\delta n}}_F^2 + o_n(1) \\
        &=_{\E[xx^\sT] = \Pi_{\delta n}}\norm{\frac{1-o(n(1))}{\delta n}\Pi_{\delta n} - \frac{1}{\delta n}\Pi_{\delta n}}_F^2 + o_n(1) \\
        &=\norm{-\frac{o_n(1)}{\delta n}\Pi_{\delta n}}_F^2 + o_n(1) = o_n(1) + o_n(1) = o_n(1)\, ,
    \end{align*}
    where we used the fact that the covariance of the Gaussian is given by an orthogonal transformation into the top $\delta n$-dimensional subspace of $\nabla^2H(w_{i-1})$, which in this case is just $\Pi_{\delta n} = \Pi_{\mathsf{eigen}\left(\nabla^2_\perp H(w_{i-1})\right)}$.
    For every $q \in \N$,
    \begin{align*}
        \E[(v_iv_i^\sT)^{\ot q} \mid w_{i-1}] &=_{\text{\Cref{lem:denom-conc}}} \frac{1-o_n(1)}{\left(\delta n\right)^q}\E_{x \sim \cN(0, \Id_{A(\eps)n} \oplus0_{(1-A(\eps))n})}\left[(xx^\sT)^{\ot q}\right] \\
        &=_{\text{\Cref{lem:isserlis}}} (q-1)!!\,\sym\left[\left(\E[xx^\sT]\right)^{\ot q}\right] \\
        &= (q-1)!!\,\sym\left[\left((1-o_n(1)) \E[v_iv_i^\sT \mid w_{i-1}]\right)^{\ot q}\right] \\
        &= (1-o_n(1))(q-1)!!\,\sym\left[\left(\E[v_iv_i^\sT \mid w_{i-1}]\right)^{\ot q}\right] \, . 
    \end{align*}  
    Rewrite the conditional covariance as $\E[v_iv_i^\sT \mid w_{i-1}] = p(w_{i-1}) + B$  where $B$ is a matrix with $o_n(1)$ Frobenius norm (as shown above). This yields that,
    \[
        \E[(v_iv_i^\sT)^{\ot q} \mid w_{i-1}] = (1-o_n(1))(q-1)!!\,\sym\left[\left(p(w_{i-1}) + B\right)^{\ot q}\right]\, ,
    \]
    and, therefore,
    \begin{align*}
        \norm{\E[(v_iv_i^\sT)^{\ot q} \mid w_{i-1}] - p(w_{i-1})^{\ot q}}_F^2 &\le (q-1)!!\norm{\sum_{i \in [q], j = q-i}\sym\left(B^{\ot i}\ot p(w_{i-1})^{\ot j}\right)}_F^2 \\
        &\le_{\text{$\sym$ is contractive}} (q-1)!!\norm{\sum_{i \in [q], j = q-i}B^{\ot i}\ot p(w_{i-1})^{\ot j}}_F^2 \\
        &=_{\text{\Cref{lem:l1-to-lt}}} q^2(q-1)!!\sum_{i \in [q], j = q-i}\norm{B^{\ot i}\ot p(w_{i-1})^{\ot j}}_F^2 \\
        &\le_{\text{\Cref{lem:scalar-am-gm} + \Cref{prop:matrix-bernstein-linear-ramp}}}\frac{q^3(q-1)!!}{\delta^q} \frac{1}{n^q}\, .  
    \end{align*}

    \paragraph{Soundness \& Optimality of Taylor Expansion} The last step is to invoke soundness for all distributions~(\Cref{sec:actual-conditional}) which implies that the corresponding pseudo-distribution comprising the low-degree moments of the ``Subag distribution'' implicit in~\Cref{alg:subag-2} satisfies the entropy constraints. Invoking~\cite[Corollary 59]{arous2020geometry} in conjunction with a Taylor series argument immediately implies that the hessian has the largest contribution, and the final energy is lower-bounded by $\sum_{i=1}^k \frac{1}{k}\nu''(q_i)^{1/2} - \eps\cdot k(\eps) - O\left(\frac{1}{k(\eps)^{1/2}}\right)$, which is arbitrarily close to $\int_{0}^1 \nu''(q)^{1/2}dq$ with the appropriate choice of $k(\eps)$. The upper bound follows from~\Cref{lem:chen-sen}, as no configuration can have value greater than the true maximum.
\end{proof}

\subsection{Equivalence of rescaled spherical and Gaussian moments}
We state and prove a small lemma that allows us to freely go between the low-degree moments of the uniform measure on the real unit sphere and an appropriate rescaling of the comparable Gaussian moments.

\begin{lemma}[Equivalence of re-scaled gaussian and uniform moments]\label[Lemma]{lem:denom-conc}
    Given $x \sim \cN(0, \Id_n)$ and $v \sim \mathsf{Unif}(\cS^{n-1}(1))$, the following holds,
    \begin{equation}
        \E_{v}[(vv^\sT)^{\ot q}] = \frac{(1 - O_n(1/n))}{n^q}\E_x[(xx^\sT)^{\ot q}]\, .
    \end{equation}
\end{lemma}
\begin{proof}
    The proof will use the rotational invariance of both distributions and an explicit entry-by-entry comparison between the rescaled Gaussian moment matrix and the real Haar distribution moment matrix.
    \begin{align*}
        \E_x[(xx^\sT)^{\ot q}] = \sum_{a_1,\dots,a_{2q}=1}^n \E_{x}[x_{a_1}\cdots x_{a_{2q}}](e_{a_1}\ot \cdots \ot e_{a_q})(e_{a_{q+1}} \ot \cdots \ot e_{a_{2q}})^\sT\, .
    \end{align*}
    By exactly the same argument,
    \begin{align*}
        \E_v[(vv^\sT)^{\ot q}] = \sum_{a_1,\dots,a_{2q}=1}^n \E_{v}[v_{a_1}\cdots v_{a_{2q}}](e_{a_1}\ot \cdots \ot e_{a_q})(e_{a_{q+1}} \ot \cdots \ot e_{a_{2q}})^\sT\, .
    \end{align*}
    It is elementary to see from~\Cref{lem:standard-gaussian-moments} and the definition of the $\Gamma$-function that the $2m_i$-th moment of $x_i \sim \cN(0,1)$ is,
    \begin{equation}\label{eq:gaussian-mmoments}
        \E_{x_i \sim \cN(0,1)}[x^{2m_i}_i] = (2m_i-1)(2m_i-3)\cdots (3)(1) = \frac{(2m_i)!}{2^{m_i} m_i!} = \frac{\Gamma(m_i + 1/2)2^{m_i}}{\sqrt{\pi}}\, .
    \end{equation}
    Symmetry implies that the odd moments of the $1$-marginals of the uniform distribution are $0$. By
    \cite[Theorem, Pg. 2]{folland2001integrate}, the even moments are
    \begin{align}\label{eq:spherical-moments}
        \E_{v}[v_{a_1}\cdots v_{a_{2q}}] &= \E_v\left[\prod_{i=1}^{\ell}(v_i^2)^{m_i}\right] \nonumber\\
        &= \int_{v \in \R^n}\left(\prod_{i=1}^{\ell}(v_i^2)^{m_i}\right)d\sigma = \frac{\int_{v \in \cS^{n-1}(1)}v^{\alpha}d\sigma}{\int_{v \in \cS^{n-1}(1)}d\sigma} \nonumber\\
        &= \frac{2\Gamma(m_1 + 1/2)\Gamma(m_2 + 1/2)\cdots\Gamma(m_\ell + 1/2)\left(\Gamma(1/2)\right)^{n - \ell}}{\Gamma(q + n/2)}\frac{\Gamma(n/2)}{2\Gamma(1/2)^n} \nonumber\\
        &= \frac{\E_{x}[x_{a_1}\cdots x_{a_{2q}}]}{2^q}\frac{\Gamma(n/2)}{\Gamma(q + n /2)}\nonumber \\
        &= \frac{\E_{x}[x_{a_1}\cdots x_{a_{2q}}]}{n(n+2)(n+4)\cdots (n+2q-2)}
        \,,
    \end{align}

    where $2(m_1 + \cdots + m_\ell) = 2q$ with $m_i$ denoting the number of times an index $a_i \in [n]$ appears in $\{a_1,\dots,a_{2q}\}$ and $\ell$ being the number of distinct indices.
    Combining~\Cref{eq:spherical-moments} and~\Cref{eq:gaussian-mmoments} immediately yields that,
    \begin{equation}
        \E_{v}[v_{a_1}\cdots v_{a_{2q}}] = \frac{(1 - O_n(1/n))}{n^q}\E_x[x_{a_1}\cdots x_{a_2q}]\, ,
    \end{equation}
    for every $a_1,\dots,a_{2q} \in \{1,\dots,n\}$.
\end{proof}

\section{Local iterative algorithms vs. HES SoS hierarchy}\label{sec:robustness}
In this section, we provide three concrete examples of inputs from a ``modified'' spherical spin glass ensemble, and then give proofs for the first two (and a sketch for for the latter) of the sub-optimality of local iterative algorithms on these (such as Hessian ascent~\cite{subag2021following} and its randomized generalization~(\Cref{alg:subag-2})) along with arguments for why the HES SoS relaxation provides certificates with better values. The last example, in particular, motivates the conjecture that the HES SoS relaxation is capable of ``disentangling'' arbitrary correlations in the input, provided they are sufficiently smooth.

\subsection{Invariance with respect to the spherical ideal}
The first example is one that ``fools'' the randomized version~(\Cref{alg:subag-2}) of the Hessian ascent~\cite{subag2021following} algorithm by a scaling of the degree-2 component.
\begin{example}[Homogeneous degree scaling]\label[example]{example:degree-scaling}
Let the Hamiltonian $H$ be,
\[
    H(\sigma) = \alpha_2\norm{\sigma}^{100}_2 H_2(\sigma) + H_4(\sigma)\, , 
\]
for some choice of $\alpha > 0$. Then, the algorithm given in~\Cref{alg:subag-2} is provably suboptimal, but the HES SoS relaxation does better.
\end{example}
\begin{proof}
A proof sketch for the claim above is given, whereby the geometry and scaling of the top $\delta n$-dimensional eigenspace of the Hessian is calculated, and then it is demonstrated that~\Cref{alg:subag-2} initially follows the \emph{incorrect} direction, thereby accumulating an incorrect amount of energy before it changes direction. However, the HES relaxation is able to use the Schatten $p$-norm of the collective Hessian to ``distinguish'' between the space where the large (and dominant) amount of signal from the top part of the spectrum resides, while the rank-$1$ terms contribute signal that is suppressed by the moment scalings implied by the HES constraints.

\paragraph{Eigenspectrum of the Hessian} When $\sigma \in \cS^{n-1}(1)$, note that $H(\sigma)$ is indistinguishable from an instance of a spherical spin glass Hamiltonian.
The Hessian of $H(\sigma)$ is,
\begin{align*}
    \frac{1}{n}\nabla^2 H(\sigma) =&\,\,\alpha_2\norm{\sigma}^{100}_2\frac{1}{n}\nabla^2 H_2(\sigma) + \frac{1}{n}\nabla^2 H_4(\sigma) + 100\alpha_2\norm{\sigma}_2^{98}\frac{1}{n}H_2(\sigma)\Id_n \\
    &+ \underbrace{9800\alpha\norm{\sigma}_2^{96}\frac{1}{n}H_2(\sigma)\sigma\sigma^{\sT} + 100\alpha_2\norm{\sigma}_2^{98}\left(\sigma\left(\nabla \frac{1}{n}H_2(\sigma)\right)^\sT + \frac{1}{n}\nabla H_2(\sigma)\sigma^\sT\right)}_{\text{rank-1 terms}}\, .
\end{align*}
By~\cite[Proof of Lemma-3]{subag2021following}, it is clear that,
\[
    \alpha\norm{\sigma}_2^{100}\frac{1}{n}\nabla^2 H_2(\sigma) \sim \sqrt{2\alpha_2}\norm{\sigma}_2^{50}\mathsf{GOE}(n)\,,\text{ and},\,\frac{1}{n}\nabla^2 H_4(\sigma) \sim \sqrt{12}\mathsf{GOE}(n)\, .
\]
Furthermore,
\begin{align*}
    H_2(\sigma) \le \sup_{\sigma \in \cS^{n-1}(\norm{\sigma}^2_2)} H_2(\sigma) =_{\text{\Cref{prop:frsb-equivalent}}} \sqrt{2}\int_0^{\norm{\sigma}^2_2}dx = \sqrt{2}\norm{\sigma}^2_2\, . 
\end{align*}
This immediately implies that \emph{all} the (redundant) eigenvalues of $100\alpha_2\norm{\sigma}^{98}_2\frac{1}{n}H_2(\sigma)\Id_n$ are no more than $ 100\sqrt{2}\alpha_2 \norm{\sigma}^{100}_2$.

\paragraph{Suboptimality of the rank-$1$ components}
Suboptimality of the rank-$1$ terms is a consequence of the analysis of Subag~\cite{subag2018free, subag2021following} but can also be rigorously justified by noticing that they have large alignment with the gradient which leads to larger steps (under the HES constraints) and sub-optimal gain of energy after a certain number of steps\footnote{This can be made rigorous via an analysis based on a Taylor expansion argument which gives a precise iterate $i$ after which it becomes suboptimal to follow the gradient term.}.

\paragraph{Threshold for suboptimality of local algorithm} We show the sub-optimality of the variant of Subag's algorithm stated in~\Cref{alg:subag-2}. Note that, since $v \sim \cS^{\mathsf{eigen}_{\delta n}(\nabla^2 H(\sigma))}(1)_{\perp \sigma}$, by construction,
\begingroup
\allowdisplaybreaks
\begin{align*}
    &\left\langle 9800\alpha_2\norm{\sigma}_2^{96}\frac{1}{n}H_2(\sigma)\sigma\sigma^{\sT} + 100\alpha_2\norm{\sigma}_2^{98}\left(\sigma\left(\nabla \frac{1}{n}H_2(\sigma)\right)^\sT + \frac{1}{n}\nabla H_2(\sigma)\sigma^\sT\right), vv^\sT\right\rangle \\
    &= 0 + \left\langle 9800\alpha_2\norm{\sigma}_2^2\sigma\left(\frac{1}{n}H_2(\sigma)\right)^\sT, vv^\sT\right\rangle + 0 = 0 + 0 + 9800\alpha_2\norm{\sigma}^2_2\left\langle\frac{1}{n}H(\sigma),v\right\rangle\langle\sigma,v\rangle = 0\, .
\end{align*}
\endgroup
Consequently, the algorithm chooses the next step uniformly at random from the unit sphere in the space spanned by the top $\delta n$-dimensional eigenspace of $\alpha\norm{\sigma}^2_2\frac{1}{n}\nabla^2_{\perp \sigma} H_2(\sigma) + \frac{1}{n}\nabla^2_{\perp \sigma} H_4(\sigma) + 100\alpha\norm{\sigma}^{98}_2\frac{1}{n}H_2(\sigma)\left(\Id_n - \sigma\sigma^\sT\right)$, where the Hessians have been projected orthogonal to $\sigma$. Both the Hessians are symmetric with independent Gaussian coefficients, and therefore, by~\cite[Proof of Lemma 3]{subag2021following},
\begingroup
\allowdisplaybreaks
\begin{align*}
    &\alpha_2\norm{\sigma}^2_{100}\frac{1}{n}\nabla^2_{\perp \sigma} H_2(\sigma) + \frac{1}{n}\nabla^2_{\perp \sigma} H_4(\sigma) + 100\alpha\norm{\sigma}^{98}_2\frac{1}{n}H_2(\sigma)\left(\Id_n - \sigma\sigma^\sT\right) \\
    &\sim \left(\sqrt{2\alpha_2}\norm{\sigma}_2^{50}\right)\mathsf{GOE}(n-1) \boxplus \sqrt{12}\mathsf{GOE}(n-1) + 100\alpha\norm{\sigma}_2^{98}\frac{1}{n}H_2(\sigma)\\
    &\sim \left(\sqrt{2\alpha_2}\norm{\sigma}_2^{50}\right)\mathsf{GOE}(n-1) \boxplus \sqrt{12}\mathsf{GOE}(n-1) + 100\alpha\norm{\sigma}_2^{98}\frac{1}{n}H_2(\sigma) \\
    &\sim_{\text{free convolution~\cite{mingo2017free}}} \left(\sqrt{2\alpha_2\norm{\sigma}_2^{100} + 12}\right)\mathsf{GOE}(n-1) + 100\alpha_2\norm{\sigma}_2^{98}\frac{1}{n}H_2(\sigma)\, .    
\end{align*}
\endgroup
Now, note that, for $\norm{\sigma}^2_2$ sufficiently small, as the top $\delta n$-dimensional eigenbases of $\frac{1}{n}\nabla^2_{\perp \sigma}H_2(\sigma)$ and $\frac{1}{n}\nabla^2_{\perp \sigma}H_4(\sigma)$ will be mostly orthogonal (which can be verified using the fact that symmetric $\mathsf{GOE}$ matrices are orthogonally invariant), one can see that the local algorithm will (with high probability) move in the direction of $\frac{1}{n}\nabla^2_{\perp \sigma}H_4(\sigma)$. However, as $\norm{\sigma}^2_2 \to 1$, the dominating signal will come (for sufficiently large $\alpha_2 > 0$) from $\frac{1}{n}\nabla^2_{\perp\sigma}H(\sigma)$. 

\paragraph{Robustness of the HES relaxation} The HES SOS relaxation will use the following low-degree extension,
\[
    \Tilde{H}(\sigma) = \alpha_2H_2(\sigma) + H_4(\sigma)\, ,
\]
which is simply equivalent to the original spherical spin-glass Hamiltonian. Consequently,~\Cref{thm:main} implies that the HES SoS hierarchy outputs certificates with value better than direct Hessian ascent.
\end{proof}

\underline{\textbf{Note}:} While the above argument demonstrates that the HES relaxation is automatically robust to the sort of perturbed Hamiltonian shown, a small modification to the local algorithm to quotient the input Hamiltonian by the ideal $\norm{\sigma}_2^{2p} = \left(\frac{i}{k}\right)^p$ (which is just the $p$-th power of the length of the current step $\sigma_i$) fixes this example.

\subsection{Global information: The direct-sum case}
We now use projectors $\{\Pi, \Pi^{\perp}\}$ to project into subspaces, and then define different degree components to act in complementary subspaces. This leads to examples of Hamiltonians which confuse the Hessian ascent algorithm in a way that doesn't seem to be fixable with \emph{local} information, since one would require global knowledge of a path in the \emph{future} (on average) to prevent the algorithm from getting ``stuck'' in the wrong subspace. 
\begin{example}[Disjoint-subspace projector]\label[example]{example:direct-sum}
    Let the hamiltonian $H$ be,
    \[
        H(\sigma) = \alpha_1 H_2(\Pi\sigma) + \alpha_2 H_4(\Pi^{\perp}\sigma)\, .
    \]
    Then, there exists a choice of $\alpha_1,\alpha_2 \ge 0$ which causes~\Cref{alg:subag-2} to be provably suboptimal, but the HES SoS relaxation does better. 
\end{example}
\begin{proof}
We will normalize contributions to the energy in the both sub-spaces, so that the ``history'' of energy gains from the degree-$2$ component can \emph{never} fool the HES SoS hierarchy to \emph{only} follow the degree-$2$ steps, but the Hessian ascent~\cite{subag2021following} algorithm (and its randomized version) will always follow the degree-$2$ component initially, thereby making it suboptimal.

\paragraph{Eigenspectrum of the Hessian} It is clear to see that the Hessian of the as-is Hamiltonian is,
\[
    \frac{1}{n}\nabla^2 H(\sigma) = \frac{1}{n}\left(\alpha_2 \Pi \nabla^2 H_2(\Pi\sigma) \Pi + \alpha_4\Pi^\perp H_4(\Pi^\perp\sigma)\Pi^\perp\right)\, , 
\]
and that it has a spectral law given by,
\[
    \frac{1}{n}\nabla^2 H(\sigma) \sim \sqrt{2\alpha_2}\mathsf{GOE}(\Tr[\Pi]) \oplus \norm{\Pi^\perp\sigma}^2_2\sqrt{12\alpha_4}\mathsf{GOE}(\Tr[\Pi^\perp])\, .        
\]

\paragraph{Normalization of contributions to $\Pi$ \& $\Pi^\perp$} We now choose $\alpha_2, \alpha_4 \ge 0$, such that, the contribution from the degree-2 part forces the Hessian ascent algorithm into the $\Pi$-subspace, thereby, only following the degree-2 signal. Note that, near the origin (at the $i$-th step), we want
\[
    \sqrt{2\alpha_2} >> \norm{\Pi^\perp\sigma}^2_2\sqrt{12 \alpha_4}\, ,
\]
which holds for a choice of $\alpha_4 << \frac{1}{\norm{\Pi^\perp\sigma}^2_2}\frac{\alpha_2}{6}$. A concrete choice of $\alpha_4 = \frac{1}{\norm{\Pi^\perp\sigma}^2_2}\frac{\alpha_2}{12}$ suffices to obstruct the local algorithm (since it will be mislead within the first few steps).

\paragraph{Robustness of the SoS HES relaxation} The previous choice of normalization motivates the following extended hamiltonian in the ball,
\[
    \tilde{H}(\sigma) = \alpha_2\norm{\sigma}^2_2H_2(\Pi\sigma) + \alpha_4H_4(\Pi^\perp\sigma)\, .
\]
After some algebra, it is easy to see that the Hessian, ignoring rank-$1$ terms, is given as,
\[
    \frac{1}{n}\nabla^2 \Tilde{H}(\sigma) = \frac{1}{n}\left(\alpha_2\norm{\sigma}^2_2 \Pi \nabla^2 \Tilde{H}_2(\Pi\sigma)\Pi + \alpha_4 \Pi^\perp \nabla^2 \Tilde{H}_4(\Pi^\perp\sigma) \Pi^\perp + \alpha_2 \Tilde{H}_2(\Pi\sigma)\Id_n\right)\, .
\]
Assuming that $\sigma$ is the $i$-th step, the eigenspectrum of the above Hessian (modulo the rank-$1$ contributions) is,
\begin{align*}
    \left\{\sqrt{2\alpha_2\norm{\sigma}^2_2}\mathsf{GOE}(\Tr[\Pi]) \oplus \sqrt{12\alpha_4\norm{\Pi^\perp\sigma}^2_2}\mathsf{GOE}(\Tr[\Pi^\perp])\right\} + \alpha_2\left(\sum_{j=0}^{i}\frac{(j-i)}{n}\left\langle \nabla^2 \Tilde{H}_2(\Pi\sigma_{j-1}),(\Pi v_j)(\Pi v_j)^\sT\right\rangle\right)\, ,
\end{align*}
where we intrinsically used the fact that the contributions of the higher-order derivatives are suppressed under the HES distributions. Given the normalizations above, this can be simplified further to purely depend on $\alpha_2$ as,
\[
    \left\{\sqrt{2\alpha_2\norm{\sigma}^2_2}\mathsf{GOE}(\Tr[\Pi]) \oplus \sqrt{\alpha_2}\mathsf{GOE}(\Tr[\Pi^\perp]) \right\} + \alpha_2\left(\sum_{j=0}^{i}\frac{(j-i)}{n}\left\langle \nabla^2 \Tilde{H}_2(\Pi\sigma_{j-1}),(\Pi v_j)(\Pi v_j)^\sT\right\rangle\right)
\]
Since $\norm{\sigma}^2_2 < \frac{1}{2}$ for the first half of the iterates and the energy from the $\alpha_2 \Tilde{H}_2(\Pi\sigma)\Id_n$ term contributes equally to $\Pi$ and $\Pi^\perp$ (while being independent of $v$), it is clear that the SoS HES certificates for the Hessian will be dominated by the signals from the $\Pi^\perp$ space.

\paragraph{Suboptimality of local algorithm} Since the local algorithm can only see the as-is extended Hamiltonian, for the first few steps it will \emph{only} follow the $\Pi$-subspace (because of the reasoning above). Note that this implies that the algorithm will go into the wrong subspace for a constant fraction of the time, and output a solution strictly worse than the HES relaxation.
\end{proof}

\textbf{Note:} The only missing part of the argumentation presented is that the certifiable threshold rank condition of~\Cref{thm:main} is proxied for using the extremal point in the support of the bulk of the spectrum. Given that the eigenspectrum involves a block-decomposed sum of (independent) $\mathsf{GOE}$ matrices and~\Cref{lem:hessian-schatten-norm} already certifies the $p$-Schatten norm of matrices that are Wigner-like, this is a reasonable proxy. To see why, note that,
\begingroup
\allowdisplaybreaks
\begin{align*}
    \norm{\frac{1}{n}\nabla^2 \Tilde{H}(\sigma)}^p_p &= \|\frac{1}{n}\left\{\sqrt{2\alpha_2\norm{\sigma}^2_2}\mathsf{GOE}(\Tr[\Pi]) \oplus \sqrt{\alpha_2}\mathsf{GOE}(\Tr[\Pi^\perp]) \right\} + \\
    &\alpha_2\left(\sum_{j=0}^{i}\frac{(j-i)}{n}\left\langle \nabla^2 \Tilde{H}_2(\Pi\sigma_{j-1}),(\Pi v_j)(\Pi v_j)^\sT\right\rangle\right)\Id_n\|^p_p \\
    &= \frac{1}{n^p}\sum_{k=0}^p\left(\alpha_2\tilde{H}_2(\sigma)\right)^{p-k} n\norm{\sqrt{2\alpha_2\norm{\sigma}^2_2}\mathsf{GOE}(\Tr[\Pi]) \oplus \sqrt{\alpha_2}\mathsf{GOE}(\Tr[\Pi^\perp])}^k_k \\
    &= \frac{1}{n^{p-1}}\sum_{k=0}^p\left(\alpha_2\tilde{H}_2(\sigma)\right)^{p-k}\left(\norm{\sigma}^{2k}_2 \norm{\alpha_2\Pi \nabla^2 \Tilde{H}_2(\Pi\sigma)\Pi}^k_k + \norm{\alpha_4 \Pi^\perp \nabla^2 \Tilde{H}_4(\Pi^\perp\sigma) \Pi^\perp}^k_k\right)\,, 
\end{align*}
\endgroup
at which point one can invoke~\Cref{lem:hessian-schatten-norm} to bound the Schatten norms of each of the individual $\mathsf{GOE}$ components and then use the HES constraint on the moment matrices of $\sigma$ to apply a $(p,q)$-H\"older's inequality as in~\Cref{cor:shifted-hessian-holders}.

\subsection{More involved HES SoS proofs}
For this example, a SoS proof of stepwise boundedness will require upper bounding the value gained over all of a set of possible choices that the algorithm could've made in previous steps.
Here, the Hamiltonian will involve the direct sum of two usual spin glass Hamiltonians in different subspaces, along with an interaction term between the two subspaces.
These terms will be set up so that the optimal choice will be to evenly balance $\sigma$ between the two subspaces at every step, but showing an upper bound will require considering every possible choice of how to split $\sigma$ between the two subspaces at each step.

This will only be a very rough sketch of a possible argument, depending on the following assumption:

\begin{assumption}
\label[Assumption]{ass:free-prob-sos-proof}
    Let $A \in \R^{(n/2) \times (n/2)}$ and $C \in \R^{(n/2) \times (n/2)}$ be random dense symmetric-matrix-valued polynomials of $\sigma$ and $\tau$ respectively, and $B \in \R^{(n/2) \times (n/2)}$ be a random dense asymmetric-random-valued matrix polynomial of both $\sigma$ and $\tau$.
    In this way, $A$, $B$, and $C$ will be GOE-distributed with respect to the randomness of the polynomial coefficients and a HES distribution over $(\sigma,\tau)$.
    Let $a$, $b$, and $c$, be the expected operator norms of $A$, $B$, and $C$, which may depend on $\norm{\sigma}_2^2$ and $\norm{\tau}_2^2$.

    There exists some function $f(\mu^2,a,b,c)$ that expresses the with-high-probability maximum over $(u,v)$ satisfying $\norm{u}_2^2 + \norm{v}_2^2 = 1$ and $\norm{u}_2^2 = \mu^2$ of
    \[
    \left(u\;\;v\right)M\left(\!\!\begin{array}{c}
        u \\
        v
    \end{array}\!\!\right)
    :=
    \left(u\;\;v\right)\left(\begin{array}{cc}
        A & B \\
        B^{\sT} & C
    \end{array}\right)\left(\!\!\begin{array}{c}
        u \\
        v
    \end{array}\!\!\right)\]
    There also exist $g_X(\mu^2,a,b,c)$ for $X \in \{A,B,C\}$ that express the with-high-probability values of $u^{\sT}Au$, $2u^{\sT}Bv$, and $v^{\sT}Cv$, respectively, for the aforementioned maximizing $(u,v)$.

    The assumption is that there are HES SoS proofs that capture the fact that the quantities described by $f(\mu^2,a,b,c)$ and $g_X(\mu^2,a,b,c)$ actually concentrate strongly at the true values of those functions.

    Furthermore, there is $b^2/(ac)$ large enough so that $f(\mu^2,a,b,c)$ is maximized at $\mu^2 = 1/2$.
    
\end{assumption}

One might hope for the existence of this HES SoS proof because there are combinatorial arguments that compute the spectral moments of block-structured Wigner-type random matrices via a trace power method~\cite{zhu2020graphon} which has structure very similar to the proof of~\Cref{lem:hessian-schatten-norm}.

\begin{example}[Shared-subspace projector]\label[example]{example:shared-sum}
    Let the Hamiltonian $H$ be,
\[
    H(\sigma) = \alpha_4H_4(\Pi\sigma) + \alpha_6 H_{6}(\Pi\sigma, \Pi^{\perp}\sigma) + \alpha_8 H_8(\Pi^{\perp}\sigma)\, ,
\]
where $H_{6}(\rho,\tau) = \iprod{g_6, \rho^{\ot 2}\ot\tau^{\ot 4}}$ for $g_6$ an order-6 tensor of i.i.d. unit Gaussian entries. Then, there exists a choice of $\alpha_4, \alpha_6, \alpha_8 \ge 0$ which causes~\Cref{alg:subag-2} to be provably suboptimal, but the HES SoS relaxation does better with \Cref{ass:free-prob-sos-proof}. 
\end{example}
\begin{proofsketch}
The argument for the failure of the local algorithm will be very similar as for the previous example, so we focus here on the success of the HES hierarchy, especially on the stepwise upper bound on the extended Hamiltonian.
    Let the extension be 
\[ \tilde H(\sigma) = \alpha_4\norm{\sigma}_2^4H_4(\Pi\sigma) + \alpha_6\norm{\sigma}_2^2H_{6}(\Pi\sigma, \Pi^{\perp}\sigma) + \alpha_8 H_8(\Pi^{\perp}\sigma)\,.\]
Then its Hessian is 
\begin{align*} \nabla^2\tilde H(\sigma) ={}& \alpha_4\left(\norm{\sigma}_2^4\nabla^2H_4(\Pi\sigma)+12\norm{\sigma}_2^2H_4(\Pi\sigma)\Id\right)
\\&+ 2\alpha_6H_{6}(\Pi\sigma, \Pi^{\perp}\sigma)\Id
\\&+ \alpha_6\norm{\sigma}_2^2\left(2\nabla_1^2H_{6}(\Pi\sigma, \Pi^{\perp}\sigma) + 12\nabla_2^2H_{6}(\Pi\sigma, \Pi^{\perp}\sigma) + 16\nabla_1\nabla_2H_{6}(\Pi\sigma, \Pi^{\perp}\sigma) \right)
\\&+ \alpha_8 \nabla^2H_8(\Pi^{\perp}\sigma)
+ \text{[rank-1 terms]}
\,,\end{align*}
where $\nabla_1$ refers to the gradient in the first argument of a multi-argument function and similarly for $\nabla_2$, and the rank-1 terms make negligible contribution over an HES distribution.

Now we fit this into \Cref{ass:free-prob-sos-proof}, characterizing each of these terms by their norm and whether they are part of $A$, $B$, $C$, or constant with respect to $v_i$.
We halve the contributions to $B$ due to the fact that $B$ needs to appear twice.

\begin{tabular}{l|c|c}
Term & Norm & Type
\\\hline
     $\alpha_4\norm{\sigma}_2^4\nabla^2H_4(\Pi\sigma)$
     & $\sqrt{6}\,\alpha_4\norm{\sigma}_2^4\norm{\Pi\sigma}_2^2$ & $A$
     \\
     $12\alpha_4\norm{\sigma}_2^2H_4(\Pi\sigma)\Id$
     & $12\alpha_4\norm{\sigma}_2^2H_4(\Pi\sigma)$ & constant
     \\
     $2\alpha_6H_{6}(\Pi\sigma, \Pi^{\perp}\sigma)\Id$
     & $2\alpha_6H_{6}(\Pi\sigma, \Pi^{\perp}\sigma)$ & constant
     \\
     $2\alpha_6\norm{\sigma}_2^2\nabla_1^2H_{6}(\Pi\sigma, \Pi^{\perp}\sigma)$
     & $\sqrt{2}\alpha_6\norm{\sigma}_2^2\norm{\Pi^{\perp}\sigma}_2^4$ & $A$
     \\
     $12\alpha_6\norm{\sigma}_2^2  \nabla_2^2H_{6}(\Pi\sigma, \Pi^{\perp}\sigma)$
     & $\sqrt{12}\alpha_6\norm{\sigma}_2^2\norm{\Pi\sigma}_2^2\norm{\Pi^{\perp}\sigma}_2^2$ & $C$
     \\
     $16\alpha_6\norm{\sigma}_2^2 \nabla_1\nabla_2H_{6}(\Pi\sigma, \Pi^{\perp}\sigma)$
     & $\sqrt{8}\alpha_6\norm{\sigma}_2^2\norm{\Pi\sigma}_2^1\norm{\Pi^{\perp}\sigma}_2^3$ & $B$
     \\
     $\alpha_8 \nabla^2H_8(\Pi^{\perp}\sigma)$
     & $\sqrt{28}\norm{\Pi^{\perp}\sigma}_2^6$ & $C$
     \\
\end{tabular}

At this point, we invoke \Cref{ass:free-prob-sos-proof} and unfurl it over $i$ steps to get a formula for the total expected value of the Hamiltonian $H(\sigma_i) \approx \sum_{j \in [i]} v_j^{\sT}\nabla^2 H(\sigma_{j-1})v_j$ at the end of $i$ steps, as a function of $\norm{\Pi v_j}_2^2$ for each $j \in [i]$.

Showing that this formula is maximized when $\norm{\Pi v_j}_2^2 = 1/2$ for all $j$ will involve a ``payment plan'' argument where, while the value gain from the $H_8$ term is immediate, the value gain from the $H_4$ and $H_6$ terms are spread out over all future steps via the $\Id$ terms in the Hessian.
The task will be to show that the extra value from the ``payment plans'' when taking $\norm{\Pi v_j}_2^2 > 1/2$ fails to offset the future losses to the $H_8$ term, since the $H_8$ term scales with a relatively high power of $\norm{\Pi^{\perp}\sigma}_2^2$, and similarly but conversely when taking $\norm{\Pi v_j}_2^2 < 1/2$.
\end{proofsketch}

\underline{\textbf{Note:}} The parameters in the examples above (the subspaces projected to) are carefully balanced (for instance, $\{\Pi, \Pi^{\perp}\}$ are chosen so that $\Tr[\Pi] = \Tr[\Pi^{\perp}]$) to make the argument easier, but we expect that the argument generalizes to unbalanced choices as well, in part by using the $R$-transform from free probability theory~\cite{speicher2019lecture} to more precisely calculate spectral distributions of sums of matrices. Given this flexibility of choice one has to ``fool'' the local algorithm in ways that the SoS HES relaxation seems to be robust to, it leads to a natural conjecture - That the SoS HES relaxation can ``peel'' out correlations between the entries (when drawn from some sub-gaussian distribution), provided these correlations are arbitrarily well-approximated by a low-degree polynomial. This conjecture is formally stated in~\Cref{conj:distributional-robustness}.

\section{Discussion}\label{sec:discussion} 

\subsection{Open problems}
\subsubsection{Hardness equivalence between LDPs and \texorpdfstring{$O(1)$}{O(1)}-degree SoS?}\label{sec:discussion-hardness}
In the introduction, we described the desire in the literature to relate various meta-algorithmic families, such as LDPs and SoS hierarchies. To demonstrate two-way reductions between these families for the family of random optimization problems of interest, it is not enough to merely show that the (for instance) SoS HES hierarchy can match the various LDP algorithms in use on the family of problems, one must also demonstrate that any SoS hierarchy cannot do \emph{better} in terms of outputting better solutions (under the utilization of some reasonable rounding schemes).

Partly, the ``hardness'' half of this task would involve demonstrating that, under the presence of various so-called overlap-gap properties~\cite{chen2019suboptimality, gamarnik2020low, chou2022limitations, huang2021tight, jones2022random} (OGPs), there exist \emph{no} $O(1)$-degree SoS relaxations that perform above a certain algorithmic threshold~\cite[Eq. (1.9) \& Eq. (2.5)]{huang2021tight}\cite[Corollary 6.9]{jones2022random}. Even more ambitiously, one would perhaps attempt to rule out the possibility of $\poly(n)$~sized semi-definite programs (SDPs) in the presence of OGPs altogether.

However, as it stands, this seems to be a rather daunting technical challenge and a good first step in this direction would be to demonstrate the overlap-concentration of (certain families of) low-degree polynomial algorithms~\cite[Conjecture 7.1]{jones2022random} on mean-field spin glasses and sparse random Max-CSPs.

\subsubsection{SoS algorithms for the SK model \& bounded-degree CSPs under fRSB?}\label{sec:hes-sos-cube}

There is a precise way to round the parameterized HES distribution, so that the SoS relaxation can be adapted to work when the solution domain is the discrete hypercube.~This requires certifying a bound on a Hessian with a low-degree potential function added, and defining the SoS program over a newly parameterized HES distribution~\cite{auffinger2015parisi, montanari2021optimization,panchenko2013parisi}.~Doing this would require generalizing the graph matrix framework~\cite{ahn2016graph} to the setting of random matrices that are conjugated by diagonal matrices with \emph{non} i.i.d.~entries. This is a significant technical challenge, but ideas from asymptotic freeness introduced in the free probability theory literature in conjunction with appropriately certifiable properties (such as unitary-invariance) should allow the current HES SoS relaxation template to work for spin glasses on the hypercube as well.

While an AMP \cite{montanari2021optimization} algorithm and Hessian ascent-type algorithm~\cite{jekel2024potential} are already known for the SK model, it would be interesting to know if a SoS algorithm using some parameterized family of certificates can be made to work in the setting of \emph{all} bounded degree CSPs. 
Adapting the current SoS program for the sparse setting, however, would require surmounting the formidable technical challenge of bounding the operator norms of \emph{extremely} sparse graph matrices.

\subsubsection{SoS certificate complexity of the Ghirlanda-Guerra identities?}\label{sec:sos-cert-gg-identities}
Note that the success of the HES SoS relaxation given in our work crucially relies on the information about the distribution of the eigenspectrum of the Hessian and its scaling equivalence matching the optimal value given by the Chen and Sen formula~(\Cref{lem:chen-sen}). More precisely, the precise optimal value is known to us by knowledge that is \emph{external} to the SoS solver. A natural question then is whether one can certify key mathematical properties of many spin glass models that are known to be the reason why the ground state configurations reside in an ultrametric space. As mentioned before, due to a landmark result of Panchenko~\cite{panchenko2010connection, panchenko2013parisi}, it is known that ultrametricity is implied by the so-called Ghirlanda-Guerra (GG) identities~\cite{ghirlanda1998general}. As a result, a first natural question along this direction is asking whether there even \emph{exist} SoS certificates to check whether an instance of certain spin glass models satisfies some low-degree weakening of the GG identities, written over \emph{some} parameterizable step distribution.
  
\begin{openquestion}[SoS Certificate for ``weakened'' Ghirlanda-Guerra Identities of Low-Degree Moments of the Gibbs Distribution]\label[Open Question]{open-prob:sos-cert-gg-identities}
    Does there exist a SoS relaxation with degree parametrized by the support of the overlap CDF $\mu$ (and $n$) as defined in the Parisi functional $\cP_\beta(\mu)$ that can certify whether an instance $H_n(.)$ of a mean-field spin glass (approximately) satisfies the GG identities on the \emph{low degree} moments of some parameterizable distribution (to proxy those of the Gibbs distribution)? Are there other weakenings of the GG identities that a SoS relaxation can certify?
\end{openquestion}

Note that it is unclear how to even state this problem in the sum-of-squares proof system, let alone find a certificate for it, as the Ghirlanda-Guerra identities are a property satisfied by the Gibbs measure of (typical) instances of spin glass models. However, if it is possible to understand the low-degree cumulants of the Gibbs precisely, one may state a relaxed version of this property over the low-degree moments of the Gibbs distribution. An answer to the question above sheds light on the power of the SoS proof system to ``recognize'' when solutions to a problem instance satisfy a certain geometric structure (albeit under a different relaxation), which in this case is an ultrametric structure. In a regime (if any) in which certificates of low-complexity exist, one can conceivably show the existence of a SoS algorithm that relies \emph{even less} on external knowledge about the problem structure.

\subsection{Future directions}\label{sec:future}

\subsubsection{Equivalences between algorithm classes}
By subsuming local iterative algorithms on problems that previously defied a SoS approach, the HES hierarchy could make progress toward unification of meta-algorithmic families such as local algorithms~\cite{chen2019suboptimality, chou2022limitations} and the low-degree polynomial family (which include AMP and certain spectral algorithms)~\cite{hopkins2018statistical, gamarnik2020low}. This advances a goal of Schramm to classify these approaches under an \emph{algorithmic universality}~\cite{youtube}.

\subsubsection{Other spin glass ansatze}
The pseudo-distribution structure of the HES hierarchy allows for statistical mixtures of HES distributions that capture near-extrema (up to a particular precision) when solutions exhibit a particular overlap structure. This is compatible with the structure described by the Mezard-Parisi ansatz~\cite{mezard2001bethe, panchenko2015hierarchical, panchenko2016structure, coja2019spin}, involving non-constructive statistical mixtures of hierarchically-organized solutions.
A prototypical example of a problem whose solutions are described by the Mezard-Parisi ansatz is the  $\mathsf{MAX}$-$\mathsf{CUT}$ of random $d$-regular graphs for \emph{fixed} $d \in \N$.
The structural similarity suggests that the HES hierarchy could optimize these $\mathsf{MAX}$-$\mathsf{CUT}$ instances in the fRSB regime~\cite{concetti2018full, concetti2019full}.

\subsubsection{Bridging the HES and Lasserre Hierarchies}

We observe that \Cref{lem:nuc-norm-moments-informal} suggests that the HES approach might generally work when there are no large spectral weights in the higher-degree components of the input polynomial.
The existence of these heavy spectral directions would suggest that techniques based on reweighing or lifting into higher-dimensional spaces would be non-trivial~\cite{hopkins2016fast,ma2016polynomial}.
It would be interesting to bridge the gap between these two regimes.

\newpage

\section*{Acknowledgements}

\ifnum\blindmode=0
We thank Alexandra Kolla, Boaz Barak \& Peter J. Love for helpful comments and feedback.
We are grateful to Sam Hopkins for pointing out a simplification of the notation for pseudo-distributions over distributions and for various helpful discussions. We are also indebted to David Steurer for extended insightful discussions while hosting us at ETH Zurich during Summer 2022. We would like to thank Chris Jones for helpful comments on graph matrix norm bounds and certifications, as well as suggesting various improvements in exposition. We would also like to thank Pravesh Kothari for pointing us to an error in an earlier version of this argument and for providing helpful feedback on a prior draft. We are grateful to the anonymous reviewers for suggestions on how to improve presentation and include a discussion on input-oblivious rounding, and we'd like to specially thank Kunal Marwaha for giving detailed feedback and suggestions on the introduction. 

Much of this work began after a long series of discussions with Antares Chen in Fall 2021 and with Tommaso d'Orsi in Spring and Summer 2022, seeding the development of these ideas, for which we are extremely grateful. 

JS would like to thank Luca Trevisan for instigating his introduction to spin glass theory and for discussions on SoS and local-improvement-based methods. JSS would like to thank Peter J. Love for excellent mentorship and sharing thoughts on ``algorithmic  equivalence" which (in part) motivated this paper. JSS partly worked on this project as a visiting scholar at Bocconi University, hosted by Luca Trevisan during Summer 2022.

Lastly, JS \& JSS would like to extend many thanks to Prof. SueYeon Chung for graciously hosting them at the Flatiron Institute in New York City over various periods in Summer-2023 and Fall-2023 where the final stages of this paper were completed.

JSS would like to dedicate this work and the companion work~\cite{ss24rounding} to his late grandfather, Dr. Rashpal Malhotra.
\fi

\bibliography{ms.bib}
\bibliographystyle{alpha-beta}

\newpage

\appendix
\section{Univariate Matrix Bernstein Polynomials}
\label{sec:bernsteiny-stuff}
In this section, we define the \emph{matrix} bernstein basis polynomials over symmetric matrices with eigenvalues in a compact interval $I \subset \R$. The matrix Bernstein polynomial $B^f_d(M)$ of degree $d$ approximating a continuous function $f: I \to \R$ will have the same desirable properties as in the scalar regime when the underlying inputs are symmetric matrices, including a low-degree SoS proof of positivity over the Loewner order~(\Cref{app:sos-matrix-weierstrass}).

\begin{definition}[Matrix Bernstein Basis Polynomials]\label[Definition]{def:matrix-bernstein-basis}
    Given a degree $d \ge 0$, the $d+1$ matrix Bernstein basis polynomials on the interval $I = [a, b] \subset \R$ are defined as,
    \begin{equation}
        B_{d, i} = \binom{d}{i}\frac{(M - a\Id)^i(b\Id - M)^{d-i}}{(b-a)^d}\, ,
    \end{equation}
    where $M$ is a symmetric matrix such that $a\Id \preceq M \preceq b\Id$.
\end{definition}

Given a step-function $S_\alpha: I = [a,b] \to \R$ with the step at $\alpha \in (a,b)$, we define the linear ramp function $L_{\alpha, \gamma}: [a,b] \to \R$ as a continuous approximation of the step-function as follows,
\begin{equation}
    L_{\alpha, \gamma}(x) = \begin{cases}
        0 & x \in [a, \alpha - \gamma]\, , \\
        x & x \in (\alpha - \gamma, \alpha + \gamma)\, , \\
        1 & x \in [\alpha + \gamma,b]\, .
    \end{cases}
\end{equation}
Typically, $\gamma$ will be chosen to be significantly smaller than $\alpha$. For the Hessian $\nabla^2_\perp H(\cdot)$, the linear ramp function approximates its top $\delta n$-dimensional eigenspace with the following choice: $a = -2\nu''(\cdot)^{1/2}$, $b = 2\nu''(\cdot)^{1/2}$, $\alpha = 2\nu''(\cdot)^{1/2} - 2\epsilon/3$ and $\gamma = \epsilon/3$. The choice of $\epsilon$ and $\delta$ correspond to~\cite[Lemma 3]{subag2021following}.

\begin{proposition}[Matrix Bernstein Approximation of $L_{\alpha, \gamma}$]\label[Proposition]{prop:matrix-bernstein-linear-ramp}
    Let $I = [a,b] \subset \R$. Then, given a linear-ramp function $L_{\alpha, \gamma}: I \to \R$, the Bernstein polynomial $B^{L_{\alpha, \gamma}}_{d}(\cdot)$ defined over symmetric matrices $M$ with $a\Id \preceq M \preceq b\Id$ that approximates it is given as,
    \begin{equation}\label{eq:matrix-bernstein-approx}
        B^{L_{\alpha, \gamma}}_d(M) := \sum_{j=0}^d L_{\alpha, \gamma}\left(a + \frac{j(b-a)}{d}\right)\binom{d}{i}\frac{(M - a\Id)^j(b\Id - M)^{d-j}}{(b-a)^d}\, .
    \end{equation}
    Furthermore, $B^{L_{\alpha, \gamma}}_d$ has the following properties:
    \begin{description}
        \item[1.~Symmetric Positive Semi-Definiteness:] $B^{L_{\alpha, \gamma}}_d(M) \succeq 0$.
        \item[2.~Operator Norm Bound:] $B^{L_{\alpha, \gamma}}_d(M) \preceq \Id$.
        \item[3.~Close in Frobenius Norm:] Let $\Pi_{\alpha + \gamma, b}$ be the orthogonal projector to the subspace of $M$ with eigenvalues in $[\alpha + \gamma, b]$ and $p = \frac{B^{L_\alpha, \gamma}_d(M)}{\delta' n}$. Then, for any $d \ge \frac{8\nu''(q)^{1/2}}{\epsilon^3}$,
        \begin{equation*}
            \norm{p - \frac{1}{\delta' n}\Pi_{\alpha + \gamma, b}}_{F} \le \sqrt{\frac{2\delta'''}{\delta'^2}}\frac{1}{\sqrt{n}}\, ,
        \end{equation*}
        where,
        \begin{equation*}\label{eq:delta-eps-max-frobenius}
            \delta''' = \max\left\{\eps^2, \frac{3}{4}(\delta - \delta'')\right\}\, .
        \end{equation*}
    \end{description}
\end{proposition}
\begin{proof}
    We prove each item in sequence.
    \begin{description}
        \item[1.~Symmetric Positive Semi-Definiteness:] This follows immediately from the following observations:
        \begin{itemize}
            \item The value of the scaled linear ramp function $\binom{d}{i}\frac{L_{\alpha, \gamma}(a + j(b-a)/d)}{(b-a)^d} \ge 0$ for every $j \in [d]$ since the function is non-negative and $b > a$.
            \item By invoking~\Cref{lem:prod-psd-matrices-sos} with $p_1(M - a\Id) = (M-a\Id)^j$ and $p_2(b\Id - M) = (b\Id - M)^{d-j}$, we have that $p_1(\cdot)p_2(\cdot)$ is a symmetric and positive definite matrix since $a\Id \preceq M \preceq b\Id$.
            \item A positive linear sum of symmetric and positive semi-definite matrices is a symmetric positive semi-definite matrix.
        \end{itemize}
        \item[2.~Operator Norm Bound:] Diagonalizing the symmetric matrix immediately implies that,
        \begin{equation}
            B^{L_{\alpha, \gamma}}_d(M) = O\begin{bmatrix} B^{L_{\alpha, \gamma}}_d(\lambda_1) & & \\ & \ddots & \\ & & B^{L_{\alpha, \gamma}}_d(\lambda_n)\end{bmatrix}O^\sT\, ,
        \end{equation}
        where the Bernstein polynomials applied to the eigenvalue are as defined in~\Cref{eq:Bernstein-polynomial}. Then, by~\Cref{def:scalar-bernstein},
        \begin{align*}
            B^{L_{\alpha, \gamma}}_d(\lambda_i) &= \sum_{j=0}^d L_{\alpha, \gamma}(a + j(b-a)/d)\binom{d}{j}\left(\frac{\lambda_i - a}{b-a}\right)^j\left(\frac{b - \lambda_i}{b-a}\right)^{d-j}\, , \\
            &\le \sum_{j=0}^d \binom{d}{j}\left(\frac{\lambda_i - a}{b-a}\right)^j\left(\frac{b - \lambda_i}{b-a}\right)^{d-j} \\
            &= \sum_{j=0}^d \binom{d}{j}y^j(1-y)^{d-j} = (y + (1-y))^d \\
            &= 1\, ,
        \end{align*}
        where we used the facts that $L_{\alpha, \gamma}(a + j(b-a)/d) \le 1$ for every $\lambda_i \in [a,b]$ and the substitution $y = \frac{\lambda_i - a}{b-a} \in [0,1]$. The last equality follows by the binomial formula. This immediately implies that,
        \[
            B^{L_{\alpha, \gamma}}_d(M) \le O\Id O^\sT = \Id\, .
        \]
        \item[3.~Close in Frobenius Norm:]
        We apply~\cite[Lemma 3]{subag2021following} repetitively to obtain $0 < \delta'' < \delta' < \delta < 1$ corresponding to the choices of $0 < \eps/3 < 2\eps/3 < \eps < 1$ respectively. Then, applying~\Cref{thm:bernstein-approx} for $f = L_{\alpha, \gamma}$ with the choice of parameters $a, b, \gamma$ stated above, we have,
        \begingroup
        \allowdisplaybreaks
        \begin{align*}
            \norm{p - \frac{1}{\delta' n}\Pi_{\alpha + \gamma, b}}_{F} &= \frac{1}{\delta' n} \sqrt{\sum_{i=1}^{\delta'' n}\eps^2 + \sum_{i = \delta'' n  + 1}^{\delta' n}(1/2 + \eps)^2 + \sum_{i=\delta' n + 1}^{\delta n}(1/2 + \eps)^2 + \sum_{i = \delta n}^n \eps^2} \\
            &= \frac{1}{\delta' n}\left(\sqrt{(\delta'' n)\eps^2 + (\delta' - \delta'')n\left(\frac{1}{2} + \eps\right)^2 + (\delta - \delta')n\left(\frac{1}{2} + \eps\right)^2 + (1-\delta)n\eps^2}\right) \\
            &= \frac{1}{\delta' n}\left(\sqrt{(1-\delta - \delta'')\eps^2  n + (\delta - \delta'')\left(\frac{1}{2} + \eps\right)^2n}\right) \\
            &= \frac{1}{\delta' n}\left(\sqrt{(1-2\delta'')\eps^2 n + \left(\frac{\delta - \delta''}{4}\right)n + \left(\frac{\delta - \delta''}{2}\right)\eps n}\right) \\
            &\le \frac{1}{\delta' n}\left(\sqrt{(1 - 2\delta'')\eps^2 n + \frac{3n}{4}(\delta - \delta'')}\right) \\
            &\le \frac{1}{\delta' n}\left(\sqrt{\eps^2 n + \frac{3n}{4}(\delta - \delta'')}\right) \le \sqrt{\frac{2\delta'''}{\delta'^2}}\frac{1}{\sqrt{n}}\, ,
        \end{align*}
        \endgroup
        where we used the definition of the Frobenius norm, in conjunction with a rewrite of the projector diagonalized into the eigenbasis of $B^{L_{\alpha, \gamma}}_d(M)$ and the fact that $\eps < 1$.
    \end{description}
\end{proof}

\section{Analytic Sparsity Proofs}\label{app:sos-matrix-weierstrass}

\subsection{Normalizing the Bernstein polynomial under the semi-circle measure}

\begin{lemma}[Average Area of Bernstein Polynomial under $\mu_{sc}$]\label[Lemma]{lem:bernstein-normalize-sc}
    For the semi-circle measure $d\mu_{sc}$ supported on $[-1,1]$ and $B^{L^\alpha_\gamma}(x)$ as defined in~\Cref{eq:matrix-bernstein-approx} with a choice of $L^\alpha_\gamma(x)$ given in~\Cref{eq:l_alpha_gamma}, the following holds,
    \[
        \E_{x\sim\mu_{sc}}\left[B^{L^\alpha_\gamma}(x)\right] \le 1 + \eps - \mathsf{CDF}(1-\phi)\, ,
    \]
    and,
    \[
        \E_{x\sim\mu_{sc}}\left[B^{L^\alpha_\gamma}(x)\right] \ge (1+2\eps)\frac{4\sqrt{2}}{3\pi}\phi^{3/2} + O(\phi^{5/2}) + o(\phi^{5/2}) + \eps\, .
    \]
\end{lemma}
\begin{proof}
    We compute the area of the semi-circle where the Bernstein polynomial ``lights up'' to normalize the correlation appropriately. To avoid unnecessary complications, we will use the uniform lower bound of $0$ and uniform upper bound of $\phi^2$ on the terms, leading us to a lower and upper bound of the correct \emph{leading} order.

    \paragraph{The upper bound} We being by computing the upper bound by using the uniform upper bound of $\phi^2$ on the ramp part,
    \begingroup
    \allowdisplaybreaks
    \begin{align*}
        &\E_{x \sim \mu_{sc}}\left[B^{L^{\alpha}_\gamma}(x)\right] = \int_{-1}^1 B^{L^\alpha_\gamma}(x)\frac{2}{\pi}\sqrt{1-x^2}dx \le \int_{-1}^1 \left(L^\alpha_\gamma(x) + \eps \right)\frac{2}{\pi}\sqrt{1-x^2}dx \\
        &= \int_{-1}^1 L^\alpha_\gamma(x)d\mu_{sc}(x) + \eps\int_{-1}^1d\mu_{sc}(x) \\
        &= \int_{-1}^{1-\phi}0\cdot d\mu_{sc}(x) + \frac{1}{\phi^2}\int_{1-\phi}^{1-\phi +\phi^2}\left(x - (1-\phi)\right)d\mu_{sc}(x) + \int_{1-\phi+\phi^2}^{1}d\mu_{sc}(x) + \eps \\
        &\le 0 + \frac{1}{\phi^2}\phi^2\int_{1-\phi}^{1-\phi+\phi^2}d\mu_{sc}(x) + \int_{1-\phi+\phi^2}^1 d\mu_{sc}(x) + \eps \\
        &=  \int_{1-\phi}^1d\mu_{sc}(x) + \eps = (1+\eps) - \mathsf{CDF}(1-\phi)\, .
    \end{align*}
    \endgroup
    To evaluate the above, we take the Taylor expansion of $\mathsf{CDF}(1-\phi)$ around $1$ (using the fact that $\phi(\eps) << 1$) and have,
    \begingroup
    \allowdisplaybreaks
    \begin{align*}
        1 - \mathsf{CDF}(1-\phi) &= 1 - \left(\frac{1}{2} + \frac{(1-\phi)\sqrt{1 - (1-\phi)^2}}{\pi} + \frac{\arcsin{1-\phi}}{\pi}\right) \\
        &\overset{\text{Taylor expansion}}{=} 1 - \left(1 - \frac{4\sqrt{2}}{3\pi}\phi^{3/2} + O(\phi^{5/2}) + o(\phi^{5/2})\right) \\
        &= \frac{4\sqrt{2}}{3\pi}\phi^{3/2} + O(\phi^{5/2}) + o(\phi^{5/2})\, ,
    \end{align*}
    \endgroup
    where $o(\phi^{5/2})$ denotes terms of higher power than $5/2$ and utilizes the fact that $\phi(\eps) << 1$.

    \paragraph{The lower bound} The lower bound follows a similar computation, using the uniform lower bound of $0$ on the ramp part,
    \begingroup
    \allowdisplaybreaks
    \begin{align*}
        &\E_{x \sim \mu_{sc}}\left[B^{L^{\alpha}_\gamma}(x)\right] = \int_{-1}^1 B^{L^\alpha_\gamma}(x)\frac{2}{\pi}\sqrt{1-x^2}dx \\
        &\ge \int_{-1}^{1-\phi}\left(L^\alpha_\gamma + \eps\right)\frac{2}{\pi}\sqrt{1-x^2}dx + \int_{1-\phi}^1 \left(L^\alpha_\gamma - \eps\right)\frac{2}{\pi}dx \\
        &= \int_{1-\phi}^{1-\phi + \phi^2}\frac{1}{\phi^2}\left(x - (1-\phi)\right)d\mu_{sc}(x) + \int_{1-\phi+\phi^2}^1 d\mu_{sc} + \eps\left(\int_{-1}^{1-\phi}d\mu_{sc}(x) - \int_{1-\phi}^1d\mu_{sc}(x)\right) \\
        &\ge 1 - \mathsf{CDF}(1- \phi+\phi^2) + \eps\left(\mathsf{CDF}(1-\phi) - 0 - 1 + \mathsf{CDF}(1-\phi)\right) \\
        &= 1 - \mathsf{CDF}(1-\phi+\phi^2) + 2\eps(\mathsf{CDF}(1-\phi) - \frac{1}{2}) \\
        &= 1 - \left(\frac{1}{2} + \frac{(1-\phi(1-\phi))\sqrt{1 - (1-\phi(1-\phi))^2}}{\pi} + \frac{\arcsin{1-\phi(1-\phi)}}{\pi}\right) \\
        &\qquad\qquad\qquad\qquad + 2\eps\left(\frac{1}{2}\frac{4\sqrt{2}}{3\pi}\phi^{3/2} + O(\phi^{5/2}) + o(\phi^{5/2})\right) \\
        &= \frac{4\sqrt{2}}{3\pi}\phi^{3/2}(1-\phi)^{3/2} + O(\phi^{5/2}(1-\phi)^{5/2}) + o(\phi^{5/2}(1-\phi^{5/2}) + 2\eps\frac{4\sqrt{2}}{3\pi}\phi^{3/2} + \eps \\
        &\ge (1+2\eps)\frac{4\sqrt{2}}{3\pi}\phi^{3/2} + O(\phi^{5/2}) + o(\phi^{5/2}) + \eps\, ,
    \end{align*}
    \endgroup
    where fourth equality uses a similar taylor expansion as the one used in computing the upper bound.
\end{proof}

\subsection{Auxiliary lemmata}
We state and prove some helper lemmata, used to reason about matrix Bernstein polynomials and low-degree SoS proofs.

\begin{lemma}[Products of univariate non-negative symmetric matrix polynomials]\label[Lemma]{lem:prod-psd-matrices-sos}
    Given univariate matrix polynomials of symmetric matrices $p_1, p_2: M_n(\R) \to \R$ with $\mathsf{deg}(p_1), \mathsf{deg}(p_2) \le d$ and $0 \preceq p_1(M)$, $0 \preceq p_2(M)$ for every $M \in M_n(\R)$, the matrix polynomial $p_1(M)p_2(M)$ is symmetric and non-negative in the Loewner order with $\mathsf{deg}(p_1p_2) \le 2d$.
\end{lemma}
\begin{proof}
    By~\Cref{prop:non-zero-eig-psd}, $p_1(M)p_2(M)$ has non-negative non-zero eigenvalues. It remains to check that $p_1(M)p_2(M)$ is symmetric, as~\Cref{prop:non-zero-eig-psd} then also implies that $p_1(M)p_2(M)$ is a symmetric PSD matrix. This follows by expressing the matrix polynomials explicitly over the monomial basis with their real coefficients $\{a_i\}_{i=0}^d$ and $\{b_i\}_{i=0}^d$.
    \begingroup
    \allowdisplaybreaks
    \begin{align*}
        p_1(M)p_2(M) &= \left(\sum_{i=0}^d a_i M^i\right)\left(\sum_{j=0}^d b_j M^j\right) \\
        &= \sum_{i, j = 0}^d a_ib_j M^i M^j = \sum_{i=0}^d\sum_{j=0}^d a_i b_j M^{i+j} \\
        &= \sum_{i=0}^d\sum_{j=0}^d a_ib_j (M^\sT)^i (M^\sT)^j = \sum_{i=0}^d\sum_{j=0}^d a_i b_j (M^i)^\sT (M^j)^\sT = \sum_{i=0}^d\sum_{j=0}^d a_ib_j(M^j M^i)^\sT \\
        &= \sum_{i=0}^d\sum_{j=0}^d a_i b_j (M^i M^j)^\sT = \left(\sum_{i=0}^d\sum_{j=0}^d a_i b_j M^i M^j\right)^\sT = \left(p_1(M)p_2(M)\right)^\sT\, .
    \end{align*}
    \endgroup
    Note that $\mathsf{deg}(p_1p_2) \le 2d$ also follows from the reasoning above.
\end{proof}

Stated below is an elementary proposition which shows that the products of matrices that commute preserves the Lowener order.
\begin{proposition}[Non-Zero Eigenspectrum Equivalence of conjugated PSD matrices]\label[Proposition]{prop:non-zero-eig-psd}
    Let $A$ and $B$ by symmetric PSD matrices in $M_n(\R)$. Then,
    \begin{align}
        \lambda_i(AB) = \lambda_i(A^{1/2}BA^{1/2})\, ,
    \end{align}
    for all non-zero eigenvalues $\lambda_i(\cdot)$.
    Furthermore, if $AB = (AB)^\sT$, then $AB$ is a symmetric PSD matrix.
\end{proposition}
\begin{proof}
    Since $A$ is PSD,
    \begin{align*}
        AB &= A^{1/2}A^{1/2}B\, .
    \end{align*}
    For every non-zero eigenvector $u_i$ of $A^{1/2}BA^{1/2}$ with eigenvalue $\lambda_i$, the vector $v_i = A^{1/2} u_i$ is an eigenvector with eigenvalue $u_i$ as,
    \begingroup
    \allowdisplaybreaks
    \begin{align*}
        AB(A^{1/2}u_i) = ABv_i &= A^{1/2}(A^{1/2}B)v_i = A^{1/2}(A^{1/2}BA^{1/2})u_i \\
        &= A^{1/2} \lambda_i u_i = \lambda_i (A^{1/2}u_i)\, .
    \end{align*}
    \endgroup
    For every $w \in \mathsf{ker}(A^{1/2})$ that is also a non-zero eigenvector of $AB$, we can choose $w' = A^{1/2}Bw$ as an eigenvector of $A^{1/2}B A^{1/2}$ with the same eigenvalue. If there is some vector $w \in \mathsf{ker}(A^{1/2}) \cap \mathsf{ker}(B)$, then it is a zero eigenvector of both $AB$ and $A^{1/2}A^{1/2}B$. \\
    It remains to show that $A^{1/2}BA^{1/2}$ is PSD. This follows by the elementary observation that, for every $v \in \R^n$,
    \begin{align*}
        v^\sT(A^{1/2}BA^{1/2})v = \norm{B^{1/2}w}^2_2 \ge 0\, ,
    \end{align*}
    where $w = A^{1/2}v$. This implies that the non-zero eigenvalues of $AB$ are non-negative. \\
    If, furthermore,
    \begin{align*}
        AB = (AB)^\sT = B^\sT A^\sT = BA\, ,
    \end{align*}
    the spectral theorem for symmetric matrices immediately implies that $AB$ is a symmetric PSD matrix.
\end{proof}

\begin{lemma}[Asymptotic evaluation of Catalan Numbers]\label[Lemma]{lem:catalan-stirling}
    For all $q > 0$, the Catalan number $C_q$ can be equivalently expressed as,
    \begin{equation}
        C_q = \left(1 + O\left(\frac{1}{q}\right)\right)\left(\frac{4^q}{\sqrt{\pi}q^{3/2}}\right)\, .
    \end{equation}
\end{lemma}
\begin{proof}
    The proof follows by a simple application of the Stirling bounds for the factorial function on the definition of the Catalan numbers.
    \begingroup
    \allowdisplaybreaks
    \begin{align*}
        C_q &:= \binom{2q}{q}\frac{1}{q+1} = \frac{(2q)!}{(q!)^2 (q+1)} \\
        &= \frac{\sqrt{4\pi q}\left(\frac{2q}{e}\right)^{2q}\left(1 + \Theta\left(\frac{1}{q}\right)\right)}{(q+1)\left(\sqrt{2\pi q}\left(\frac{q}{e}\right)^q\left(1 + \Theta\left(\frac{1}{q}\right)\right)\right)^2} \\
        &= \frac{2^{2q}\left(1 + \Theta\left(\frac{1}{q}\right)\right)}{\sqrt{\pi}\sqrt{q}(q + 1)\left(1 + \Theta\left(\frac{1}{q}\right)\right)^2} \\
        &= \frac{4^q}{\sqrt{\pi}q^{3/2}}\frac{1 + \Theta\left(\frac{1}{q}\right)}{\left(1 + \Theta\left(\frac{1}{q}\right)\right)^2\left(1+\frac{1}{q}\right)} \\
        &= \left(1 + O\left(\frac{1}{q}\right)\right)\left(\frac{4^q}{\sqrt{\pi}q^{3/2}}\right)\, .
    \end{align*}
    \endgroup
\end{proof}

\section{Sum-of-Squares Proofs for Commutative Analytic Inequalities}\label{sec:sos-commutative-inequalities}

\subsection{Pseudo-Jensen inequalities}

\textbf{Note}: Each of the lemmata proved below about pseudo-distributions can \emph{also} be used to give a SoS proof for polynomials of moments that are of even degree, provided they are valid elements of the polynomial grammar~(\Cref{sec:sos-poly-basis}), meaning that these proofs can also be point-wise over the atoms of the HES SoS hierarchy.

\begin{lemma}[Pseudo-Jensen Inequality for arbitrary functions]\label[Lemma]{lem:matrix-pseudo-jensen}
    Given a real polynomial $f: \R^n \to \R$ with degree $\le d$ and a pseudo-distribution $\cD$ of degree $r > 2d$ over $\R^n$,
    the following holds,
    \begin{equation}\label{eq:matrix-pseudo-jensen}
        \pE_{\cD}\left[\left(f(x)\right)^{2}\right] \ge \left(\pE_{\cD}[f(x)]\right)^{2}\,.
    \end{equation}
\end{lemma}
\begin{proof}
    We prove the statement using a ``variance'' like computation.
    \begingroup
    \allowdisplaybreaks
    \begin{align*}
        0 \le \pE_\cD\left[\left(f(x) - \pE_\cD[f(x)]\right)^2\right] &= \pE_\cD\left[f(x)^2 -2f(x)\pE_\cD[f(x)] + \left(\pE_\cD[f(x)]\right)^2\right] \\
        &= \pE_\cD[f(x)^2] - 2\pE_{\cD}\left[f(x)\pE_{\cD}f(x)\right] + \left(\pE_\cD[f(x)]\right)^2 \\
        &= \pE_\cD[f(x)^2] - 2\pE_\cD[f(x)]\pE_\cD[f(x)] + \left(\pE_\cD[f(x)]\right)^2 = \pE_\cD[f(x)^2] - \left(\pE_\cD[f(x)]\right)^2\, ,
    \end{align*}
    \endgroup
    where the first inequality uses the definition of a pseudo-distribution and the second-last equality uses the fact that the pseudo-expectation operator outputs a real number and is a linear operator.
\end{proof}
\paragraph{Remark} Extending the above proof to hold for even degree monomials is straightforward and follows by setting $f(x) = g(x)^p$ when attempting to prove pseudo-Jensen for $g(x)^{2p}$, thereby yielding $\left(\pE[g(x)^p]\right)^2 \le \pE[g(x)^{2p}]$.

\begin{lemma}[Pseudo-Jensen inequality for even powers]\label[Lemma]{lem:pseudo-jensen-even}
    Given a real polynomial $f: \R^n \to \R$ with degree $\le d$, an even number $2p > 0$ and a pseudo-distribution of degree $> 4dp$, the following holds,
    \[
        \pE_\cD\left[f(x)^{2p}\right] \ge \left(\pE_\cD\left[f(x)\right]\right)^{2p}\, .
    \]
\end{lemma}
\begin{proof}
     We prove this by induction, and use~\Cref{lem:reverse-cauchy-schwarz} in conjunction with a direct application of~\Cref{lem:matrix-pseudo-jensen} to prove the inductive step.

     \paragraph{Base case} The case when $p=1$ simply reduces to~\Cref{lem:matrix-pseudo-jensen}.

     \paragraph{Inductive step} We assume the statement is true up to $2 \le 2p-2 = 2(p-1)$. For the case of $2p \ge 4$, we obtain the following,
    \begingroup
    \allowdisplaybreaks
    \begin{align*}
        \pE_\cD\left[f(x)^{2p}\right] - \left(\pE_\cD\left[f(x)\right]\right)^{2p} &= \pE_\cD\left[f(x)^{2p}\right] - \pE_\cD\left[f(x)^2\right]\pE_\cD\left[f(x)\right]^{2p-2} + \pE_\cD\left[f(x)^2\right]\pE_\cD\left[f(x)\right]^{2p-2} - \left(\pE_\cD\left[f(x)\right]\right)^{2p} \\
        &= \underbrace{\pE_\cD\left[f(x)^2\left(f(x)^{2p-2} - \pE_\cD\left[f(x)\right]^{2p-2}\right)\right]}_{\mathsf{I}} + \underbrace{\pE_\cD\left[f(x)\right]^{2p-2}}_{=\,\pE_\cD\left[(f(x)^{p-1})^2\right] \ge 0}\underbrace{\left(\pE_\cD\left[f(x)^2\right] - \pE_\cD[f(x)]^2\right)}_{\ge 0\text{ by~\Cref{lem:matrix-pseudo-jensen}}}\, .
    \end{align*}
    \endgroup
    To evaluate $\mathsf{I}$ we will apply~\Cref{lem:reverse-cauchy-schwarz} in conjunction with the inductive hypothesis. First, note that,
    \begingroup
    \allowdisplaybreaks
    \begin{align*}
        \mathsf{I} &= \pE_\cD[f(x)^{2p}] - \pE_\cD\left[f(x)^2\right]\pE_\cD[f(x)]^{2p-2} \ge 0 \\
        &\iff \pE_\cD[f(x)^{2p}] \ge \pE_\cD\left[f(x)^2\right]\pE_\cD[f(x)]^{2p-2}\, .
    \end{align*}
    \endgroup
    Now,
    \[
        \pE_\cD[f(x)^{2p}] \ge_{\text{\Cref{lem:reverse-cauchy-schwarz} with } s=1 } \pE_\cD\left[f(x)^2\right]\pE_\cD[f(x)^{2p-2}] \ge_{\text{by I.H.}} \pE_\cD\left[f(x)^2\right]\pE_\cD\left[f(x)\right]^{2p-2}\, ,
    \]
    which concludes the proof.
\end{proof}

\begin{lemma}[Factored Even Powers]\label[Lemma]{lem:reverse-cauchy-schwarz}
    Given an even number $2p > 0$ and a number $s >0$, such that $1 \le 2s \le p$, the following holds,
    \[
        \pE_\cD\left[f(x)^{2p}\right] \ge \pE_\cD\left[f(x)^{2s}\right]\pE_\cD\left[f(x)^{2p-2s}\right]
    \]
\end{lemma}
\begin{proof}
    We prove this by induction. For the base case, we consider $p=1$.

    \paragraph{Base case} When $p=1$ and $2s=1$, the following
    \[
        \pE_\cD\left[f(x)^{2}\right] \ge \pE_\cD\left[f(x)\right]\pE_\cD\left[f(x)\right] = \pE_\cD\left[f(x)\right]^2\, ,
    \]
    holds by~\Cref{lem:matrix-pseudo-jensen}.

    \paragraph{Inductive step} We assume the statement is true up to $2p-2 = 2(p-1)$. For the case of $2p$, we obtain the following by generalizing the calculation in~\Cref{lem:matrix-pseudo-jensen},
    \begingroup
    \allowdisplaybreaks
    \begin{align*}
        0 \le &\pE_\cD\left[f(x)^{2p-4s}\left(f^{2s}(x)-\pE_\cD\left[f^{2s}(x)\right]\right)^2\right] \\
        &= \left(\pE_\cD\left[f(x)^{2p}\right] - \pE_\cD\left[f(x)^{2p-2s}\right]\pE_\cD\left[f(x)^{2s}\right]\right) - \pE_\cD\left[f(x)^{2s}\right]\left(\pE_\cD\left[f(x)^{2p-2s}\right] - \pE_\cD\left[f(x)^{2p-4s}\right]\pE_\cD\left[f(x)^{2s}\right]\right)\, .
    \end{align*}
    \endgroup
    Note that, since $s \ge 1$, we have that $2p-2s \le 2(p-1)$ which allows us to invoke the inductive hypothesis and conclude that,
    \[
        \pE_\cD\left[f(x)^{2p-2s}\right] - \pE_\cD\left[f(x)^{2p-4s}\right]\pE_\cD\left[f(x)^{2s}\right] \ge 0\, .
    \]
    Since $\mathsf{deg}(f(x)^{2s}) = \mathsf{deg}((f(x)^s)^2) \le 2dp$, we immediately have that $\pE_\cD[f(x)^{2s}] \ge 0$ and this implies,
    \[
        \pE_\cD\left[f(x)^{2p}\right] - \pE_\cD\left[f(x)^{2p-2s}\right]\pE_\cD\left[f(x)^{2s}\right] \ge \pE_\cD\left[f(x)^{2s}\right]\left(\pE_\cD\left[f(x)^{2p-2s}\right] - \pE_\cD\left[f(x)^{2p-4s}\right]\pE_\cD\left[f(x)^{2s}\right]\right) \ge 0\,.
    \]
\end{proof}

\subsection{\texorpdfstring{$\ell_1$}{l1}-to-\texorpdfstring{$\ell_t$}{lt} norm inequality}

Stated below is an elementary $1$-to-$t$-norm inequality that permits a very simple Sum-of-Squares proof of $O(t)$.
\begin{lemma}[$\ell_1$-to-$\ell_t$ norm inequality, Mild generalization of {\cite[SoS triangle inequality]{hopkins2018clustering}}]\label[Lemma]{lem:l1-to-lt}
    Given $a_1, \dots, a_m \in \R$ and $t = 2^k$ for some $k > 0$, the following holds,
    \[
        \proves_t (a_1 + \dots + a_m)^t \le m^{t-1}(a_1^t + \dots + a_m^t)\, .
    \]
\end{lemma}
\begin{proof}
    The proof is by induction in conjunction with~\cite[Proposition (squaring SoS proofs)]{hopkins2018clustering}.

    \paragraph{Base case ($t = 2$)} In this case, we set $a = a_1$ and $b = a_2 + \dots + a_m$, and observe that,
    \[
        (a_1 + \dots + a_m)^2 = \sum_{i=1}^m a_i^2 + 2\sum_{i < j}a_ia_j, ,
    \]
    and that,
    \[
        m(a_1^2 + \dots + a_m^2) - (a_1 + \dots + a_m)^2 = \sum_{i < j}(a_i - a_j)^2\, .
    \]
    The above immediately implies that,
    \[
        (a_1 + \dots + a_m)^2 \le m(a_1^2 + \dots + a_m^2)\, ,
    \]
    where the SoS proof is clearly degree-$2$.

    \paragraph{Inductive step} For the inductive step, we assume the following,
    \[
        \proves_{t/2} (a_1 + \dots + a_m)^{t/2} \le m^{t/2 - 1}(a_1^{t/2} + \dots + a_m^{t/2})\, .
    \]
    By~\cite[Proposition (squaring SoS proofs)]{hopkins2018clustering}, this immediately implies that,
    \[
        \proves_t (a_1 + \dots + a_m)^t \le m^{t-2}(a_1^{t/2} + \dots + a_m^{t/2})^2\, .
    \]
    Using the base case with $t = 2$ and $a'_1 = a_1^{t/2}, \dots, a'_m = a_m^{t/2}$, this immediately yields the following,
    \[
        \proves_2 (a_1^{t/2} + \dots + a_m^{t/2})^2 \le m(a_1^t + \dots + a_m^t)\, ,
    \]
    which finishes the proof.
\end{proof}

\subsection{Semi-algebraic degree reduction inequalities}\label{sec:comm-semi-algebraic-degree-reduction}

\textbf{Note}: Once again, the statements below of~\Cref{prop:quadratic-ideal-reduction},~\Cref{prop:odd-bounds-from-even-bounds} and~\Cref{lem:induction-ideal-reduction} can all be made to work as SoS proofs for polynomials of moments in the grammar for the underlying HES SoS relaxation.

\begin{lemma}[Quadratic degree reduction]
\label[Lemma]{lem:ideal-reduction-2}
Let $c > 0$ be a constant.
There is a SoS proof that if $(x-y)^2 \le c^2$ then $x-y \le c$.
\end{lemma}
\begin{proof}
$(x-y)^2 \le c^2$ is equivalent to $c^2 - (x-y)^2 \ge 0$, which is equivalent by factoring to $[c - (x-y)][c + (x-y)] \ge 0$.

We can add $[c - (x-y)]^2$ to both sides to get $[c - (x-y)][2c] \ge [c - (x-y)]^2 \ge 0$.
Dividing both sides by $2c$, we get $c \ge (x-y)$.
\end{proof}
 
\textbf{Note:} The above can also be amended to immediately conclude that $x + y \ge -c$, and the same conclusion follows about the ``two-sidedness'' of all semi-algebraic degree reduction inequalities.

\begin{lemma}
\label[Lemma]{lem:ideal-reduction-4}
Let $c > 0$ be a constant.
There is a SoS proof that if $(x-y)^4 \le c^4$ then $x-y \le c$.
\end{lemma}
\begin{proof}
$(x-y)^4 \le c^4$ is equivalent to $c^4 - (x-y)^4 \ge 0$, which is equivalent by factoring to $[c - (x-y)][c^3 + c^2(x-y) + c(x-y)^2 + (x-y)^3] \ge 0$.

We can add $[c - (x-y)]^2[c + (x-y)]^2$ to both sides to get $[c - (x-y)][2c^3 + 2c^2(x-y)] \ge [c - (x-y)]^2[c + (x-y)]^2 \ge 0$.

We can then add $2[c - (x-y)]^2c^2$ to both sides to get $[c - (x-y)][4c^3] \ge 2[c - (x-y)]^2c^2 \ge 0$.
Dividing both sides by $4c^3$, we get $c \ge (x-y)$.
\end{proof}

\begin{lemma}
\label[Lemma]{lem:ideal-reduction-2p}
Let $c > 0$ be a constant.
There is a SoS proof that if $(x-y)^{2p} \le c^{2p}$ then $x-y \le c$.
\end{lemma}
\begin{proof}
$(x-y)^{2p} \le c^{2p}$ is equivalent to $c^{2p} - (x-y)^{2p} \ge 0$, which is equivalent by factoring to $[c - (x-y)][c^{2p-1} + c^{2p-2}(x-y) + c^{2p-3}(x-y)^2 + \dots + (x-y)^{2p-1}] \ge 0$. It suffices to recurisvely generalize the case for degree-$2$ and degree-$4$.
\end{proof}

\begin{proposition}
\label[Proposition]{prop:degree-reduction-nonconstant}
    Suppose that $x_i^2 \le y_i^2$ and $y_i \ge 0$ by sum-of-squares proofs for all $i$.
    Then there is a degree-$3$ sum-of-squares proof that
    \[ \sum_i x_i \le \sum_i y_i \]
    and
    \[ -\sum_i x_i \le \sum_i y_i \,.\]
\end{proposition}
\begin{proof}
    The inequality $x_i^2 \le y_i^2$ can be factored and summed to get
    \[ \sum_i (y_i - x_i)(y_i + x_i) \ge 0 \,.\]
    Add $(y_i-x_i)(y_i-x_i)$, which is non-negative, to find
    \[ \sum_i (y_i - x_i)y_i \ge 0 \]
    
    Since $y_j \ge 0$, we know that 
    \[(y_i - x_i)y_i\left(\sum_{j\ne i} y_j\right) \ge 0\,.\]
    This implies by a weighted AM-GM that 
    \[(y_i - x_i)\left(\sum_{j\ne i} \frac{1}{\sqrt{n-1}}y_i^2 + \sqrt{n-1}y_j^2\right) \ge 0\,,\]
    where $n$ is the number of possible values of $i$.
    This implies 
    \[(y_i - x_i)\left(\sum_{j} y_j^2\sqrt{n-1}\right) \ge 0\,.\]
    Finally, sum the preceding equation over $i$.
    This is the form of a SoS proof that demonstrates that either $\sum_i x_i \le \sum_i y_i$ or $y_i = 0$ for all $i$.
    The latter case implies that $x_i = 0$ for all $i$ as well.
    The $-\sum_{i}x_i$ case follows symmetrically.
\end{proof}

One again, note that each of the lemmata proved below about pseudo-distributions can \emph{also} be used to give a SoS proof for polynomials of moments that are of even degree, provided they are valid elements of the polynomial grammar~(\Cref{sec:sos-poly-basis}), meaning that these proofs can also be point-wise over the atoms of the HES SoS hierarchy.

\begin{proposition}[Quadratic degree reduction with pseudo-expectations]\label[Proposition]{prop:quadratic-ideal-reduction}
    Given non-SoS program variables $y = an, c(n) = o(n)$ for $a \ge 1$ a constant, a non-SoS program variable $c(n) = o(n)$, and a low-degree SoS proof that,
    \[
        \pE_x\left(x - y\right)^2 \le c(n)^2\, ,
    \]
    there are low-degree SoS proofs that,
    \begin{align}
        &\pE_x(x^2) \le y^2 + 2yc(n) + c^2(n) = y^2 + o(n^2)\, ,\\
        &\pE_x(x^2) \ge y^2 - o(n^2)\, ,\\
        &\pE_x(x) \le y + o(n)\, ,\\
        &\pE_x(x) \ge y - o(n)\, .
    \end{align}
\end{proposition}
\begin{proof}
    We first use quadratic ideal reduction to show the linear bounds, and then use those to obtain the quadratic bounds.

    \paragraph{Linear bounds} By the pseudo-Jensen inequality~(\Cref{lem:matrix-pseudo-jensen}) and the given SoS proof,
    \[
        (\pE_x(x-y))^2 = \left(\pE_x\left[x\right] - y\right)^2 \le \pE_x(x-y)^2 \le c(n)^2\, .
    \]
    Then, by~\Cref{lem:ideal-reduction-2} on $(\pE_x(x) - y)^2 \le c(n)^2$, we have that,
    \begin{equation}\label{eq:linear-bound-via-pseudo-jensen}
        \pE_x(x-y) \le c(n)\, .
    \end{equation}
    Using the fact that $(\pE_x(y-x))^2 \le \pE_x(y-x)^2 = \pE_x(x-y)^2 \le c(n)^2$, one can again use~\Cref{lem:ideal-reduction-2} to conclude that,
    \[
        y - \pE_x(x) \le o(n) \implies \pE_x(x) \ge y - o(n)\, .
    \]

    \paragraph{Quadratic bounds}
    Now,
    \[
        \pE_x(x-y)^2 = \pE(x^2) + y^2 - 2y\pE_x(x) \le c(n)^2\, .
    \]
    Using the implication in~\Cref{eq:linear-bound-via-pseudo-jensen} and the given bound,
    \[
        \pE_x[x^2]\le c(n)^2 - y^2 + 2y(c(n) + y) = y^2 + 2yc(n) + c(n)^2 \le y^2 + o(n^2)\, .
    \]
    Using the fact that $\pE_x(x-y)^2 \ge 0$ and $\pE_x(x) \ge y -o(n)$,
    \begin{align*}
        &\pE_x(x-y)^2 = \pE_x(x^2) + y^2 -2y\pE_x(x) \ge 0\\
        &\implies \pE(x^2) \ge 2y\pE_x(x) - y^2 \ge  2y(y - o(n)) - y^2 = y^2 - o(n^2)\, .
    \end{align*}
\end{proof}

\begin{proposition}[Odd pseudo-moment bounds from degree Reductions]\label[Proposition]{prop:odd-bounds-from-even-bounds}
    Given non-SoS program variables $y = an, c(n) = o(n)$ for $a \ge 1$ a constant, a pseudo-distribution of degree $\ge 4p+2$, and low-degree SoS proofs that,
    \[
        \pE_x[(x-y)^{2j}] \le o(n^{2j})\, ,\,\forall 1 \le j \le 2p+1\, ,
    \]
    there exist low-degree SoS proofs that,
    \begingroup
    \allowdisplaybreaks
    \begin{align}
        &\pE_x(x-y)^{2j+1} \le o(n^{2j+1})\, , \,\forall 1 \le j \le p\\
        &\pE_x(y-x)^{2j+1} \le o(n^{2j+1})\, ,\,\forall 1 \le j \le p\, .
    \end{align}
    \endgroup
\end{proposition}
\begin{proof}
    We begin the following observation using the pseudo-Jensen inequality~(\Cref{lem:matrix-pseudo-jensen}) and the given bound over even pseudo-moments,
    \[
        0 \le \left(\pE_x(y-x)^{2j+1}\right)^2 = \left(\pE_x(x-y)^{2j+1}\right)^2 \le \pE_x(x-y)^{4j+2} \le o(n^{4j+2})\, , \forall\, 1 \le j \le 2p-1\, .
    \]
    We then apply a similar argument as in~\Cref{lem:ideal-reduction-2} to $0 \le \left(\pE_x(y-x)^{2j+1}\right)^2 \le o(n^{4j+2})$ as follows,
    \begingroup
    \allowdisplaybreaks
    \begin{align*}
        &o(n^{4j+2}) - \left(\pE_x(y-x)^{2j+1}\right)^2 \ge 0 \\
        &\implies \left(o(n^{2j+1}) - \pE_x(y-x)^{2j+1}\right)\left(o(n^{2j+1}) + \pE_x(y-x)^{2j+1}\right) \ge 0 \\
        &\implies \left(o(n^{2j+1}) - \pE_x(y-x)^{2j+1}\right)\left(o(n^{2j+1}) + \pE_x(y-x)^{2j+1}\right) + \left(o(n^{2j+1}) - \pE_x(y-x)^{2j+1}\right)^2 \ge 0 \\
        &\implies \left(o(n^{2j+1}) - \pE_x(y-x)^{2j+1}\right)\left\{\left(o(n^{2j+1}) + \pE_x(y-x)^{2j+1}\right) + \left(o(n^{2j+1}) - \pE_x(y-x)^{2j+1}\right)\right\} \ge 0 \\
        &\implies \left(o(n^{2j+1}) - \pE_x(y-x)^{2j+1}\right)o(n^{2j+1}) \ge 0 \\
        &\implies \left(o(n^{2j+1}) - \pE_x(y-x)^{2j+1}\right) \ge 0\, ,
    \end{align*}
    \endgroup
    which yields the desired bound.
\end{proof}

\begin{lemma}[Pseudo-moment bounds via inductive degree reductions]\label[Lemma]{lem:induction-ideal-reduction}
    Given non-SoS program variables $y = an, c(n) = o(n)$ for $a \ge 1$ a constant, a pseudo-distribution with degree $\ge (4p+2)\cdot\mathsf{deg}(x)$ and a low-degree SoS proof that,
    \[
        \pE_x(x-y)^{2j} \le c(n)^{2j} = o(n^{2j})\, ,\,\forall 1\le j \le 2p+1
    \]
    there are low-degree SoS proofs that,
    \begingroup
    \allowdisplaybreaks
    \begin{align}
        &y^{2p} - o(n^{2p}) \le \pE_x(x^{2p}) \le y^{2p} + o(n^{2p})\, ,\\
        &y^{2p+1} - o(n^{2p+1}) \le \pE_x(x^{2p+1}) \le y^{2p+1} + o(n^{2p+1})\, .
    \end{align}
    \endgroup
\end{lemma}
\begin{proof}
    We prove this by appealing to strong induction. The base case corresponds to~\Cref{prop:quadratic-ideal-reduction}. The assumption will be guaranteed for the particular use case by the certificates provided in~\Cref{lem:hessian-schatten-norm}.

    \paragraph{Base case} For $j=1$, the base case is immediately implied by
    \Cref{prop:quadratic-ideal-reduction} which gives the following bounds,
    \begingroup
    \allowdisplaybreaks
    \begin{align*}
        &\pE_x(x^2) \le y^2 + 2yc(n) + c^2(n) = y^2 + o(n^2)\, ,\\
        &\pE_x(x^2) \ge y^2 - o(n^2)\, ,\\
        &\pE_x(x) \le y + o(n)\, ,\\
        &\pE_x(x) \ge y - o(n)\, .
    \end{align*}
    \endgroup
    \paragraph{Inductive hypothesis} Assume that,
    \[
        y^{2j} - o(n^{2j}) \le \pE_x(x^{2j}) \le y^{2j} + o(n^{2j})\, , \forall\, j \in \{2,3,4,\dots,p-1\}\, ,
    \]
    and
    \[
        y^{2j+1} - o(n^{2j+1}) \le \pE_x(x^{2j+1}) \le y^{2j+1} + o(n^{2j+1})\, ,\forall\, j \in \{1,2,3,\dots,p-1\}\,.
    \]

    \paragraph{Inductive step} We will demonstrate lower and upper bounds for the case that $j=p$ for both the even ($2j = 2p$) and odd ($2j+1 = 2p+1$) case. We will need the former (even) for the latter (odd), and therefore, we begin with the even case. By the definition of a pseudo-distribution and the certificates given in the theorem condition,
    \[
        0 \le \pE_x(x-y)^{2p} \le c(n)^{2p} = o(n^{2p})\, .
    \]
    expanding the LHS yields,
    \begingroup
    \allowdisplaybreaks
    \begin{align*}
        \pE_x(x-y)^{2p} &= \pE\left(\sum_{j=0}^{2p}(-1)^{2p-j}\binom{2p}{j}x^jy^{2p-j}\right) \\
        &= y^{2p} + \pE_x\left[x^{2p}\right] + \sum_{j=1}^{2p-1}(-1)^{2p-j}\binom{2p}{j}y^{2p-j}\pE_x[x^j]\\
        &=_{(-1)^{2p-j}=(-1)^j} y^{2p} + \pE_x\left[x^{2p}\right] + \sum_{j=1}^{2p-1}(-1)^j\binom{2p}{j}y^{2p-j}\pE_x[x^j]\, .
    \end{align*}
    \endgroup
    For the lower bound, we use the expansion of the LHS with the inductive hypothesis and positivity,
    \begingroup
    \allowdisplaybreaks
    \begin{align*}
        &y^{2p} + \pE_x\left[x^{2p}\right] + \sum_{j=1}^{2p-1}(-1)^j\binom{2p}{j}y^{2p-j}\pE_x[x^j] \ge 0 \\
        &\implies \pE_x\left[x^{2p}\right] \ge -y^{2p} - \sum_{j=1}^{2p-1}(-1)^j\binom{2p}{j}y^{2p-j}\pE_x[x^j] \\
        &\implies \pE_x[x^{2p}] \ge -y^{2p} -\sum_{j \in \{2,4,\dots,2p-2\}}\binom{2p}{j}y^{2p-j}\pE_x[x^j] + \sum_{j\in\{1,3,\dots,2p-1\}}\binom{2p}{j}y^{2p-j}\pE_x[x^j] \\
        &\implies \pE_x[x^{2p}] \ge_{\text{ind. hyp.}} -y^{2p} -\sum_{j \in \{2,4,\dots,2p-2\}}\binom{2p}{j}y^{2p-j}(y^{j} + o(n^{j})) + \sum_{j\in\{1,3,\dots,2p-1\}}\binom{2p}{j}y^{2p-j}(y^{j} - o(n^j)) \\
        &\implies \pE_x[x^{2p}] \ge -y^{2p} - y^{2p}\sum_{j=1}^{2p-1}(-1)^j\binom{2p}{j} - \sum_{j=1}^{2p-1}\binom{2p}{j}y^{2p-j}o(n^j) \\
        &\implies \pE_x[x^{2p}] \ge -y^{2p} - y^{2p}\left(\sum_{j=0}^{2p}(-1)^{j}\binom{2p}{j} - 2\right) - (1+a)^{p}o(n^{2p}) \\
        &\implies \pE_x[x^{2p}] \ge_{(1-1)^{2p} = 0} -y^{2p} + 2y^{2p} - o(n^{2p1}) = y^{2p} - o(n^{2p})\, .
    \end{align*}
    \endgroup

    For the upper bound, we combine the expansion of the LHS with the inductive hypothesis and the upper bound given by the condition,
    \begingroup
    \allowdisplaybreaks
    \begin{align*}
        \pE[x^{2p}] &\le o(n^{2p}) - y^{2p} -\sum_{j=1}^{2p-1}(-1)^j\binom{2p}{j}y^{2p-j}\pE_x\left[x^j\right] \\
        &= o(n^{2p}) + \sum_{j=1,3,\dots,2p-1}\binom{2p}{j}y^{2p-j}\pE_x[x^j] - \sum_{j=2,4,\dots,2p-2}\binom{2p}{j}y^{2p-j}\pE_x[x^j] - y^{2p}\\
        &\le_{\text{ind. hyp.}} o(n^{2p}) + \sum_{j=1,3,\dots,2p-1}\binom{2p}{j}y^{2p-j}(y^j + o(n^j)) - \sum_{j=2,4,\dots,2p-2}\binom{2p}{j}y^{2p-j}(y^j - o(n^j)) - y^{2p}\\
        &= o(n^{2p}) - y^{2p} - y^{2p}\sum_{j=1}^{2p-1}(-1)^j\binom{2p}{j} + \sum_{j=1}^{2p-1}\binom{2p}{j}y^{2p-j}o(n^j) \\
        &= o(n^{2p}) - y^{2p} - y^{2p}\left(\sum_{j=0}^{2p}(-1)^{j}\binom{2p}{j} - 2\right) + o(n^{2p})\sum_{j=1}^{2p-1}\binom{2p}{j}a^{2p-j} \\
        &\le_{(1-1)^{2p} = 0} o(n^{2p}) - y^{2p} + 2y^{2p} + (1+a)^{2p}o(n^{2p}) \le y^{2p} + o(n^{2p})\, .
    \end{align*}
    \endgroup
    Using the bounds,
    \[
        -o(n^{2p+1}) \le \pE_x(x-y)^{2p+1} \le o(n^{2p+1})\, ,
    \]
    guaranteed by~\Cref{prop:odd-bounds-from-even-bounds} and the bounds for the case that $j=2p$ as proved in the above (even) induction step, a similar computation gives the desired result for the odd case.
\end{proof}

\subsection{Strong convexity for \texorpdfstring{$x \to x^q$}{x to x} via linear interpolation}

Stated below is a low-degree SoS proof that the function $f: x \to x^{q}$ for $q \in (1,2]$ is strongly convex when $p$ is dyadic as the H\"older conjugate. This proof is presented for the purpose of modularization and is ``lifted'' to a proof of \emph{trace} strong convexity for the same function in the companion paper, and ends up being crucial in the rounding analysis.

\begin{lemma}[Strong convexity of $x \to x^q$]\label[Lemma]{lem:scalar-strong-convexity}
    Let the function $f: x \to x^{q}$ with $\mathsf{dom}(f) \in \left[0, \frac{1}{(A(\eps)n)}\right]$ and $q = \frac{2^k}{2^k-1}$ for some $k \ge 1$. Then, $\forall x \ne y \in \mathsf{dom}(f)$,
    \[
        \proves \{y^{q} - x^{q} \ge q(y-x)x^{q-1} + C(y-x)^2\}\, ,
    \]
    where $C = \frac{(q-1)}{(p-1)(\delta n)^{q-2}}$.
\end{lemma}

\begin{proof}
    We want to show (with $p/(p-1) = q \implies 1/(q-1) = (p-1)$) that, \\
    \[
        y^q - x^q \ge q(y-x)x^{q-1} + \frac{(q-1)}{(p-1)(\delta n)^{q-2}}(y-x)^2,\, \forall q \in (1,2].
    \]
    Let $x = a^{1/(q-1)}$ and $y = b^{1/(q-1)}$. Then,
    \[
        b^{p} - a^{p} \ge q(b^{p-1} - a^{p-1})a + \frac{(q-1)}{(p-1)(\delta n)^{q-2}}(b^{p-1}-a^{p-1})^2\, .
    \]
    Dividing both sides by $(b-a)$ (assuming $b > a$ WLOG):
    \[
        \sum_{i=1}^{p}b^{i-1}a^{p-i} \ge \sum_{j=1}^{p-1}b^{j-1}a^{p-1-j}\left(qa + \frac{(q-1)}{(p-1)(\delta n)^{q-2}}(b^{p-1} - a^{p-1})\right)\, .
    \]
    We now linearly interpolate between the summands in the LHS by breaking each into a weighted sum of two terms as,
    \[ \sum_{i=1}^{p}b^{i-1}a^{p-i}
    = \sum_{i = 0}^{p-2} \left[ \left(1-\frac{j}{p-1}\right)(b^{j}a^{p-1-j}) +\frac{j+1}{p-1} (b^{j+1}a^{p-2-j})\right]\, ,
    \]
    where we use the fact that $p$ is even and $p-1$ is odd. Rewriting the sum in the LHS and RHS from zero indices yields the following final form,
    \begingroup
    \allowdisplaybreaks
    \begin{align*}
        \sum_{j = 0}^{p-2} \left[ \left(1-\frac{i}{p-1}\right)(b^{i}a^{p-1-i}) +\frac{i+1}{p-1} (b^{i+1}a^{p-2-i})\right] \ge \sum_{j=0}^{p-2}b^ja^{p-2-j}\left(aq + \frac{(q-1)}{(p-1)(\delta n)^{q-2}}(b^{p-1} - a^{p-1})\right)\, .
    \end{align*}
    \endgroup
    Re-arranging the terms immediately yields that,
    \[
        \sum_{i=0}^{p-2}b^{j}a^{p-2-j}\left[\left(\frac{p-1-j}{p-1}\right)a + \left(\frac{j+1}{p-1}\right)b - \left(aq + \frac{(q-1)}{(p-1)(\delta n)^{q-2}}(b^{p-1} - a^{p-1})\right)\right] \ge 0\, ,
    \]
    which, with a bit of algebra becomes equivalent to showing that,
    \begingroup
    \allowdisplaybreaks
    \begin{align*}
        &\sum_{i=0}^{p-2}b^{j}a^{p-2-j} \left[\left(1-\frac{j}{p-1}\right)a +\left(\frac{j+1}{p-1}\right)b - qa - \frac{(q-1)}{(p-1)(\delta n)^{q-2}}(b^{p-1}-a^{p-1})\right] \\
        &= \sum_{i=0}^{p-2}b^{j}a^{p-2-j}\left[a\left(1 - j(q-1)\right) + b((j+1)(q-1)) - aq - \frac{(q-1)}{(p-1)(\delta n)^{q-2}}(b^{p-1}-a^{p-1})\right] \\
        &= \sum_{i=0}^{p-2}b^{j}a^{p-2-j}\left[(q-1)(b(j + 1) - aj) + a - aq - (q-1)\left(\frac{b^{p-1} - a^{p-1}}{(p-1)(\delta n)^{q-2}}\right) \right] \\
        &= \sum_{i=0}^{p-2}b^{j}a^{p-2-j}\left[-a(q-1) + (q-1)(j(b-a) + b) - (q-1)\left(\frac{b^{p-1} - a^{p-1}}{(p-1)(\delta n)^{q-2}}\right)\right] \\
        &= (q-1)\sum_{i=0}^{p-2}b^{j}a^{p-2-j}\left[(b-a) + j(b-a) - \left(\frac{b^{p-1} - a^{p-1}}{(p-1)(\delta n)^{q-2}}\right)\right] \\
        &=_{\text{divide by }(b-a)} (q-1)\sum_{i=0}^{p-2}b^{j}a^{p-2-j}\left[(j + 1) - \frac{1}{(p-1)(\delta n)^{q-2}}\left( \sum_{i=0}^{p-2}b^i a^{p-2-i}\right)\right]
    \end{align*}
    \endgroup
    Now, clearly $j + 1\ge 1$. Furthermore, since, $x = a^{p-1}$ and $b = y^{p-1}$ with the given constraints that $a^{p-1} \le \frac{1}{\delta n}$ and $b^{p-1} \le \frac{1}{\delta n}$,
    \begin{align*}
        a, b \le_{\text{semi-algebraic degree reduction by~\Cref{lem:ideal-reduction-2p}}} \frac{1}{(\delta n)^{1/(p-1)}}\, .
    \end{align*}
    Note that $p-1$ is odd, so technically one needs to apply a slightly different version of~\Cref{lem:ideal-reduction-2p} for the case that $y = 0$, $c > 0$, $x > 0$ and there is an \emph{odd} power. Factoring $x^{p-1} - c^{p-1} = (x-c)\sum_{i=0}^{p-2} x^ic^{p-2-i}$, separating the terms with even powers as being squared, and then taking $xc$ as a factor for the odd terms (yielding sums of even squares for the factored terms) yields a SoS proof for odd-degree ideal reduction in the desired special case. The concluding step in the strong convexity argument is to use this semi-algebraic degree reduction and observe that,
    \begingroup
    \allowdisplaybreaks
    \begin{align*}
        b^i a^{p - 2 -i} &\le \frac{1}{(\delta n)^{i/(p-1)}}\frac{1}{(\delta n)^{(p-2-i)/(p-1)}} = \frac{1}{(\delta n)^{(p-2)/(p-1)}} = \frac{1}{(\delta n)^{1 - 1/(p-1)}} \\
        &= \frac{1}{(\delta n)^{1 - (q - 1)}} = \frac{1}{(\delta n)^{2-q}} \\
        &\implies -\sum_{i=0}^{p-2}b^{i}a^{p-2-i} \ge -(p-1)\frac{1}{(\delta n)^{2-q}} \\
        &\implies \frac{-1}{(p-1)(\delta n)^{q-2}}\sum_{i=0}^{p-2}b^{i}a^{p-2-i} \ge \frac{1}{(p-1)}\frac{1}{(\delta n)^{q-2}}\frac{-(p-1)}{(\delta n)^{2-q}} = 1\, .
    \end{align*}
    \endgroup
    The above set of inequalities in conjunction with the facts that $a, b\ge 0$ and $q > 1$ finishes the proof.
\end{proof}

\subsection{AM-GM-type inequality}
We state and prove a version of the AM-GM inequality and briefly comment on how to straight-forwardly extend it to a non-commutative version for symmetric matrices.

\begin{lemma}[Scalar AM-GM inequality]\label[Lemma]{lem:scalar-am-gm}
    Given any $k \in \Z_{\mathsf{even}}$ and real-valued variables $x_1,\dots,x_k$, the following holds,
    \[
        \proves_{4k} \left\{\prod_{i=1}^k x_i \le \frac{1}{k}\sum_{i=1}^k x^k_i\right\}\, .
    \]
\end{lemma}
\begin{proof}
    We proceed by induction on $k$.

    \paragraph{Base case $k=2$} For the case that $k=2$, the inequality follows by the non-negativity of $\frac{1}{4}(x_1-x_2)^2$ in conjunction with~\Cref{lem:l1-to-lt}.
    \begin{align*}
        \frac{1}{4}(x_1+x_2)^2 - x_1x_2 = \frac{1}{4}(x_1^2 + x_2^2) + \frac{1}{2}x_1x_2 - x_1x_2 = \frac{1}{4}(x_1^2 + x_2^2) - \frac{1}{2}x_1x_2 = \frac{1}{4}(x_1 - x_2)^2 \ge 0\, .
    \end{align*}
    Furthermore, invoking~\Cref{lem:l1-to-lt} with $m=2$ and $t=2$ gives,
    \[
        \frac{1}{4}(x_1+x_2)^2 \le \frac{1}{4}2^{1}(x_1^2 + x_2^2) = \frac{1}{2}(x_1^2 + x_2^2)\, .
    \]
    Combing the above two inequalities yields the desired result.

    \paragraph{Inductive hypothesis} Assume the inequality holds for even powers up to $k = 2^m$,
    \[
        x_1\cdots x_k \le \frac{1}{k}(x_1^k + \dots + x^k_k)\, .
    \]

    \paragraph{Inductive step} Set $k' = 2^{m+1} = 2k$. Now,
    \begin{align*}
        &\frac{1}{k}\sum_{i=1}^{2k}x^{2k}_i = \frac{\frac{\sum_{i=1}^k x^{2k}_i}{k} + \frac{\sum_{j=k+1}^{2k}x^{2k}_j}{k}}{2} \\
        &\ge_{\text{by~\Cref{lem:l1-to-lt}}} \frac{\frac{\left(x^k_1 + \dots + x^k_k\right)^2}{k^2} + \frac{\left(x_{k+1}^k + \dots + x_{2k}^k\right)^2}{k^2}}{2} \\
        &\ge_{\text{by inductive step}} \frac{(x_1\cdots x_k)^2 + (x_{k+1}\cdots x_{2k})^2}{2} \\
        &\ge_{\text{by~\Cref{lem:l1-to-lt} with $a_1 = x_1\cdots x_k$, $a_2 = x_{k+1}\cdots x_{2k}$}} \left(\frac{x_1\cdots x_k + x_{k+1}\cdots x_{2k}}{2}\right)^2 = \frac{1}{4}\left(x_1\cdots x_k + x_{k+1}\cdots x_{2k}\right)^2 \\
        &\ge_{\text{by $\frac{1}{4}(a+b)^2 \ge ab$}} (x_1\cdots x_k)(x_{k+1}\cdots x_{2k}) = x_1\cdots x_{2k}\,.
    \end{align*}

    In the event that the even number is not a power of $2$, we reduce to the case that it is by adding ``dummy'' variables that have a specific value. Let $k = 2m$ and let $k'$ be the smallest number that is a power of $2$ such that $k' > k$. Then, set $x^k_{k+1} = \dots = x^k_{k'} = \frac{1}{k}\sum_{i=1}^k {x^{k}_i} = \alpha$\, and $y_i = (x_i^{1/k'})^k$ as a polynomial of radicals of $x_i$. By straightforward algebraic manipulation,
    \begin{align*}
        &\alpha^{k'} = \left(\frac{1}{k}\sum_{i=1}^{k}x^{k}_i\right)^{k'} = \left(\frac{\sum_{i=1}^{k'} x^k_i}{k'}\right)^{k'} = \left(\frac{\sum_{i=1}^{k'}y_i^{k'}}{k'}\right)^{k'} \\
        &\ge _{\text{by power-of-2 version}} (y_1\cdots y_{k'})^{k'} = y^{k'}_1\cdots y^{k'}_{k'} \\
        &= \left(x^k_1\cdots x^k_k\right)\left(x^k_{k+1}\cdots x^k_{k'}\right) = \left(x^k_1\cdots x^k_k\right)\alpha^{k' - k} \\
        &\implies \alpha^k \ge \left(x^k_1\cdots x^k_k\right) \implies_{\text{by~\Cref{lem:ideal-reduction-2p} with $y=0$ and $c = \eps$}}\,\alpha \ge x_1\cdots x_k + \eps\, .
    \end{align*}
\end{proof}

\section{Sum-of-Squares proofs for Non-Commutative Analytic Inequalities}\label{sec:non-commutative-inequalities}

\subsection{Matrix Cauchy-Schwarz inequality}
We begin by stating a straightforward proof of a matrix Cauchy-Schwarz inequality, which works on symmetric matrices where the entries are low-degree polynomials in the SoS program variables.
\begin{lemma}[Matrix Cauchy-Schwarz]\label[Lemma]{lem:matrix-cauchy-schwarz}
    Given two symmetric matrices $A, B \in M_n(p(\sigma))$ where $p \in \R_{\le d}[\sigma_1,\dots \sigma_n]$, the following holds:
    \[
        \left\{A = A^\sT, B = B^\sT\right\} \proves \left\{\langle A, B \rangle^2 \le \norm{A}^2_F \norm{B}^2_F\right\}\, .
    \]
\end{lemma}
\begin{proof}
    The proof follows by simply ``flattening'' the matrices into a $n^2$-dimensional vector, and applying the SoS proof for the standard Cauchy-Schwarz inequality.
    \begin{align*}
        & \norm{A}^2_F \norm{B}^2_F - \langle A, B \rangle^2 = \Tr[A^\sT A]\Tr[B^\sT B] - \left(\Tr[A^\sT B]\right)^2 \\
        &= \left(\sum_{i_1, j_1 = 1}^{n} A^2_{i_1,j_1}\right)\left(\sum_{i_2,j_2=1}^{n^2} B^2_{i_2,j_2}\right)  - \left(\sum_{j=1}^n \langle\mathsf{col}_j(A), \mathsf{col}_j(B) \rangle\right)^2 \\
        &= \left(\sum_{i_1, j_1 = 1}^{n} A^2_{i_1,j_1}\right)\left(\sum_{i_2,j_2=1}^{n^2} B^2_{i_2,j_2}\right) - \left(\sum_{i,j=1}^n A_{ij}B_{ij}\right)^2 \\
        &= \left(\sum_{(i_1, j_1) \preceq (i_2,j_2)} A_{i_1,j_1}B_{i_2,j_2} - A_{i_2,j_2}B_{i_1,j_1}\right)^2 \ge 0\, .
    \end{align*}
\end{proof}

\subsection{The \texorpdfstring{$2\text{-}\infty\text{-}2$}{2-infty-2} matrix H\"older's inequality}
We state low-degree SoS proofs for Matrix versions of H\"older inequalities with $(2,\infty, 2)$ and $(1,\infty)$ conjugate pairs. The general form of H\"older's inequality for matrices is typically proved using the weak majorization of singular values, however, these proofs are not low-degree SoS proofs.

\begin{lemma}[($2,\infty,2$)-Matrix H\"older's Inequality]\label[Lemma]{lem:holders-2-inf-2}
    Given two symmetric matrices $A, B \in M_n(\R)$, the following holds,
    \[
        \{0 \preceq B, B = \lambda\Id - CC^\sT\}\proves_{4\mathsf{deg}(x^\sT B x)} \Tr[ABA^\sT] \le \lambda\norm{A}^2_F\, .
    \]
\end{lemma}

\begin{proof}
    We begin by giving a $\mathsf{deg}(x^\sT B x)$ SoS proof that,
    \[
        \Tr[ABA^\sT] \le \Tr[A(\lambda\Id)A^\sT] = \lambda\Tr[AA^\sT] =\lambda\norm{A}_F^2\, .
    \]
    Observe the fact that $B = \lambda\Id - CC^\sT$ by the system of axioms and, therefore, for every $x \in \R$,
    \begingroup
    \allowdisplaybreaks
    \begin{align*}
        x^\sT(A(\lambda\Id)A^\sT - ABA^\sT)x&= \lambda x^\sT(AA^\sT)\Id x - (A^\sT x)^\sT B (A^\sT x) \\
        &= \lambda y^\sT y - y^\sT B y = \lambda\norm{y}^2_2 - y^\sT B y \\
        &=\lambda\norm{y}^2_2 - (\lambda\norm{y}^2_2 - y^\sT CC^\sT y) \\
        &= \norm{C^\sT y}^2_2 \ge 0\, ,
    \end{align*}
    \endgroup
    where we used the fact that $0 \le \opnorm{B} = \lambda\Id - CC^\sT$ by our axioms. We now show that $\Tr[ABA^\sT] \le \lambda\Tr[AA^\sT] = \lambda\norm{A}_F^2$ permits a low-degree SoS proof.
    \begingroup
    \allowdisplaybreaks
    \begin{align*}
        \lambda\Tr[AA^\sT] - \Tr[ABA^\sT] &= \lambda\left(\sum_{i=1}^n e_i^\sT (AA^\sT)e_i\right) - \left(\sum_{i=1}^n e_i^\sT(ABA^\sT)e_i\right)\\
        &= \left(\sum_{i=1}^n e_i^\sT\left(A(\lambda\Id)A^\sT - ABA^\sT\right)e_i\right)  \ge 0\, .
    \end{align*}
    \endgroup
\end{proof}

\begin{lemma}[Asymmetric ($2,\infty,2$)-Matrix H\"older's Inequality]\label[Lemma]{lem:holders-2-inf-2-asym}
    Given three real matrices $A, B, X$, the following holds:
    \[
        \{XX^T = \lambda^2\Id - CC^\sT\}\proves_{6} \{(\Tr[A^\sT XB])^2 \le \lambda^2\norm{A}_F^2\norm{B}_F^2\}\, .
    \]
\end{lemma}
\begin{proof}
    By one step of matrix Cauchy-Schwarz,
    \[(\Tr[A^\sT XB])^2 \le \norm{A}_F^2\norm{BX}_F^2 \,.\]
    Now,
    \[\norm{BX}_F^2 = \Tr BXX^TB^T = \lambda^2\Tr BB^T - \Tr BCC^TB^T \le \lambda^2\Tr BB^T = \lambda_2\norm{B}_F^2 \,.\]
\end{proof}

\begin{lemma}[($2,\infty,2,\infty$)-Matrix H\"older's Inequality]\label[Lemma]{lem:holders-2-inf-2-inf}
    Given three real matrices $A, B, X$, the following holds:
    \[
        \{XX^T = \lambda^2\Id - CC^\sT\}\proves_{8} \left\{(\Tr[AXBX])^2 \le \lambda^4\norm{A}_F^2\norm{B}_F^2\right\}\, .
    \]
    Consequently,
    \[
        \{XX^T = \lambda^2\Id - CC^\sT\}\proves_{12} \left\{ \pm\Tr[AXBX] \le \frac{\lambda^2}{2}(\norm{A}_F^2 + \norm{B}_F^2)\right\}\, .
    \]
\end{lemma}
\begin{proof}
    By one step of matrix Cauchy-Schwarz,
    \[(\Tr[AXBX])^2 \le \norm{AX}_F^2\norm{BX}_F^2 \,.\]
    Now,
    \[\norm{BX}_F^2 = \Tr BXX^TB^T = \lambda^2\Tr BB^T - \Tr BCC^TB^T \le \lambda^2\Tr BB^T = \lambda_2\norm{B}_F^2\]
    and similarly for $\norm{AX}_F^2$.

    For the consequence, we apply \Cref{prop:degree-reduction-nonconstant} along with the SoS proof that
    \[ \lambda^4\norm{A}_F^2 \norm{B}_F^2 \le \frac{\lambda^4}{4}(\norm{A}_F^2 + \norm{B}_F^2) + \frac{\lambda^4}{2}\norm{A}_F^2 \norm{B}_F^2 \le \frac{\lambda^4}{4}(\norm{A}_F^2 + \norm{B}_F^2)^2\,,  \]
    where we applied AM-GM to $\frac{\lambda^4}{2}\norm{A}_F^2 \norm{B}_F^2$ and then factored the result.
\end{proof}

\subsection{The \texorpdfstring{$1\text{-}\infty$}{1-infty} matrix H\"older's inequality}

\begin{lemma}[($1, \infty$)-Matrix H\"older's Inequality]\label[Lemma]{lem:holders-one-inf}
    Given two symmetric matrices $A, B \in M_n(\R)$, the following holds,
    \[
        \{A = \lambda\Id- CC^\sT , 0 \preceq B = DD^\sT\} \proves_{4}\langle A, B \rangle \le \opnorm{A}\norm{B}_1\, ,
    \]
    when $B$ is either PSD (as stated in the axioms above) \emph{or} permits a rank-decomposing SoS proof (as defined in~\Cref{def:rank-decomposing}) that $\norm{B}_1 \le C$.
\end{lemma}
\begin{proof}
    The proof for this is similar to the proof of~\Cref{lem:holders-2-inf-2}. Note that,
    \[
        \langle A, B \rangle = \Tr[A^\sT B] \le \Tr[\lambda\Id B] = \lambda\Tr[B]\, .
    \]
    A SoS proof for this is as follows,
    \begingroup
    \allowdisplaybreaks
    \begin{align*}
        \opnorm{A}\norm{B}_1 - \langle A, B \rangle &= \lambda\Tr[B] - \Tr[A^\sT B] = \Tr[(\lambda\Id B - A^\sT B)] = \Tr[(\lambda\Id - A^T)B] \\
        &= \Tr[(\lambda\Id - (\lambda\Id - CC^\sT)^\sT)B] = \Tr[(CC^\sT)^\sT B] \\
        &= \Tr[CC^\sT DD^\sT] = \Tr[C(D^\sT C)^\sT D^\sT] = \Tr[D^\sT C (D^\sT C)^\sT] \\
        &= \Tr[((D^\sT C)^\sT)^\sT (D^\sT C)^\sT] = \norm{(D^\sT C)^\sT}^2_F \geq 0\,.
    \end{align*}
    \endgroup
    In the rank-decomposing setting, we use a generalization of \Cref{lem:holders-2-inf-2} to the three-matrix case, along with the definition of $\norm{B}_1 \le C$ having a rank-decomposing proof.
\end{proof}

\subsection{Matrix H\"older's inequality for dyadic norms}

The following is the Schatten norm version of H\"older's inequality reformulated to be a polynomial inequality.
We recover the usual form by upper-bounding the inner product of $A^{p-1}$ with $B$ by the product of the $p$-norm of $B$ with the $q$-norm of $A^{p-1}$, which is equal to $(\tr A^p)^{(p-1)/p}$.

\begin{lemma}\label[Lemma]{lem:matrix-p-q-holders}
Let $A$ and $B$ be symmetric matrix variables.
    Let $p=2^k$ for $k \in \N$. Then
    \[ \proves \left\{ (\tr A^{p-1}B)^{p} \le (\tr A^{p})^{p-1}(\tr B^{p}) \right\}\,. \]
\end{lemma}
\begin{proof}
    We show this inductively.
    The base case where $p=1$ is trivial.
    The inductive hypothesis tells us that for $X$ and $Y$ symmetric matrix variables,
    \[ \proves \left\{ (\tr X^{p/2-1}Y)^{p/2} \le (\tr X^{p/2})^{p/2-1}(\tr Y^{p/2}) \right\}. \]
    Going back to the lemma statement, by Cauchy-Schwarz inequality, $\iprod{A^{p/2},A^{p/2-1}B} \le \norm{A^{p/2}}_F^2\norm{A^{p/2-1}B}_F^2$, which we can write by trace cyclicity as
    \[ \proves \left\{ (\tr A^{p-1}B)^{p} \le (\tr A^{p})^{p/2} (\tr A^{p-2}B^2)^{p/2}\right\}. \]
    Now we can invoke the inductive hypothesis with $X := A^2$ and $Y := B^2$ to find
    \[ \proves \left\{ (\tr A^{p-2}B^2)^{p/2} \le (\tr A^p)^{p/2-1}(\tr B^{p}) \right\}. \]
    This concludes the proof.
\end{proof}

\section{Semi-Definite Encodings: High-Entropy Steps and Cumulants}\label{sec:appendix-sdp}

\subsection{Encoding high-entropy step constraints into the SDP}\label{sec:appendix-sdp-entropy}
We briefly describe how to encode the maximum-entropy constraints into the semi-definite program by introducing the required auxiliary variables and using the definition of hierarchical pseudo-distributions.

\paragraph{Radicals of the polynomial grammar} We ``radicalize'' various elements of the polynomial grammar in many proofs above. To see why this is reasonable, note that, by introducing families of dummy variables $\mathsf{dum}_i \succeq 0$ one may split the (non-negative) fractional power $q = \frac{a}{b}$ into integer powers, which is a standard technique~\cite{barak2012hypercontractivity, lei2017sum}.

\begin{align*}
    &\{\mathsf{dum}_i \succeq 0\}\, , \\
    &\left\{\mathsf{dum}_i^b = \fE_i\left[v_iv_i^\sT\right]\right\}\, , \\
    &\left\{ \opnorm{\mathsf{dum}_i}^b = \opnorm{\mathsf{dum}_i^b} \le \frac{1+o_n(1)}{\delta n}\right\}\, .
\end{align*}

\paragraph{The $\opnorm{\fE_i[v_iv_i^\sT]} \le \frac{1 + o_n(1)}{A(\eps) n}$ constraint} This constraint uses the matrix Sum-of-Squares framework~(\Cref{sec:matrix-sos-system}) to encode the constraint on the operator norm of the conditional pseudo-variance matrix of every iterate variable $v_i$. The constraint translates to the following two constraints,
\[
    \fE_i[v_iv_i^\sT] \preceq \frac{1 + o_n(1)}{A(\eps) n}\Id\, ,\;\; -\fE_i[v_iv_i^\sT] \preceq \frac{1 + o_n(1)}{A(\eps) n}\Id
\]
which is encoded using slack matrices $S_{i,1}$ and $S_{i,2}$ $\in \R^{n \times k}$ as,
\[
    \fE_i[v_iv_i^\sT] = \frac{1 + o_n(1)}{A(\eps) n}\Id - S_{i,1}S_{i,1}^\sT\, ,
\]
and,
\[
    \fE_i[v_iv_i^\sT] = S_{i,2}S_{i,2}^\sT - \frac{1 + o_n(1)}{A(\eps) n}\Id\, .
\]

\paragraph{The ``approximate-LDP'' constraint} We now give the encoding of the constraint that the conditional covariance is close in Frobenius norm to a low-degree polynomial in the Wiener-Hermite basis.
\[ \left\{\norm{\left(\fE_i\left[v_iv_i^\sT\right]^{> d'}\right)^{\ot d}}^2_F \le \frac{o_n(1)}{n^d}\right\} \]
is interpreted as, for all $j\in[i-1]$ and all $d'' \in [d]$,
\[ \left\{\fE_j\norm{\left(\left[\fE_i\left[v_iv_i^\sT\right] -\sum_{\alpha:\;|\alpha| > (d')^{i-j}} \mathscr{W}_j\left[\fE_{j+1}\fE_i\left[v_iv_i^\sT\right]\right](\alpha)\,\He_{\alpha}\left(v_j \mid \fE_j v_jv_j^{\sT}\right)\right]\right)^{\ot d''}}^2_F \le \frac{o_n(1)}{n^{d''}}\right\}, \]
using the definitions in \Cref{sec:prelims-multivar-hermites} and \Cref{sec:encoding-hermite-decomps}.

\subsection{Sum-of-squares proofs for Hermite-Weiner polynomials via cumulants}\label{sec:hermites-and-cums-in-sos}

It is not hard to verify that the cumulant constraints presereve the orthogonality of the Hermites and the (weak) orthogonality of the Hermite-Weiner expansion.
It is quite straightforward to see by the recurrence relation defintion of the Hermites that their expectation is a valid moment in the polynomial grammar, and that it continues to be $0$.

\begin{proposition}[SoS zero-expectation for Hermites]\label[Proposition]{prop:sos-hermite-exp}
  Given the cumulant constraints in~\Cref{eqs:the-program-intro-cum}, it is the case that,
  \[
      \E[\He_n(x_{1},\dots,x_{n})] = \begin{cases}
        1 & \text{if }n = 0\, , \\
        0 & \text{if }n \ge 1\, .
      \end{cases}
  \]
\end{proposition}
\begin{proof}
  The proof is a consequence of~\cite[Theorem 4.1]{terdik2021multivariate} with the choice of $m=1$ and the observation that any closed diagram of this type cannot exist. First, note that $\E[\He_n(x_1,\dots,x_n)]$ is an element of the polynomial grammar under the right choice of filtration. The fact that this is a SoS proof follows as the derivation of~\cite[Theorem 4.1]{terdik2021multivariate} does not invoke the density function of the gaussian, and the moments of a single gaussian can be counted as the number of perfect matchings. Therefore, by the equivalence of cumulants and moments via the Bell polynomials, the proof is a sum-of-squares proof in the variables of the polynomial grammar.
\end{proof}

One can substitute $x_1,\dots,x_n$ to be indices over any multi-set $\alpha$ with $|\alpha| \le d$ and $\alpha_i \in \{v_{i1},\dots,v_{in}\}$ for $i \in [k]$ and evaluate the Hermites conditioned on prior iterates $v_1,\dots,v_{i-1}$ to apply the above lemma to the formal polynomials in the grammar of the grammar of the SDP.

A simple proposition similar to~\Cref{prop:sos-hermite-exp} is that Isserlis' lemma can be invoked in a sum-of-squares proof under the HES constraints.

\begin{proposition}[SoS Isserlis Lemma]\label[Proposition]{prop:sos-isserlis}
  Given the constraints of~\Cref{eqs:the-program-intro-cum}, for any degree sequence $\alpha = (\alpha_1,\dots, \alpha_n)$ with $|\alpha| \le d$ it is the case that,
  \[
      \E\left[\prod_{i \in \alpha}x_i^{\alpha_i}\right] = \sum_{\pi \in P_2\{i_1,\dots,i_k\}, \alpha_{i_j} \ne 0} \prod_{(a, b) \in \pi}\cum{2}{x_a x_b}\, .
  \]
\end{proposition}
\begin{proof}
    The proof follows by using the equivalence between moments and cumulants using the Bell polynomials and then using that defintition to invoke the moment-cumulant formula with the constraints in~\Cref{eqs:the-program-intro-cum}.
    \begingroup
    \allowdisplaybreaks
    \begin{align*}
        \E\left[\prod_{i \in \alpha}x_i^{\alpha_i}\right] &=_{\text{distinct indices}} \E[y_1,\dots,y_{|\alpha|}] \\
        &=_{\text{moment-cumulant formula}} \sum_{\pi \in \cL[|\alpha|]}\prod_{b \in \pi}\cum{|b|}{x_{b_1},\dots,x_{b_{|b|}}} \\
        &=_{\text{\Cref{eqs:the-program-intro-cum}}} \sum_{\pi \in P^2\{1,\dots,|\alpha|\}}\prod_{(j,k) \in \pi}\cum{2}{y_jy_k} \\
        &=_{\text{substitute index map}} \sum_{\pi \in P^2\{1,\dots,|\alpha|\}}\prod_{(j,k) \in \pi} \cum{2}{x_{\alpha(j)}x_{\alpha(k)}}\, .
    \end{align*}
    \endgroup
\end{proof}

\begin{proposition}[SoS orthogonality for Hermites]\label[Proposition]{prop:sos-orthgonality-hermites}
    Given the entropy constraints of the program,
    \[
        \{\text{HES constraints}\} \proves_d \left\langle \He_\alpha(x), \He_\beta(x) \right\rangle = \fE_1[\He_\alpha(x)\He_\beta(x)] = \begin{cases}
          1 & \text{if }\alpha = \beta\, , \\
          0 & \text{ otherwise}\, .
        \end{cases}
    \]
\end{proposition}
\begin{proof}
    The proof will directly invoke the evaluation of the product of the Hermite polynomials as given in~\cite[Proposition 4.1]{terdik2021multivariate} in terms of the underlying covariances. As one can derive from the recurrence relation definition of the Hermites in conjunction with the moment-cumulant formula, that the coefficients of the Hermites depend entirely on the underlying covariance matrix, the expectation of the product depends on the \emph{products} of the covariances. \\
    Set $m = 2$ with $b_1 = (y_{1},\dots,y_{|\alpha|})$ and $b_2 = (y_1,\dots,y_{|\beta|})$ (using distinct indices) in~\cite[Proposition 4.1]{terdik2021multivariate} and note that,
    \begingroup
    \allowdisplaybreaks
    \begin{align*}
        &\langle \He_\alpha(x), \He_\beta(x) \rangle = \E\left[\He_\alpha(x)\He_\beta(x)\right] \\
        &= \E\left[\He_{|\alpha|}(y_1,\dots,y_{|\alpha|})\He_{|\beta|}(y_1,\dots,y_{|\beta|})\right] = \delta_{|\alpha||\beta|}\,. 
    \end{align*}
    \endgroup
\end{proof}

A similar argument as the one above can be provided to demonstrate the weak orthogonality of the Hermite-Weiner polynomials~\cite[Definition 3]{rahman2017wiener}, which is the orthogonal decomposition critically used for bounding the polynomials corresponding to the higher-order derivatives, as well as reasoning about ``closeness'' of moments and pseudo-moments in the rounding procedure.

\begin{fact}[SoS weak-orthogonality for Hermite-Weiner polynomials]\label{fact:sos-weak-orthogonality}
    The HES SoS hierarchy provides the weak-orthogonality of Hermite-Weiner polynomials up to the requisite degree bound. More formally,
    \[
        \{\text{HES constraints}\} \proves_d \fE_1 \He_\alpha(v_i) \He_\beta(v_i) = \begin{cases}
          \sum_{\pi \in S_{|\alpha|}}\prod_{i=1}^{|\alpha|}\left(\fE_i[v_iv_i^\sT| \mid v_1,\dots,v_{i-1}]\right)_{i, \pi(i)} & \text{if }|\alpha| = |\beta|\, , \\
          0 & \text{ otherwise}\, ,
        \end{cases}
    \]
    for every $\alpha, \beta \in [n]$ with $|\alpha|, |\beta| \le d$.
\end{fact}

The proof for this follows for similar reasons as~\Cref{prop:sos-orthgonality-hermites} with the use of ~\cite[Proposition 4.3]{terdik2021multivariate} with the recursive definition of the Hermites as an Appell sequence. Since this is an equality and ultimately all terms become polynomials of the entries of the conditional covariance (and constants), the equality is a SoS proof over the grammar described in~\Cref{sec:sos-poly-basis}.

\end{document}